\documentclass[12pt,english]{amsart}
\usepackage{xcolor}
\usepackage{amsfonts}
\usepackage{amsmath, amsfonts, amsthm, amssymb, amscd}
\usepackage[mathcal]{euscript}
\usepackage{mathrsfs}
\usepackage{ulem}
\usepackage{mathabx,a4wide}
  
\usepackage{dsfont}
\usepackage{extarrows}
\usepackage{perpage} 
\MakePerPage{footnote} 
\makeatletter
\def\blfootnote{\xdef\@thefnmark{}\@footnotetext}
\makeatother

\newtheorem{theorem}{Theorem}[section]
\newtheorem{proposition}[theorem]{Proposition}
\newtheorem{lemma}[theorem]{Lemma}

\newtheorem{definition}[theorem]{Definition}

\numberwithin{equation}{section} 
\newcommand \bel {\be\label}
\newcommand \be {\begin{equation}}
\newcommand \ee {\end{equation}}
\newcommand \newperp  {\underline{\partial}_\perp}
\newcommand \Lcal {\mathcal L}
\newcommand \xb {\overline {x}}
\newcommand \delb {\overline {\del}}
\newcommand \hb{\overline h}

\newcommand \Tb {\overline {T}}
\newcommand \minb{\overline{m}}
\newcommand \Phib{\overline{\Phi}}
\newcommand \Psib{\overline{\Psi}}
\newcommand \Kcal {\mathcal K}
\newcommand \Kical {\mathcal {K}^{\text{int}}}

\newcommand \Hcal {\mathcal H}
\newcommand \Boxt {\widetilde {\Box}}
\newcommand \del \partial
\newcommand \delu {\underline{\del}}
\newcommand \Tu {\underline{T}}
\newcommand \Hu {\underline{H}}
\newcommand \hu {\underline{h}}
\newcommand {\Gammau}{\underline{\Gamma}}
\newcommand {\Thetau}{\underline{\Theta}}
\newcommand {\thetau}{\underline{\theta}}

\newcommand {\Pu}{\underline{P}}
\newcommand {\gu}{\underline{g}}
\newcommand {\minu}{\underline{m}}
\newcommand \RR{\mathbb{R}}
\newcommand {\vep}{\varepsilon}

\newcommand {\gb}{\overline{g}}
\newcommand {\Mb}{\overline {M}}
\newcommand {\Rb}{\overline{R}}
\newcommand {\nablab}{\overline{\nabla}}

\let\oldmarginpar\marginpar
\renewcommand\marginpar[1]{\- \oldmarginpar[\raggedleft\footnotesize #1]%
{\raggedright\footnotesize #1}}
\newcommand \eps {\epsilon}

\newcommand \sbar {{\overline{s}}}

\newcommand \gSch {{g_S}}

\begin{document}

\

\title[Stability of Minkowski space for self-gravitating massive fields]
{\large The global nonlinear stability 
\\
\vskip.4cm 
of 
Minkowski space
\\
\vskip.8cm 
for self-gravitating massive fields}

\blfootnote{PLF: Laboratoire Jacques-Louis Lions \& Centre National de la Recherche Scientifique, 
Universit\'e Pierre et Marie Curie, 
4 Place Jussieu, 75252 Paris, France. Email: {contact@philippelefloch.org}
\newline 
YM: 
School of Mathematics and Statistics, Xi'an Jiaotong University, Xi'an, 710049 Shaanxi, People's Republic of China.
Email: {yuemath@mail.xjtu.edu.cn}
}


\author[Philippe G. L{\smaller e}FLOCH and Yue MA]{\large \vskip.8cm 
P{\small hilippe} G. L{\small e}F{\small loch} {\small and} Y{\small ue} M{\small a}}

\date{} 

\vfill 

\maketitle
               
\

\vfill

\newpage

\

\vskip1.cm 

\centerline{\scshape \large Preface.} 

\

The theory presented in this Monograph establishes the first mathematically rigorous result on the {\sl global nonlinear stability of self-gravitating matter} under small perturbations.
Indeed, it allows us to exclude the existence of dynamically unstable, self-gravitating massive fields and, therefore, solves a long-standing open problem in General Relativity. 

We establish that Minkowski spacetime is nonlinearly stable  in presence of a massive scalar field under suitable smallness conditions (for, otherwise, black holes might form). We formulate the initial value problem for the Einstein-massive scalar field equations, when the initial slice is a perturbation of an asymptotically flat, spacelike hypersurface in Minkowski space, and we prove that this perturbation disperses in future timelike directions so that the associated Cauchy development is future geodesically complete. 

Our method of proof which we refer to as the {\sl Hyperboloidal Foliation Method}, extends the standard `vector field method' developed for massless fields and, importantly, does not use the scaling vector field of Minkowski space. We construct a foliation (of the interior of a light cone) by spacelike and asymptotically hyperboloidal hypersurfaces and we rely on a decomposition of the Einstein equations expressed in {\sl wave gauge} and in a {\sl semi-hyperboloidal frame,} in a sense defined in this Monograph. We focus here on the problem of the evolution of a spatially compact matter field, and we consider initial data coinciding, in a neighborhood of spacelike infinity, with a spacelike slice of Schwarzschild spacetime. We express the Einstein equations as a system of coupled nonlinear wave-Klein-Gordon equations (with differential constraints) posed on a curved space (whose metric is one of the unknowns). 

The main challenge is to establish a global-in-time existence theory for {\sl coupled wave-Klein-Gordon systems} in Sobolev-type spaces defined from the translations and the boosts of Minkowski spacetime, only. To this end, we rely on the following novel and robust techniques: new commutator estimates for hyperboloidal frames, sharp decay estimates for wave and Klein-Gordon equations, Sobolev and Hardy inequalities along the hyperboloidal foliation, quasi-null hyperboloidal structure of the Einstein equations, as well as integration arguments along characteristics and radial rays. Our proof also relies on an iterative procedure involving the components of the metric and the Klein-Gordon field, and on a hierarchy of low- and high-order energy estimates, which distinguishes between the metric components and between 
several levels of time dependency and regularity for the metric coefficients and the massive field.

\

\hfill Philippe G. LeFloch (Paris) and Yue Ma (Xi'an)

\

\newpage

\vskip4.cm 

\setcounter{tocdepth}{5}
\tableofcontents
 
\newpage

\section{Introduction}  

\subsection{The nonlinear stability problem for the Einstein-Klein-Gordon system}

We consider Einstein's field equations 
of General Relativity for self-gravitating massive scalar fields and formulate the initial value problem when the initial data set is a perturbation of an asymptotically flat, spacelike hypersurface in Minkowski spacetime. We then establish the existence of an Einstein development associated with this initial data set, which is proven to be an asymptotically flat and future geodesically complete spacetime.
Recall that, in the case of vacuum spacetimes or massless scalar fields, such a nonlinear stability theory for Minkowski spacetime was first established by Christodoulou and Klainerman in their breakthrough work \cite{CK}, which was later revisited by Lindblad and Rodnianski \cite{LR2} via an alternative approach. Partial results on the global existence problem for the Einstein equations was also obtained earlier by Friedrich \cite{Friedrich81,Friedrich83}. 

Let us emphasize that the {\sl vacuum} Einstein equations are currently under particularly active development: this is illustrated by the recent contributions by Christo\-doulou \cite{Christodoulou} and Klainerman and Rodnianski \cite{KR} (on the formation of trapped surfaces) and by Klainerman, Rodnianski and Szeftel \cite{KRS} (on the $L^2$ curvature theorem). The Einstein equations coupled with massless fields such as the Maxwell field were also extensively studied; see for instance Bieri and Zipser \cite{BieriZipser} and Speck \cite{Speck}; existence under slow decay conditions was established by Bieri  \cite{BieriZipser}. 
 
The present Monograph offers a new method for the global analysis of the Einstein equations, which we refer to as the {\sl Hyperboloidal Foliation Method} 
 and allows us to investigate the global dynamics of massive fields and, especially, the {\sl coupling between wave and Klein-Gordon equations.} This method was first outlined in \cite{PLF-YM-book,PLF-YM-one}, together with references to earlier works, especially by Friedrich \cite{Friedrich81,Friedrich83}, Klainerman \cite{Klainerman85}, and H\"ormander \cite{Hormander}.
 We hope that the present contribution will open a further direction of research concerning {\sl matter spacetimes,} which need not be not Ricci-flat and may contain {\sl massive fields}. In this direction, we refer to LeFloch et al.  \cite{BLF,BuLF,GLF,LeFloch-lectures,LFR} for existence results on weakly regular matter spacetimes. 

The {\sl nonlinear stability problem for self-gravitating massive fields,} solved in the present Monograph\footnote{We present here our method for a restricted class of initial data, while more general data as well as the theory of $f(R)$--modified gravity are treated in \cite{PLF-YM-three}.}, was a long-standing open problem for
the past twenty five years since the publication of Christodoulou-Klainerman's book \cite{CK}. In the physics literature, blow-up mechanisms were proposed which suggest possible instabilities for self-gravitating massive fields. While the most recent numerical investigations \cite{OCP} gave some confidence that Minkowski spacetime should be nonlinearly stable, the present work provides the {\sl first mathematically rigorous proof that dynamically unstable solutions to the Einstein equations do not exist in presence of massive fields} (under suitable smallness conditions specified below). On the other hand, nonlinear stability would not hold when the mass is sufficiently large, since trapped surfaces and presumably black holes form from (large) perturbations of Minkowski spacetime \cite{Christodoulou}. 

Mathematically, the problem under consideration can be formulated (in the so-called wave gauge, see below) as a quasilinear system of {\sl coupled nonlinear wave-Klein-Gordon equations,} supplemented with differential constraints and posed on a curved spacetime. The spacetime (Lorentzian) metric together with the scalar field defined on this spacetime are the unknowns of the Einstein-matter system. The Hyperboloidal Foliation Method introduced in this Monograph leads us to a {\sl global-in-time theory} for this wave-Klein-Gordon system when initial data are provided on a spacelike hypersurface. 
Our proof is based on a substantial modification of the so-called vector field method, which have been applied to massless problems, only. Importantly, we do not use the scaling vector field of Minkowski spacetime, which is required to be able to handle Klein-Gordon equations.  

In order to simplify the presentation of the method, in this Monograph we are interested in spatially compact matter fields and, therefore, we assume that the initial data coincide, in a neighborhood of spacelike infinity, with an asymptotically flat spacelike slice of Schwarzschild spacetime in wave coordinates. Our proof relies on several novel contributions: sharp time-decay estimates for wave equations and Klein-Gordon equations on a curved spacetime, Sobolev and Hardy's inequalities on hyperboloids, quasi-null hyperboloidal structure of the Einstein equations and estimates based on integration along characteristics and radial rays. We also distinguish between low- and high-order energies for the metric coefficients and the massive field. 

We refer to \cite{PLF-YM-book,PLF-YM-CRAS,PLF-YM-one} for earlier work by the authors and to the companion work \cite{PLF-YM-three} for an extension to more general data and to the theory of modified gravity. We focus on $(3+1)$-dimensional problems since this is the dimension of main interest. As already mentioned, in the context of the Einstein equations, hyperboloidal foliations were introduced first by Friedrich \cite{Friedrich81,Friedrich83}. Of course, hyperboloidal foliations can be introduced in any number of dimensions, and should also lead to interesting results (see \cite{Ma} in $(2 +1)$ dimensions), but we do not pursue this here since the Einstein equations have rather different properties in these other dimensions. 

For a different approach to the nonlinear stability of massive fields, we refer the reader to an ongoing research project by Q. Wang (outlined in ArXiv:1607.01466) which is aimed at generalizing Christodoulou-Klainerman's geometric method. 
An important recent development is provided by Fajman, Joudioux, and Smulevici \cite{FJS,FJS2}, who recently introduced
 a new vector field method based on a hyperboloidal foliation and aimed at dealing with global existence problems for massive kinetic equations; for this technique, we also refer to Smulevici \cite{Smulevici}. Hyperboloidal foliations are also useful to analyze the blow-up of solutions for, for instance, focusing wave equations, as investigated by Burtscher and Donninger \cite{BD}. 

Furthermore, we also recall that nonlinear wave equations of Klein-Gordon-type posed on possibly curved spacetimes have been the subject of extensive research in the past two decades, and we will not try to review this vast literature and we refer the interested reader to, for instance, Bachelot \cite{Bachelot88,Bachelot94}, 
Delort et al. \cite{Delort01,Delort04}, 
Katayama \cite{Katayama12a,Katayama12b}, and 
Shatah \cite{Shatah82,Shatah85}, 
as well as Germain \cite{G} and Ionescu and Pasauder \cite{IP}; 
see also \cite{Hormander,HoshigaKubo,Strauss} and the references cited therein.
Importantly, the use of hyperboloidal foliations leads to robust and efficient numerical methods, as demonstrated by a variety of approaches by Ansorg and Macedo \cite{AnsorgMacedo}, Frauendiener \cite{Frauendiener,FrauendienerH}, Hilditch et al. \cite{Hilditch,VHH}, Moncrief and Rinne \cite{MoncriefRinne}, Rinne \cite{Rinne}, and Zenginoglu \cite{Zenginoglu,Zenginoglu2}. 


\subsection{Statement of the main result}

We thus consider the {\sl Einstein equations} for an unknown spacetime $(M,g)$, that is, 
\be
\label{eq main geo 0}
G_{\alpha \beta} := R_{\alpha \beta} - {R \over 2} g_{\alpha \beta} = 8 \pi \, T_{\alpha \beta},
\ee
where $R_{\alpha \beta}$ denotes the Ricci curvature of $(M,g)$, $R=g^{\alpha \beta} R_{\alpha \beta}$ its scalar curvature, and $G_{\alpha \beta}$ is referred to as the Einstein tensor.
Our main unknown in \eqref{eq main geo 0} is a Lorentzian metric $g_{\alpha \beta}$ defined on a topological $4$-manifold $M$. By convention, Greek indices $\alpha, \beta, \ldots$ take values $0,1,2,3$, while Latin indices $i,j, \ldots$ takes values $1,2,3$ (as, for instance, in \eqref{eq constraint} below). In this work, we  are interested in non-vacuum spacetimes when the matter content is described by a massive scalar field denoted by $\phi: M\to \RR$ with potential $V=V(\phi)$. The stress-energy tensor of such a field reads 
\bel{eq tensor T}
T_{\alpha \beta} := \nabla_\alpha \phi \nabla_\beta \phi 
- \Big( {1 \over 2} \nabla_\gamma \phi \nabla^\gamma \phi + V(\phi) \Big) g_{\alpha \beta}.
\ee
Recall that from the contracted Bianchi identities $\nabla^\alpha G_{\alpha \beta} =0$, we can derive an evolution equation for the scalar field and, in turn, formulate the Einstein--massive field system as the system of quasilinear partial differential equations (in any choice of coordinates at this stage) 
\begin{subequations}
\label{eq main geo}
\be
\label{eq main geo a}
\aligned
R_{\alpha \beta}  = 8\pi\big(\nabla_\alpha \phi\nabla_{\beta} \phi + V(\phi) \, g_{\alpha \beta} \big),
\endaligned
\ee
\bel{eq main geo b}
\aligned
\Box_g\phi - V'(\phi) = 0.
\endaligned
\ee
\end{subequations}
Without loss of generality, throughout we assume that the potential is quadratic in $\phi$, i.e.
\be
V(\phi) = \frac{c^2}{2} \phi^2,
\ee
where $c^2>0$ is referred to as the mass density of the scalar field. The equation \eqref{eq main geo b} is nothing but a Klein-Gordon equation posed on an (unknown) curved spacetime.

The Cauchy problem for the Einstein equations can be formulated as follows; cf., for instance, Choquet-Bruhat's textbook \cite{CB}. First of all, let us recall that an {\sl initial data set} for the Einstein equations consists of a Riemannian $3$-manifold $(\Mb, \gb)$, a symmetric $2$-tensor field $K$ defined on $\Mb$, and two scalar fields $\phi_0$ and $\phi_1$ also defined on $\Mb$. 
A {\sl Cauchy development of the initial data set} $(\Mb, \gb, K, \phi_0, \phi_1)$, by definition, is a $(3+1)$-dimensional Lorentzian manifold $(M,g)$ satisfying the following two properties:
\begin{itemize}

\item There exists an embedding $i: \Mb\to M$ such that the (pull-back) induced metric $i^*(g) = \gb$ coincides with the prescribed metric $\gb$, while  the second fundamental form of $i(\Mb) \subset M$ coincides with the prescribed $2$-tensor $K$. In addition, by denoting by $n$ the (future-oriented) unit normal to $i(\Mb)$, the restriction (to the hypersurface $i(\Mb)$) of the field $\phi$ and its Lie derivative $\Lcal_n \phi$ coincides with the data $\phi_0$ and $\phi_1$ respectively.

\item The manifold $(M,g)$ satisfies the Einstein equations \eqref{eq main geo a} and, consequently, the scalar field $\phi$ satisfies the Klein-Gordon equation \eqref{eq main geo b}.
\end{itemize}

\noindent As is well-known, in order to fulfill the equations \eqref{eq main geo a}, the initial data set cannot be arbitrary but must satisfy Einstein's constraint equations:
\bel{eq constraint}
\aligned
\Rb - K_{ij} \, K^{ij} + (K_i^i)^2   &= 8\pi T_{00},
\qquad 
\nablab^i K_{ij} - \nablab_j K_l^l = 8 \pi T_{0j},
\endaligned
\ee 
where $\Rb$ and $\nablab$ are the scalar curvature and Levi-Civita connection of the manifold $(\Mb, \gb)$, respectively,
while the mass-energy density $T_{00}$ and the momentum vector $T_{0i}$ are determined from the data $\phi_0, \phi_1$ (in view of the expression \eqref{eq tensor T} of the stress-energy tensor).

Our main result established in the present Monograph can be stated as follows.

\begin{theorem}[Nonlinear stability of Minkowski spacetime for self-gravitating massive fields. Geo\-metric 
 version]
\label{MAIN-TH}

Consider the Einstein-massive field system \eqref{eq main geo} when the initial data set $(\Mb, \gb, K, \phi_0, \phi_1)$ satisfies Einstein's constraint equations \eqref{eq constraint} and is close to an asymptotically flat slice of the (vacuum) Minkowski spacetime and, more precisely, coincides in a neighborhood of spacelike infinity with a spacelike slice of a Schwarzschild spacetime with sufficiently small ADM mass. The corresponding initial value problem admits a globally hyperbolic Cauchy development, which represents an asymptotically flat and future geodesically complete spacetime.
\end{theorem}

We observe that the existence of initial data sets satisfying the conditions above was established by Corvino and Schoen \cite{CorvinoSchoen}; see also Chrusciel and Delay \cite{CD} and the recent review \cite{Chrusciel}. Although the main focus therein is on vacuum spacetimes, it is straightforward to include matter fields by observing\footnote{The authors thank J. Corvino for pointing this out to them.}  that classical existence theorems \cite{CB} provide the existence of non-trivial initial data in the ``interior region'' and that Corvino-Schoen's glueing construction is purely local in space.  

We are going to formulate the Einstein-massive field system as coupled partial differential equations. This is achieved by introducing {\sl wave coordinates} denoted by $x^\alpha$, satisfying the wave equation
$\Box_g x^\alpha = 0$ ($\alpha = 0, \ldots, 3$). 
From \eqref{eq main geo}, we will see that, in wave coordinates, the Ricci curvature operator reduces to the wave operator on the metric coefficients and, in fact, (cf.~Lemma \ref{lem Ricci}, below)
\begin{subequations} 
\label{eq main PDE}
\bel{eq main PDE a}
\Boxt_g h_{\alpha \beta} = F_{\alpha \beta}(h;\del h, \del h) - 16\pi \del_{\alpha} \phi\del_{\beta} \phi - 16\pi V(\phi)g_{\alpha \beta},
\ee 
\bel{eq main PDE b}
\Boxt_g\phi - V'(\phi) = 0,
\ee
\end{subequations}
where $\Boxt_g := g^{\alpha \beta} \del_{\alpha} \del_{\beta}$ is referred to as the {\sl reduced wave operator,} and $h_{\alpha \beta} := g_{\alpha \beta} - m_{\alpha \beta}$ denotes the curved part of the unknown metric. The nonlinear terms $F_{\alpha \beta}(h;\del h, \del h)$ are quadratic in first-order derivatives of the metric. Of course, that the system \eqref{eq main PDE} must be supplemented with Einstein's constraints \eqref{eq constraint} as well as the wave gauge conditions $\Box_g x^\alpha = 0$, which both are first-order differential constraints on the metric. 

In order to establish a global-in-time existence theory for the above system, several major challenges are overcome in the present work:

\begin{itemize}

\item Most importantly, we cannot use the scaling vector field $S := r\del_r  +  t\del_t$, since the Klein-Gordon equation is not kept conformally invariant by this vector field.

\item In addition to null terms which are standard in the theory of quasilinear wave equations, in the nonlinearity $F_{\alpha \beta}(h;\del h, \del h)$ we must also handle {\sl quasi-null terms,} as we call them, which will be controlled by relying on the wave gauge condition.

\item The structure of the nonlinearities in the Einstein equations must be carefully studied in order to exclude instabilities that may be induced by the {\sl massive scalar field.}  

\end{itemize}
In addition to the refined estimates on the commutators for hyperboloidal frames\footnote{A fundamental observation  in \cite{PLF-YM-book} is that commutators of hyperboloidal frames and Lorentz boosts enjoy good properties, which had not been observed in earlier works on the subject. In the notation presented below, this is especially true of $[L_a, (s/t)\del_t]$ and $[L_a,\delb_b]$.}
 and  
the sharp $L^\infty$-$L^\infty$ estimates for wave equations and Klein-Gordon equations already introduced by the authors in the first part \cite{PLF-YM-one}, we need the following new arguments of proof (further discussed below):
\begin{itemize}

\item Formulation of the Einstein equations in wave gauge in the semi-hyperboloidal frame.

\item Energy estimates at arbitrary order on a background Schwarzschild space in wave gauge. 

\item Refined estimates for nonlinear wave equations, that are established by integration along characteristics or radial rays.

\item Estimates of quasi-null terms in wave gauge, for which we rely on, both, the tensorial structure of the Einstein equations and the wave gauge condition.

\item New weighted Hardy inequality along the hyperboloidal foliation.

 \end{itemize}
A precise outline of the content of this Monograph will be given at the end of the following section, after introducing further notation. 
 
\ 

\centerline{\scshape Acknowledgements}

\vskip.3cm

The results established in this Monograph were presented in July 2015 at the Mini-Symposium on `Mathematical Problems of General Relativity'' (EquaDiff Conference), organized by S. Klainerman and J. Szeftel at the University of Lyon (France). They also led to a graduate course \cite{LeFloch-lectures} taught by the first author at the Institute Henri Poincar\'e (Paris) during the quarter program ``Mathematical General Relativity: the 100th anniversary'' organized by L. Andersson, S. Klainerman, and P.G. LeFloch in the Fall 2015. The first author (PLF) gratefully acknowledges support from the Simons Center for Geometry and Physics, Stony Brook University, during the one-month Program ``Mathematical Method in General Relativity' hold  in January 2015  and organized by M. Anderson, S. Klainerman, P.G. LeFloch, and J. Speck.


\section{Overview of the Hyperboloidal Foliation Method} 

\subsection{The semi-hyperboloidal frame and the hyperboloidal frame}

Consider the $(3+1)$-dimensional Minkowski spacetime with signature $(-,+,+,+)$. In Cartesian coordinates, we write $(t, x) = (x^0, x^1, x^2, x^3)$ with $r^2 := |x|^2 = (x^1)^2 + (x^2)^2 + (x^3)^2$, and we use the partial derivative fields $\del_0$ and $\del_a$, as well as the Lorentz boosts
$L_a := x^a \del_t + t\del_a$ and their ``normalized'' version  
${L_a \over t} = {x^a \over t} \del_t + \del_a$. 
We primarily deal with functions defined in the interior of the future light cone from the point $(1,0,0,0)$, denoted by 
$$
\Kcal := \{(t, x) \, / \, r<t-1\}.
$$ 
To foliate this domain, we consider the hyperboloidal hypersurfaces with hyperbolic radius $s>0$, defined by 
$$
\Hcal_s := \big\{(t, x) \, / \, t^2-r^2 = s^2; \quad t>0 \big\}
$$
 with $s \geq 1$. 
In particular, we can introduce the following subset of $\Kcal$ limited by two hyperboloids (with $s_0 < s_1$) 
$$
\Kcal_{[s_0,s_1]} := \big\{ (t, x) \, / \, s_0^2 \leq t^2-r^2 \leq s_1^2; \quad r<t-1 \big\}
$$
whose boundary contains a section of the light cone $\Kcal$. 

With these notations, the {\sl semi-hyperboloidal frame} is, by definition, 
\be
\delu_0 := \del_t,
\qquad \delu_a:= \frac{x^a}{t} \del_t + \del_a, \qquad a=1,2,3. 
\ee
Note that the three vectors $\delu_a$ generate the tangent space to the hyperboloids. For some of our statements (for instance in Proposition~\ref{Linfini KG}), It will be convenient to also use the  vector field $\delu_{\perp} : = \del_t + \frac{x^a}{t} \del_a$,
which is orthogonal to the hyperboloids (and is proportional to the scaling vector field). 

Furthermore, given a multi-index $I = (\alpha_n, \alpha_{n-1}, \dots, \alpha_1)$ with $\alpha_i\in\{0,1,2,3\}$, we use the notation
$\del^I := \del_{\alpha_n} \del_{\alpha_{n-1}} \ldots \del_{\alpha_1}$
for the product of $n$ partial derivatives and, similarly, for $J = (a_n,a_{n-1}, \dots, a_1)$ with $a_i\in \{1,2,3\}$ we write 
$L^J = L_{a_n}L_{a_{n-1}} \ldots L_{a_1}$ for the product of $n$ Lorentz boosts.

Associated with the semi-hyperboloidal frame, one has the dual frame $\thetau^0:= dt - \frac{x^a}{t} \, dx^a$, $\thetau^a: = dx^a$. The (dual) semi-hyperboloidal frame and the (dual) natural Cartesian frame are related via
$$
\aligned
&\delu_{\alpha} = \Phi_{\alpha}^{\alpha'} \del_{\alpha'},
\quad \del_{\alpha}
= \Psi_{\alpha}^{\alpha'} \delu_{\alpha'},
\qquad 
\thetau^{\alpha} = \Psi_{\alpha'}^{\alpha} \, dx^{\alpha'},
\quad
dx^{\alpha} = \Phi^{\alpha}_{\beta'} \thetau^{\alpha'}, 
\endaligned
$$
in which the transition matrix $\left(\Phi_{\alpha}^{\beta} \right)$ and its inverse $\left(\Psi_{\alpha}^{\beta} \right)$ are 
$$
\big(\Phi_{\alpha}^{\beta} \big)
=
\left(
\aligned
&1 &&0 &&&0 &&&&0
\\
&x^1/t &&1 &&&0 &&&&0
\\
&x^2/t &&0 &&&1 &&&&0
\\
&x^3/t &&0 &&&0 &&&&1
\endaligned
\right),
\qquad
\qquad
\big(\Psi_\alpha^{\beta} \big) 
=
\left(
\aligned
&1 &&0 &&&0 &&&&0
\\
-&x^1/t &&1 &&&0 &&&&0
\\
-&x^2/t &&0 &&&1 &&&&0
\\
-&x^3/t &&0 &&&0 &&&&1
\endaligned
\right).
$$
With this notation, for any two-tensor $T_{\alpha \beta} \, dx^\alpha \otimes dx^{\beta} = \Tu_{\alpha \beta} \thetau^{\alpha} \otimes \thetau^{\beta}$, we can write 
$\Tu_{\alpha \beta} = T_{\alpha'\beta'} \Phi_\alpha^{\alpha'} \Phi_{\beta}^{\beta'}$ and $T_{\alpha \beta} = \Tu_{\alpha'\beta'} \Psi_\alpha^{\alpha'} \Psi_{\beta}^{\beta'}$. We also have the similar decompositions
$\Tu^{\alpha \beta} = T^{\alpha'\beta'} \Phi^\alpha_{\alpha'} \Phi^{\beta}_{\beta'}$ and 
$T^{\alpha \beta} = \Tu^{\alpha'\beta'} \Psi^\alpha_{\alpha'} \Psi^{\beta}_{\beta'}$.

\begin{lemma}[Decomposition of the wave operator]
For every smooth function $u$ defined in the future light-cone $\Kcal$, the flat wave operator in the semi-hyperboloidal frame reads
\be
\label{eq 1 wave decompo}
\Box u = - \frac{s^2}{t^2} \del_t\del_t u  - \frac{3}{t} \del_t u
- \frac{x^a}{t} \big( \del_t  \delu_a u + \delu_a \del_t u \big)
+ \sum_a \delu_a \delu_a u.
\ee
\end{lemma}

Within the  future cone $\Kcal$, we introduce the change of variables
$\xb^0 = s: = \sqrt{t^2 - r^2}$ and $\xb^a = x^a$ 
and the associated frame which we refer to as the {\sl hyperboloidal frame} :
\bel{Hyper frame}
\aligned
& \delb_0 := \del_s = \frac{s}{t} \del_t = \frac{\xb^0}{t} \del_t = \frac{\sqrt{t^2-r^2}}{t} \del_t,
\qquad 
 \delb_a := \del_{\xb^a} = \frac{\xb^a}{t} \del_t + \del_a = \frac{x^a}{t} \del_t + \del_a.
\endaligned
\ee
The transition matrices between the hyperboloidal frame and the Cartesian frame read
$$
\big(\Phib^{\beta}_\alpha \big) 
= \left(
\begin{array}{cccc}
s/t &0 &0 &0
\\
x^1/t &1 &0 &0
\\
x^2/t &0 &1 &0
\\
x^3/t &0 &0 &1
\end{array}
\right), 
\qquad
\big(\Psib^{\ \beta}_\alpha \big)
:= \big(\Phib^{\beta}_\alpha \big)^{-1}
 = \left(
\begin{array}{cccc}
t/s &0 &0 &0
\\
-x^1/s &1 &0 &0
\\
-x^2/s &0 &1 &0
\\
-x^3/s &0 &0 &1
\end{array}
\right),
$$
so that
$
\delb_\alpha = \Phib^{\beta}_\alpha \del_{\beta}
$
and
$
\del_\alpha = \Psib^{\beta}_\alpha \delb_{\beta}.
$
Observe also that the dual hyperboloidal frame is
$d\xb^0 := ds = \frac{t}{s} \, dt - \frac{x^a}{s} \, dx^a$ and $d\xb^a := dx^a$, while the Minkowski metric in the hyperboloidal frame reads 
$$
\big( 
\minb^{\alpha \beta} \big) = \left(
\begin{array}{cccc}
-1 & \quad - x^1/s & \quad - x^2/s & \quad - x^3/s
\\
-x^1/s &1 &0 &0
\\
-x^2/s &0 &1 &0
\\
-x^3/s &0 &0 &1
\end{array}
\right).
$$

A given tensor can be expressed in any of the above three frames:
the standard frame $\{\del_\alpha \}$,
the semi-hyperboloidal frame $\{\delu_\alpha \}$,
and the hyperboloidal frame $\{\delb_\alpha \}$.
We use Roman letters, underlined Roman letters and overlined Roman letters for the corresponding components of a tensor expressed in different frame. For example, $T^{\alpha \beta} \del_\alpha \otimes \del_{\beta}$ also reads
$
T^{\alpha \beta} \del_\alpha \otimes\del_{\beta} = \Tu^{\alpha \beta} \delu_\alpha \otimes \delu_{\beta} = \Tb^{\alpha \beta} \delb_\alpha \otimes\delb_{\beta},
$
where $\Tb^{\alpha \beta} = \Psib_{\alpha'}^{\alpha} \Psib_{\beta'}^{\beta}T^{\alpha'\beta'}$
and, moreover, by setting $M :=\max_{\alpha \beta}|T^{\alpha \beta}|$, in the hyperboloidal frame we have the uniform bounds\footnote{Here and in the rest of this paper, the notation $A \lesssim B$ is used when $A \leq C B$ and $C$ is already known to be bounded (at the stage of the analysis).}  
$ (s/t)^2 \, |\Tb^{00}| + (s/t) \, |\Tb^{a0}| + |\Tb^{ab}| \lesssim M$.


\subsection{Spacetime foliation and initial data set} \label{subsec foliation}

We now discuss the construction of the initial data by following the notation in \cite[Sections VI.2 and VI.3]{CB}.  We are interested in a time-oriented spacetime $(M,g)$ that is endowed with a Lorentzian metric $g$ with signature $(-, +, +, +)$ and admits a global foliation by spacelike hypersurfaces $M_t \simeq \{t\} \times \RR^3$. 
The foliation is determined by a time function $t: M \to [0, +\infty)$. 
We introduce local coordinates adapted to the above product structure, that is, $(x^{\alpha}) = (x^0=t, x^i)$, and we choose the basis of vectors $(\del_i)$ as the `natural frame' of each slice $M_t$, and this also defines the `natural frame' $(\del_t, \del_i)$ on the spacetime $M$. By definition, the `Cauchy adapted frame' is $e_i = \del_i$ and $e_0 = \del_t - \beta^i\del_i$, where $\beta = \beta^i\del_i$ is a time-dependent field, tangent to $M_t$ and is called the {\sl shift vector,} and we impose the restriction that $e_0$ is orthogonal to each hypersurface $M_t$.
The dual frame $(\theta^{\alpha})$ of the Cauchy adapted frame $(e_{\alpha})$, by definition, is
$\theta^0 := dt$ and $\theta^i := dx^i + \beta^i dt$ and the spacetime metric reads
\be
g = -N^2 \theta^0\theta^0 + g_{ij} \theta^i\theta^j,
\ee
where the function $N>0$ is referred to as the {\sl lapse function} of the foliation.

We denote by $\gb=\gb_t$ the induced Riemannian metric associated with the slices $M_t$ and by $\nablab$ the Levi-Civita connection of $\gb$.
We also introduce the {\sl second fundamental form} $K = K_t$ defined by
$$
K(X,Y) := - g(\nabla_X n, Y)
$$
for all vectors $X,Y$ tangent to the slices $M_t$, where $n$ denotes the future-oriented, unit normal to the slices. In the Cauchy adapted frame, it reads
$$
K_{ij} = - \frac{1}{2N} \Big(\langle e_0,g_{ij} \rangle - g_{lj} \del_i\beta^l - g_{il} \del_j\beta^l \Big).
$$ 
Here, we use the notation $\langle e_0,g_{ij} \rangle$ for the action of the vector field $e_0$ on the function $g_{ij}$. Next, we define the {\sl time-operator} $D_0$
acting on a two-tensor defined on the slice $M_t$ by
$
D_0 T_{ij} = \langle e_0,T_{ij} \rangle - T_{lj} \del_i\beta^l - T_{il} \del_j\beta^l$, 
which is again a two-tensor on $M_t$. With this notation, we have
$$
K = - \frac{1}{2N}D_0 \gb.
$$

In order to express the field equations \eqref{eq main geo} as a system of partial differential equations (PDE) in wave coordinates, we need first to turn the geometric initial data set $(\Mb, \gb,K, \phi_0, \phi_1)$ into a ``PDE initial data set''.
Since the equations are second-order, we need to know the data
$g_{\alpha \beta}|_{\{t=2\}} = g_{0, \alpha \beta}$,
$\del_t g_{\alpha \beta} |_{\{t=2\}} = g_{1, \alpha \beta}$,
$\phi|_{\{t=2\}} = \phi_0$,
$\del_t \phi|_{\{t=2\}} = \phi_1$,
that is, the metric and the scalar field and their time derivative evaluated on the initial hypersurface $\{t=2\}$.
We claim that these data can be precisely determined from the prescribed geometric data $(\gb, K, \phi_0, \phi_1)$, as follows. The PDE initial data satisfy:
\begin{itemize}

\item 4 Gauss-Codazzi equations which form the system of Einstein's constraints, and  
 
\item 4 equations deduced from the (restriction of the) wave gauge condition.
  
\end{itemize}
For the PDE initial data we have to determine $22$ components, and the geometric initial data provide us with $(\gb_{ab},K_{ab}, \phi_0, \phi_1)$, that is, $14$ components in total. The remaining degrees of freedom are exactly determined by the above $8$ equations.
The well-posedness of the system composed by the above $8$ equations is a trivial property. In this work, we are concerned with the evolution part of the Einstein equations and our discussion is naturally based directly on the PDE initial data set.

The initial data sets considered in the present article are taken to be ``near'' initial data sets generating the Minkowski metric (i.e. without  matter field). More precisely, we consider initial data sets which coincide, outside a spatially compact set $\{|x| \leq 1\}$, with an asymptotically flat, spacelike hypersurface in a Schwarzschild spacetime with sufficiently small ADM mass. The following observation is in order. The main challenge overcome by the hyperboloidal foliation method applied to \eqref{eq main PDE} concerns the part of the solution supported in the  region $\Kcal_{[2,+\infty)}$ or, more precisely, the global evolution of initial data posed on an asymptotically hyperbolic hypersurface. (See \cite{PLF-YM-three} for further details.)
To guarantee this, the initial data posed on the hypersurface $\{t=2\}$ should have its support contained in the unit ball $\{r < 1\}$. Of course, in view of the positive mass theorem (associated with the constraint equation \eqref{eq constraint}), admissible non-trivial initial data must have a non-trivial tail at spatial infinity, that is,
\be
m_S := \lim_{r\rightarrow +\infty}\int_{\Sigma_r}\big(\del_jg_{ij} - \del_ig_{jj}\big)n^i d\Sigma,
\ee 
where $n$ is the outward unit norm to the sphere $\Sigma_r$ with radius $r$.
Therefore, an initial data (unless it identically vanishes) cannot be supported in a compact region.

To bypass this difficulty, we make the following observation: first, the Schwarzschild spacetime provides us with an exact solution to \eqref{eq main geo}, that is, the equations \eqref{eq main PDE} (when expressed with wave coordinates). So, we assume that our initial data $g_0$ and $g_1$ coincide with the restriction of the Schwarzschild metric and its time derivative, respectively (again in wave coordinates) on the initial hypersurface $\{t=2\}$ outside the unit ball $\{r<1\}$. 
Outside the region $\Kcal_{[2,+\infty)}$, we prove that the solution coincides with Schwarzschild spacetime and the global existence problem can be posed in the region $\Kcal_{[2,+\infty)}$.

We can also formulate the Cauchy problem directly with initial data posed on a hyperboloidal hypersurface. This appears to be, both, geometrically and physically natural. As we demonstrated earlier in \cite{PLF-YM-book}, the analysis of nonlinear wave equations is also more natural in such a setup and may lead us to {\sl uniform bounds} for the energy of the solutions. 
Yet another approach would be to pose the Cauchy problem on a light cone, but while it is physically appealing, such a formulation would introduce spurious technical difficulties (i.e.~the regularity at the tip of the cone) and does not appear to be very convenient from the analysis viewpoint.

The Schwarzschild metric in standard wave coordinates $(x^0, x^1, x^2, x^3)$ takes the form (cf.~\cite{Asanov}):
\bel{eq Sch-wave}
\aligned
\gSch_{00} &= - \frac{r-m_S}{r+m_S},
\qquad\qquad
\gSch_{ab} = \frac{r+m_S}{r-m_S} \omega_a \omega_b + \frac{(r+m_S)^2}{r^2}(\delta_{ab} - \omega_a \omega_b)
\endaligned
\ee 
with $\omega_a := x_a/r$. 
Furthermore, in order to distinguish between the behavior in the small and in the large, we introduce a smooth cut-off function $\chi: \RR^+  \to \RR$ (fixed once for all) satisfying $\chi(\tau) = 0$ for $\tau \in [0, 1/3]$ while  $\chi(\tau) = 1$ for $\tau \in [2/3, +\infty)$.

\begin{definition}
An initial data set for the Einstein-massive field system posed on the initial hypersurface $\{t=2\}$
is said to be a {\rm spatially compact perturbation of Schwarzschild spacetime} or
a compact Schwarzschild perturbation, in short,
if outside a compact set it coincides with the (vacuum) Schwarzschild space.
\end{definition}

The proof of the following result is postponed to Section~\ref{subsec proof-localization}, after investigating the nonlinear structure of the Einstein-massive field system.

\begin{proposition} \label{prop basic-localization}
Let $(g_{\alpha \beta}, \phi)$ be a solution to the system \eqref{eq main PDE} whose initial data
is a compact Schwarzschild perturbation,  
then $(g_{\alpha \beta} - {\gSch}_{\alpha \beta})$ is supported in the region $\Kcal$ 
and vanishes in a neighborhood of the 
boundary $\del_B\Kcal:= \{r=t-1, t\geq 2\}$.
\end{proposition}


\subsection{Coordinate formulation of the nonlinear stability property}

We introduce the restriction
$$
\Hcal_s^* := \Hcal_s\cap\Kcal
$$
of the hyperboloid to the light cone
and we consider the energy functionals
$$
\aligned
E_{g,c^2}(s,u)
:& =  \int_{\Hcal_s} \Big(-g^{00}|\del_t u|^2 + g^{ab} \del_au\del_bu + \sum_a \frac{2x^a}{t}g^{a \beta} \del_{\beta}u\del_t u + c^2u^2\Big) \, dx,
\\
E_{g,c^2}^*(s,u)
:& = 
 \int_{\Hcal^*_s} \Big(-g^{00}|\del_t u|^2 + g^{ab} \del_au\del_bu + \sum_a \frac{2x^a}{t}g^{a \beta} \del_{\beta}u\del_t u + c^2u^2\Big) \, dx,
\endaligned
$$
and, for the flat Minkowski background,
$$
\aligned
E_{M,c^2}(s,u)
:& = 
 \int_{\Hcal_s} \Big(|\del_t u|^2 + \sum_a|\del_au|^2 + \sum_a \frac{2x^a}{t} \del_au\del_t u + c^2u^2\Big) \, dx,
\\
E^*_{M,c^2}(s,u)
:& = 
 \int_{\Hcal_s^*} \Big(|\del_t u|^2 + \sum_a|\del_au|^2 + \sum_a \frac{2x^a}{t} \del_au\del_t u + c^2u^2\Big) \, dx.
\endaligned
$$
We have the alternative form
$$
\aligned
E_{M,c^2}(s,u) & =  \int_{\Hcal_s} \Big((s/t)^2 |\del_t u|^2 + \sum_a|\delu_a u|^2 + c^2u^2\Big) \, dx
\\
& =  \int_{\Hcal_s} \Big(|\del_t u + (x^a/t) \del_a u|^2 + \sum_{a<b}|t^{-1} \Omega_{ab}u|^2 + c^2u^2\Big) \, dx, 
\endaligned
$$
where $\Omega_{ab} :=x^a \del_b-x^b\del_a$ denotes the spatial rotations.
When the parameter $c$ is taken to vanish, we also use the short-hand notation
$E^*_g(s,u) := E^*_{g,0}(s,u)$ and $E_g(s,u):=E_{g,0}(s,u)$. 
In addition, for all  $p \in [1,+\infty)$, the $L^p$ norms on the hyperboloids endowed with the (flat) measure $dx$  are denoted by 
$$
\|u\|_{L_f^p(\Hcal_s)}^p : = \int_{\Hcal_s}|u|^p dx =  \int_{\RR^3} \big|u\big(\sqrt{s^2 +r^2}, x\big) \big|^pdx
$$
and the $L^P$ norms on the interior of $\Hcal_s$ by
$$
\|u\|_{L^p(\Hcal_s^*)}^p := \int_{\Hcal_s\cap \Kcal}|u|^p dx = \int_{r\leq (s^2-1)/2} \big|u\big(\sqrt{s^2 +r^2}, x\big) \big|^pdx.
$$

We are now in a position to state our main result for the Einstein system \eqref{eq main PDE}. The principal part of our system is the reduced wave operator associated with the curved metric $g$  and we can write the decomposition 
\be
\Boxt_g = g^{\alpha \beta} \del_{\alpha} \del_{\beta} = \Box + H^{\alpha \beta} \del_{\alpha} \del_{\beta}, 
\ee
in which $H^{\alpha \beta} := m^{\alpha \beta} - g^{\alpha \beta}$ are functions of $h = (h_{\alpha \beta})$. When $h$ is sufficiently small, $H^{\alpha \beta}(h)$ can be expressed as a power series in the components $h_{\alpha \beta}$ and vanishes at first-order at the origin. Our analysis will (only) use the translation and boost Killing fields associated with the flat wave operator $\Box$ in the coordinates under consideration.

\begin{theorem}[Nonlinear stability of Minkowski spacetime for self-gravitating massive fields. Formulation in coordinates]
\label{thm main-PDE}
Consider the Einstein-massive field equations \eqref{eq main PDE} together with an initial data set satisfying the constraints and prescribed on the hypersurface $\{t=2\}$:
\bel{eq main initial-data}
\aligned
& g_{\alpha \beta}|_{\{t=2\}} = g_{0, \alpha \beta},
\qquad
&&&
\del_t g_{\alpha \beta}|_{\{t=2\}} = g_{1, \alpha \beta},
\\
& \phi|_{\{t=2\}} = \phi_0,
\qquad
&&&
\del_t \phi |_{\{t=2\}} = \phi_1, 
\endaligned
\ee
which, on $\{t=2\}$ outside the unit ball $\{ r < 1\}$, is assumed to coincide with the restriction of Schwarzschild spacetime of mass $m_S$ (in the wave gauge \eqref{eq Sch-wave}), i.e. 
$$
g_{\alpha \beta}(2, \cdot) = {\gSch}_{\alpha \beta}, \qquad \del_tg_{\alpha \beta}(2, \cdot) = \phi(2, \cdot) = \del_t\phi(2, \cdot) = 0
\quad \text{in } \, \big\{ r=|x| \geq 1 \big\}.
$$
Then, for any sufficiently large integer $N$, there exist constants $\vep_0,C_1, \delta>0$ and such that provided
\bel{eq main smallness-condition}
\sum_{\alpha, \beta} \| \del g_{0, \alpha \beta}, g_{1, \alpha \beta} \|_{H^N(\{r < 1\})} 
+ \|\phi_0 \|_{H^{N+1}(\{ r <1\})} 
+ \|\phi_1 \|_{H^{N}(\{ r <1\})} 
+ m_S \leq \vep\leq \vep_0
\ee
holds at the initial time, then the solution associated with the initial data \eqref{eq main initial-data} exists for all times $t \geq 2$ and, furthermore, 
\bel{eq main energy-bound}
\aligned
E_M(s, \del^IL^Jh_{\alpha \beta})^{1/2}
& \leq C_1\vep s^{\delta},
\qquad&&&& | I | + |J| \leq N,
\\
 E_{M,c^2}(s, \del^IL^J\phi)^{1/2}
& \leq C_1\vep s^{\delta+1/2},
\qquad &&&&| I | + |J| \leq N,
\\
 E_{M,c^2}(s, \del^IL^J\phi)^{1/2}
& \leq C_1\vep s^{\delta},
\qquad &&&&| I | + |J| \leq N-4.
\endaligned
\ee
\end{theorem}


\subsection{Bootstrap argument and construction of the initial data}

We will rely on a bootstrap argument, which can be sketched as follows. We begin with our main system \eqref{eq main PDE} supplemented with initial data on the initial hyperboloid $\Hcal_2$, that is, $g_{\alpha \beta}|_{\Hcal_2}$, 
$\del_tg_{\alpha \beta}|_{\Hcal_2}$, 
$\phi|_{\Hcal_2}$, and $\del_t\phi|_{\Hcal_2}$. 
First of all, since the initial data is posed on $\{t=2\}$ and is sufficiently small, we need first to construct its restriction on the initial hyperboloid $\Hcal_2$. Since the data are compactly supported, this is immediate by the standard local existence theorem (see \cite[Chap.~11]{PLF-YM-book} for the details). We also observe that when the initial data posed on $\{t=2\}$ are sufficiently small, i.e.~\eqref{eq main smallness-condition} holds, then the corresponding data on $\Hcal_2$ satisfies the bounds
$$
\aligned
\| \del_a \del^IL^Jh_{\alpha \beta} \|_{L^2(\Hcal_2^*)} + \|\del_t\del^IL^Jh_{\alpha \beta} \|_{L^2(\Hcal_2^*)} 
& \leq C_0 \, \vep,  \qquad &&& |I|+|J| \leq N,
\\
\|\del^IL^J\phi\|_{L^2(\Hcal_2^*)} + \|\del_t\del^IL^J\phi\|_{L^2(\Hcal_2^*)} &
 \leq C_0 \, \vep, \qquad &&& |I|+|J| \leq N.
\endaligned
$$ 

We outline here the bootstrap argument and refer to \cite[Section~2.4]{PLF-YM-book} for further details. Throughout we fix a sufficiently large integer $N$
and we proceed by assuming that the following energy bounds have been established within a hyperbolic time interval $[2,s^*]$:
\begin{subequations} \label{eq 1 bootstrap}
\bel{eq 1 bootstrap a}
\aligned
E_M(s, \del^IL^Jh_{\alpha \beta})^{1/2} 
& \leq C_1\vep s^{\delta}, \quad &&& N-3\leq |I|+|J| \leq N,
\\
E_{M,c^2}(s, \del^IL^J\phi)^{1/2} 
& \leq C_1\vep s^{1/2 +\delta},
\quad &&& N-3\leq |I|+|J| \leq N,
\endaligned
\ee 
\bel{eq 1 bootstrap b}
E_M(s, \del^IL^Jh_{\alpha \beta})^{1/2} + E_{M,c^2}(s, \del^IL^J\phi)^{1/2} \leq C_1\vep s^{\delta},
\quad  |I|+|J| \leq N-4,
\ee
\end{subequations}
and, more precisely, we choose 
$$
s^* := \sup \Big\{s_1 \, \big| \, \text{for all }2\leq s\leq s_1, \text{the bounds \eqref{eq 1 bootstrap} hold} \Big\}.
$$
Since standard arguments for local existence do apply (see \cite[Chap. 11]{PLF-YM-book}) and, 
clearly, $s^* $ is not trivial in the sense that, if we choose $C_1>C_0$, then by continuity we have $s^*>2$.

By continuity, when $s=s^*$ at least one of the following equalities holds:
\bel{eq 2 bootstrap}
\aligned
E_M(s, \del^IL^Jh_{\alpha \beta})^{1/2}
& = C_1\vep s^{\delta},
\quad &&& N-3\leq |I|+|J| \leq N,
\\
E_{M,c^2}(s, \del^IL^J\phi)^{1/2}
& = C_1\vep s^{1/2 +\delta},
\quad &&& N-3\leq |I|+|J| \leq N,
\\
E_M(s, \del^IL^Jh_{\alpha \beta})^{1/2} + E_{M,c^2}(s, \del^IL^J\phi)^{1/2}
& = C_1\vep s^{\delta},
\quad &&& |I|+|J| \leq N-4.
\endaligned
\ee
Our main task for the rest of this work is to derive from \eqref{eq 1 bootstrap} the {\sl improved energy bounds} :
\bel{eq 3 bootstrap}
\aligned
E_M(s, \del^IL^Jh_{\alpha \beta})^{1/2} 
& \leq \frac{1}{2}C_1\vep s^{\delta},  
\quad &&&N-3\leq |I|+|J| \leq N,
\\
E_{M,c^2}(s, \del^IL^J\phi)^{1/2} 
& \leq \frac{1}{2}C_1\vep s^{1/2 +\delta},
\quad &&&N-3\leq |I|+|J| \leq N,
\\
E_M(s, \del^IL^Jh_{\alpha \beta})^{1/2} + E_{M,c^2}(s, \del^IL^J\phi)^{1/2} 
& \leq \frac{1}{2}C_1\vep s^{\delta},
\quad  &&&|I|+|J| \leq N-4.
\endaligned
\ee 
By comparing with  \eqref{eq 2 bootstrap}, we will be able to conclude that the interval $[2,s^*]$ extends to the maximal time of existence of the local solution. Then by a standard local existence argument, this local solution extends to all time values $s$.


\subsection{Outline of the Monograph}

We must therefore derive the improved energy bounds \eqref{eq 3 bootstrap} and, to this end, the rest of this work is organized as follows. In Section~\ref{section-2}, we begin by presenting various analytical tools which are required for the analysis of (general functions or) solutions defined on the hyperboloidal foliation. In particular, we establish first an energy estimate for wave equations and Klein-Gordon equations on a curved spacetime, then a sup-norm estimate based on characteristic integration, and next sharp $L^\infty$--$L^\infty$ estimates for wave equations and for Klein-Gordon equations, as well as Sobolev and Hardy inequalities on hyperboloids. 

In Section~\ref{section-3-QUASI}, we discuss the reduction of the Einstein-massive field system and we establish the quasi-null structure in wave gauge. We provide a classification of all relevant nonlinearities arising in the problem and we carefully study the nonlinear structure of the Einstein equations in the semi-hyperboloidal frame.  

Next, in Section~\ref{section-4} we formulate our full list of bootstrap assumptions and we write down basic estimates that directly follow from these assumptions. In Section~\ref{section-5}, we are in a position to provide a preliminary control of the nonlinearities of the Einstein equations in the $L^2$ and $L^\infty$ norms.  In Section~\ref{sect--7}, we establish estimates which are tight to the wave gauge condition.
 
 An estimate of the second-order derivatives of the metric coefficients is then derived in Section~\ref{section-6.5}, while in Section~\ref{section-7} we obtain a sup-norm estimate based on integration on characteristics and we apply it to the control of quasi-null terms. 

We are then able, in Section~\ref{section-8}, to derive the low-order ``refined'' energy estimate for the metric and next, in Section~\ref{section-9}, to control the low-order sup-norm of the metric as well as of the scalar field. In Section~\ref{section-10}, we improve our bound on the high-order energy for the metric components and the scalar field. In Section~\ref{section-12}, based on this improved energy bound at high-order, we establish high-order sup-norm estimates. Finally, in Section~\ref{section-13}, we improve the low-order energy bound on the scalar field and we conclude our bootstrap argument.


\section{Functional Analysis on Hyperboloids of Minkowski Spacetime}
\label{section-2}

\subsection{Energy estimate on hyperboloids}

In this section, we need to adapt the techniques we introduced earlier in \cite{PLF-YM-book,PLF-YM-one} to 
the compact Schwarzschild perturbations under consideration in the present Monograph, since these techniques were established for compactly supported initial data. Here, the initial data is not supported in the unit ball but coincides with Schwarzschild space outside the unit ball. 
As mentioned in the previous section, the curved part of the metric (for a solution of the Einstein-massive field system with a compact Schwarzschild perturbation) is not compactly supported in the light-cone $\Kcal$, while the hyperboloidal energy estimate developed in \cite{PLF-YM-book} 
were assuming this. Therefore, we need to revisit the energy estimate and take
suitable boundary terms into account. 

\begin{proposition}[Energy estimate. I]\label{prop 1 energy-W}
Let $(h_{\alpha \beta}, \phi)$ be a solution of the Einstein-massive field system associated with an initial data set that is a compact Schwarzschild perturbation with mass $m_S \in (0,1)$. Assume that there exists a constant $\kappa>1$ such that
\bel{eq energy 1}
\kappa^{-1} E_M^*(s,u)^{1/2} \leq E_g^*(s,u)^{1/2} \leq \kappa E_M^*(s,u)^{1/2}. 
\ee
Then, there exists a positive constant $C$ (depending upon $N$ and $\kappa$) such that the following energy estimate holds (for all $\alpha, \beta \leq 3$, and $|I|+|J| \leq N$):
\bel{eq energy 2}
\aligned
E_M(s, \del^IL^J h_{\alpha \beta})^{1/2} & \leq   CE_g(2, \del^IL^J h_{\alpha\beta} )^{1/2} + C m_S
+ C\int_2^s\|\del^IL^J F_{\alpha \beta} \|_{L^2(\Hcal_\tau^*)}d\tau
\\
& \quad +  C\int_2^s\|[\del^IL^J,H^{\mu\nu} \del_\mu\del_\nu]h_{\alpha \beta} \|_{L^2(\Hcal_\tau^*)}d\tau
+  C\int_2^s M_{\alpha \beta}[\del^IL^J h](\tau) \, d\tau 
\\
& \quad +  C\int_2^s
\Big(\|\del^IL^J(\del_{\alpha} \phi\del_{\beta} \phi) \|_{\Hcal_{\tau}^*} + \|\del^IL^J(\phi^2 g_{\alpha\beta}) \|_{\Hcal_{\tau}^*} \Big) \, d\tau,
\endaligned
\ee
in which $M_{\alpha \beta}[\del^IL^J h](s)$ is a positive function such that
\bel{eq energy 3}
\aligned
&\int_{\Hcal_s^*}(s/t) \big|\del_{\mu}g^{\mu\nu} \del_{\nu} \big(\del^IL^Jh_{\alpha \beta} \big) \del_t\big(\del^IL^Jh_{\alpha \beta} \big)- \frac{1}{2} \del_tg^{\mu\nu} \del_{\mu} \big(\del^IL^Jh_{\alpha \beta} \big) \del_{\nu} \big(\del^IL^Jh_{\alpha \beta} \big) \big| \, dx
\\
& \leq M_{\alpha \beta}[\del^IL^Jh](s)E_M^*(s, \del^IL^Jh_{\alpha \beta})^{1/2}.
\endaligned
\ee
\end{proposition}

The proof of this estimate is done as follows: in the exterior part of the hyperboloid (i.e. $\Hcal_s\cap \Kcal^c$), the metric coincides with the Schwarzschild metric and we can calculate the energy by an explicit expression. On the other hand, the interior part is bounded as follows.

\begin{lemma} \label{lem energy 1}
Under the assumptions in Proposition \ref{prop 1 energy-W}, one has 
\bel{eq energy 4}
\aligned
E^*_M(s, \del^IL^J h_{\alpha \beta})^{1/2} & \leq  CE_g^*(2, \del^I L^Jh_{\alpha \beta})^{1/2} + C m_S
+ C\int_2^s M_{\alpha \beta}(\tau, \del^IL^Jh_{\alpha \beta}) \, d\tau
\\
& \quad +  C\int_2^s\|\del^IL^J F_{\alpha \beta} \|_{L^2(\Hcal_\tau^*)}d\tau
+ C\int_2^s\|[\del^IL^J,H^{\mu\nu} \del_\mu\del_\nu]h_{\alpha \beta} \|_{L^2(\Hcal_\tau^*)}d\tau
\\
& \quad +  C\int_2^s\big(\|\del^IL^J(\del_{\alpha} \phi\del_{\beta} \phi) \|_{L^2(\Hcal_\tau^*)} + \|\del^IL^J\big(\phi^2 g_{\alpha\beta} \big) \|_{L^2(\Hcal_\tau^*)} \big) d\tau. 
\endaligned
\ee
\end{lemma}

\begin{proof} We consider the wave equation
$g^{\mu\nu} \del_{\mu} \del_{\nu}h_{\alpha \beta} = F_{\alpha \beta} - 16\pi \del_{\alpha} \phi\del_{\beta} \phi - 8\pi c^2\phi^2 g_{\alpha\beta}$  satisfied by the curved part of the metric
and differentiate it (with $\del^IL^J$ with $|I|+|J| \leq N$): 
$$
\aligned
g^{\mu\nu} \del_{\mu} \del_{\nu} \del^IL^Jh_{\alpha \beta}
 =
&  -[\del^IL^J,H^{\mu\nu} \del_{\mu} \del_{\nu}]h_{\alpha \beta}
  +   \del^IL^JF_{\alpha \beta} 
\\
&- 16\pi \del^IL^J\big(\del_{\alpha} \phi\del_{\beta} \phi\big) -8\pi c^2\del^IL^J\big(\phi^2 g_{\alpha\beta} \big). 
\endaligned
$$
Using the multiplier $- \del_t \del^IL^J h_{\alpha \beta}$, we obtain the general identity 
\bel{eq energy 5}
\aligned
&\del_t\big(-(1/2)g^{00}|\del_t\del^IL^Jh_{\alpha \beta}|^2 + (1/2)g^{ab} \del_a \del^IL^Jh_{\alpha \beta} \del_b \del^IL^Jh_{\alpha \beta} \big)
-  \del_a \big(g^{a \nu} \del_{\nu} \del^IL^Jh_{\alpha \beta} \del_t\del^IL^Jh_{\alpha \beta} \big)
 \\
& = \frac{1}{2} \del_tg^{\mu\nu} \del_{\mu} \del^IL^Jh_{\alpha \beta} - \del_{\mu}g^{\mu\nu} \del_t\del^IL^Jh_{\alpha \beta} \del_{\nu} \del^IL^Jh_{\alpha \beta}
\\
& \quad + [\del^IL^J,H^{\mu\nu} \del_{\mu} \del_{\nu}]h_{\alpha \beta} \del_t\del^IL^Jh_{\alpha \beta}
 - \del^IL^JF_{\alpha \beta} \del_t\del^IL^Jh_{\alpha \beta}
\\
& \quad + 16\pi \del^IL^J\big(\del_{\alpha} \phi\del_{\beta} \phi\big) \del_t\del^IL^J h_{\alpha \beta}
 +8\pi c^2\del^IL^J\big(\phi^2 g_{\alpha\beta}  \big) \del_t\del^IL^Jh_{\alpha \beta}. 
\endaligned
\ee
For simplicity, we write $u = \del^IL^J h_{\alpha \beta}$ and
$W := \big(-(1/2)g^{00}|\del_t u|^2 + (1/2)g^{ab} \del_au\del_bu,  - g^{a \nu} \del_{\nu}u\del_t u\big)$
for 
the energy flux, while 
$$
\aligned
\mathcal{F} :=
&\frac{1}{2} \del_tg^{\mu\nu} \del_{\mu} \del^IL^Jh_{\alpha \beta} - \del_{\mu}g^{\mu\nu} \del_t\del^IL^Jh_{\alpha \beta} \del_{\nu} \del^IL^Jh_{\alpha \beta}
\\
& \quad + [\del^IL^J,H^{\mu\nu} \del_{\mu} \del_{\nu}]h_{\alpha \beta} \del_t\del^IL^Jh_{\alpha \beta}
 - \del^IL^JF_{\alpha \beta} \del_t\del^IL^Jh_{\alpha \beta}
\\
& \quad + 16\pi \del^IL^J\big(\del_{\alpha} \phi\del_{\beta} \phi\big) \del_t\del^IL^J h_{\alpha \beta}
 +8\pi c^2\del^IL^J\big(\phi^2 g_{\alpha\beta}  \big) \del_t\del^IL^Jh_{\alpha \beta}. 
\endaligned
$$
Then, by defining $\text{Div}$ with respect to the Euclidian metric on $\RR^{3+1}$, \eqref{eq energy 5} reads 
$\text{Div}W = \mathcal{F}$ 
and we can next integrate this equation in the region $\Kcal_{[2,s]}$ and write 
$
\int_{\Kcal_{[2,s]}} \text{Div}W dxdt = \int_{\Kcal[2,s]} \mathcal{F} dxdt.
$
In the left-hand side, we apply Stokes' formula:
$$
\int_{\Kcal_{[2,s]}} \text{Div}W dxdt
 = \int_{\Hcal_s^*}W\cdot {\sl n}d\sigma
 + \int_{\Hcal_2^*}W\cdot {\sl n}d\sigma
 + \int_{B_{[2,s]}}W\cdot {\sl n}d\sigma, 
$$
where $B_{[2,s]}$ is the  
boundary of $\Kcal_{[2,s]}$, which is $\big\{(t, x)|t=r+1, 3/2 \leq r \leq (s^2-1)/2 \big\}$.
An easy calculation shows that
\bel{eq energy 7}
\aligned
\int_{\Kcal_{[2,s]}} \text{Div}W dxdt 
& =  \frac{1}{2} \Big(E_g^*(s, \del^IL^J h_{\alpha \beta}) - E_g^*(2, \del^IL^J h_{\alpha \beta}) \Big)
\\
 & \quad +  \int_{3/2\leq r\leq (s^2-1)/2} \int_{\mathbb{S}^2}W\cdot (- \sqrt{2}/2, \sqrt{2}x^a/2r) \sqrt{2}r^2drd\omega ds, 
\endaligned
\ee
where $d\omega$ is the standard Lebesgue measure on $\mathbb{S}^2$. Recall that $g_{\alpha \beta} = {\gSch}_{\alpha \beta}$  in a neighborhood of $B_{[2,s]}$.
An explicit calculation shows that 
$W = \big((1/2){\gSch}^{ab} \del_a \del^IL^J {h_S}_{\alpha \beta} \del_b\del^I L^J {h_S}_{\alpha \beta},0\big)$
on $B_{[2,s]}$.
We have 
$$
\aligned
&
\int_{3/2\leq r\leq (s^2-1)/2} \int_{\mathbb{S}^2}W\cdot (- \sqrt{2}/2, \sqrt{2}x^a/2r) \sqrt{2}r^2drd\omega
 \\
& =-2\pi \int_{3/2}^{(s^2-1)/2}\gSch^{ab} \del_a \del^IL^J {h_S}_{\alpha \beta} \del_b \del^IL^J {h_S}_{\alpha \beta}r^2dr ds
\endaligned
$$
with ${h_S}_{\alpha \beta} := {\gSch}_{\alpha \beta} - m_{\alpha \beta}$. This leads us to
$$
 \frac{d}{ds} \int_{B_{[2,s]}}W\cdot {\sl n}d\sigma
 = - \frac{\pi}{2} s(s^2-1)^2  \gSch^{ab} \del_a \del^IL^J {h_S}_{\alpha \beta} \del_b \del^IL^J {h_S}_{\alpha \beta} \bigg|_{r=\frac{s^2-1}{2}}. 
$$
Assuming that $m_S$ is sufficiently small, we see that
$$
\big| \gSch^{ab} \del_a \del^IL^J {h_S}_{\alpha \beta} \del_b \del^IL^J {h_S}_{\alpha \beta} \big| \leq Cm_S^2r^{-4} \leq Cm_S^2 s^{-8},
\qquad 3/2\leq r.
$$
We have 
\bel{eq energy 8}
\bigg|\frac{d}{ds} \int_{B_{[2,s]}}W\cdot {\sl n}d\sigma \bigg| \leq Cm_S^2 s^{-3}.
\ee

Now, we combine $\text{Div}W = \mathcal{F}$ and \eqref{eq energy 7} and differentiate in $s$:
$$
\frac{1}{2} \frac{d}{ds}E_g^*(s, \del^IL^J h_{\alpha \beta}) + \frac{d}{ds} \int_{B_{[2,s]}}W\cdot {\sl n}d\sigma = \frac{d}{ds} \int_{\Kcal_{[2,s]}} \mathcal{F} \, dxdt, 
$$
which leads us to
$$
E_g^*(s, \del^IL^J h_{\alpha \beta})^{1/2} \frac{d}{ds} \big(E_g^*(s, \del^IL^J h_{\alpha \beta})^{1/2} \big) 
=
 - \frac{d}{ds} \int_{B_{[2,s]}}W\cdot {\sl n}d\sigma + \frac{d}{ds} \int_2^s\int_{\Hcal_s^*}(s/t) \mathcal{F} \, dxds. 
$$
Then, in view of \eqref{eq energy 8} we have 
\bel{eq energy 9}
E_g^*(s, \del^IL^J h_{\alpha \beta})^{1/2} \frac{d}{ds} \big(E_g^*(s, \del^IL^J h_{\alpha \beta})^{1/2} \big) \leq
\int_{\Hcal_s^*}(s/t)|\mathcal{F}| \, dx + Cm_S^2 s^{-3}.
\ee
In view of the notation and assumptions in Proposition \ref{prop 1 energy-W}, we have
$$
\aligned
&\int_{\Hcal_s^*} \big|(s/t) \mathcal{F} \big| \, dx
\leq \int_{\Hcal_s^*}|(s/t) \del_t\del^IL^Jh_{\alpha \beta} \del^IL^JF_{\alpha \beta}| dx
\\
&+ \int_{\Hcal_s^*}|(s/t) \del_t\del^IL^Jh_{\alpha \beta}[\del^IL^J,H^{\mu\nu} \del_{\mu} \del_{\nu}]h_{\alpha \beta}| dx
 + 16\pi\int_{\Hcal_s^*}|(s/t) \del_t\del^IL^Jh_{\alpha \beta} \del^IL^J(\del_{\alpha} \phi\del_{\beta} \phi)| \, dx
\\
&+ 8\pi c^2\int_{\Hcal_s^*}|(s/t) \del_t\del^IL^Jh_{\alpha \beta} \del^IL^J\big(\phi^2 g_{\alpha\beta}  \big)| \, dx
 + M[\del^IL^J h](s)E_M^*(s, \del^IL^J h_{\alpha \beta})^{1/2}
\\
&\leq \|(s/t) \del_t\del^IL^J h_{\alpha \beta} \|_{L^2(\Hcal_s^*)} \big(\|\del^IL^JF_{\alpha \beta} \|_{L^2(\Hcal_s^*)} + \|\del^IL^J,[H^{\mu\nu} \del_\mu\del_\nu]h_{\alpha \beta} \|_{L^2(\Hcal_s^*)} \big)
\\
& \quad + C\|(s/t) \del_t\del^IL^J h_{\alpha \beta} \|_{L^2(\Hcal_s^*)} \big(\|\del^IL^J(\del_{\alpha} \phi\del_{\beta} \phi) \|_{L^2(\Hcal_s^*)} + \|\del^IL^J\big(\phi^2 g_{\alpha\beta} \big) \|_{L^2(\Hcal_s^*)} \big)
\\
&\quad+M[\del^IL^J h](s)E^*_M(s, \del^IL^J h_{\alpha \beta})^{1/2},
\endaligned
$$
 so that
$$
\aligned
\int_{\Hcal_s^*} \big|(s/t) \mathcal{F} \big| \, dx
& \leq CE_M^*(s, \del^IL^J h_{\alpha \beta})^{1/2} \, 
\bigg(\|\del^IL^JF_{\alpha \beta} \|_{L^2(\Hcal_s^*)} + \|\del^IL^J,[H^{\mu\nu} \del_\mu\del_\nu]h_{\alpha \beta} \|_{L^2(\Hcal_s^*)}
\\
& \quad +  \|\del^IL^J(\del_{\alpha} \phi\del_{\beta} \phi) \|_{L^2(\Hcal_s^*)} + \|\del^IL^J\big(\phi^2  g_{\alpha\beta} \big) \|_{L^2(\Hcal_s^*)} +M[\del^IL^J h](s) \bigg). 
\endaligned
$$
For simplicity, we write 
$$
\aligned
L(s) :& = 
\|\del^IL^JF_{\alpha \beta} \|_{L^2(\Hcal_s^*)} + \|\del^IL^J,[H^{\mu\nu} \del_\mu\del_\nu]h_{\alpha \beta} \|_{L^2(\Hcal_s^*)}
\\
& \quad +  \|\del^IL^J(\del_{\alpha} \phi\del_{\beta} \phi) \|_{L^2(\Hcal_s^*)} + \|\del^IL^J\big(\phi^2  g_{\alpha\beta} \big) \|_{L^2(\Hcal_s^*)} + M[\del^IL^J h](s)
\endaligned
$$
and
$y(s):= E_g^*(s, \del^IL^Jh_{\alpha \beta})^{1/2}.$
In view of \eqref{eq energy 1}, we have 
$$
E_M^*(s, \del^IL^J h_{\alpha \beta})^{1/2} \leq C\kappa E_g^*(s, \del^IL^J h_{\alpha \beta})^{1/2} 
$$
and \eqref{eq energy 9} leads us to
$
y(s)y'(s) = C\kappa y(s)L(s) + Cm_S^2s^{-3}. 
$
By Lemma \ref{lem energy 2} stated shortly below, we conclude that (with $m_S = \vep$ and $\sigma=2$ therein)
$$
y(s) \leq y(0) + Cm_S + C\kappa \int_2^sL(s)ds.
$$
By recalling \eqref{eq energy 1}, the above inequality leads us to \eqref{eq energy 4}.
\end{proof}

\begin{lemma} \label{lem energy 2}
The nonlinear inequality 
$
y(\tau)y'(\tau) \leq g(\tau)y(\tau) + C^2\vep^2 \tau^{-1- \sigma},  
$
in which the function $y:[2,s] \to \RR^+$ is sufficiently regular, 
the function $g$ is positive and locally integrable, and $C, \vep, \sigma$ are positive constants, 
implies the linear inequality 
$$
y(\tau) \leq y(2) + C\vep\left(1+\sigma^{-1} \right) + \int_2^\tau g(\eta)d\eta.
$$
\end{lemma}

\begin{proof}
We denote by $I = \{\tau\in[2,s]|y(s)>C\vep \}$. In view of the continuity of $y$, $I = \bigcup_{i\in\mathbb{N}}(I_n\cap[2,s])$ where $I_n$ are open intervals disjoint from each other. For $\tau \notin I$, $y(\tau) \leq C\vep $.
For $\tau \in I$, there exists some integer $i$ such that $\tau\in I_i\cap [2,s]$. Let $\inf (I_i\cap[2,s]) = s_0\geq 2$, then on $I_n\cap[2,s]$,
$$
y'(\tau) \leq g(\tau) + \frac{C^2\vep^2\tau^{-1- \sigma}}{y(\tau)} \leq g(\tau) + C\vep \tau^{-1- \sigma}.
$$
This leads us to
$$
\aligned
\int_{s_0}^\tau y'(\eta)d\eta & \leq  \int_{s_0}^\tau g(\eta)d\eta + C\vep \int_{s_0}^\tau s^{-1- \sigma}ds
\leq \int_2^\tau g(\eta)d\eta + C\vep\int_2^\infty s^{-1- \sigma}ds
 \leq  \int_2^\tau g(\eta)d\eta + C\vep\sigma^{-1}
\endaligned
$$
and 
$y(\tau) - y(s_0) \leq \int_2^\tau g(\eta)d\eta + C\vep\sigma^{-1}$. 
By continuity, either $s_0 \in(2,s)$ which leads us to $y(s_0) = C\vep$, or else $s_0 = 2$ which leads us to $y(s_0) = y(2)$. Then, we obtain 
$$
y(\tau) \leq \max\{y(2),C\vep\} + C\vep\sigma^{-1} + \int_2^\tau g(\eta)d\eta. 
$$ 
\end{proof}

To complete the proof of Proposition \ref{prop 1 energy-W}, we need the following additional observation, which is checked
by an explicit calculation (omitted here).

\begin{lemma} \label{lem energy 3}
The following uniform estimate holds (for all $a, \alpha, \beta$, all relevant $I,J$, and for some $C=C(I,J)$)
\be
\int_{\Hcal_s\cap\Kcal^c}|\delu_a \del^IL^J{h_S}_{\alpha \beta}|^2dx + \int_{\Hcal_s\cap\Kcal^c}(s/t)|\del_t \del^IL^J{h_S}_{\alpha \beta}|^2dx \leq Cm_S^2.
\ee
\end{lemma}

\begin{proof}[Proof of Proposition \ref{prop 1 energy-W}]
We observe that
$$
\aligned
& E_g(s, \del^IL^J h_{\alpha \beta}) 
\\
& \leq   E^*_g(s, \del^IL^Jh_{\alpha \beta})
 + C\int_{\Hcal_s\cap\Kcal^c}|\delu_a \del^IL^J{h_S}_{\alpha \beta}|^2dx + \int_{\Hcal_s\cap\Kcal^c}(s/t)|\del_t \del^IL^J{h_S}_{\alpha \beta}|^2dx. 
\endaligned
$$
Combining \eqref{eq energy 4} with Lemma \ref{lem energy 3} allows us to complete the proof of \eqref{eq energy 2}.
\end{proof}

For all solutions to the Einstein-massive field system associated with compact Schwarzschild perturbations, the scalar field $\phi$ is also supported in $\Kcal$. So the energy estimate for $\phi$ remains identical to the one in \cite{PLF-YM-one}.

\begin{proposition}[Energy estimate. II]
\label{prop energy 2KG}
Under the assumptions in Proposition \ref{prop 1 energy-W}, the scalar field $\phi$ satisfies
\bel{eq energy 10}
\aligned
E_{M,c^2}(s, \del^IL^J\phi)^{1/2} & \leq  CE_{g,c^2}(2, \del^IL^J \phi)^{1/2}
\\
& \quad +  \int_2^s\big|[\del^IL^J,H^{\mu\nu} \del_{\mu} \del_{\nu}]\phi\big|d\tau
+ \int_2^sM[\del^IL^J \phi](\tau) \, d\tau,
\endaligned
\ee
in which $M[\del^IL^J \phi](s)$ denotes a positive function such that
\be
\label{eq energy 11}
\aligned
&\int_{\Hcal_s}(s/t) \big|\del_{\mu}g^{\mu\nu} \del_{\nu} \big(\del^IL^J\phi\big) \del_t\big(\del^IL^J\phi\big)- \frac{1}{2} \del_tg^{\mu\nu} \del_{\mu} \big(\del^IL^J\phi\big) \del_{\nu} \big(\del^IL^J\phi\big) \big| \, dx
\\
& \leq M[\del^IL^J\phi](s)E_{M, c^2}(s, \del^IL^J\phi)^{1/2}.
\endaligned
\ee
\end{proposition}


\subsection{Sup-norm estimate based on curved characteristic integration}

We now revisit an important technical tool introduced first in Lindblad and Rodnianski~\cite{LR1}. This is an $L^\infty$ estimate on the gradient of solutions to a wave equation posed in a curved background. For our problem, we must adapt this tool to the hyperboloidal foliation and we begin by stating without proof the following identity.  

\begin{lemma}[Decomposition of the flat wave operator in the null frame]\label{lem LR-sup 1}
For every smooth function $u$, the following identity holds:
\bel{eq LR-sup 1}
- \Box u = r^{-1}(\del_t+\del_r) \big(\del_t- \del_r \big)(ru) - \sum_{a<b} \big(r^{-1} \Omega_{ab} \big)^2u 
\ee 
with 
$
\Omega_{ab} = x^a \del_b - x^b\del_a = x^a \delu_b - x^b\delu_a$ (defined earlier). 
\end{lemma}

We then write $\del_t = \frac{t}{t+r}(\del_t - \del_r) + \frac{x^at}{(t+r)r} \delu_a$
and thus 
$$
\aligned
\del_t\del_t & =  \frac{t^2}{(t+r)^2}(\del_t- \del_r)^2 + \frac{t}{t+r}(\del_t- \del_r) \bigg(\frac{x^at\delu_a}{r(t+r)} \bigg) + \frac{x^at}{r(t+r)} \delu_a \bigg(\frac{t}{t+r}(\del_t- \del_r) \bigg)
\\
& \quad +  \bigg(\frac{x^at}{r(t+r)} \delu_a \bigg)^2 + \frac{\del_t- \del_r}{t+r}.
\endaligned
$$
Consequently, we have found the decomposition 
\bel{eq:form9-1}
\aligned
r\del_t\del_t u & =  \frac{t^2}{(t+r)^2}(\del_t- \del_r)^2(ru) + \frac{2t^2}{(t+r)^2}(\del_t- \del_r)u +  \frac{rt}{t+r}(\del_t- \del_r) \bigg(\frac{x^at}{r(t+r)} \delu_au\bigg)
\\
& \quad  + \frac{x^at}{(t+r)} \delu_a \bigg(\frac{t}{t+r}(\del_t- \del_r)u\bigg)
 + r\bigg(\frac{x^at}{r(t+r)} \delu_a \bigg)^2 u + \frac{r(\del_t- \del_r)u}{t+r}
\\
& =:  \frac{t^2}{(t+r)^2}(\del_t- \del_r)^2(ru) + W_1[u]. 
\endaligned
\ee

On the other hand, the curved part of the reduced wave operator $H^{\alpha \beta} \del_{\alpha} \del_{\beta}$ can be decomposed in the semi-hyperboloidal frame as follows: 
$$
\aligned
H^{\alpha \beta} \del_{\alpha} \del_{\beta}u & =  \Hu^{\alpha \beta} \delu_{\alpha} \delu_{\beta}u
+ H^{\alpha \beta} \del_{\alpha} \big(\Psi_{\beta}^{\beta'} \big) \delu_{\beta'}u
\\
& =  \Hu^{00} \del_t\del_t u
 +  \Hu^{a0} \delu_a \del_t u + \Hu^{0a} \del_t\delu_a u + \Hu^{ab} \delu_a \delu_b u
 + H^{\alpha \beta} \del_{\alpha} \big(\Psi_{\beta}^{\beta'} \big) \delu_{\beta'}u.
\endaligned
$$
The ``good'' part of the curved wave operator (i.e. terms containing one derivative tangential to the hyperboloids) is defined to be 
\bel{eq:form9}
R[u,H]:=  \Hu^{a0} \delu_a \del_t u + \Hu^{0a} \del_t\delu_a u + \Hu^{ab} \delu_a \delu_b u
 + H^{\alpha \beta} \del_{\alpha} \big(\Psi_{\beta}^{\beta'} \big) \delu_{\beta'}u, 
\ee
and, with this notation together with \eqref{eq:form9-1},  
\bel{eq LR-sup 2}
\aligned
rH^{\alpha \beta} \del_{\alpha} \del_{\beta}u & =  \frac{t^2\Hu^{00}}{(t+r)^2}(\del_t- \del_r) \big((\del_t- \del_r)(ru) \big)
+\Hu^{00}W_1[u] + rR[u,H].
\endaligned
\ee
Then, by combining \eqref{eq LR-sup 1} for the flat wave operator and \eqref{eq LR-sup 2} for the curved part, we reach the following conclusion. 

\begin{lemma}[Decomposition of  the reduced wave operator $\Boxt_g$]
\label{lem LR-sup 2}
Let $u$ be a smooth function defined in $\RR^{3+1}$ and $H^{\alpha \beta}$ be functions in $\RR^{3+1}$. Then the following identity holds:
\bel{eq LR-sup 3}
\aligned
&\Big((\del_t+\del_r) - t^2(t+r)^{-2} \Hu^{00}(\del_t- \del_r) \Big) \Big( \big(\del_t- \del_r \big)(ru) \Big) 
\\
&= -r \, \Boxt_g u + r\sum_{a<b} \big(r^{-1} \Omega_{ab} \big)^2 u + \Hu^{00}W_1[u] + rR[u,H]
\endaligned
\ee
with the notation above. 
\end{lemma}

Now we are ready to establish the desired estimate of this section. For convenience, we set  
$$
\Kical := \big\{(t, x)|r \leq \frac{3}{5}t \big\} \cap \Kcal,
\quad 
\Kical_{[s_0,s_1]} := \big\{(t, x) \in \Kical \, / \,  s_0^2\leq t^2-r^2\leq s_1^2 \big\}
$$
and we denote by $\del_B\Kical_{[s_0,s_1]}$ the following ``boundary'' of $\Kical_{[s_0,s_1]}$
$$
\del_B\Kical_{[s_0,s_1]} := \big\{(t, x) \, / \, r = (3/5)t, \, (5/4)s_0 \leq t \leq (5/4)s_1 \big\}. 
$$ 
We will now prove the following sharp decay property for solutions to the wave equation on a curved spacetime. 

\begin{proposition}[Sup-norm estimate based on characteristic integration]
\label{prop LR-sup 1}
Let $u$ be a solution to the wave equation on curved spacetime 
$- \Box u - H^{\alpha \beta} \del_{\alpha} \del_{\beta} u = F,
$
where $H^{\alpha \beta}$ are given functions.
Given any point $(t_0, x_0)$, denote by $\left(t, \varphi(t;t_0, x_0) \right)$ the integral curve of the vector field
$$
\del_t + \frac{(t+r)^2 +t^2\Hu^{00}}{(t+r)^2 - t^2\Hu^{00}} \del_r
$$
passing through  $(t_0, x_0)$, that is, 
$\varphi(t_0;t_0, x_0) = x_0$. 
Then, there exist two positive constants $\vep_s$ and $a_0\geq 2$ such that for $t\geq a_0$
\bel{ineq LR-sup 2}
|\Hu^{00}| \leq \vep_s (t-r)/t,
\ee
then for all $s \geq a_0$ and $(t, x) \in \Kcal\backslash\Kical_{[2,s]}$
one has 
\bel{ineq LR-sup 1}
\aligned
|(\del_t- \del_r)u(t, x)|
& \leq  t^{-1} \sup_{\del_B\Kical_{[2,s]} \cup \del\Kcal} \Big( |(\del_t- \del_r)(ru)|\Big) 
 + Ct^{-1}|u(t, x)|
\\
& \quad +  t^{-1} \int_{a_0}^t \tau|F(\tau, \varphi(\tau;t, x))|d\tau
+ t^{-1} \int_{a_0}^t \big|M_s[u,H]|_{(\tau, \varphi(\tau;t, x))}d\tau, 
\endaligned
\ee
where $F=- \Box u - H^{\alpha \beta} \del_{\alpha} \del_{\beta} u$ is the right-hand side of the wave equation, 
$$
M_s[u,H]:= r\sum_{a<b} \big(r^{-1} \Omega_{ab} \big)^2 u + \Hu^{00}W_1[u] + rR[u,H],
$$
in which one can guarantee that the associated integral curve satisfies 
$
(\tau, \varphi(\tau;t, x)) \in \Kcal\backslash\Kical_{[2,s]}
$
for $2\leq a_0<\tau<t$, but
$
(a_0, \varphi(a_0;t, x)) \in\del_B\Kical_{[2,s_0]} \cup \del\Kcal
$ at the initial time $a_0$. 
\end{proposition}

\begin{proof} Under the condition \eqref{ineq LR-sup 2}, the decomposition \eqref{eq LR-sup 3} can be rewritten in the form 
\bel{eq LR-sup 4}
\aligned
&\bigg(\del_t + \frac{1+t^2(t+r)^{-2} \Hu^{00}}{1-t^2(t+r)^{-2} \Hu^{00}} \del_r\bigg) \big((\del_t - \del_r)(ru) \big)
=: \mathcal{L} \big((\del_t - \del_r)(ru) \big)
\\
& = \frac{-r\Boxt_g u + r\sum_{a<b} \big(r^{-1} \Omega_{ab} \big)^2 u + \Hu^{00}W_1[u] + rR[u,H]}{1-t^2(t+r)^{-2} \Hu^{00}} =: \mathcal{F}.
\endaligned
\ee
In other words, \eqref{eq LR-sup 4} reads 
$
\mathcal{L} \big((\del_r- \del_r)(ru) \big) = \mathcal{F} 
$
and by writing 
$$
v_{t_0, x_0}(t) := \big( (\del_r- \del_r)(ru) \big)(t, \varphi(t;t_0, x_0)),
$$
we have 
$$
\frac{d}{dt} v_{t_0, x_0}(t) = \mathcal{L} \big((\del_t - \del_r)(ru) \big)(t, \varphi(t;t_0, x_0))
= \mathcal{F} (t, \varphi(t;t_0, x_0)).
$$
By integration, we have 
$
v_{t_0, x_0}(t_0) = v_{t_0, x_0}(a) + \int_{a}^{t_0} \mathcal{F}(t, \varphi(t;t_0, x_0)) \, dt.
$

Fix $s_0^2 = t_0^2-r_0^2$ with $s_0>0$ and take $(t_0, x_0) \in \Kcal_{[2,s]} \backslash \Kical$, that is $\{(t_0, x_0)|(3/5)t_0\leq r_0<t_0-1\}$. We will prove that there exists some $a \geq 2$ such that for all $t\in[a,t_0]$, $(t, \varphi(t;t_0, x_0)) \in \Kcal_{[2,s]} \backslash\Kical$ and $(a, \varphi(a;t_0, x_0)) \in \del_B\Kical_{[2,s_0]} \cup \del\Kcal$, that is, for $t<t_0$, $(t, \varphi(t;t_0, x_0))$ will not intersect $\Hcal_{s_0}$ again before leaving the region $\Kcal_{[2,s_0]} \backslash \Kical$. This is due to the following observation: denote by $|\varphi(t;t_0, x_0)|$ the Euclidian norm of $\varphi(t;t_0, x_0)$, and by the definition of $\mathcal{L}$, we have 
$$
\frac{d|\varphi(t;t_0, x_0)|}{dt} = \frac{1+t^2(t+r)^{-2} \Hu^{00}}{1-t^2(t+r)^{-2} \Hu^{00}}.
$$
Also, we observe that for a point $(t, x)$ on the hyperboloid $\Hcal_{s_0}$, we have $r(t) = |x(t)| = \sqrt{t^2-s_0^2}$, and this leads us to
$
\frac{dr}{dt} = \frac{t}{r}. 
$
Then we have
$$
\aligned
\frac{d\big(|\varphi(t;t_0, x_0)| - r\big)}{dt} & =  \frac{1+t^2(t+r)^{-2} \Hu^{00}}{1-t^2(t+r)^{-2} \Hu^{00}} - \frac{t}{r}
 = \frac{2t^2(t+r)^{-2} \Hu^{00}}{1-t^2(t+r)^{-2} \Hu^{00}} - \frac{t-r}{r}.
\endaligned
$$
So, there exists a constant $\vep_s$ such that if $|\Hu^{00}| \leq \frac{\vep_s(t-r)}{t}$, then 
$
\frac{d\big(|\varphi(t;t_0, x_0)| - r\big)}{dt}< 0.
$
Recall that at $t=t_0$, $|\varphi(t_0;t_0, x_0)| = |x_0| = r(t_0)$. We conclude that for all $t< t_0$, $|\varphi(t;t_0, x_0)|> r(t)$ which shows that $(t, \varphi(t;t_0, x_0))$ will never intersect $\Hcal_{s_0}$ again. Furthermore we see that there exists a time $a_0$ sufficiently small (but still $a_0\geq 3$) such that $(t, \varphi(t;t_0, x_0))$ leaves $\Kcal_{[2,s]} \backslash \Kical$ by intersecting the boundary $\del_B\Kical_{[2,s_0]} \cup \del\Kcal$ at $t=a_0$. So we see that
$
v_{t_0, x_0}(t_0) = v_{t_0, x_0}(a_0) + \int_{a_0}^{t_0} \mathcal{F}(t, \varphi(t;t_0, x_0)) \, dt, 
$
which leads us to
$$
\aligned
|v_{t_0, x_0}(t_0)|
& \leq  \sup_{(t, x) \in\del_B\Kical_{[2,s_0]} \cup \del\Kcal} \{|(\del_t- \del_r)(ru)|_{(t, x)}|\}
\\
& \quad +  \int_2^{t_0} \big|-r\Boxt_g u + r\sum_{a<b} \big(r\Omega_{ab} \big)^2 u + \Hu^{00}W_1[u] + rR[u,H]\big|_{(t, \varphi(t;t_0, x_0))} \, dt.
\endaligned
$$ 
\end{proof}


\subsection{Sup-norm estimate for wave equations with source}

Our sup-norm estimate for the wave equation, established earlier in \cite{PLF-YM-one} and based on an explicit formula for solutions (cf. also the Appendix at the end of this monograph), is now revisited and adapted to the problem of compact Schwarzschild perturbations. By applying $\del^IL^J$ to  the Einstein equations \eqref{eq main PDE a}, we obtain 
\bel{eq:SJO}
\aligned
\Box \del^IL^J h_{\alpha \beta} 
& = - \del^IL^J\big(H^{\mu\nu} \del_{\mu} \del_{\nu}h_{\alpha \beta} \big) + \del^IL^JF_{\alpha \beta} -  16\pi\del^IL^J\big(\del_{\alpha} \phi\del_{\beta} \phi\big) - 8\pi c^2\del^IL^J\big(\phi^2 g_{\alpha\beta} \big)
\\
& =: S_{\alpha \beta}^{I,J}  = S^{W,I,J}_{\alpha \beta} + S^{KG,I,J}_{\alpha \beta}, 
\endaligned
\ee
with 
$$
\aligned
S^{W,I,J}_{\alpha \beta} 
& := - \del^IL^J\big(H^{\mu\nu} \del_{\mu} \del_{\nu}h_{\alpha \beta} \big) + \del^IL^JF_{\alpha \beta},
\\
S^{KG,I,J}_{\alpha \beta} 
& := -  16\pi\del^IL^J\big(\del_{\alpha} \phi\del_{\beta} \phi\big) - 8\pi c^2\del^IL^J\big(\phi^2 g_{\alpha\beta} \big).
\endaligned
$$
We denote by $\mathds{1}_{\Kcal}: \RR^4 \to \{0,1\}$ the characteristic function of the set $\Kcal$,
and  introduce the corresponding decomposition into interior/exterior contributions of the wave source of the Einstein equations:  
$$
S^{W,I,J}_{\text{Int}, \alpha \beta} := \mathds{1}_{\Kcal}  S^{W,I,J}_{\alpha \beta}, 
\qquad 
S^{W,I,J}_{\text{Ext}, \alpha \beta}  := (1- \mathds{1}_{\Kcal})S^{W,I,J}_{\alpha \beta}, 
$$
while $S^{KG,I,J}_{\alpha \beta}$ is compactly supported in $\Kcal$ and need not be decomposed. We thus have 
\be
S_{\alpha \beta}^{I,J} = S^{W,I,J}_{\text{Ext}, \alpha \beta}  + S^{KG,I,J}_{\alpha \beta} + S^{W,I,J}_{\text{Int}, \alpha \beta} .
\ee
Outside the region $\Kcal$, the metric $g_{\alpha \beta}$ coincides with the Schwarzschild metric so that an easy calculation leads us to the following estimate. 

\begin{lemma} \label{lem 1 sup-norm-W}
One has
$
|S^{W,I,J}_{\text{Ext}, \alpha \beta} | \leq Cm_S^2(1- \mathds{1}_\Kcal)r^{-4}. 
$
\end{lemma}

We next decompose the initial data for the equations \eqref{eq:SJO}. Recall that on the initial hypersurface $\{t=2\}$ and outside the unit ball, the metric  coincides with the Schwarzschild metric. We write 
$$
\aligned 
\del^IL^J h_{\alpha \beta}(2, \cdot) &:= I_{\text{Int}, \alpha}^{0,I,J}  + I_{\text{Ext}, \alpha \beta}^{0,I,J}, 
\\
I_{\text{Int}, \alpha}^{0,I,J}  &:= \widetilde{\chi}(r) \del^IL^J h_{\alpha \beta}(2, \cdot),
\qquad 
I_{\text{Ext}, \alpha \beta}^{0,I,J} := (1- \widetilde{\chi}(r)) \del^IL^J h_{\alpha \beta}(2, \cdot), 
\endaligned
$$
in which $\widetilde{\chi}(\cdot):\RR^+  \to \RR^+$ is a smooth cut-off function with 
$$
\widetilde{\chi}(r) =
\left\{
\aligned
&1, \quad r\leq 1, 
\\
&0, \quad r\geq 2. 
\endaligned
\right.
$$
On the other hand, the initial data $\del_t\del^IL^Jh_{\alpha \beta}(2, \cdot) =: I^1[\del^IL^J]$ is supported in $\{r\leq 1\}$ since the metric is initially static outside the unit ball. We are in a position to state our main sup-norm estimate. 

\begin{proposition}[Sup-norm estimate for the Einstein equations]
\label{prop 1 sup-norm-W}
Let $(g_{\alpha \beta}, \phi)$ be a solution of the Einstein-massive field system associated with a compact Schwarzschild initial data. 
Assume that the source terms in \eqref{eq:SJO} satisfy 
\be
|S_{\text{Int}, \alpha \beta}^{W,I,J}| + |S_{\alpha \beta}^{KG,I,J}| \leq C_*  t^{-2- \nu}(t-r)^{-1+\mu}.
\ee
Then, when $0<\mu\leq 1/2$ and $0<\nu\leq 1/2$, one has 
\be
|\del^IL^Jh_{\alpha \beta}(t, x)| \leq \frac{CC_* (\alpha, \beta)}{\mu|\nu|}t^{-1}(t-r)^{\mu- \nu} + C m_St^{-1}, 
\ee
while, when $0<\mu\leq 1/2$ and $-1/2\leq \nu<0$, 
\be
|\del^IL^Jh_{\alpha \beta}(t, x)| \leq \frac{CC_* (\alpha, \beta)}{\mu|\nu|}t^{-1- \nu}(t-r)^{\mu} + C m_St^{-1}.
\ee
\end{proposition}

For the proof of this result, we will rely on the decomposition $\del^IL^Jh_{\alpha \beta} = \sum_{k=1}^5 h_{\alpha \beta}^{IJ,k}$ with 
\begin{subequations} 
\be
\aligned
&\Box h_{\alpha \beta}^{IJ,1} = S^{W,I,J}_{\text{Int}, \alpha \beta} ,
\qquad h_{\alpha \beta}^{IJ,1}(2, \cdot) = 0, 
\quad \del_t h_{\alpha \beta}^{IJ,1}(2, \cdot) = 0,
\endaligned
\ee
\be
\aligned
&\Box h_{\alpha \beta}^{IJ,2} = S^{KG,I,J}_{\alpha \beta},
\qquad
h_{\alpha \beta}^{IJ,2}(2, \cdot) = 0,
\quad \del_t h_{\alpha \beta}^{IJ,2}(2, \cdot) = 0,
\endaligned
\ee
\be
\aligned
&\Box h_{\alpha \beta}^{IJ,3} = S^{W,I,J}_{\text{Ext}, \alpha \beta},
\qquad
h_{\alpha \beta}^{IJ,3}(2, \cdot) = 0, 
\quad \del_t h_{\alpha \beta}^{IJ,3}(2, \cdot) = 0,
\endaligned
\ee
\be
\aligned
&
\Box h_{\alpha \beta}^{IJ,4} = 0,
\qquad \quad
h_{\alpha \beta}^{IJ,4}(2, \cdot) = I^{0,I,J}_{\text{Int}, \alpha \beta}, 
\qquad
\del_th_{\alpha \beta}^{IJ,4}(2, \cdot) = I^{1,I,J}_{\alpha \beta},
\endaligned
\ee
\be
\aligned
&\Box h_{\alpha \beta}^{IJ,5} = 0,
\qquad
h_{\alpha \beta}^{IJ,5}(2, \cdot) = I^{0,I,J}_{\text{Ext}, \alpha \beta}, 
\quad 
\del_th_{\alpha \beta}^{IJ,5}(2, \cdot) =0. 
\endaligned
\ee
\end{subequations} 
The proof of Proposition \ref{prop 1 sup-norm-W} is immediate once we control each term. 

First of all, the estimates for $h_{\alpha \beta}^{IJ,1}$ and $h_{\alpha \beta}^{IJ,2}$ are immediate from Proposition~3.1 in \cite{PLF-YM-one}, since they concern compactly supported sources. The control of $h_{\alpha \beta}^{IJ,4}$ is standard for the homogeneous wave equation with {\sl compact} initial data. 

\begin{lemma} 
\label{lem 2 sup-norm-W}
The metric coefficients satisfy the inequality
\bel{ineq lem 2 sup-norm-W}
|h_{\alpha \beta}^{IJ,4}(t, x)| 
\leq 
C t^{-1} \Big(
\|\del^IL^J h_{\alpha \beta} (2, \cdot) \|_{W^{1,\infty}(\{r\leq 1\})} 
+ \| \del_t  \del^IL^J h_{\alpha \beta} (2, \cdot)  \|_{L^{\infty}(\{r\leq 1\})} 
\Big) \mathds{1}_{\{|t+2-r| \leq 1\}}(t, x).
\ee
\end{lemma}

We thus need to study the behavior of $h_{\alpha \beta}^{IJ,3}$ and $h_{\alpha \beta}^{IJ,5}$. 
We treat first the function $h_{\alpha \beta}^{IJ,5}$ and observe that 
\bel{eq 1 prop 1 sup-norm-W}
\aligned
& h_{\alpha \beta}^{IJ,5}(t, x)
\\
& =  \frac{1}{4\pi (t-2)^2} \int_{|y-x| = t-2} \left(I_{\text{Ext}, \alpha \beta}^{0,I,J}(y) - \langle \nabla I_{\text{Ext}, \alpha \beta}^{0,I,J}(y), x-y\rangle \right) d\sigma(y)
\\
& =  \frac{1}{4\pi (t-2)^2} \int_{|y-x| = t-2}I_{\text{Ext}, \alpha \beta}^{0,I,J}(y) d\sigma(y)
    -  \frac{1}{4\pi (t-2)^2} \int_{|y-x|=t-2} \langle \nabla I_{\text{Ext}, \alpha \beta}^{0,I,J}(y), x-y\rangle d\sigma(y). 
\endaligned
\ee
We now estimate  the two integral terms successively. 

\begin{lemma} \label{lem 3 sup-norm-W}
One has
$
\left|\int_{|y-x| = t}I_{\text{Ext}, \alpha \beta}^{0,I,J}(y) d\sigma(y) \right| \leq Cm_St.
$
\end{lemma}

\begin{proof} Since $g_{\alpha \beta}$ coincides with the Schwarzschild metric outside $\{r\geq 1\}$, we have immediately 
$|I_{\text{Ext}, \alpha \beta}^{0,I,J}| \leq Cm_S(1+r)^{-1}$ 
and thus
\bel{ineq 1 proof lem 3 sup-norm-W}
\left|\int_{|y-x| = t}I_{\text{Ext}, \alpha \beta}^{0,I,J}(y) d\sigma(y) \right| \leq Cm_S \int_{|y-x|=t} \frac{d\sigma(y)}{1+|y|}
=: Cm_S \, \Theta(t,x).
\ee
Assume that $r>0$ and, without loss of generality, $x = (r,0,0)$. Introduce the parametrization of the sphere $\{|y-x|=t\}$ such that: 
\begin{itemize}

\item $\theta \in[0, \pi]$ is the angle from $(-1,0,0)$ to $y-x$. 

\item $\varphi\in[0,2\pi)$ is the angle from the plane determined by $(1,0,0)$ and $(0,1,0)$ to the plane determined by $y-x$ and $(1,0,0)$.
\end{itemize}
With this parametrization, $d\sigma(y) = t^2\sin\theta d\theta d\varphi$ and the above integral reads 
$$
\Theta(t,x)
= 
\int_{|y-x|=t} \frac{d\sigma(y)}{1+|y|}
= t^2\int_0^{2\pi} \int_0^{\pi} \frac{\sin\theta d\theta d\varphi}{1+t\big(1+(r/t)^2 - (2r/t) \cos\theta \big)^{1/2}}, 
$$
where the law of cosines was applied to $|y|$. Then, we have
$$
\aligned
\Theta(t,x) 
& = 2\pi t^2\int_0^\pi \frac{\sin\theta d\theta}{1+t\big(1+(r/t)^2 - (2r/t) \cos\theta \big)^{1/2}} 
\\
& = 2\pi t^2\int_{-1}^1\frac{d\sigma}{1+t|1+(r/t)^2 - (2r/t) \sigma|^{1/2}},
\endaligned
$$
with the change of variable $\sigma:= \cos \theta$, so that $\lambda: = t|1+(r/t)^2 - (2r/t) \sigma|^{1/2}$ and 
$$
\aligned
\Theta(t,x)
& =  2\pi tr^{-1} \int_{t-r}^{t+r} \frac{\lambda d\lambda}{1+\lambda}
=4\pi t - 2\pi tr^{-1} \ln\left(\frac{t+r+1}{t-r+1} \right). 
\\
\endaligned
$$
The second term is bounded by the following observation. When $r\geq t/2$, this term is bounded by $\ln(t+1)$. When $r\leq t/2$, according to the mean value theorem, there exists $\xi$ such that 
$$
r^{-1} \ln\left(\frac{t+r+1}{t-r+1} \right) = 2 \, \frac{\left(\ln(1+t+r) - \ln(1+t-r) \right)}{2r}  = \frac{2}{1+t+\xi}.
$$  
By recalling $r \leq t/2$, we deduce that 
$
\left|r^{-1} \ln\left(\frac{t+r+1}{t-r+1} \right) \right| \leq \frac{C}{1+t} 
$
and we conclude that the first term in the right-hand side of \eqref{ineq 1 proof lem 3 sup-norm-W} is bounded by
$$
 Cm_S \int_{|y-x|=t} \frac{d\sigma(y)}{1+|y|} \leq Cm_St.
$$
We also observe that, when $r=0$, we have 
$\int_{|y|=t} \frac{d\sigma(y)}{1+|y|} = \frac{4\pi t^2}{1+t}$
and thus 
$
 Cm_S \int_{|y-x|=t} \frac{d\sigma(y)}{1+|y|} \leq Cm_St.
$
\end{proof}

The proof of the following lemma is similar to the one abve and we omit the proof. 

\begin{lemma} \label{lem 4 sup-norm-W} One has 
$$
\left|\int_{|y-x|=t} \langle \nabla I_{\text{Ext}, \alpha \beta}^{0,I,J}(y), x-y\rangle d\sigma(y) \right| \leq Cm_S t.
$$
\end{lemma}

From the above two lemmas, we conclude that $\big|h_{\alpha \beta}^{IJ,5}(t, x) \big| \leq Cm_St^{-1}$
as expected, and we can finally turn our attention to the last term $h_{\alpha \beta}^{IJ,3}$. 

\begin{lemma} One has 
$
|h_{\alpha \beta}^{IJ,3}(t, x)| \leq Cm_S^2 t^{-1}.
$
\end{lemma}

\begin{proof} This estimate is based on Lemma~\ref{lem 1 sup-norm-W} and on the explicit formula 
$$
h_{\alpha \beta}^{IJ,3}(t, x) = \frac{1}{4\pi} \int_2^t\frac{1}{t-s} \int_{|y|=t-s} S^{W,I,J}_{\text{Ext}, \alpha \beta} d\sigma(y)ds, 
$$
which yields us 
$$
\aligned
|h_{\alpha \beta}^{IJ,3}(t, x)| 
& \leq Cm_S^2\int_2^t\frac{1}{t-s} \int_{|y|=t-s} \frac{\mathds{1}_{\{|x-y|\geq s-1\}}d\sigma}{|x-y|^4}ds
\\
& = Cm_S^2 t^{-2} \int_{2/t}^1\frac{1}{1- \lambda} \int_{|y|=1- \lambda} \frac{\mathds{1}_{\{|y-x/t|\geq \lambda -1/t\}}d\sigma}{|y-x/t|^4}d\lambda
\endaligned
$$
thanks to the change of variable $\lambda: = s/t$. 
Without loss of generality, we set $x = (r,0,0)$ and introduce the following parametrization of the sphere $\{|y| = 1- \lambda \}$:
\begin{itemize}

\item $\theta$ denotes the angle from $(1,0,0)$ to $y$. 

\item $\varphi$ denotes the angle from the plane determined by $(1,0,0)$ and $(0,1,0)$ to the plane determined  by $(1,0,0)$ and $y$.
\end{itemize}
We have $d\sigma(y) = (1- \lambda)^2\sin\theta d\theta d\varphi$ and we must evaluate the integral
$$
\aligned
|h_{\alpha \beta}^{IJ,3}(t, x)|
& \leq  Cm_S^2t^{-2} \int_{2/t}^1\frac{d\lambda}{1- \lambda} \int_0^{2\pi} \int_0^\pi\frac{\mathds{1}_{\{|y-x/t|\geq \lambda - 1/t\}}(1- \lambda)^2\sin\theta d\theta d\varphi}{|(r/t)^2 + (1- \lambda)^2 - 2(r/t)(1- \lambda) \cos\theta|^2}
\\
& \leq  Cm_S^2t^{-2} \int_{2/t}^1\frac{d\lambda}{1- \lambda} \int_0^\pi\frac{\mathds{1}_{\{|y-x/t|\geq \lambda - 1/t\}}(1- \lambda)^2\sin\theta d\theta}{|(r/t)^2 + (1- \lambda)^2 - 2(r/t)(1- \lambda) \cos\theta|^2}.
\endaligned
$$

Consider the integral expression 
$$
\aligned
I(\lambda)
:=& \int_0^\pi\frac{\mathds{1}_{\{|y-x/t|\geq \lambda - 1/t\}}(1- \lambda)^2\sin\theta d\theta}{|(r/t)^2 + (1- \lambda)^2 - 2(r/t)(1- \lambda) \cos\theta|^2}
\\
=&
(1- \lambda)tr^{-1} \int^{1- \lambda+r/t}_{|1- \lambda - r/t|} \frac{\mathds{1}_{\{\tau\geq \lambda - 1/t\}}d\tau}{\tau^3}, 
\endaligned
$$
where we used the change of variable $\tau : = |(r/t)^2 + (1- \lambda)^2 - 2(r/t)(1- \lambda) \cos\theta|^{1/2}$. 
We see that when $1- \lambda + r/t \leq \lambda - 1/t$, $I(\lambda) =0$. We only need to discuss the case $1- \lambda+r/t\geq \lambda-1/t$ which is equivalent to $\lambda \leq \frac{t+r+1}{2t}$. We distinguish between the following cases: 

$\bullet$ Case $1\leq t-r\leq 3$. In this case, when $\lambda \in[2/t,(t+r+1)/2t]$, we observe that $|1- \lambda-r/t| \leq \lambda - 1/t$. Then, we find 
    $
    I(\lambda) = (1- \lambda)tr^{-1} \int^{1- \lambda+r/t}_{\lambda-1/t} \frac{\mathds{1}_{\{\tau\geq \lambda - 1/t\}}d\tau}{\tau^3},
    $
    which leads us to
$$
\aligned
    I(\lambda) & =  (1- \lambda)tr^{-1} \int^{1- \lambda+r/t}_{\lambda-1/t} \frac{d\tau}{\tau^3}
 =  \frac{t(1- \lambda)}{ 2 r} \left((\lambda-1/t)^{-2} - (1- \lambda+r/t)^{-2} \right). 
\endaligned
$$
    Then we conclude that
$$
\aligned
    |h_{\alpha \beta}^{IJ,3}(t, x)| & \leq  Cm_S^2t^{-2} \int_{2/t}^{(t+r+1)/2t}(1- \lambda)^{-1}I(\lambda)d\lambda
\\
& =  Cm_S^2 r^{-1}t^{-1} \int_{2/t}^{(t+r+1)/2t} \left((\lambda-1/t)^{-2} -(1- \lambda + r/t)^{-2} \right)d\lambda
\\
& =  Cm_S^2 r^{-1} \left(1- \frac{1}{t+r-2} \right) \leq Cm_S^2t^{-1}.
\endaligned
$$

$\bullet$ Case $t-r> 3$ and $\frac{t-r}{t} \leq \frac{t+r+1}{2t} \Leftrightarrow r\geq \frac{t-1}{3}$. In this case the interval $\left[2/t, \frac{t+r+1}{2t} \right]$ is divided into two parts: $\left[2/t, \frac{t-r}{t} \right]\cup [\frac{t-r}{t}, \frac{t+r+1}{2t}]$. In the first subinterval, $|1- \lambda - r/t| = 1- \lambda-r/t$ while in the second $|1- \lambda-r/t| = r/t-1+\lambda$

    Again in the subinterval $\left[2/t, \frac{t-r}{t} \right]$, we see that when $2/t\leq \lambda \leq \frac{t-r+1}{2t}$, $\lambda-1/t\leq 1- \lambda-r/t$, when $\frac{t-r+1}{2t} \leq \lambda \leq \frac{t-r}{t}$, $\lambda-1/t\geq 1- \lambda-r/t$. In the subinterval $[\frac{t-r}{t}, \frac{t+r+1}{2t}]$, we see that $\lambda-1/t\geq  r/t-1+\lambda$.

\

    \noindent {\bf Case 1}. When $\lambda \in \left[2/t, \frac{t-r+1}{2t} \right]$, we have 
$$
    I(\lambda) = (1- \lambda)tr^{-1} \int_{1- \lambda-r/t}^{1- \lambda+r/t} \frac{d\tau}{\tau^3}
    = \frac{2(1- \lambda)^2}{\left((1- \lambda)^2-(r/t)^2\right)^2}. 
$$
    
    \noindent {\bf Case 2}. When $\lambda \in\left[\frac{t-r+1}{2t}, \frac{t-r}{t} \right]$, we have 
$$
    I(\lambda) = (1- \lambda)tr^{-1} \int_{\lambda-1/t}^{1- \lambda+r/t} \frac{d\tau}{\tau^3}
    = \frac{t(1- \lambda)}{2r} \left((\lambda-1/t)^{-2} - (1- \lambda + r/t)^{-2} \right).
$$
\\
    \noindent {\bf Case 3}. When $\lambda \in[\frac{t-r}{t}, \frac{t+r+1}{2t}]$, we have 
$$
    I(\lambda) = (1- \lambda)tr^{-1} \int_{\lambda-1/t}^{1- \lambda+r/t} \frac{d\tau}{\tau^3}
    = \frac{t(1- \lambda)}{2r} \left((\lambda-1/t)^{-2} - (1- \lambda + r/t)^{-2} \right).
$$
    We obtain 
$$
\aligned
&    |h_{\alpha \beta}^{IJ,3}(t, x)| \leq  Cm_S^2t^{-2} \int_{2/t}^{(t+r+1)/2t}(1- \lambda)^{-1}I(\lambda)d\lambda \
\\
& = Cm_S^2t^{-2} \int_{2/t}^{\frac{t-r+1}{2t}}+ \int_{\frac{t-r+1}{2t}}^{\frac{t+r+1}{2t}} (1- \lambda)^{-1}I(\lambda)d\lambda
    = Cm_S^2t^{-2} \int_{2/t}^{\frac{t-r+1}{2t}} \frac{2(1- \lambda)}{\left((1- \lambda)^2-(r/t)^2\right)^2}d\lambda
\\
& \quad +  Cm_S^2r^{-1}t^{-1} \int_{\frac{t-r+1}{2t}}^{\frac{t+r+1}{2t}} \left((\lambda-1/t)^{-2} - (1- \lambda + r/t)^{-2} \right)d\lambda
\endaligned
$$
and we observe that 
$$
    \int_{2/t}^{\frac{t-r+1}{2t}} \frac{(1- \lambda)d\lambda}{\left((1- \lambda)^2-(r/t)^2\right)^2} = \frac{2t^2}{(t-r-1)(t+3r-1)} - \frac{t^2}{2(t-r-2)(t+r-2)} \simeq Ct
$$
    and
$$
\aligned
    \int_{\frac{t-r+1}{2t}}^{\frac{t+r+1}{2t}} \left((\lambda-1/t)^{-2} - (1- \lambda + r/t)^{-2} \right)d\lambda
& =  \frac{4rt}{(t-r-1)(t+r-1)} - \frac{4tr}{(t+r-1)(t+3r-1)}
\\
&\simeq Cr.
\endaligned
$$
    We conclude that
    $
     |h_{\alpha \beta}^{IJ,3}(t, x)| \leq  Cm_S^2t^{-1}.
    $

$\bullet$ Case $1-r/t\geq \frac{t+r+1}{2t} \Leftrightarrow r\leq \frac{t-1}{3}$. In this case, for $\lambda \in \left[2/t, \frac{t+r+1}{2t} \right]$, $|1- \lambda - r/t| = 1- \lambda - r/t$. We also observe that when $2/t\leq \lambda \leq \frac{t-r+1}{2t}$, $|1- \lambda-r/t|\geq\lambda -1/t$ and when $\frac{t-r+1}{2t} \leq \lambda \leq \frac{t+r+1}{2t}$, $|1- \lambda-r/t| \leq\lambda -1/t$. So, similarly to the above case, we find 
$$
\aligned
    |h_{\alpha \beta}^{IJ,3}(t, x)| 
& \leq  Cm_S^2t^{-2} \int_{2/t}^{(t+r+1)/2t}(1- \lambda)^{-1}I(\lambda)d\lambda \
     = Cm_S^2t^{-2} \int_{2/t}^{\frac{t-r+1}{2t}}+ \int_{\frac{t-r+1}{2t}}^{\frac{t+r+1}{2t}} (1- \lambda)^{-1}I(\lambda)d\lambda
\\
& = Cm_S^2t^{-2} \int_{2/t}^{\frac{t-r+1}{2t}} \frac{(1- \lambda)}{\left((1- \lambda)^2-(r/t)^2\right)^2}d\lambda
\\
& \quad +  Cm_S^2r^{-1}t^{-1} \int_{\frac{t-r+1}{2t}}^{\frac{t+r+1}{2t}} \left((\lambda-1/t)^{-2} - (1- \lambda + r/t)^{-2} \right)d\lambda,
\endaligned
$$ 
$$
    \int_{2/t}^{\frac{t-r+1}{2t}} \frac{(1- \lambda)d\lambda}{\left((1- \lambda)^2-(r/t)^2\right)^2} = \frac{2t^2}{(t-r-1)(t+3r-1)} - \frac{t^2}{2(t-r-2)(t+r-2)} \simeq C,
$$
    and
$$
\aligned
  &  \int_{\frac{t-r+1}{2t}}^{\frac{t+r+1}{2t}} \left((\lambda-1/t)^{-2} - (1- \lambda + r/t)^{-2} \right)d\lambda
  \\  & =  \frac{4rt}{(t-r-1)(t+r-1)} - \frac{4tr}{(t+r-1)(t+3r-1)}
    \simeq C.
\endaligned
$$
    So, we obtain
    $
     |h_{\alpha \beta}^{IJ,3}(t, x)| \leq Cm_S^2t^{-1}$, which completes the proof. 
\end{proof}


\subsection{Sup-norm estimate for Klein-Gordon equations}
\label{subsec KG-sup}

Our next statement, first presented in \cite{PLF-YM-one}, was motivated by a pioneering work by Klainerman \cite{Klainerman85} for Klein-Gordon equations. In more recent years, Katayama \cite{Katayama12a,Katayama12b} also made some important contribution on the global existence problem for Klein-Gordon eqations. Furthermore, a related estimate in two spatial dimensions in Minkowski spacetime was established earlier by Delort  et al. \cite{Delort04}. (Our approach below could also be applied \cite{Ma} in $2+1$ dimensions.)

For compact Schwarz\-schild perturbations, the scalar field $\phi$ is supported in $\Kcal$,  and the sup-norm estimate in \cite{PLF-YM-one} remains valid for our purpose and we only need to state the corresponding result. Namely, let us consider the Klein-Gordon problem on a curved spacetime
\bel{Linfini KG eq}
\aligned
& \quad -  \Boxt_g v + c^2 v = f,
\qquad 
&v|_{\Hcal_2} = v_0,
\quad
\del_t v|_{\Hcal_2} = v_1,
\endaligned
\ee
with initial data $v_0, v_1$ which are prescribed on the hyperboloid $\Hcal_2$ and are assumed to be compactly supported in $\Hcal_2\cap \Kcal$, while the curved metric has the form $g^{\alpha \beta} = m^{\alpha \beta} + h^{\alpha \beta}$
with  $\sup |\hb^{00}| \leq 1/3$.

We consider the coefficient $\hb^{00}$ along lines from the origin and, more precisely, we set 
$$
h_{t, x}(\lambda):= \hb^{00}\Big(\lambda {t \over s}, \lambda {x \over s} \Big), \qquad s = \sqrt{t^2 - r^2}, 
$$
while $h_{t, x}'(\lambda)$ stands for the derivative with respect to the variable $\lambda$. We also set 
$$
s_0 :=\left
\{
\aligned
& 2, \quad &&0\leq r/t \leq 3/5,
\\
& \sqrt{\frac{t+r}{t-r}}, \quad &&3/5\leq r/t\leq 1,
\endaligned
\right.
$$
Fixing some constant $C>0$, we introduce the following function $V$ by distinguishing between the regions ``near" and ``far" from the light cone: 
$$
V:=
\left\{
\aligned
& \Big( \|v_0\|_{L^\infty(\Hcal_2)} + \|v_1\|_{L^\infty(\Hcal_2)} \Big)
\Big(1+\int_2^s|h_{t, x}'(\sbar)|e^{C\int_\sbar^s|h_{t, x}'(\lambda)|d\lambda} \, d\sbar \Big)
\\
& \hskip4.cm  + F(s) + \int_2^s F(\sbar)|h_{t, x}'(\sbar)|e^{C\int_\sbar^s|h_{t, x}'(\lambda)|d\lambda} \, d\sbar,
\hskip.cm  && 0\leq r/t\leq 3/5,
\\
& 
F(s) + \int_{s_0}^s F(\sbar)|h_{t, x}'(\sbar)|e^{C\int_\sbar^s|h_{t, x}'(\lambda)|d\lambda} \, d\sbar,
\hskip2.cm  &&3/5<r/t<1,
\endaligned
\right.
$$ 
where the function $F$ takes the right-hand side of the Klein-Gordon equation into account, as well as the curved part of the metric (except the $\hb^{00}$ contribution), that is, 
$$
F(\sbar):
= \int_{s_0}^\sbar \Big( |R_1[v]| + |R_2[v]| + |R_3[v]| + \lambda^{3/2} |f| \Big) \, (\lambda t/s, \lambda x/s)  \, d\lambda
$$
with
$$
\aligned
R_1[v] &= s^{3/2} \sum_a \delb_a \delb_a v  + \frac{x^ax^b}{s^{1/2}} \delb_a \delb_b v + \frac{3}{4s^{1/2}} v + \sum_a \frac{3x^a}{s^{1/2}} \delb_a v,
\\
R_2[v] &=\hb^{00} \bigg(\frac{3v}{4s^{1/2}} + 3s^{1/2} \delb_0 v\bigg)
+ s^{3/2} \big(2\hb^{0b} \delb_0\delb_bv + \hb^{ab} \delb_a \delb_bv + h^{\alpha \beta} \del_\alpha \Psib^{\beta'}_{\beta} \, \delb_{\beta'}v\big),
\\
R_3[v] &= \hb^{00} \bigg(2x^as^{1/2} \delb_0\delb_a v + \frac{2x^a}{s^{1/2}} \delb_a v +\frac{x^ax^b}{s^{1/2}} \delb_a \delb_bv\bigg).
\endaligned
$$ 

\begin{proposition}[A sup-norm estimate for Klein-Gordon equations on a curved spacetime]
\label{Linfini KG}
Spatially compact solutions $v$ to the Klein-Gordon problem \eqref{Linfini KG eq}
defined the region $\Kcal_{[2, +\infty)}$ satisfy the decay estimate (for all relevant $(t, x)$)
\label{Linfty KG ineq}
\bel{Linfty KG ineq a}
s^{3/2}|v(t, x)| + (s/t)^{-1}s^{3/2}|\delu_{\perp} \ v(t, x)| \leq C \, V(t, x). 
\ee 
\end{proposition}

We postpone the proof to the Appendix.  


\subsection{Weighted Hardy inequality along the hyperboloidal foliation}

We now derive a modified version of the Hardy inequality, formulated on hyperboloids, which is nothing but a weighted version of Proposition 5.3.1 in \cite{PLF-YM-book}. This inequality will play an essential role in our derivation of a key $L^2$ estimate for the metric component $\hu^{00}$. (Cf. Section \ref{subsubsec L-2-h00}, below.)

\begin{proposition}[Weighted Hardy inequality on hyperboloids]\label{prop Hardy 2}
For every smooth function $u$ supported in the cone $\Kcal$, one has 
(for any given $0\leq \sigma \leq 1$):
\bel{ineq Hardy 2}
\aligned
\|(s/t)^{- \sigma}s^{-1}u\|_{L_f^2(\Hcal_s)} & \leq  C\|(s_0/t)^{- \sigma}s_0^{-1}u\|_{L^2(\Hcal_{s_0})} + C\sum_a \|\delu_a u\|_{L_f^2(\Hcal_s)}
\\
& \quad + C\sum_a \int_{s_0}^s\tau^{-1} 
\Big(\|(s/t)^{1- \sigma} \del_au\|_{L^2(\Hcal_\tau)} + \|\delu_a u\|_{L^2(\Hcal_\tau)} \Big) \, d\tau.
\endaligned
\ee
\end{proposition}

The proof is similar to that of Proposition 5.3.1 in \cite{PLF-YM-book} (but we must now cope with the parameter $\sigma$) and uses the following inequality, established in \cite[Chapter 5, Lemma 5.3.1]{PLF-YM-book}.

\begin{lemma} \label{lem Hardy 1}
For all (sufficiently regular) functions $u$ supported in the cone $\Kcal$, one has 
\bel{ineq Hardy 3}
\|r^{-1}u\|_{L_f^2(\Hcal_s)} \leq C\sum_a \|\delu_a u\|_{L_f^2(\Hcal_s)}.
\ee
\end{lemma}

\begin{proof}[Proof of Proposition \ref{prop Hardy 2}]
Consider the vector field 
$
W:= \big(0,-(s/t)^{-2\sigma} \frac{x^at u^2 \chi(r/t)^2}{(1+r^2)s^2} \big)
$ defined on $\RR^4$ 
and, similarly to what we did in the proof of Proposition 5.3.1 in \cite{PLF-YM-book}, let us calculate its divergence:  
$$
\aligned
\text{div} \hskip.1cm W
& =  -2s^{-1} (s/t)^{- \sigma} \sum_a \del_a u  (s/t)^{- \sigma} \frac{r\chi(r/t) u}{(1+r^2)^{1/2}s}  \frac{x^at\chi(r/t)}{r(1+r^2)^{1/2}}
\\
& \quad - 2 s^{-1} (s/t)^{- \sigma}r^{-1}u  (s/t)^{- \sigma} \frac{r\chi(r/t) u}{(1+r^2)^{1/2}s}   \frac{\chi'(r/t)r}{(1+r^2)^{1/2}}
\\
& \quad - (s/t)^{-2\sigma} \big(u\chi(r/t) \big)^2\bigg(\frac{r^2t + 3t}{(1+r^2)^2s^2} + \frac{2r^2t}{(1+r^2)s^4} \bigg)
 \\
& \quad - 2\sigma(s/t)^{-1-2\sigma} \big(u\chi(r/t) \big)^2\frac{r^2}{(1+r^2)s^3}.
\endaligned
$$
We integrate this identity within $\Kcal_{[s_0,s_1]}$ and, after recalling the relation $dxdt = (s/t) \, dxds$, we obtain 
$$
\aligned
\int_{\Kcal_{[s_0,s_1]}} \text{div} \hskip.1cm W dxdt
& =  -2\int_{\Kcal_{[s_0,s_1]}} s^{-1} (s/t)^{1- \sigma} \sum_a \del_a u  (s/t)^{- \sigma} \frac{r\chi(r/t) u}{(1+r^2)^{1/2}s}  \frac{x^at\chi(r/t)}{r(1+r^2)^{1/2}} \, dxds
\\
& \quad - 2\int_{\Kcal_{[s_0,s_1]}} s^{-1} (s/t)^{1- \sigma}r^{-1}u  (s/t)^{- \sigma} \frac{r\chi(r/t) u}{(1+r^2)^{1/2}s}   \frac{\chi'(r/t)r}{(1+r^2)^{1/2}} \, dxds
\\
& \quad -  \int_{\Kcal_{[s_0,s_1]}}(s/t)^{1-2\sigma} \big(u\chi(r/t) \big)^2\bigg(\frac{r^2t + 3t}{(1+r^2)^2s^2} + \frac{2r^2t}{(1+r^2)s^4} \bigg) \, dxds
\\
& \quad - 2\sigma \int_{\Kcal_{[s_0,s_1]}}(s/t)^{-2\sigma} \big(u\chi(r/t) \big)^2\frac{r^2}{(1+r^2)s^3} \, dxds.
\endaligned
$$
We thus find 
$$
\aligned
\int_{\Kcal_{[s_0,s_1]}} \text{div} \hskip.1cm W dxdt
& =  -2\int_{s_0}^{s_1} ds\int_{\Hcal_s}s^{-1} (s/t)^{1- \sigma} \sum_a \del_a u  (s/t)^{- \sigma} \frac{r\chi(r/t) u}{(1+r^2)^{1/2}s}  \frac{x^at\chi(r/t)}{r(1+r^2)^{1/2}} \, dx
\\
& \quad - 2\int_{s_0}^{s_1} ds\int_{\Hcal_s} s^{-1} (s/t)^{1- \sigma}r^{-1}u  (s/t)^{- \sigma} \frac{r\chi(r/t) u}{(1+r^2)^{1/2}s}   \frac{\chi'(r/t)r}{(1+r^2)^{1/2}} \, dx
\\
& \quad -  \int_{s_0}^{s_1} ds\int_{\Hcal_s}(s/t)^{1-2\sigma} \big(u\chi(r/t) \big)^2\bigg(\frac{r^2t + 3t}{(1+r^2)^2s^2} + \frac{2r^2t}{(1+r^2)s^4} \bigg) \, dx
\\
& \quad - 2\sigma \int_{s_0}^{s_1} ds\int_{\Hcal_s}(s/t)^{-2\sigma} \big(u\chi(r/t) \big)^2\frac{r^2}{(1+r^2)s^3} dx
\\
& =: \int_{s_0}^{s_1} \big(T_1 +T_2 +T_3 + T_4\big) ds.
\endaligned
$$
On the other hand, we apply Stokes' formula to the left-hand side of this identity. Recall that the flux vector vanishes in a neighborhood of the boundary of $\Kcal_{[s_0,s_1]}$, which is $\{r=t-1,s_0\leq \sqrt{t^2-r^2} \leq s_1\}$ and, by a calculation similar to the one in the proof of Lemma \ref{lem energy 1},
$$
\aligned
\bigg{\|}(s/t)^{- \sigma} \frac{r\chi(r/t) u}{(1+r^2)^{1/2}s} \bigg{\|}_{L^2(\Hcal_{s_1})}^2
 - \bigg{\|}(s/t)^{- \sigma} \frac{r\chi(r/t) u}{(1+r^2)^{1/2}s} \bigg{\|}_{L^2(\Hcal_{s_0})}^2
= \int_{s_0}^{s_1} \big(T_1 +T_2 +T_3 + T_4\big) ds.
\endaligned
$$
After differentiation with respect to $s$, we obtain 
\bel{ineq 1 proof prop Hardy 2}
2\bigg{\|}(s/t)^{- \sigma} \frac{r\chi(r/t) u}{(1+r^2)^{1/2}s} \bigg{\|}_{L^2(\Hcal_{s_1})}
\frac{d}{ds} \bigg{\|}(s/t)^{- \sigma} \frac{r\chi(r/t) u}{(1+r^2)^{1/2}s} \bigg{\|}_{L^2(\Hcal_{s_1})} = T_1 +T_2 +T_3 + T_4.
\ee

We observe that
$$
\aligned
|T_1| & \leq  2\sum_a \int_{\Hcal_s}s^{-1} (s/t)^{1- \sigma}|\del_a u|  (s/t)^{- \sigma} \frac{r\chi(r/t) |u|}{(1+r^2)^{1/2}s}  \frac{|x^a|t\chi(r/t)}{r(1+r^2)^{1/2}} \, dx
\\
& \leq  2\sum_a s^{-1} \|(s/t)^{1- \sigma} \del_au\|_{L_f^2(\Hcal_s)}
 \bigg{\|}(s/t)^{- \sigma} \frac{r\chi(r/t) u}{(1+r^2)^{1/2}s} \bigg{\|}_{L_f^2(\Hcal_s)}
  \bigg{\|} \frac{x^at\chi(r/t)}{r(1+r^2)^{1/2}} \bigg{\|}_{L^\infty(\Hcal_s)}
\\
& \leq  Cs^{-1} \sum_a \|(s/t)^{1- \sigma} \del_au\|_{L_f^2(\Hcal_s)}
  \bigg{\|}(s/t)^{- \sigma} \frac{r\chi(r/t) u}{(1+r^2)^{1/2}s} \bigg{\|}_{L_f^2(\Hcal_s)}, 
\endaligned
$$
where we have observed that
$
\bigg{\|} \frac{x^at\chi(r/t)}{r(1+r^2)^{1/2}} \bigg{\|}_{L^\infty(\Hcal_s)} \leq C, 
$
since the support of $\chi(\cdot)$ is contained in $\{r\geq t/3\}$. Similarly, we find 
$$
\aligned
|T_2| & \leq  Cs^{-1} \|(s/t)^{1- \sigma}r^{-1}u\|_{L_f^2(\Hcal_s)}
  \bigg{\|}(s/t)^{- \sigma} \frac{r\chi(r/t) u}{(1+r^2)^{1/2}s} \bigg{\|}_{L_f^2(\Hcal_s)}
\\
& \leq Cs^{-1} \|r^{-1}u\|_{L_f^2(\Hcal_s)}  \bigg{\|}(s/t)^{- \sigma} \frac{r\chi(r/t) u}{(1+r^2)^{1/2}s} \bigg{\|}_{L_f^2(\Hcal_s)}
\\
& \leq  Cs^{-1} \sum_a \|\delu_a u\|_{L_f^2(\Hcal_s)}  \bigg{\|}(s/t)^{- \sigma} \frac{r\chi(r/t) u}{(1+r^2)^{1/2}s} \bigg{\|}_{L_f^2(\Hcal_s)}, 
\endaligned
$$
where we have applied \eqref{ineq Hardy 3}. We also observe that $T_3\leq 0$ and $T_4\leq 0$. Then,  \eqref{ineq 1 proof prop Hardy 2} leads us to
\bel{ineq 2 proof prop Hardy 2}
\frac{d}{ds} \bigg{\|}(s/t)^{- \sigma} \frac{r\chi(r/t) u}{(1+r^2)^{1/2}s} \bigg{\|}_{L^2(\Hcal_{s_1})} \leq Cs^{-1} \sum _a \big(\|(s/t)^{1- \sigma} \del_a u\|_{L_f^2(\Hcal_s)} + \|\delu_a u\|_{L_f^2(\Hcal_s)} \big). 
\ee
Then  by integrating on the interval $[s_0,s]$, we have
\bel{ineq 3 proof prop Hardy 2}
\aligned
\bigg{\|}(s/t)^{- \sigma} \frac{r\chi(r/t) u}{(1+r^2)^{1/2}s} \bigg{\|}_{L_f^2(\Hcal_s)}
& \leq  \bigg{\|}(s/t)^{- \sigma} \frac{r\chi(r/t) u}{(1+r^2)^{1/2}s} \bigg{\|}_{L^2(\Hcal_{s_0})}
\\
& \quad +  C\sum _a \int_{s_0}^s\tau^{-1} \big(\|(s/t)^{1- \sigma} \del_a u\|_{L^2(\Hcal_\tau)} + \|\delu_a u\|_{L^2(\Hcal_\tau)} \big) \, d\tau, 
\endaligned
\ee
which is the desired estimate in the outer part of $\Hcal_s$.

For the inner part, $r\leq t/3$ leads us to $\frac{2\sqrt{2}}{3} \leq s/t\leq 1$. Then by Lemma \ref{lem Hardy 1}, we find 
\bel{ineq 4 proof prop Hardy 2}
\bigg{\|}(s/t)^{- \sigma} \frac{r\big(1- \chi(r/t) \big) u}{(1+r^2)^{1/2}s} \bigg{\|}_{L_f^2(\Hcal_s)} \leq \|r^{-1}u\|_{L_f^2(\Hcal_s)} \leq C\sum_a \|\delu_a u\|_{L_f^2(\Hcal_s)}
\ee
and it remains to combine \eqref{ineq 3 proof prop Hardy 2} and \eqref{ineq 4 proof prop Hardy 2}.
\end{proof}


\subsection{Sobolev inequality on hyperboloids}

In order to turn an $L^2$ energy estimate into an $L^\infty$ estimate, we will rely on the following version of the Sobolev inequality (Klainerman \cite{Klainerman85}, H\"ormander \cite[Lemma 7.6.1]{Hormander}; see also 
LeFloch and Ma \cite[Section 5]{PLF-YM-book}).  

\begin{proposition}[Sobolev-type estimate on hyperboloids]
\label{pre lem sobolev}
For any sufficiently smooth function $u=u(t,x)$ which is defined in the future of $\Hcal_2$ and 
is spatially compactly supported, one has 
\be
\label{eq 1 sobolev}
\sup_{(t,x) \in \Hcal_s} t^{3/2} \, |u(t,x)|
\lesssim
\sum_{|I|\leq 2} \| L^I u (t, \cdot) \|_{L_f^2(\Hcal_s)},  \qquad s \geq 2,
\ee
where the implied constant is uniform in $s \geq 2$, and one recalls that $t = \sqrt{s^2 + |x|^2}$ on $\Hcal_s$.
\end{proposition}

\begin{proof} Consider the function $w_s(x) := u(\sqrt{s^2+|x|^2},x)$. 
Fix $s_0$ and a point $(t_0,x_0)$ in $\Hcal_{s_0}$ (with
$t_0 = \sqrt{s_0^2 + |x_0|^2}$), and observe that
\bel{eq:gh6007}
\del_aw_{s_0}(x) = \delu_au\big(\sqrt{s_0^2+|x|^2},x\big) = \delu_au(t,x),
\ee
with $t = \sqrt{s_0^2 + |x|^2}$ and
$
t\del_aw_{s_0}(x) = t\delu_au\big(\sqrt{s_0^2+|x|^2},t\big) = L_au(t,x).
$
Then, introduce $g_{s_0,t_0}(y) := w_{s_0}(x_0 + t_0\,y)$ and write
$$
g_{s_0,t_0}(0) = w_{s_0}(x_0) = u\big(\sqrt{s_0^2+|x_0|^2},x_0\big)=u(t_0,x_0).
$$
From the standard Sobolev inequality applied to the function $g_{s_0,t_0}$, we get 
$$
\big|g_{s_0,t_0}(0)\big|^2\leq C\sum_{|I|\leq 2}\int_{B(0,1/3)}|\del^Ig_{s_0,t_0}(y)|^2 \, dy,
$$
$B(0, 1/3) \subset \RR^3$ being the ball centered at the origin with radius $1/3$.

In view of (with $x = x_0 + t_0y$)
$$
\aligned
\del_ag_{s_0,t_0}(y)
& = t_0\del_aw_{s_0}(x_0 + t_0y) 
\\
& = t_0\del_aw_{s_0}(x) = t_0\delu_au\big(t,x)
\endaligned
$$
in view of \eqref{eq:gh6007}, we have (for all $I$)
$\del^Ig_{s_0,t_0}(y) = (t_0\delu)^I u(t,x)$ and, therefore, 
$$
\aligned
\big|g_{s_0,t_0}(0)\big|^2 \leq& C\sum_{|I|\leq 2}\int_{B(0,1/3)}\big|(t_0\delu)^I u\big(t,x)\big)\big|^2dy
\\
= & C t_0^{-3}\sum_{|I|\leq 2}\int_{B((t_0,x_0),t_0/3)\cap \Hcal_{s_0}}\big|(t_0\delu)^I u\big(t,x)\big)\big|^2dx.
\endaligned
$$

Note that
$$
\aligned
(t_0\delu_a(t_0\delu_b w_{s_0}))
& = t_0^2\delu_a\delu_bw_{s_0}
\\
& = (t_0/t)^2(t\delu_a)(t\delu_b) w_{s_0} - (t_0/t)^2 (x^a/t)L_b w_{s_0}
\endaligned
$$
and $x^a/t = x^a_0/t + yt_0/t = (x^a_0/t_0 + y)(t_0/t)$. In the region $y\in B(0,1/3)$, the factor $|x^a/t|$ is bounded by $C(t_0/t)$ and thus (for $|I| \leq 2$)
$$
|(t_0 \delu)^I u| \leq \sum_{|J| \leq |I|} | L^J u| (t_0/t)^2.
$$

In the region $|x_0|\leq t_0/2$, we have $t_0\leq \frac{2}{\sqrt{3}}s_0$ so 
$$
t_0 \leq C s_0 \leq C \sqrt{|x|^2 + s_0^2} = Ct
$$
for some $C>0$.  When $|x_0|\geq t_0/2$, in the region $B((t_0,x_0),t_0/3)\cap \Hcal_{s_0}$ we get
$t_0 \leq C|x| \leq C\sqrt{|x|^2 + s_0^2} =Ct$ and thus 
$$
|(t_0 \delu)^I u| \leq C \, \sum_{|J| \leq |I|} | L^J u|
$$
and
$$
\aligned
\big|g_{s_0,t_0}(y_0)\big|^2
\leq& Ct_0^{-3}\sum_{|I|\leq 2}\int_{B(x_0,t_0/3)\cap\Hcal_{s_0}}\big|(t\delu)^I u\big(t,x)\big)\big|^2 \, dx
\\
\leq& Ct_0^{-3}\sum_{|I|\leq 2}\int_{\Hcal_{s_0}}\big|L^I u(t,x)\big|^2 \, dx. 
\endaligned
$$ 
\end{proof}


\subsection{Hardy inequality for hyperboloids}

We now bound the norm $
\|r^{-1} \del^IL^Jh_{\alpha \beta} \|_{L^2(\Hcal_s^*)}.
$
If $\del^IL^Jh_{\alpha \beta}$ were compactly supported in $\Hcal_s\cap\Kcal$, we could directly apply the standard Hardy inequality to the function
$u_s(x): = \big( \del^IL^Jh_{\alpha \beta} \big) (\sqrt{s^2 +r^2}, x)$ 
and we would obtain 
$$
\|r^{-1} \del^IL^Jh_{\alpha \beta} \|_{L^2(\Hcal_s^*)} \leq C\|\delu \del^IL^Jh_{\alpha \beta} \|_{L^2(\Hcal_s^*)}.
$$
However, since $\del^IL^Jh_{\alpha \beta}$ is not compactly supported in $\Kcal$, we must take a boundary term into account. 

\begin{lemma}[Adapted Hardy inequality]\label{lem 0 Hardy}
Let $(h_{\alpha \beta}, \phi)$ be a solution to the Einstein-massive field system associated with a compact Schwarzschild
perturbation. Then, one has 
\bel{ineq 0 Hardy}
\|r^{-1} \del^IL^Jh_{\alpha \beta} \|_{L^2(\Hcal_s^*)} \leq C\sum_a \|\delu_a \del^IL^Jh_{\alpha \beta} \|_{L^2(\Hcal_s^*)} + Cm_Ss^{-1}.
\ee
\end{lemma}

\begin{proof} With the notation
$u_s (x) := \big( \del^IL^Jh_{\alpha \beta} \big) (\sqrt{s^2 +r^2}, x )$, we obtain 
$$
\del_a u_s(x) = \delu_a \del^IL^Jh_{\alpha \beta} \left(\sqrt{s^2 +r^2}, x\right).
$$
Consider the identity
$
r^{-2}u_s^{-2} = - \del_r\left(r^{-1}u_s^2\right) + 2u_s r^{-1} \del_r u_s
$
and integrate it in the region
$
C_{[\vep,( s^2-1)/2]} :=\left\{\vep \leq r\leq  \frac{s^2-1}{2} \right\}
$
with spherical coordinates. We have
\bel{eq 1 proof lem 0 Hardy}
\int_{C_{[\vep,( s^2-1)/2]}}|r^{-1}u_s|^2 dx 
= \int_{r= (s^2-1)/2}r^{-1}u_s^2 d\sigma  
- \int_{r=\vep}r^{-1}u_s^2d\sigma
+ 2\int_{C_{[\vep,( s^2-1)/2]}}u_s r^{-1} \del_r u_s dx. 
\ee
Letting now $\vep \to 0^+$, we have $\int_{r=\vep}r^{-1}u_s^2 d\sigma \to 0$. Observe that on the sphere
$r = (s^2-1)/2$,
$$
\sqrt{s^2 +r^2} - r = \frac{s^2 +1}{2} - \frac{s^2-1}{2} = 1,
$$
that is the point $\left(\sqrt{s^2 +r^2}, x\right)$ is on the cone $\{r = t-1\}$. We know that, on this cone, $h_{\alpha \beta}$ coincides with the Schwarzschild metric, so that
$$
\int_{r= (s^2-1)/2}r^{-1}u_s^2d\sigma \leq Cm_S^2 s^{-2}.
$$
Then, \eqref{eq 1 proof lem 0 Hardy} yields us
$$
\|r^{-1}u_s \|^2_{L^2(C_{[0,(s^2-1)/2]})} \leq 2\|r^{-1}u_s \|_{L^2(C_{[0,(s^2-1)/2]})}  \|\del_r u_s \|_{L^2(C_{[0,(s^2-1)/2]})} + Cm_S^2s^{-2}.
$$
And this inequality leads us to
$$
\|r^{-1}u_s \|_{L^2(C_{[0,(s^2-1)/2]})} \leq  C\|\del_r u_s \|_{L^2(C_{[0,(s^2-1)/2]})} + Cm_Ss^{-1}.
$$
By recalling that
$$
\aligned
\|r^{-1}u_s \|^2_{L^2(C_{[0,(s^2-1)/2]})} & 
=  \int_{r\leq (s^2-1)/2} \big|r^{-1} \del^IL^Jh_{\alpha \beta} \left(\sqrt{s^2 +r^2}, x\right) \big|^2dx
\\
& = \int_{\Kcal\cap \Hcal_s} \big|r^{-1} \del^IL^Jh_{\alpha \beta}(t, x) \big|^2dx = \|r^{-1} \del^IL^Jh_{\alpha \beta} \|^2_{L^2(\Hcal_s^*)}
\endaligned
$$
and
$
\del_r u_s = \frac{x^a}{r} \del_au_s = \frac{x^a}{r} \delu_a \del^IL^Jh_{\alpha \beta}(\sqrt{s^2 +r^2}, x),
$
the proof is completed. 
\end{proof}
 

\subsection{Commutator estimates for admissible vector fields}

We recall the following identities first established in \cite{PLF-YM-one}; see also Appendix~\ref{appendix-COM} at the end of this monograph. 

\begin{lemma}[Algebraic decomposition of commutators]  
One has
\bel{pre lem commutator pr0}
[\del_t, \delu_a] = - \frac{x^a}{t^2} \del_t, \quad [\delu_a, \delu_b] = 0.
\ee
There exist constants $\lambda_{aJ}^I$ such that
\bel{pre lem commutator pr1-ZZ}
[\del^I, L_a] = \sum_{|J| \leq|I|} \lambda^I_{aJ} \del^J.
\ee
There exist constants $\theta_{\alpha J}^{I\gamma}$ such that
\be
\label{pre lem commutator pr2-ZZ}
[L^I, \del_\alpha] = \sum_{|J|<|I|, \gamma} \theta_{\alpha J}^{I\gamma} \del_{\gamma}L^J.
\ee
In the future light-cone $\Kcal$, the following identity holds:
\bel{pre lem commutator pr2 NEW-ZZ}
[\del^IL^J, \delu_{\beta}]
 = \sum_{|J'| \leq |J|\atop|I'| \leq|I|} \thetau_{\beta I'J'}^{IJ \gamma} \del_{\gamma} \del^{I'}L^{J'},
\ee
where the coefficients $\thetau_{\beta I'J'}^{IJ\gamma}$ are smooth functions
and satisfy (in $\Kcal$)
\bel{pre lem commutator pr4a-ZZ}
\aligned
\big|\del^{I_1}L^{J_1} \thetau_{\beta I'J'}^{IJ\gamma} \big| & \leq  C\big(|I|,|J|,|I_1|,|J_1|\big) \, t^{-|I_1|}, \quad |J'|<|J|,
\\
\big|\del^{I_1}L^{J_1} \thetau_{\beta I'J'}^{IJ\gamma} \big| & \leq  C\big(|I|,|J|,|I_1|,|J_1|\big) \, t^{-|I_1|-1}, \quad |I'|<|I|.
\endaligned
\ee
Within the future light-cone $\Kcal$, the following identity holds:
\bel{pre lem commutator pr3-ZZ}
[L^I, \delu_c]
 = \sum_{|J|<|I|} \sigma^{Ia}_{cJ} \delu_aL^J,
\ee
where the coefficients $\sigma_{c J}^{Ia}$ are
smooth functions and
satisfy (in $\Kcal$)
\bel{pre lem commutator pr3b-ZZ}
\big|\del^{I_1}L^{J_1} \sigma_{c J}^{Ia} \big| \leq C(|I|,|J|,|I_1|,|J_1|)t^{-|I_1|}.
\ee
Within the future light-cone $\Kcal$, the following identity holds:
\bel{pre lem commutator pr4-ZZ}
[\del^I, \delu_c]
=  t^{-1} \!\!\!\!\sum_{|J| \leq|I|} \rho_{cJ}^{I} \del^{J},
\ee
where the coefficients $\rho_{cJ}^I$ are smooth functions 
and satisfy (in $\Kcal$)
\bel{pre lem commutator pr4b-ZZ}
\big|\del^{I_1}L^{J_1} \rho_{cJ}^{I} \big| \leq C(|I|,|J|,|I_1|,|J_1|)t^{-|I_1|}.
\ee
\end{lemma}

\begin{lemma} \label{lem com 1}
For all indices $I$, the function
\bel{ineq com 1.1}
\Xi^I  := (t/s) \del^IL^J(s/t)
\ee
defined in the closed cone $\overline{\Kcal} = \{|x| \leq t-1\}$, is smooth and all of its derivatives (of any order)
are bounded in $\overline{\Kcal}$. Furthermore, it is homogeneous of degree $\eta$ with $\eta \leq 0$ (in the sense recalled in Definition \ref{def 1} below).
\end{lemma}

\begin{lemma}[Commutator estimates]\label{lem com 2} 
For all sufficiently smooth functions $u$ defined in the cone $\Kcal$, the following identities hold:
\bel{ineq com 2.1}
\big|[\del^IL^J, \del_{\alpha}]u\big| \leq C(|I|, |J|) \sum_{|J'|<|J|, \beta}|\del_{\beta} \del^IL^{J'}u|,
\ee
\bel{ineq com 2.2}
\big|[\del^IL^J, \delu_c]u\big| \leq C(|I|,|J|) \sum_{|J'|<|J|,a \atop |I'| \leq |I|} |\delu_a \del^{I'}L^{J'}u|
+ C(|I|,|J|)t^{-1} \!\!\!\!\sum_{|I| \leq|I'|\atop |J| \leq|J'|}|\del^{I'}L^{J'}u|.
\ee
\bel{ineq com 2.3}
\big|[\del^IL^J, \delu_{\alpha}] u\big| \leq C(|I|,|J|)t^{-1} \sum_{\beta,|I'|<|I|\atop |J'| \leq|J|} \big|\del_{\beta} \del^{I'}L^{J'} u\big|
+ C(|I|,|J|) \sum_{\beta,|I'| \leq|I|\atop |J'|<|J|} \left|\del_\beta \del^{I'}L^{J'}u\right|,
\ee 
\bel{ineq com 2.4}
\big|[\del^IL^J, \del_\alpha \del_{\beta}] u \big|
\leq  C(|I|,|J|) \sum_{\gamma, \gamma'\atop |I| \leq|I'|,|J'|<|I|} \big|\del_{\gamma} \del_{\gamma'} \del^{I'}L^{J'} u\big|,
\ee
\bel{ineq com 2.5}
\aligned
&\big|[\del^IL^J, \delu_a \delu_{\beta}] u\big| + \big|[\del^IL^J, \delu_\alpha \delu_b] u\big|
\\
&\leq C(|I|,|J|) \Bigg(
\sum_{c, \gamma,|I'| \leq |I|\atop |J'| < |J|} \big|\delu_c \delu_{\gamma} \del^{I'}L^{J'}u\big|
+
 \sum_{c, \gamma,|I'| < |I|\atop |J'| \leq |J|}  t^{-1} \big|\delu_c \delu_{\gamma} \del^{I'}L^{J'}u\big|
+ \sum_{\gamma,|I'| \leq|I|\atop |J'| \leq|J|}   t^{-1} \big|\del_{\gamma} \del^{I'}L^{J'}u\big|
\Bigg).
\endaligned
\ee
\end{lemma}


\section{Quasi-Null Structure of the Einstein-Massive Field System on Hyperboloids}
\label{section-3-QUASI}

\subsection{Einstein equations in wave coordinates}
\label{section-3}

Our next task is to derive an explicit expression for the curvature. We set $\Gamma^\gamma := 
g^{\alpha \beta} \Gamma_{\alpha \beta}^\gamma = 0$ and $\Gamma_\alpha := 
g_{\alpha \beta} \Gamma^\beta$. 

\begin{lemma}[Ricci curvature of a $4$-manifold]
\label{lem Ricci}
In arbitrary local coordinates, one has the decomposition:
$$
R_{\alpha \beta} = - \frac{1}{2}g^{\lambda \delta} \del_{\lambda} \del_{\delta}g_{\alpha \beta}
 + \frac{1}{2} \big(\del_{\alpha} \Gamma_{\beta} + \del_{\beta} \Gamma_{\alpha} \big)
  +\frac{1}{2}F_{\alpha \beta},
$$
where $F_{\alpha \beta} := P_{\alpha \beta} + Q_{\alpha \beta} + W_{\alpha \beta}$ is a sum of null terms, that is,
$$
\aligned
Q_{\alpha \beta} :& = 
  g^{\lambda \lambda'}g^{\delta \delta'} \del_{\delta}g_{\alpha \lambda'} \del_{\delta'}g_{\beta \lambda}
-g^{\lambda \lambda'}g^{\delta \delta'} \big
(\del_{\delta}g_{\alpha \lambda'} \del_{\lambda}g_{\beta \delta'} - \del_{\delta}g_{\beta \delta'} \del_{\lambda}g_{\alpha \lambda'} \big)
\\
& \quad + g^{\lambda \lambda'}g^{\delta \delta'}
\big(\del_{\alpha}g_{\lambda'\delta'} \del_{\delta}g_{\lambda \beta} - \del_{\alpha}g_{\lambda \beta} \del_{\delta}g_{\lambda'\delta'} \big)
+\frac{1}{2}g^{\lambda \lambda'}g^{\delta \delta'}
\big(\del_{\alpha}g_{\lambda \beta} \del_{\lambda'}g_{\delta \delta'} - \del_{\alpha}g_{\delta \delta'} \del_{\lambda'}g_{\lambda \beta} \big)
\\
& \quad + g^{\lambda \lambda'}g^{\delta \delta'}
\big(\del_{\beta}g_{\lambda'\delta'} \del_{\delta}g_{\lambda \alpha} - \del_{\beta}g_{\lambda \alpha} \del_{\delta}g_{\lambda'\delta'} \big)+\frac{1}{2}g^{\lambda \lambda'}g^{\delta \delta'}
\big(\del_{\beta}g_{\lambda \alpha} \del_{\lambda'}g_{\delta \delta'} - \del_{\beta}g_{\delta \delta'} \del_{\lambda'}g_{\lambda \alpha} \big),
\endaligned
$$
quasi-null term (as they are called by the authors)
$$
P_{\alpha \beta} := - \frac{1}{2}g^{\lambda \lambda'}g^{\delta \delta'} \del_{\alpha}g_{\delta \lambda'} \del_{\beta}g_{\lambda \delta'}
+\frac{1}{4}g^{\delta \delta'}g^{\lambda \lambda'} \del_{\beta}g_{\delta \delta'} \del_{\alpha}g_{\lambda \lambda'}
$$
and a remainder
$
W_{\alpha \beta} := g^{\delta \delta'} \del_{\delta}g_{\alpha \beta} \Gamma_{\delta'} - \Gamma_{\alpha} \Gamma_{\beta}.
$
\end{lemma}

Let us make some observations based on this lemma. Note that the Einstein equation $R_{\alpha \beta} = 0$ now reads
\be
\Boxt_g h_{\alpha \beta} = P_{\alpha \beta} + Q_{\alpha \beta} + W_{\alpha \beta}
 +  \big(\del_{\alpha} \Gamma_{\beta} + \del_{\beta} \Gamma_{\alpha} \big).
\ee
Furthermore, if the coordinates are assumed to satisfy the wave condition
$
\Gamma^\gamma = 0,
$
so that $\Gamma_{\beta} = 0$ and, by specifying the dependence of the right-hand sides in $(g;\del h)$,
\bel{eq Einstein wave}
\Boxt_{g}g_{\alpha \beta} = P_{\alpha \beta}(g;\del h)  + Q_{\alpha \beta}(g;\del h),
\ee
which is a standard result.

For the Einstein-massive field system
\be
\aligned
G_{\alpha \beta} & =  8\pi T_{\alpha \beta},
\\
T_{\alpha \beta} & =  \del_{\alpha} \phi\del_{\beta} \phi - \frac{1}{2}g_{\alpha \beta} \left(g^{\mu\nu} \del_{\mu} \phi\del_{\nu} \phi + c^2\phi^2\right),
\endaligned
\ee
we obtain
$$
R_{\alpha \beta} = 8\pi\left(\nabla_{\alpha} \phi\nabla_{\beta} \phi + \frac{1}{2}c^2\phi^2g_{\alpha \beta} \right)
$$
and, by the above lemma, the Einstein-massive field system in a wave coordinate system reads
\bel{eq main PDE'}
\aligned
&\Boxt_g g_{\alpha \beta} = P_{\alpha \beta}(g;\del h)  + Q_{\alpha \beta}(g;\del h) - 16\pi\del_{\alpha} \phi\del_{\beta} \phi - 8\pi c^2\phi^2g_{\alpha \beta},
\\
&\Boxt_g \phi - c^2\phi = 0.
\endaligned
\ee

\begin{proof}[Proof of Lemma \ref{lem Ricci}] We need to perform straightforward but very tedious calculations, starting from the definitions
$$
\aligned
R_{\alpha \beta}
& = \del_{\lambda} \Gamma_{\alpha \beta}^{\lambda} - \del_{\alpha} \Gamma_{\beta \lambda}^{\lambda} + \Gamma_{\alpha \beta}^{\lambda} \Gamma_{\lambda \delta}^{\delta} - \Gamma_{\alpha \delta}^{\lambda} \Gamma_{\beta \lambda}^{\delta},
\\
\Gamma_{\alpha \beta}^{\lambda}
& = \frac{1}{2}g^{\lambda \lambda'} \big(\del_{\alpha}g_{\beta \lambda'}
+ \del_{\beta}g_{\alpha \lambda'} - \del_{\lambda'}g_{\alpha \beta} \big).
\endaligned
$$ 
Only the first two terms in the expression $R_{\alpha \beta}$ involves second-order derivatives of the metric, and we focus on those terms first.
In view of
$$
\aligned
&\del_{\lambda} \Gamma_{\alpha \beta}^{\lambda} = - \frac{1}{2}g^{\lambda \delta} \del_{\lambda} \del_{\delta}g_{\alpha \beta} + \frac{1}{2}g^{\lambda \delta} \del_{\lambda} \del_{\alpha}g_{\beta \delta} + \frac{1}{2}g^{\lambda \delta} \del_{\lambda} \del_{\beta}g_{\alpha \delta}
+\frac{1}{2} \del_{\lambda}g^{\lambda \delta} \big(\del_{\alpha}g_{\beta \delta} + \del_{\beta}g_{\alpha \delta} - \del_{\delta}g_{\alpha \beta} \big),
\\
&\del_{\alpha} \Gamma_{\beta \lambda}^{\lambda} = \frac{1}{2} \del_{\alpha} \del_{\beta}g_{\lambda \delta} +\frac{1}{2} \del_{\alpha}g^{\lambda \delta} \del_{\beta}g_{\lambda \delta},
\endaligned
$$
we can write 
\bel{eq pr 4-1-1}
\aligned
\del_{\lambda} \Gamma_{\alpha \beta}^{\lambda} - \del_\alpha \Gamma_{\beta \lambda}^{\lambda}
& =   - \frac{1}{2}g^{\lambda \delta} \del_{\lambda} \del_{\delta}g_{\alpha \beta}
   + \frac{1}{2}g^{\lambda \delta} \del_\alpha \del_{\lambda}g_{\delta \beta}
   +\frac{1}{2}g^{\lambda \delta} \del_{\beta} \del_{\lambda}g_{\delta \alpha}
   - \frac{1}{2}g^{\lambda \delta} \del_\alpha \del_{\beta}g_{\lambda \delta}
\\
& \quad -  \frac{1}{2} \del_{\lambda}g^{\lambda \delta} \del_{\delta}g_{\alpha \beta}
 + \frac{1}{2} \del_{\lambda}g^{\lambda \delta} \del_{\alpha}g_{\beta \delta}
 + \frac{1}{2} \del_{\lambda}g^{\lambda \delta} \del_{\beta}g_{\alpha \delta}
 - \frac{1}{2} \del_{\alpha}g^{\lambda \delta} \del_{\beta}g_{\lambda \delta},
\endaligned
\ee
in which the first line contains second-order terms and the second line contains quadratic products of first-order terms.

Let us next compute the term $\del_{\alpha} \Gamma_{\beta} + \del_{\beta} \Gamma_{\alpha}$ (which appears in our decomposition). We have
$$
\aligned
\Gamma^{\gamma} = g^{\alpha \beta} \Gamma_{\alpha \beta}^{\gamma}
&= \frac{1}{2}g^{\alpha \beta}g^{\gamma \delta} \big(\del_\alpha g_{\beta \delta} + \del_{\beta}g_{\alpha \delta} - \del_{\delta}g_{\alpha \beta} \big)
\\
&=g^{\gamma \delta}g^{\alpha \beta} \del_\alpha g_{\beta \delta} - \frac{1}{2}g^{\alpha \beta}g^{\gamma \delta} \del_{\delta}g_{\alpha \beta}
\endaligned
$$
and, therefore,
$\Gamma_{\lambda} = g_{\lambda \gamma} \Gamma^{\gamma} = g^{\alpha \beta} \del_\alpha g_{\beta \lambda}
- \frac{1}{2}g^{\alpha \beta} \del_{\lambda}g_{\alpha \beta}$, so that,
after differentiating,
$$
\aligned
\del_\alpha \Gamma_{\beta}
& =  \del_\alpha \big(g^{\delta \lambda} \del_{\delta}g_{\lambda \beta} \big)
- \frac{1}{2} \del_\alpha \big(g^{\lambda \delta} \del_{\beta}g_{\lambda \delta} \big)
\\
& =   g^{\delta \lambda} \del_\alpha \del_{\delta}g_{\lambda \beta}
- \frac{1}{2} g^{\lambda \delta} \del_\alpha \del_{\beta}g_{\lambda \delta}
- \frac{1}{2} \del_\alpha g^{\lambda \delta} \del_{\beta}g_{\lambda \delta} + \del_\alpha g^{\delta \lambda} \del_{\delta}g_{\lambda \beta}.
\endaligned
$$
The term of interest is thus found to be
\bel{eq pr 4-1-2}
\aligned
\del_\alpha \Gamma_{\beta} + \del_{\beta} \Gamma_\alpha
& =  g^{\lambda \delta} \del_\alpha \del_{\lambda}g_{\delta \beta}
+ g^{\lambda \delta} \del_{\beta} \del_{\lambda}g_{\delta \alpha}
- g^{\lambda \delta} \del_{\alpha} \del_{\beta}g_{\lambda \delta}
\\
 & \quad +  \del_{\alpha}g^{\lambda \delta} \del_{\delta}g_{\lambda \beta}
  + \del_{\beta}g^{\lambda \delta} \del_{\delta}g_{\lambda \alpha}
  - \frac{1}{2} \del_{\beta}g^{\lambda \delta} \del_{\alpha}g_{\lambda \delta}
  - \frac{1}{2} \del_{\alpha}g^{\lambda \delta} \del_{\beta}g_{\lambda \delta}.
\endaligned
\ee

We observe that the last term in \eqref{eq pr 4-1-2} coincides with the last term in \eqref{eq pr 4-1-1}.
By noting also that the second-order terms in
 $\del_\alpha \Gamma_{\beta} + \del_{\beta} \Gamma_\alpha$
 are exactly three of the (four) second-order terms arising in the expression of
$\del_{\lambda} \Gamma_{\alpha \beta}^{\lambda} - \del_\alpha \Gamma_{\beta \lambda}^{\lambda}$,
 we see that
$$
\aligned
\del_{\lambda} \Gamma_{\alpha \beta}^{\lambda} - \del_\alpha \Gamma_{\beta \lambda}^{\lambda}
& = - \frac{1}{2} g^{\lambda \delta} \del_{\lambda} \del_{\delta}g_{\alpha \beta} + \frac{1}{2} \big(\del_\alpha \Gamma_{\beta} + \del_{\beta} \Gamma_\alpha \big)
\\
& \quad -  \frac{1}{2} \del_{\lambda}g^{\lambda \delta} \del_{\delta}g_{\alpha \beta}
 + \frac{1}{2} \del_{\lambda}g^{\lambda \delta} \del_{\alpha}g_{\beta \delta}
 + \frac{1}{2} \del_{\lambda}g^{\lambda \delta} \del_{\beta}g_{\alpha \delta}
\\
& \quad -  \frac{1}{2} \del_{\alpha}g^{\lambda \delta} \del_{\delta}g_{\lambda \beta}
  - \frac{1}{2} \del_{\beta}g^{\lambda \delta} \del_{\delta}g_{\lambda \alpha}
 - \frac{1}{4} \del_{\alpha}g^{\lambda \delta} \del_{\beta}g_{\lambda \delta}
  +\frac{1}{4} \del_{\beta}g^{\lambda \delta} \del_{\alpha}g_{\lambda \delta}
\\
& = - \frac{1}{2} \del_{\lambda}g^{\lambda \delta} \del_{\delta}g_{\alpha \beta} + \frac{1}{2} \big(\del_\alpha \Gamma_{\beta} + \del_{\beta} \Gamma_\alpha \big)
\\
& \quad +  \frac{1}{2}g^{\lambda \lambda'}g^{\delta \delta'} \del_{\lambda}g_{\lambda'\delta'} \del_{\delta}g_{\alpha \beta}
 - \frac{1}{2}g^{\lambda \lambda'}g^{\delta \delta'} \del_{\lambda}g_{\lambda'\delta'} \del_{\alpha}g_{\beta \delta}
\\
& \quad -  \frac{1}{2}g^{\lambda \lambda'}g^{\delta \delta'} \del_{\lambda}g_{\lambda'\delta'} \del_{\beta}g_{\alpha \delta}
 + \sout{\frac{1}{4}g^{\lambda \lambda'}g^{\delta \delta'} \del_{\alpha}g_{\lambda'\delta'} \del_{\beta}g_{\lambda \delta}}
\\
& \quad +  \frac{1}{2}g^{\lambda \lambda'}g^{\delta \delta'} \del_{\alpha}g_{\lambda'\delta'} \del_{\delta}g_{\lambda \beta}
  +\frac{1}{2}g^{\lambda \lambda'}g^{\delta \delta'} \del_{\beta}g_{\lambda'\delta'} \del_{\delta}g_{\lambda \alpha}
  - \sout{\frac{1}{4}g^{\lambda \lambda'}g^{\delta \delta'} \del_{\beta}g_{\lambda'\delta'} \del_{\alpha}g_{\lambda \delta}},
\endaligned
$$
where we have used the identity
$\del_{\alpha}g^{\lambda \delta} = -g^{\lambda \lambda'}g^{\delta \delta'} \del_{\alpha}g_{\lambda'\delta'}$. Note that
 the two underlined terms above cancel each other. So, the quadratic terms in $\del_{\lambda} \Gamma_{\alpha \beta}^{\lambda} - \del_\alpha \Gamma_{\beta \lambda}^{\lambda}$  are
$$
\aligned
&\frac{1}{2}g^{\lambda \lambda'}g^{\delta \delta'} \del_{\lambda}g_{\lambda'\delta'} \del_{\delta}g_{\alpha \beta},
\quad
- \frac{1}{2}g^{\lambda \lambda'}g^{\delta \delta'} \del_{\lambda}g_{\lambda'\delta'} \del_{\alpha}g_{\beta \delta},
\quad
- \frac{1}{2}g^{\lambda \lambda'}g^{\delta \delta'} \del_{\lambda}g_{\lambda'\delta'} \del_{\beta}g_{\alpha \delta},
 \\
&\frac{1}{2}g^{\lambda \lambda'}g^{\delta \delta'} \del_{\alpha}g_{\lambda'\delta'} \del_{\delta}g_{\lambda \beta},
\quad
\frac{1}{2}g^{\lambda \lambda'}g^{\delta \delta'} \del_{\beta}g_{\lambda'\delta'} \del_{\delta}g_{\lambda \alpha}.
\endaligned
$$

Next, let us return to the expression of the Ricci curvature and consider
$$
\aligned 
\Gamma_{\alpha \beta}^{\lambda} \Gamma_{\lambda \delta}^{\delta}
=
& \frac{1}{4}g^{\lambda \lambda'}g^{\delta \delta'} \big(\del_{\lambda}g_{\delta \delta'} \del_{\alpha}g_{\beta \lambda'} + \del_{\beta}g_{\alpha \lambda'} \del_{\lambda}g_{\delta \delta'} - \del_{\lambda'}g_{\alpha \beta} \del_{\lambda}g_{\delta \delta'} \big),
\\
\Gamma_{\alpha \delta}^{\lambda} \Gamma_{\beta \lambda}^{\delta}
=
& \frac{1}{4}g^{\lambda \lambda'}g^{\delta \delta'} \big(
\del_{\alpha}g_{\delta \lambda'} \del_{\beta}g_{\lambda \delta'} + \del_{\alpha}g_{\delta \lambda'} \del_{\lambda}g_{\beta \delta'} - \del_{\alpha}g_{\delta \lambda'} \del_{\delta'}g_{\beta \lambda}
\\
&
+\del_{\delta}g_{\alpha \lambda'} \del_{\beta}g_{\lambda \delta'} + \del_{\delta}g_{\alpha \lambda'} \del_{\lambda}g_{\beta \delta'} - \del_{\delta}g_{\alpha \lambda'} \del_{\delta'}g_{\beta \lambda}
\\&
- \del_{\lambda'}g_{\alpha \delta} \del_{\beta}g_{\lambda \delta'} - \del_{\lambda'}g_{\alpha \delta} \del_{\lambda}g_{\beta \delta'} + \del_{\lambda'}g_{\alpha \delta} \del_{\delta'}g_{\beta \lambda}
\big)
\endaligned
$$
and deduce that
\bel{eq pr Ricci 1}
\aligned
&\Gamma_{\alpha \beta}^{\lambda} \Gamma_{\lambda \delta}^{\delta} - \Gamma_{\alpha \delta}^{\lambda} \Gamma_{\beta \lambda}^{\delta}
\\
&=- \frac{1}{4}g^{\lambda \lambda'}g^{\delta \delta'} \del_{\lambda'}g_{\alpha \beta} \del_{\lambda}g_{\delta \delta'}
  +\frac{1}{4}g^{\lambda \lambda'}g^{\delta \delta'} \del_{\delta}g_{\alpha \lambda'} \del_{\delta'}g_{\beta \lambda}
  +\frac{1}{4}g^{\lambda \lambda'} \del_{\lambda'}g_{\alpha \delta} \del_{\lambda}g_{\beta \delta'}
\\
& \quad - \frac{1}{4}g^{\lambda \lambda'}g^{\delta \delta'} \del_{\alpha}g_{\delta \lambda'} \del_{\beta}g_{\lambda \delta'}
\\
&\quad +\sout{\frac{1}{4}g^{\lambda \lambda'}g^{\delta \delta'} \del_{\lambda}g_{\delta \delta'} \del_{\alpha}g_{\beta \lambda'}}
 +\sout{\frac{1}{4}g^{\lambda \lambda'}g^{\delta \delta'} \del_{\lambda}g_{\delta \delta'} \del_{\beta}g_{\alpha \lambda'}}
 - \frac{1}{2}g^{\lambda \lambda'}g^{\delta \delta'} \del_{\delta}g_{\alpha \lambda'} \del_{\lambda}g_{\beta \delta'}.
\endaligned
\ee
Observe that the first three terms are null terms, while the fourth term is a quasi-null term.
The two underlined terms are going to cancel out with the two underlined terms in \eqref{eq pr Ricci 2}, derived below.
Hence, there remains only the last term to be treated.

In other words, we need to consider the following six terms:
\bel{eq pr Ricci 3}
\aligned
&\frac{1}{2}g^{\lambda \lambda'}g^{\delta \delta'} \del_{\lambda}g_{\lambda'\delta'} \del_{\delta}g_{\alpha \beta},
\qquad
- \frac{1}{2}g^{\lambda \lambda'}g^{\delta \delta'} \del_{\lambda}g_{\lambda'\delta'} \del_{\alpha}g_{\beta \delta},
\qquad
- \frac{1}{2}g^{\lambda \lambda'}g^{\delta \delta'} \del_{\lambda}g_{\lambda'\delta'} \del_{\beta}g_{\alpha \delta},
 \\
&\frac{1}{2}g^{\lambda \lambda'}g^{\delta \delta'} \del_{\alpha}g_{\lambda'\delta'} \del_{\delta}g_{\lambda \beta},
\qquad
\frac{1}{2}g^{\lambda \lambda'}g^{\delta \delta'} \del_{\beta}g_{\lambda'\delta'} \del_{\delta}g_{\lambda \alpha},
\qquad
- \frac{1}{2}g^{\lambda \lambda'}g^{\delta \delta'} \del_{\delta}g_{\alpha \lambda'} \del_{\lambda}g_{\beta \delta'}.
\endaligned
\ee
In view of the identities
\be
g^{\alpha \beta} \del_{\alpha}g_{\beta \delta} - \frac{1}{2}g^{\alpha \beta} \del_{\delta}g_{\alpha \beta} = \Gamma_{\delta},
\qquad
g_{\beta \delta} \del_{\alpha}g^{\alpha \beta} - \frac{1}{2}g_{\alpha \beta} \del_{\delta}g^{\alpha \beta} = \Gamma_{\delta},
\ee
the first three terms in \eqref{eq pr Ricci 3} can be decomposed as follows:
\bel{eq pr Ricci 2}
\aligned
\frac{1}{2}g^{\lambda \lambda'}g^{\delta \delta'} \del_{\lambda}g_{\lambda'\delta'} \del_{\delta}g_{\alpha \beta}
&=  \frac{1}{2}g^{\delta \delta'} \del_{\delta}g_{\alpha \beta} \Gamma_{\delta'} + \frac{1}{4}g^{\lambda \lambda'}g^{\delta \delta'} \del_{\delta}g_{\alpha \beta} \del_{\delta'}g_{\lambda \lambda'}
\\
- \frac{1}{2}g^{\lambda \lambda'}g^{\delta \delta'} \del_{\lambda}g_{\lambda'\delta'} \del_{\alpha}g_{\beta \delta}
&= - \frac{1}{2}g^{\delta \delta'} \del_{\alpha}g_{\beta \delta} \Gamma_{\delta'} - \sout{\frac{1}{4}g^{\lambda \lambda'}g^{\delta \delta'} \del_{\delta'}g_{\lambda \lambda'} \del_{\alpha}g_{\beta \delta}}
\\
- \frac{1}{2}g^{\lambda \lambda'}g^{\delta \delta'} \del_{\lambda}g_{\lambda'\delta'} \del_{\beta}g_{\alpha \delta}
&=- \frac{1}{2}g^{\delta \delta'} \del_{\beta}g_{\alpha \delta} \Gamma_{\delta'} - \sout{\frac{1}{4}g^{\lambda \lambda'}g^{\delta \delta'} \del_{\delta'}g_{\lambda \lambda'} \del_{\beta}g_{\alpha \delta}}. 
\endaligned
\ee
The last term in the first line is one of the quasi-null term stated in the proposition.
As mentioned earlier, the two underlined terms cancel out with the two underlined terms in \eqref{eq pr Ricci 1}. 
The fourth term in \eqref{eq pr Ricci 3} is treated as follows:
$$
\aligned
&\hskip-.3cm \frac{1}{2}g^{\lambda \lambda'}g^{\delta \delta'} \del_{\alpha}g_{\lambda'\delta'} \del_{\delta}g_{\lambda \beta}
\\
& = \frac{1}{2}g^{\lambda \lambda'}g^{\delta \delta'}
\big(\del_{\alpha}g_{\lambda'\delta'} \del_{\delta}g_{\lambda \beta} - \del_{\alpha}g_{\lambda \beta} \del_{\delta}g_{\lambda'\delta'} \big)
+\frac{1}{2}g^{\lambda \lambda'}g^{\delta \delta'} \del_{\alpha}g_{\lambda \beta} \del_{\delta}g_{\lambda'\delta'}
\\
& = \frac{1}{2}g^{\lambda \lambda'}g^{\delta \delta'}
\big(\del_{\alpha}g_{\lambda'\delta'} \del_{\delta}g_{\lambda \beta} - \del_{\alpha}g_{\lambda \beta} \del_{\delta}g_{\lambda'\delta'} \big)
+\frac{1}{2}g^{\lambda \lambda'} \del_{\alpha}g_{\lambda \beta} \Gamma_{\lambda'}
+\frac{1}{4}g^{\lambda \lambda'}g^{\delta \delta'} \del_{\alpha}g_{\lambda \beta} \del_{\lambda'}g_{\delta \delta'}
\\
& = \frac{1}{2}g^{\lambda \lambda'}g^{\delta \delta'}
\big(\del_{\alpha}g_{\lambda'\delta'} \del_{\delta}g_{\lambda \beta} - \del_{\alpha}g_{\lambda \beta} \del_{\delta}g_{\lambda'\delta'} \big)
+\frac{1}{4}g^{\lambda \lambda'}g^{\delta \delta'}
\big(\del_{\alpha}g_{\lambda \beta} \del_{\lambda'}g_{\delta \delta'} - \del_{\alpha}g_{\delta \delta'} \del_{\lambda'}g_{\lambda \beta} \big)
\\
& \quad + \frac{1}{2}g^{\lambda \lambda'} \del_{\alpha}g_{\lambda \beta} \Gamma_{\lambda'}
+ \frac{1}{4}g^{\lambda \lambda'}g^{\delta \delta'} \del_{\alpha}g_{\delta \delta'} \del_{\lambda'}g_{\lambda \beta}
\\
& = \frac{1}{2}g^{\lambda \lambda'}g^{\delta \delta'}
\big(\del_{\alpha}g_{\lambda'\delta'} \del_{\delta}g_{\lambda \beta} - \del_{\alpha}g_{\lambda \beta} \del_{\delta}g_{\lambda'\delta'} \big)
+\frac{1}{4}g^{\lambda \lambda'}g^{\delta \delta'}
\big(\del_{\alpha}g_{\lambda \beta} \del_{\lambda'}g_{\delta \delta'} - \del_{\alpha}g_{\delta \delta'} \del_{\lambda'}g_{\lambda \beta} \big)
\\
& \quad + \frac{1}{2}g^{\lambda \lambda'} \del_{\alpha}g_{\lambda \beta} \Gamma_{\lambda'}
+\frac{1}{4}g^{\delta \delta'} \del_{\alpha}g_{\delta \delta'} \Gamma_{\beta}
+\frac{1}{8}g^{\delta \delta'}g^{\lambda \lambda'} \del_{\alpha}g_{\delta \delta'} \del_{\beta}g_{\lambda \lambda'},
\endaligned
$$
while, for the fifth term, we have
$$
\aligned
&\hskip-.3cm \frac{1}{2}g^{\lambda \lambda'}g^{\delta \delta'} \del_{\beta}g_{\lambda'\delta'} \del_{\delta}g_{\lambda \alpha}
\\
& = \frac{1}{2}g^{\lambda \lambda'}g^{\delta \delta'}
\big(\del_{\beta}g_{\lambda'\delta'} \del_{\delta}g_{\lambda \alpha} - \del_{\beta}g_{\lambda \alpha} \del_{\delta}g_{\lambda'\delta'} \big)
+\frac{1}{4}g^{\lambda \lambda'}g^{\delta \delta'}
\big(\del_{\beta}g_{\lambda \alpha} \del_{\lambda'}g_{\delta \delta'} - \del_{\beta}g_{\delta \delta'} \del_{\lambda'}g_{\lambda \alpha} \big)
\\
& \quad + \frac{1}{2}g^{\lambda \lambda'} \del_{\beta}g_{\lambda \alpha} \Gamma_{\lambda'}
+\frac{1}{4}g^{\delta \delta'} \del_{\beta}g_{\delta \delta'} \Gamma_{\alpha}
+\frac{1}{8}g^{\delta \delta'}g^{\lambda \lambda'} \del_{\beta}g_{\delta \delta'} \del_{\alpha}g_{\lambda \lambda'}.
\endaligned
$$
For the last term in \eqref{eq pr Ricci 3}, we perform the following calculation:
$$
\aligned
&\hskip-.35cm - \frac{1}{2}g^{\lambda \lambda'}g^{\delta \delta'} \del_{\delta}g_{\alpha \lambda'} \del_{\lambda}g_{\beta \delta'}
\\
& = - \frac{1}{2}g^{\lambda \lambda'}g^{\delta \delta'} \big(\del_{\delta}g_{\alpha \lambda'} \del_{\lambda}g_{\beta \delta'} - \del_{\delta}g_{\beta \delta'} \del_{\lambda}g_{\alpha \lambda'} \big)
- \frac{1}{2}g^{\lambda \lambda'}g^{\delta \delta'} \del_{\delta}g_{\beta \delta'} \del_{\lambda}g_{\alpha \lambda'}
\\
& = - \frac{1}{2}g^{\lambda \lambda'}g^{\delta \delta'} \big(\del_{\delta}g_{\alpha \lambda'} \del_{\lambda}g_{\beta \delta'} - \del_{\delta}g_{\beta \delta'} \del_{\lambda}g_{\alpha \lambda'} \big)
- \frac{1}{2}g^{\lambda \lambda'} \del_{\lambda}g_{\alpha \lambda'} \Gamma_{\beta}
 - \frac{1}{4}g^{\lambda \lambda'}g^{\delta \delta'} \del_{\beta}g_{\delta \delta'} \del_{\lambda}g_{\alpha \lambda'}
\\
& = - \frac{1}{2}g^{\lambda \lambda'}g^{\delta \delta'} \big
(\del_{\delta}g_{\alpha \lambda'} \del_{\lambda}g_{\beta \delta'} - \del_{\delta}g_{\beta \delta'} \del_{\lambda}g_{\alpha \lambda'} \big)
- \frac{1}{2}g^{\lambda \lambda'} \del_{\lambda}g_{\alpha \lambda'} \Gamma_{\beta}
- \frac{1}{4}g^{\delta \delta'} \del_{\beta}g_{\delta \delta'} \Gamma_{\alpha}
\\
& \quad -  \frac{1}{8}g^{\lambda \lambda'}g^{\delta \delta'} \del_{\alpha}g_{\lambda \lambda'} \del_{\beta}g_{\delta \delta'}
\\
& = - \frac{1}{2}g^{\lambda \lambda'}g^{\delta \delta'} \big
(\del_{\delta}g_{\alpha \lambda'} \del_{\lambda}g_{\beta \delta'} - \del_{\delta}g_{\beta \delta'} \del_{\lambda}g_{\alpha \lambda'} \big)
- \frac{1}{2} \Gamma_{\alpha} \Gamma_{\beta}
- \frac{1}{4}g^{\delta \delta'} \del_{\alpha}g_{\delta \delta'} \Gamma_{\beta}
- \frac{1}{4}g^{\delta \delta'} \del_{\beta}g_{\delta \delta'} \Gamma_{\alpha}
\\
& \quad -  \frac{1}{8}g^{\lambda \lambda'}g^{\delta \delta'} \del_{\alpha}g_{\lambda \lambda'} \del_{\beta}g_{\delta \delta'}.
\endaligned
$$

In conclusion, the quadratic terms in $R_{\alpha \beta}$ read
$$
\aligned
& \frac{1}{2}g^{\lambda \lambda'}g^{\delta \delta'} \del_{\delta}g_{\alpha \lambda'} \del_{\delta'}g_{\beta \lambda}
\\
& \quad -  \frac{1}{2}g^{\lambda \lambda'}g^{\delta \delta'} \big
(\del_{\delta}g_{\alpha \lambda'} \del_{\lambda}g_{\beta \delta'} - \del_{\delta}g_{\beta \delta'} \del_{\lambda}g_{\alpha \lambda'} \big)
\\
& \quad + \frac{1}{2}g^{\lambda \lambda'}g^{\delta \delta'}
\big(\del_{\alpha}g_{\lambda'\delta'} \del_{\delta}g_{\lambda \beta} - \del_{\alpha}g_{\lambda \beta} \del_{\delta}g_{\lambda'\delta'} \big)
+\frac{1}{4}g^{\lambda \lambda'}g^{\delta \delta'}
\big(\del_{\alpha}g_{\lambda \beta} \del_{\lambda'}g_{\delta \delta'} - \del_{\alpha}g_{\delta \delta'} \del_{\lambda'}g_{\lambda \beta} \big)
\\
& \quad + \frac{1}{2}g^{\lambda \lambda'}g^{\delta \delta'}
\big(\del_{\beta}g_{\lambda'\delta'} \del_{\delta}g_{\lambda \alpha} - \del_{\beta}g_{\lambda \alpha} \del_{\delta}g_{\lambda'\delta'} \big)
+\frac{1}{4}g^{\lambda \lambda'}g^{\delta \delta'}
\big(\del_{\beta}g_{\lambda \alpha} \del_{\lambda'}g_{\delta \delta'} - \del_{\beta}g_{\delta \delta'} \del_{\lambda'}g_{\lambda \alpha} \big)
\\
& \quad -  \frac{1}{4}g^{\lambda \lambda'}g^{\delta \delta'} \del_{\alpha}g_{\delta \lambda'} \del_{\beta}g_{\lambda \delta'}
+\frac{1}{8}g^{\delta \delta'}g^{\lambda \lambda'} \del_{\beta}g_{\delta \delta'} \del_{\alpha}g_{\lambda \lambda'}
\\ 
& \quad +  \frac{1}{2}g^{\delta \delta'} \del_{\delta}g_{\alpha \beta} \Gamma_{\delta'}
- \frac{1}{2} \Gamma_{\alpha} \Gamma_{\beta}. 
\endaligned
$$
Finally, collecting all the terms above and observing that several cancellations take place, we arrive at the desired identity.
\end{proof}


\subsection{Analysis of the support} 
\label{subsec proof-localization}

We provide here a proof of Proposition \ref{prop basic-localization}. 

\vskip.3cm

\noindent {\bf Step I.}
We recall the structure of $F_{\alpha \beta}$ presented in Lemma \ref{lem Ricci}. We observe that both $P_{\alpha \beta}$ and $Q_{\alpha \beta}$ are linear combinations of the multi-linear terms which are product of a quadratic term in $g^{\alpha \beta}$ and a quadratic term in $\del g_{\alpha \beta}$. For convenience, we write 
$
F_{\alpha \beta} = F_{\alpha \beta}(g,g;\del g, \del g)
$
and 
$$
p_{\alpha \beta}(t, x): = \big({\gSch}_{\alpha \beta}-m_{\alpha \beta} \big)(t, x) \xi(t-r) + m_{\alpha \beta}, 
$$
where $\xi$ is a smooth function defined on $\RR$, with $\xi (r)=1$ for $r \leq 1$, while $\xi(r) = 0$ for $r \geq 3/2$. 
Hence, for $r\geq t-1$, $p_{\alpha \beta}$ coincides with the Schwarzschild metric while $r\leq t-3/2$, $p_{\alpha \beta}$ coincides with the Minkowski metric.  We also set 
\be
q_{\alpha \beta} := g_{\alpha \beta} - p_{\alpha \beta}.
\ee
So the desired result is equivalent to the following statement: {\sl If $(g_{\alpha \beta}, \phi)$ is a solution of \eqref{eq main PDE'} associated with a compact Schwarzschild perturbation, then the tensor $q_{\alpha \beta}$ above is supported in $\Kcal$.}

To establish this result, we write down the equation satisfied by $q_{\alpha \beta}$ and introduce 
$$
\aligned
(p^{\alpha \beta}) :=& (p_{\alpha \beta})^{-1}, 
\quad 
\\
q^{\alpha \beta} :=& g^{\alpha \beta} - p^{\alpha \beta} = (p_{\alpha'\beta'}-g_{\alpha'\beta'})p^{\alpha'\beta}g_{\alpha \beta'}=q_{\alpha'\beta'}p^{\alpha'\beta}g^{\alpha \beta'}.
\endaligned
$$
We observe that for $r\geq t-1$, when $q_{\alpha \beta}(t, x) =0$, then $q^{\alpha \beta}(t, x) = 0$.
In view of 
$$
\Boxt_gg_{\alpha \beta} = F_{\alpha \beta}(g,g, \del g, \del g) -16\pi\del_{\alpha} \phi\del_{\beta} \phi - 8\pi c^2\phi^2 g_{\alpha\beta},
$$
we have 
$$
\Boxt_{p+q}(p_{\alpha \beta} + q_{\alpha \beta})
= F_{\alpha \beta} \big(p+q,p+q, \del (p+q), \del (p+q) \big) -16\pi\del_{\alpha} \phi\del_{\beta} \phi - 8\pi c^2\phi^2 g_{\alpha\beta}. 
$$
By multi-linearity, the above equation leads us to
\bel{eq 1 pr prop basic-localization}
\aligned
\Boxt_p q_{\alpha \beta}  & =  - \Boxt_pp_{\alpha \beta} + F_{\alpha \beta} \big(p,p, \del p, \del p\big)
\\
& \quad + F_{\alpha \beta} \big(p,p, \del p, \del q\big) + F_{\alpha \beta} \big(p,p, \del q, \del(p+q) \big)
\\
& \quad + F_{\alpha \beta} \big(p,q, \del(p+q), \del(p+q) \big) + F_{\alpha \beta} \big(q,p+q, \del(p+q), \del(p+q) \big)
\\
& \quad -  q^{\mu\nu} \del_{\mu} \del_{\nu} \big(p_{\alpha \beta}+q_{\alpha \beta} \big)  -16\pi \del_{\alpha} \phi\del_{\beta} \phi - 8\pi c^2\phi^2 g_{\alpha\beta}.
\endaligned
\ee
Observe that for $r\geq t-1$, $p_{\alpha \beta} = \big({\gSch}_{\alpha \beta}-m_{\alpha \beta} \big) \xi(t-r) + m_{\alpha \beta}$ coincides with the Schwarzschild metric, which is a solution to the Einstein equation (in the wave gauge), so for $r\geq t-1$we have 
$
\Boxt_pp_{\alpha \beta} = F_{\alpha \beta}(p,p, \del p, \del p).
$
Setting
$
E_{\alpha \beta} = - \Boxt_pp_{\alpha \beta} + F_{\alpha \beta} \big(p,p, \del p, \del p\big)
$, 
we have obtained $E_{\alpha \beta} = 0$ for $r\geq t-1$.

Then we also observe that the third to the sixth terms are multi-linear terms, each of them contain $q$ or $\del q$ as a factor. Furthermore, we observe that the seventh term is written as
$$
- q^{\mu\nu} \del_{\mu} \del_{\nu} \big(p_{\alpha \beta}+q_{\alpha \beta} \big)
= -q_{\mu'\nu'}p^{\mu'\nu}g^{\mu\nu'} \del_{\mu} \del_{\nu} \big(p_{\alpha \beta}+q_{\alpha \beta} \big)
$$
So, the third to the seventh terms can be written in the form
$$
\del q\cdot G_1(p, \del p, q, \del q) + q\cdot G_2(p, \del p, \del\del p,q, \del q), 
$$
where $G_i$ are (sufficiently regular) multi-linear forms. 

For the equation of $\phi$, we have the decomposition 
$$
\Boxt_g\phi  = \Box_p\phi + \Boxt_q\phi = \Boxt_p\phi + q_{\mu'\nu'}p^{\mu'\nu}g^{\mu\nu'} \del_{\mu} \del_{\nu} \phi.
$$
We conclude that the metric $q_{\alpha \beta}$  satisfies 
\bel{eq 2 pr prop basic-localization}
\aligned
&\Boxt_p q_{\alpha \beta} = E_{\alpha \beta}
+ \del q\cdot G_1(p, \del p, q, \del q) + q\cdot G_2(p, \del p, \del\del p,q, \del q)
-16\pi \del_{\alpha} \phi\del_{\beta} \phi - 8\pi c^2\phi^2 g_{\alpha\beta},
\\
&\Boxt_p\phi-c^2\phi = -q_{\mu'\nu'}p^{\mu'\nu}g^{\mu\nu'} \del_{\mu} \del_{\nu} \phi.
\endaligned
\ee
Furthermore, observe that since $(g, \phi)$ describes a compact Schwarzschild perturbation, the restriction of both $q_{\alpha \beta}$ and $\phi$ on the hyperplane $\{t=2\}$ are compactly supported in the unit ball $\{r\leq 1\}$. Thus, $(q_{\alpha \beta}, \phi)$ is a regular solution to the linear wave system \eqref{eq 2 pr prop basic-localization} with initial data
$$
\aligned
q_{\alpha \beta}(2, x), \quad \phi(2, x) \qquad \text{ supported in the ball } \{r\leq 1\}.
\endaligned
$$
We want to prove that $(q_{\alpha \beta})$ and $\phi$ vanish outside $\Kcal$. This leads us to the analysis on the
domain of determinacy associated with the metric $p^{\alpha \beta}$, which is determined by the characteristics the operator $\Boxt_p$. 

\vskip.15cm

\noindent{\bf Step II. Characteristics of $\Boxt_p$.} We now analyze the domain of determinacy of a spacetime point $(t, x) \notin \Kcal$. We will prove that all characteristics passing this point do not intersect the domain $\Kcal\cap\{t\geq 2\}$. Once this is proved, we apply the standard argument on domain of determinacy (also observe that $E_{\alpha \beta}(t, x)$ vanishes outside $\Kcal$), we conclude that $q_{\alpha \beta}$ and $\phi$ vanish outside $\Kcal$.

To do so, we will prove that the  
boundary of $\Kcal$ is strictly spacelike with respect to the metric $p^{\alpha \beta}$. We observe that any vector $v$ tangent to $\{r=t-1\}$ at point $(t, x)$ satisfies  
$
v^0 = \frac{1}{r} \sum_a x^av^a = \omega_a v^a.
$
So, in view of \eqref{eq Sch-wave}, we have for all $|v|>0$
$$
\aligned
(v,v)_p(t, x) & =  (v,v)_{\gSch} =  v^0v^0g_{00} + v^av^bg_{ab}
\\
 & =  - \frac{r-m_S}{r+m_S} \omega_av^a \omega_bv^b + \omega_av^a \omega_bv^b\left(\frac{r+m_S}{r-m_S} - \frac{(r+m_S)^2}{r^2} \right) + \sum_a|v^a|^2
 \\
 & =  - \left(\frac{r-m_S}{r+m_S}- \frac{r+m_S}{r-m_S}+\frac{(r+m_S)^2}{r^2} \right) \omega_av^a \omega_bv^b + \sum_a|v^a|^2
 \\
& \geq  \left(1- \left(\frac{r+m_S}{r-m_S}- \frac{r-m_S}{r+m_S}+\frac{r^2}{(r+m_S)^2} \right) \omega_av^a \omega_bv^b\right) \sum_a|v^a|^2
 \\
 & = \frac{3r^2m_S+4rm_S^2 +m_S^3}{(r+m_S)^2(r-m_S)} \sum_a|v^a|^2>0, 
\endaligned
$$
where we have used  
$
|\omega_av^a| \leq |v| = \big(\sum_a|v^a|^2\big)^{1/2}.
$

A characteristic curve is a null curve, so a characteristic passing through $(t, x)$ with $r\geq t-1$ cannot intersect the  
boundary $\{r=t-1\}$ in the past direction (since $(t, x)$ is already in the past of  $\{r=t-1\}$). Hence, a characteristic passing through $(t, x)$ never intersects the region $\Kcal$ in the past direction, which leads to the conclusion that the domain of determinacy of $(t, x)$ does not intersect $\Kcal$ and, therefore, does not intersect $\{t=2,r\leq t-1\}$. We conclude that $q_{\alpha \beta}(t, x) =\phi(t, x) =0$.
 

\subsection{A classification of nonlinearities in the Einstein-massive field system}

First, we introduce a class of functions of particular interest.

\begin{definition} \label{def 1}
A {\rm smooth and homogeneous function} (defined in $\{r < t\}$) of degree $\alpha$ is, by definition, 
 a smooth function $\Phi$ defined in $\{r < t\}$ at least and satisfying 
\begin{itemize}

\item $\Phi(\lambda t, \lambda x) = \lambda^{\alpha} \Phi(t, x)$, for a fixed $\alpha \in \RR$ and for all $\lambda>0$,

\item $\sup_{|x| \leq 1}|\del^I\Phi(1, x)|< + \infty$ (for large enough $| I |$). 
\end{itemize}
\end{definition}

For instance, constant functions are smooth and homogeneous functions of degree $0$. We also observe that the elements of the transition matrix $\Phi_{\alpha}^{\beta}$ are smooth and homogeneous of degree $0$.

\begin{lemma}
If $\Phi$ is a smooth and homogeneous function defined in $\{r\leq t\}$ of degree $\alpha$, then there exists a constant $C$ determined by $\Phi$ and $N$ such that
$$
|\del^IL^J \Phi(t, x)| \leq Ct^{\alpha - |I|}.  
$$
Furthermore, if $\Phi$ and $\Psi$ are smooth and homogenous functions of degree $\alpha$ and $\beta$,  respectively, then the product $\Phi \, \Psi$ is smooth and homogeneous of degree $(\alpha+\beta)$.
\end{lemma}

\begin{proof} Observe that if $\Phi$ is homogeneous of degree $\alpha$, then
$
\Phi(\lambda t, \lambda x) = \lambda^{\alpha} \Phi(t, x). 
$
We differentiate the above equation with respect to $x^a$: 
$
\lambda \del_a \Phi(\lambda t, \lambda x) = \lambda^{\alpha} \del_a \Phi(t, x), 
$
which leads to
$
\del_a \Phi(\lambda t, \lambda x) = \lambda^{\alpha-1} \del_a \Phi(t, x).
$
In the same way, we obtain 
$
\del_t\Phi(\lambda t, \lambda x) = \lambda^{\alpha-1} \del_t\Phi(t, x).
$
For $L_a$, we have 
$$
\aligned
L_a \Phi(\lambda t, \lambda x) & =  (\lambda x^a) \del_t\Phi(\lambda t, \lambda x) + (\lambda t) \del_a \Phi(\lambda t, \lambda x)
\\
& =  (\lambda x^a) \lambda^{\alpha-1} \del_t\Phi(t, x) + (\lambda t) \lambda^{\alpha-1} \del_a \Phi(t, x)
  \lambda^{\alpha}L_a \Phi(t, x).
\endaligned
$$
We conclude that, after differentiation by $\del_{\alpha}$, the degree of a homogeneous function will be reduced by one while when derived by $L_a$ the degree does not change. By induction, we get the desired estimate. Finally, we observe that the relation between homogeneity and multiplication is trivial.
\end{proof}

In the following, the nonlinear terms such as $F_{\alpha \beta}$ and $[\del^IL^J,h^{\mu\nu} \del_\mu\del_\nu]h_{\alpha \beta}$ are expressed as linear combinations of some basic nonlinear terms (presented below) with smooth and homogeneous coefficients of non-positive degrees. We provide first a general classification of such nonlinear terms:
\begin{itemize}

\item
${\sl QS}_h(p,k)$ refers to at most {\sl $p$-order quadratic semi-linear terms in $h_{\alpha \beta}$.} They are linear combinations of the following terms with smooth and homogeneous coefficients of degree $\leq 0$:
$$
\del^IL^J\big(\del_{\mu}h_{\alpha \beta} \del_{\nu}h_{\alpha'\beta'} \big)
$$
with $|I|+|J| \leq p, |J| \leq k$.

\item
${\sl QS}_\phi(p,k)$ refers to {\sl $p$-order quadratic semi-linear terms in $\phi$.} They are linear combinations of the following terms with smooth and homogeneous coefficients of degree $\leq 0$:
$$
\del^IL^J\big(\del_{\mu} \phi\del_{\nu} \phi\big), \quad \del^IL^J(\phi^2 g_{\mu\nu})
$$
with
$
|I|+|J| \leq  p, |J| \leq k.
$

\item
${\sl QQ}_{hh}(p,k)$ refers to {\sl $p$-order quadratic quasi-linear terms in $h$,} which arise from the expression $[\del^IL^J,h^{\mu\nu} \del_{\mu} \del_{\nu}]h_{\alpha \beta}$. They are linear combinations of the following terms with smooth and homogeneous coefficients of degree $\leq 0$:
$$
\aligned 
\del^{I_1}L^{J_1}h_{\alpha'\beta'} \del^{I_2}L^{J_2} \del_{\mu} \del_{\nu}h_{\alpha \beta},
\qquad
h_{\alpha'\beta'} \del_{\mu} \del_{\nu} \del^IL^{J'}h_{\alpha \beta}
\endaligned
$$
with $|I_1|+|I_2| \leq p-k$, $|J_1|+|J_2| \leq k$ and $|I_2|+|J_2| \leq p-1$ and $|J'|<|J|$.

\item
${\sl QQ}_{h\phi}(p,k)$ refers to {\sl $p$-order quadratic quasi-linear terms in $h$ and $\phi$.} These terms come from the commutator $[\del^IL^J,h^{\mu\nu} \del_\mu\del_\nu]\phi$. They are linear combination of the following terms with smooth and homogeneous coefficients of degree $\leq 0$:
$$
\aligned 
\del^{I_1}L^{J_1}h_{\alpha'\beta'} \del^{I_2}L^{J_2} \del_{\mu} \del_{\nu} \phi,
\qquad
h_{\alpha'\beta'} \del_{\mu} \del_{\nu} \del^IL^{J'} \phi
\endaligned
$$
with $|I_1|+|I_2| \leq p-k$, $|J_1|+|J_2| \leq k$ and $|I_2|+|J_2| \leq p-1$, $|J'|<|J|$.
\end{itemize}

\noindent Next, we provide a list of ``good'' nonlinear terms:
\begin{itemize}
\item
${\sl Cub}(p,k)$ refers to {\sl higher-order terms of at least cubic order}, {\sl except} the cubic term $h_{\alpha \beta}h_{\gamma \delta}h_{\mu\nu}$ which does not appear in our system. This class covers all cubic terms of interest, in view of the structure of the system under consideration. Moreover, these terms are ``negligible'' as far as the analysis of global existence is concerned.

\item
${\sl GQS}_h(p,k)$ refers to {\sl ``good'' quadratic semi-linear terms in $\del h$,} that are linear combinations of the following terms with smooth and homogeneous coefficients of degree $\leq 0$:
$$
\aligned
&\del^IL^J\big(\delu_a h_{\alpha \beta} \delu_{\gamma}h_{\alpha'\beta'} \big), \quad
(s/t)^2\del^IL^J\big(\del_th_{\alpha \beta} \del_th_{\alpha'\beta'} \big)
\endaligned
$$
with $|I|+|J| \leq p$ and $|J| \leq k$.

\item
${\sl GQQ}_{hh}(p,k)$ refers to {\sl ``good'' quadratic quasi-linear terms,} that are linear combinations of the following terms with smooth and homogeneous coefficients of degree $\leq 0$:
$$
\aligned
&\del^{I_1}L^{J_1}h_{\alpha'\beta'} \del^{I_2}L^{J_2} \delu_a \delu_{\mu}h_{\alpha \beta}, \quad
 &&\del^{I_1}L^{J_1}h_{\alpha'\beta'} \del^{I_2}L^{J_2} \delu_{\mu} \delu_bh_{\alpha \beta}, 
\\
&h_{\alpha'\beta'} \del^IL^{J'} \delu_a \delu_{\mu}h_{\alpha \beta}, \quad
 &&h_{\alpha'\beta'} \del^IL^{J'} \delu_{\mu} \delu_bh_{\alpha \beta}
\endaligned
$$
with $|I_1|+|I_2| \leq p-k$, $|J_1|+|J_2| \leq k$ and $|I_2|+|J_2| \leq p-1$, $|J'|<|J|$.

\item
${\sl GQQ}_{h\phi}(p,k)$ refers to {\sl ``good'' quadratic quasi-linear terms,} that are linear combinations of the following terms with smooth and homogeneous coefficients of degree $\leq 0$:
$$
\aligned
&\del^{I_1}L^{J_1}h_{\alpha'\beta'} \del^{I_2}L^{J_2} \delu_a \delu_{\mu} \phi, \quad
&& \del^{I_1}L^{J_1}h_{\alpha'\beta'} \del^{I_2}L^{J_2} \delu_{\mu} \delu_b\phi, 
\\
&h_{\alpha'\beta'} \del^{I}L^{J'} \delu_a \delu_{\mu} \phi, \quad
 &&h_{\alpha'\beta'} \del^{I}L^{J'} \delu_{\mu} \delu_b\phi
\endaligned
$$
with $|I_1|+|I_2| \leq |I|= p-k$, $|J_1|+|J_2| \leq k$ and $|I_2|+|J_2| \leq p-1$, $|J'|<|J|$.

\item
${\sl Com}(p,k)$. These terms arise when we express a second-order derivative written in the canonical frame into the semi-hyperboloidal frame. Since the coefficients of the transition matrix $\Phi_{\alpha}^\beta$ and $\Psi_{\alpha}^\beta$ are homogeneous of degree zero, and the commutators contain at least one derivative of these coefficients as a factor, these terms are linear combinations of the following terms with homogeneous coefficients of degree $\leq 0$: 
$$
\aligned
&t^{-1}{\sl QS}_h(p,k), \quad 
&&t^{-1}{\sl QS}_{\phi}(p,k), \quad 
&&&t^{-1} \del^{I_1}L^{J_1} \del_{\mu}h_{\alpha \beta} \del^{I_2}L^{J_2} \del_{\nu} \phi, \quad
\\
&t^{-1} \del^{I_1}L^{J_1}h_{\mu\nu} \del^{I_2}L^{J_2} \del_{\gamma}h_{\mu'\nu'}, \quad
&&t^{-2} \del^{I_1}L^{J_1}h_{\mu\nu} \del^{I_2}L^{J_2} \phi, \quad
&&&t^{-2} \del^{I_1}L^{J_1}h_{\mu\nu} \del^{I_2}L^{J_2}h_{\mu'\nu'},  
\endaligned
$$
where $|I| \leq p-k, |J| \leq k$ and $|I_1|+|J_1| \leq p-1$, $|I_1|+|I_2| \leq p-k,|J_1|+|J_2| \leq k$.
\end{itemize}

With the above notation, we can decompose the commutator $[\del^IL^J,h^{\mu\nu} \del_\mu\del_\nu]u$, as follows.  

\begin{lemma}[Decomposition of quasi-linear terms]
\label{lem 1 nonlinear}
Let $|I|=p-k$ and $|J|=k$. Suppose $h^{\mu\nu} \del_\mu\del_\nu$ is a second-order operator with  sufficiently regular coefficients. Then $[\del^IL^J,h^{\mu\nu} \del_\mu\del_\nu]h_{\alpha \beta}$ is a linear combination of the following terms with smooth and homogeneous coefficients of degree $0$:
\bel{eq 1 lem 1 nolinear}
\aligned
&{\sl GQQ}_{hh}(p,k), \quad
&&
t^{-1} \del^{I_3}L^{J_3}h_{\mu\nu} \del^{I_4}L^{J_4} \del_{\gamma}h_{\mu'\nu'},
\\
&\del^{I_1}L^{J_1} \hu^{00} \del^{I_2}L^{J_2} \del_t\del_t h_{\alpha \beta}, \quad
&&L^{J_1'} \hu^{00} \del^IL^{J_2'} \del_t\del_t h_{\alpha \beta}, \quad
&&&\hu^{00} \del_{\gamma} \del_{\gamma'} \del^IL^{J'}h_{\alpha \beta},
\endaligned
\ee
where $I_1+I_2=I,J_1+J_2=J$ with $|I_1|\geq 1$, $J_1'+J_2'=J$ with $|J_1'|\geq 1$ and $|J'|<|J|$,  $|I_3|+|I_4| \leq |I|, |J_3|+|J_4| \leq |J|$.
\end{lemma}

\begin{proof} We have 
\bel{eq pr 1 lem 1 nonlinear}
\aligned
\,[\del^IL^J,h^{\mu\nu} \del_\mu\del_\nu]h_{\alpha \beta}
& =  [\del^IL^J, \hu^{\mu\nu} \delu_\mu\delu_\nu]h_{\alpha \beta} + [\del^IL^J,h^{\mu\nu} \del_{\mu} \Psi^{\nu'}_{\nu} \delu_{\nu'}]h_{\alpha \beta}
\\
& =  [\del^IL^J, \hu^{00} \del_t\del_t]h_{\alpha \beta}
\\
& \quad +  [\del^IL^J, \hu^{a0} \delu_a \del_t]h_{\alpha \beta} + [\del^IL^J, \hu^{0a} \del_t\delu_a]h_{\alpha \beta} + [\del^IL^J, \hu^{ab} \delu_a \delu_b]h_{\alpha \beta}
\\
& \quad +  [\del^IL^J,h^{\mu\nu} \del_{\mu} \Psi^{\nu'}_{\nu} \delu_{\nu'}]h_{\alpha \beta}. 
\endaligned
\ee
The second, third, and fourth terms are in class ${\sl GQQ}_{hh}(p,k)$ ($\hu^{\alpha \beta}$ being linear combinations of $h^{\alpha \beta}$ with smooth and homogeneous coefficients of degree zero) and, for the last term, we see that
$$
\aligned
\,[\del^IL^J,h^{\mu\nu} \del_{\mu} \Phi^{\nu'}_{\nu} \delu_{\nu'}]h_{\alpha \beta}
& =  \sum_{{I_1+I_2 +I_3=I\atop J_1+J_2 +J_3=J} \atop |I_3|+|J_3|<|I|+|J|} \del^{I_1}L^{J_1}h^{\mu\nu} \del^{I_2}L^{J_2} \del_{\mu} \Psi^{\nu'}_{\nu} \del^{I_3}L^{J_3} \delu_{\nu'}h_{\alpha \beta}
\\
 & \quad +  h^{\mu\nu} \del_{\mu} \Psi^{\nu'}_{\nu}[\del^IL^J, \delu_{\nu'}]h_{\alpha \beta}.
\endaligned
$$
Then by the homogeneity of $\Psi_{\nu}^{\nu'}$, the above term can be expressed as $t^{-1} \del^{I_3}L^{J_3}h_{\mu\nu} \del^{I_4}L^{J_4} \del_{\gamma}h_{\mu'\nu'}$.

Next, we treat the first term in the right-hand side of \eqref{eq pr 1 lem 1 nonlinear} :
$$
\aligned
\,[\del^IL^J, \hu^{00} \del_t\del_t]h_{\alpha \beta}
& =  \sum_{I_1+I_2=I\atop J_1+J_2= J,|I_1|\geq 1} \del^{I_1}L^{J_1} \hu^{00} \del^{I_2}L^{J_2} \del_t\del_th_{\alpha \beta}
 + \sum_{J_1+J_2=J\atop |J_1|\geq 1}L^{J_1} \hu^{00} \del^IL^{J_2} \del_t\del_th_{\alpha \beta}
\\
& \quad + \hu^{00}[\del^IL^J, \del_t\del_t]h_{\alpha \beta}.
\endaligned
$$
We observe that $[\del^IL^J, \del_t\del_t]h_{\alpha \beta}$ is a linear combination of the terms  $\del_{\alpha'} \del_{\beta'} \del^I L^{J'}h_{\alpha \beta}$ with $|J'|<|J|$. We apply the commutator identity \eqref{pre lem commutator pr2-ZZ}:
$$
\aligned
\,[\del^IL^J, \del_t\del_t]h_{\alpha \beta} & =  \del^I[L^J, \del_t\del_t]h_{\alpha \beta}
  = \del^I\left([L^J, \del_t]\del_th_{\alpha \beta} \right) + \del^I\del_t\left([L^J, \del_t]h_{\alpha \beta} \right)
\\
& = \theta_{0J'}^{J\gamma} \del_\gamma \del_t L^{J'}h_{\alpha \beta}
 + \theta_{0J'}^{J\gamma} \theta_{0J''}^{J'\gamma'} \del_{\gamma'}L^{J''}h_{\alpha \beta}
 + \theta_{0J'}^{J\gamma} \del_t\del_tL^{J'}h_{\alpha \beta}, 
\endaligned
$$
where $|J''|<|J'|<|J|$. This completes the proof.
\end{proof}

A similar decomposition is available for the commutator $[\del^IL^J,h^{\mu\nu} \del_\mu\del_\nu]\phi$: It is a linear combination of the following terms with smooth and homogeneous coefficients of degree $\leq 0$:
\bel{eq 2 lem 1 nolinear}
\aligned
&{\sl GQQ}_{h\phi}(p,k), \quad
t^{-1} \del^{I_1}L^{J_1}h_{\mu\nu} \del^{I_2}L^{J_2} \del_\gamma \phi,
\\
&\del^{I_1}L^{J_1} \hu^{00} \del^{I_2}L^{J_2} \del_t\del_t \phi, \quad
L^{J_1'} \hu^{00} \del^IL^{J_2'} \del_t\del_t \phi, \quad
\hu^{00} \del_{\alpha} \del_{\beta} \del^IL^{J'} \phi, 
\endaligned
\ee
where $I_1+I_2=I,J_1+J_2=J$ with $|I_1|\geq 1$, $J_1'+J_2'=J$ with $|J_1'|\geq 1$ and $|J'|<|J|$ and $|I_3|+|I_4| \leq |I|, |J_3|+|J_4| \leq |J|$.
In our analysis of the commutator estimates, we will make use of the decompositions \eqref{eq 1 lem 1 nolinear} and \eqref{eq 2 lem 1 nolinear}.


\subsection{Estimates based on commutators and homogeneity}
\label{subsec ineq-homo}

Let $u$ be a smooth function defined in $\Kcal$ and vanishing near the  
boundary $\{r = t-1\}$. In view of 
$
\delu_a = t^{-1}L_a, 
$
we have 
$$
\del^IL^J\delu_a u = \del^IL^J\big(t^{-1}L_a u\big) = \sum_{I_1+I_2=I\atop J_1+J_2=J} \del^{I_1}L^{J_1} \big(t^{-1} \big) \del^{I_2}L^{J_2}L_au.
$$
Since $t^{-1}$ is a smooth and homogeneous coefficient of degree $-1$, we have 
\bel{ineq 1 homo}
\big|\del^IL^J\delu_a u \big| \leq Ct^{-1} \sum_{|I'| \leq|I|\atop |J'| \leq |J|} \big|\del^{I'}L^{J'}L_au\big|.
\ee
As a direct application, for instance we have
$$
\aligned
\big|\del^IL^J\delu_a \delu_{\nu}u\big| & \leq  Ct^{-1} \sum_{|I'| \leq|I|\atop |J'| \leq |J|} \big|\del^{I'}L^{J'}L_a \delu_{\nu}u\big|
= Ct^{-1} \sum_{|I'| \leq|I|\atop |J'| \leq |J|} \big|\del^{I'}L^{J'}L_a \big(\Phi_{\nu}^{\nu'} \del_{\nu'}u\big) \big|. 
\endaligned
$$
The function $\Phi^{\nu'}_{\nu}$ is smooth and homogeneous of degree $0$, so that 
\bel{ineq 2 homo}
\big|\del^IL^J\delu_a \delu_{\nu}u\big|
\leq C(I,J)t^{-1} \sum_{\gamma,|I'| \leq|I|\atop |J'| \leq|J|}|\del^{I'}L^{J'}L_a \del_{\gamma}u|.
\ee
A similar argument holds for
\bel{ineq 3 homo}
\big|\del^IL^J\delu_\nu\delu_au\big|
\leq C(I,J)t^{-1} \sum_{\gamma,a,|I'| \leq|I|\atop |J'| \leq|J|}|\del^{I'}L^{J'}L_a \del_{\gamma}u|.
\ee
Furthermore, when there are two ``good'' derivatives, we consider 
$$
\aligned
\del^IL^J\big(\delu_a \delu_b u\big) & =  \del^IL^J\big(t^{-1}L_a(t^{-1}L_b) u\big)
= \del^IL^J\big(t^{-2}L_aL_b u\big) + \del^IL^J\big(t^{-1}L_a(t^{-1})u\big)
\\
& = \sum_{I_1+I_2=I\atop J_1+J_2=J} \del^{I_1}L^{J_1} \big(t^{-2} \big) \del^{I_2}L^{J_2}L_aL_bu
+ \sum_{I_1+I_2=I\atop J_1+J_2=J} \del^{I_1}L^{J_1} \big(t^{-1}L_a(t^{-1}) \big) \del^{I_2}L^{J_2}L_au, 
\endaligned
$$
and we find 
\bel{ineq 4 homo}
\aligned
\big|\del^IL^J\big(\delu_a \delu_b u\big) \big| 
&= \big|\del^IL^J\big(t^{-1}L_a(t^{-1}L_b) u\big) \big|
\\
& \leq Ct^{-2} \sum_{|I'| \leq|I|\atop |J'| \leq |J|} \big|\del^{I'}L^{J'}L_aL_bu\big| + Ct^{-2} \sum_{|I'| \leq|I|\atop|J'| \leq|J|} \big|\del^{I'}L^{J'}L_bu\big|. 
\endaligned
\ee


\subsection{Basic structure of the quasi-null terms}

In this section we consider the quasi-null terms $P_{\alpha \beta}$ and emphasize some important properties:

\noindent 1. The expression $P_{\alpha \beta}$ is a $2$-tensor and this tensorial structure plays a role in our analysis. 

\noindent 2. In explicit form, it reads 
$$
\aligned
P_{\alpha \beta} & =  \frac{1}{4}g^{\gamma \gamma'}g^{\delta \delta'} \del_{\alpha}h_{\gamma \delta} \del_{\beta}h_{\gamma'\delta'}
- \frac{1}{2}g^{\gamma \gamma'}g^{\delta \delta'} \del_{\alpha}h_{\gamma \gamma'} \del_{\beta}h_{\delta \delta'}, 
\endaligned
$$ 
and, in the semi-hyperboloidal frame,
$$
\Pu_{\alpha \beta} = \frac{1}{4}g^{\gamma \gamma'}g^{\delta \delta'} \delu_{\alpha}h_{\gamma \delta} \delu_{\beta}h_{\gamma'\delta'}
- \frac{1}{2}g^{\gamma \gamma'}g^{\delta \delta'} \delu_{\alpha}h_{\gamma \gamma'} \delu_{\beta}h_{\delta \delta'},
$$
so the only term to be concerned about is the $00$-component:
$$
\aligned
\Pu_{00} & =  \frac{1}{4}g^{\gamma \gamma'}g^{\delta \delta'} \del_th_{\gamma \delta} \del_th_{\gamma'\delta'}
- \frac{1}{2}g^{\gamma \gamma'}g^{\delta \delta'} \del_th_{\gamma \gamma'} \del_th_{\delta \delta'}
\\
& = \frac{1}{4} \gu^{\gamma \gamma'} \gu^{\delta \delta'} \del_t\hu_{\gamma \delta} \del_t\hu_{\gamma'\delta'}
- \frac{1}{2} \gu^{\gamma \gamma'} \gu^{\delta \delta'} \del_t\hu_{\gamma \gamma'} \del_t\hu_{\delta \delta'} + {\sl Com}(0,0).
\endaligned
$$
Here ${\sl Com}(0,0)$ represents the commutator terms:
$$
\aligned
{\sl Com}(0,0) & = \frac{1}{4}g^{\gamma \gamma'}g^{\delta \delta'} \hu_{\gamma''\delta''}
\del_t\big(\Psi_{\gamma}^{\gamma''} \Psi_{\delta}^{\delta''} \big)
\del_t\big(\Psi_{\gamma'}^{\gamma'''} \Psi_{\delta'}^{\delta'''} \big) \hu_{\gamma'''\delta'''}
\\
& \quad + \frac{1}{4}g^{\gamma \gamma'}g^{\delta \delta'}
\Psi_{\gamma}^{\gamma''} \Psi_{\delta}^{\delta''} \del_t\hu_{\gamma''\delta''}
\del_t\big(\Psi_{\gamma'}^{\gamma'''} \Psi_{\delta'}^{\delta'''} \big) \hu_{\gamma'''\delta'''}
\\
& \quad + \frac{1}{4}g^{\gamma \gamma'}g^{\delta \delta'}
\del_t\big(\Psi_{\gamma}^{\gamma''} \Psi_{\delta}^{\delta''} \big) \hu_{\gamma''\delta''}
\Psi_{\gamma'}^{\gamma''} \Psi_{\delta'}^{\delta''} \del_t\hu_{\gamma'''\delta'''}
\\
& \quad -  \frac{1}{2}g^{\gamma \gamma'}g^{\delta \delta'}
\del_t\big(\Psi_{\gamma}^{\gamma''} \Psi_{\delta}^{\delta''} \big) \hu_{\gamma''\gamma'''}
\del_t\big(\Psi_{\gamma'}^{\gamma'''} \Psi_{\delta'}^{\delta'''} \big) \hu_{\delta''\delta'''}
\\
& \quad -  \frac{1}{2}g^{\gamma \gamma'}g^{\delta \delta'}
\Psi_{\gamma}^{\gamma''} \Psi_{\delta}^{\delta''} \del_t\hu_{\gamma''\gamma'''}
\del_t\big(\Psi_{\gamma'}^{\gamma'''} \Psi_{\delta'}^{\delta'''} \big) \hu_{\delta''\delta'''}
\\
& \quad -  \frac{1}{2}g^{\gamma \gamma'}g^{\delta \delta'}
\del_t\big(\Psi_{\gamma}^{\gamma''} \Psi_{\delta}^{\delta''} \big) \hu_{\gamma''\gamma'''}
\Psi_{\gamma'}^{\gamma''} \Psi_{\delta'}^{\delta''} \del_t\hu_{\delta''\delta'''}.
\endaligned
$$
We see that
$$
\aligned
\Pu_{00}
& = \frac{1}{4} \gu^{\gamma \gamma'} \gu^{\delta \delta'} \del_t\hu_{\gamma \delta} \del_t\hu_{\gamma'\delta'}
- \frac{1}{2} \gu^{\gamma \gamma'} \gu^{\delta \delta'} \del_t\hu_{\gamma \gamma'} \del_t\hu_{\delta \delta'} + {\sl Com}(0,0)
\\
& =  \frac{1}{4} \minu^{\gamma \gamma'} \minu^{\delta \delta'} \del_t\hu_{\gamma \delta} \del_t\hu_{\gamma'\delta'}
- \frac{1}{2} \gu^{\gamma \gamma'} \gu^{\delta \delta'} \del_t\hu_{\gamma \gamma'} \del_t\hu_{\delta \delta'} + {\sl Com}(0,0) + {\sl Cub}(0,0).
\endaligned
$$
Here the terms ${\sl Cub}(0,0)$ stands for the high-order terms:
$$
\aligned
{\sl Cub}(0,0) & = \frac{1}{4} \hu^{\gamma \gamma'} \minu^{\delta \delta'} \del_t\hu_{\gamma \delta} \del_t\hu_{\gamma'\delta'}
+\frac{1}{4} \minu^{\gamma \gamma'} \hu_{\gamma \delta'} \del_t\hu_{\gamma \delta} \del_t\hu_{\gamma'\delta'}
+\frac{1}{4} \hu^{\gamma \gamma'} \hu^{\delta \delta'} \del_t\hu_{\gamma \delta'} \del_t\hu_{\gamma \delta'}.
\endaligned
$$
We summarize our conclusion.

\begin{lemma}[Structure of the quasi-null terms]\label{lem P 1}
The quasi-null term $\Pu_{00}$ are linear combinations of the following terms with smooth and homogeneous coefficients of degree $\leq 0$:
\bel{eq P 1}
{\sl GQS}_h(0,0), \quad {\sl Cub}(0,0), \quad {\sl Com}(0,0), \quad  \gu^{\gamma \gamma'} \gu^{\delta \delta'} \del_t\hu_{\gamma \gamma'} \del_t\hu_{\delta \delta'}, \quad
\minu^{\gamma \gamma'} \minu^{\delta \delta'} \del_t\hu_{\gamma \delta} \del_t\hu_{\gamma'\delta'}.
\ee
The quasi-null term $\Pu_{a \beta}$ are linear combinations of ${\sl GQS}_h(0,0)$ and ${\sl Cub}(0,0)$ terms.
\end{lemma}

So, the only problematic terms in $P_{\alpha \beta}$ are $\gu^{\gamma \gamma'} \gu^{\delta \delta'} \del_t\hu_{\gamma \gamma'} \del_t\hu_{\delta \delta'}$ and $\minu^{\gamma \gamma'} \minu^{\delta \delta'} \del_t\hu_{\gamma \delta} \del_t\hu_{\gamma'\delta'}$. They will be controlled by using the wave gauge condition.


\subsection{Metric components in the semi-hyperboloidal frame}

In this subsection, we derive the equation satisfied by the metric components within the semi-hyperboloidal frame. To do so, we need the identity
$$
\Boxt_g(uv) = u\Boxt_gv + v\Boxt_gu + 2g^{\alpha \beta} \del_{\alpha}u\del_{\beta}v.
$$
Then, we have 
$$
\Boxt_g\hu_{\alpha \beta} = \Boxt_g\big(\Phi_{\alpha}^{\alpha'} \Phi_{\beta}^{\beta'}h_{\alpha'\beta'} \big) = \Phi_{\alpha}^{\alpha'} \Phi_{\beta}^{\beta'} \Boxt_g h_{\alpha'\beta'} + 2g^{\mu\nu} \del_{\mu} \big(\Phi_{\alpha}^{\alpha'} \Phi_{\beta}^{\beta'} \big) \del_{\nu}h_{\alpha'\beta'}
+h_{\alpha'\beta'} \Boxt_g\big(\Phi_{\alpha}^{\alpha'} \Phi_{\beta}^{\beta'} \big).
$$ 
Then we calculate explicitly the correction terms concerning the derivatives of $\Phi_{\alpha}^{\alpha'} \Phi_{\beta}^{\beta'}$:
\begin{itemize}

\item Case $\alpha = \beta = 0$:
$$
\Phi_0^0\Phi_0^0 = 1, \text{ the other ones vanish, }
$$
$$
\Box \big(\Phi_0^{\alpha'} \Phi_0^{\beta'} \big) = 0, \quad \del \big(\Phi_0^0\Phi_0^0\big) = 0.
$$

\item Case $\alpha = a>0, \beta = 0$:
$$
\Phi_a^0\Phi_0^0 = x^a/t, \quad \Phi_a^a \Phi_0^0 = 1,
$$
$$
\Box \big(\Phi_a^0\Phi_0^0 \big) = - \frac{2x^a}{t^3}, \quad \del_t\big(\Phi_a^0\Phi_0^0 \big) = - \frac{x^a}{t^2}, \quad \del_a \big(\Phi_a^0\Phi_0^0 \big) = \frac{1}{t}.
$$

\item Case $\alpha = a>0, \beta = b>0$:
$$
\Phi_a^0\Phi_b^0 = x^ax^b/t^2, \quad \Phi_a^0\Phi_b^b = x^a/t, \quad \Phi_a^a \Phi_b^b = 1.
$$
$$
\aligned
&\Box \big(\Phi_a^0\Phi_b^0\big) = - \frac{6x^ax^b}{t^4} + \frac{2\delta_{ab}}{t^2}, \quad
\del_t \big(\Phi_a^0\Phi_b^0\big) = - \frac{2x^ax^b}{t^3}, \quad
\del_c\big(\Phi_a^0\Phi_b^0\big) = \frac{\delta_{ca}x^b + \delta_{cb}x^a}{t^2},
\\
&\Box \big(\Phi_a^0\Phi_b^b\big) = - \frac{2x^a}{t^3}, \quad \del_t\big(\Phi_a^0\Phi_b^b\big) = - \frac{x^a}{t^2}, \quad \del_a \big(\Phi_a^0\Phi_b^b\big) = \frac{1}{t},
\\
&\text{while the other ones vanish. }
\endaligned
$$
\end{itemize}

Then we calculate the remaining terms (up to second-order):
$$
\aligned
\Boxt_g\hu_{00} & =  \Phi_0^{\alpha'} \Phi_0^{\beta'}Q_{\alpha'\beta'} + \Pu_{00} - 16\pi\del_t\phi\del_t\phi - 8\pi c^2 \minu_{00} \phi^2 + {\sl Cub}(0,0),
\\
\Boxt_g\hu_{0a} & =  \Phi_0^{\alpha'} \Phi_a^{\beta'}Q_{\alpha'\beta'} + \Pu_{0a} - 16\pi \delu_a \phi\del_t\phi - 8\pi c^2 \minu_{a0} \phi^2
+ \frac{2}{t} \delu_a h_{00} - \frac{2x^a}{t^3}h_{00} + {\sl Cub}(0,0),
\\
\Boxt_g\hu_{aa} & =  \Phi_a^{\alpha'} \Phi_a^{\beta'}Q_{\alpha'\beta'} + \Pu_{aa} - 16\pi \delu_a \phi\delu_a \phi - 8\pi c^2\minu_{aa} \phi^2,
\\
& \quad +  \frac{4x^a}{t^2} \delu_ah_{00} + \frac{4}{t} \delu_ah_{0a} - \frac{4x^a}{t^3}h_{0a}  + \bigg(\frac{2}{t^2} - \frac{6|x^a|^2}{t^4} \bigg)h_{00} + {\sl Cub}(0,0),
\\
\Boxt_g\hu_{ab} & =  \Phi_a^{\alpha'} \Phi_b^{\beta'}Q_{\alpha'\beta'} + \Pu_{ab} - 16\pi \delu_a \phi\delu_b\phi - 8\pi c^2\minu_{ab} \phi^2,
\\
& \quad +  \frac{2x^b}{t^2} \delu_ah_{00}  + \frac{2x^a}{t^2} \delu_bh_{00}  + \frac{2}{t} \delu_ah_{0b} + \frac{2}{t} \delu_bh_{0a} - \frac{6x^ax^b}{t^4}h_{00} 
\\
& \quad - \frac{2x^a}{t^3}h_{0b} - \frac{2x^b}{t^3}h_{0a}+ {\sl Cub}(0,0)
\\
&(a \neq b).
\endaligned
$$
The most important point is that for the components $\hu_{a \beta}$, the quasi-null terms $P_{\alpha \beta}$ become {\sl null terms}. This tensorial structure will lead us to the fact that these metric components do have better decay rate compared to $\hu_{00}$. In Section \ref{section-7}, these equations will be used to derive sharp decay estimates for these components. For clarity, we state the following conclusion:
\bel{eq tensorial 1}
\aligned
\Boxt_g\hu_{0a} & = \frac{2}{t} \delu_a h_{00} - \frac{2x^a}{t^3}h_{00}  + {\sl GQS}_h(0,0) + {\sl GQS}_\phi(0,0) + {\sl Cub}(0,0),
\\
\Boxt_g\hu_{aa} & = \frac{4x^a}{t^2} \delu_ah_{00}
+ \bigg(\frac{2}{t^2} - \frac{6|x^a|^2}{t^4} \bigg)h_{00} + \frac{4}{t} \delu_ah_{0a} - \frac{4x^a}{t^3}h_{0a}
\\
 & \quad +  {\sl GQS}_h(0,0) + {\sl GQS}_\phi(0,0) + {\sl Cub}(0,0),
\\
\Boxt_g\hu_{ab} & = \frac{2x^b}{t^2} \delu_ah_{00}  + \frac{2x^a}{t^2} \delu_bh_{00} - \frac{6x^ax^b}{t^4}h_{00}
+ \frac{2}{t} \delu_ah_{0b} - \frac{2x^a}{t^3}h_{0b} + \frac{2}{t} \delu_ah_{0a} - \frac{2x^b}{t^3}h_{0a}
\\
& \quad +  {\sl GQS}_h(0,0) + {\sl GQS}_{\phi}(0,0) + {\sl Cub}(0,0). 
\endaligned
\ee


\subsection{Wave gauge condition in the semi-hyperboloidal frame}

Our objective in the rest of this section is to establish some estimates based on the wave condition
$g^{\alpha \beta} \Gamma_{\alpha \beta}^{\gamma} = 0$, 
which is equivalent to saying 
\bel{eq wave-condition 0}
g_{\beta \gamma} \del_{\alpha}g^{\alpha \beta} = \frac{1}{2}g_{\alpha \beta} \del_{\gamma}g^{\alpha \beta}.
\ee
We have introduced 
\be
\aligned
&h^{\alpha \beta} = g^{\alpha \beta} - m^{\alpha \beta}, \quad h_{\alpha \beta} = g_{\alpha \beta} - m_{\alpha \beta}, 
\\
&\hu^{\alpha \beta} = \gu^{\alpha \beta} - \minu^{\alpha \beta}, \quad \hu_{\alpha \beta} = \gu_{\alpha \beta} - \minu_{\alpha \beta}, 
\endaligned
\ee
in which
$
\hu^{\alpha \beta} = h^{\alpha'\beta'} \Psi_{\alpha'}^{\alpha} \Psi_{\beta'}^{\beta}$
and $\hu_{\alpha \beta} = h_{\alpha'\beta'} \Phi_{\alpha}^{\alpha'} \Phi_{\beta}^{\beta'}$.

\begin{lemma} \label{lem wave-condition 1}
Let $(g_{\alpha \beta})$ be a metric satisfying the wave gauge condition \eqref{eq wave-condition 0}.  Then $\del_t\hu^{00}$ is a linear combination of the following terms with smooth and homogeneous coefficients of degree $\leq 0$:
\bel{eq wave-condition 0.5}
(s/t)^2\del_{\alpha} \hu^{\beta \gamma}, \quad \delu_a \hu^{\beta \gamma}, \quad t^{-1} \hu^{\alpha \beta}, \quad \hu^{\alpha \beta} \del_{\gamma} \hu^{\alpha'\beta'}, \quad t^{-1}h_{\alpha \beta} \hu^{\alpha'\beta'}.
\ee
\end{lemma}

\begin{proof} The wave gauge condition \eqref{eq wave-condition 0} can be written in the semi-hyperboloidal frame as
\bel{eq wave-condition 1}
\gu_{\beta \gamma} \delu_{\alpha} \hu^{\alpha \beta} + g_{\beta'\gamma'} \Phi_{\gamma}^{\gamma'} \hu^{\alpha \beta} \del_{\alpha'} \big(\Phi_{\alpha}^{\alpha'} \Phi_{\beta}^{\beta'} \big)
=
\frac{1}{2} \gu_{\alpha \beta} \delu_{\gamma} \hu^{\alpha \beta} +  \frac{1}{2}g_{\alpha \beta} \hu^{\alpha'\beta'} \delu_{\gamma} \big(\Phi_{\alpha'}^{\alpha} \Phi_{\beta'}^{\beta} \big).
\ee
This leads us to
\bel{eq wave-condition 2}
\minu_{\beta \gamma} \delu_{\alpha} \hu^{\alpha \beta} =
\frac{1}{2} \gu_{\alpha \beta} \delu_{\gamma} \hu^{\alpha \beta} +  \frac{1}{2}g_{\alpha \beta} \hu^{\alpha'\beta'} \delu_{\gamma} \big(\Phi_{\alpha'}^{\alpha} \Phi_{\beta'}^{\beta} \big)
- g_{\beta'\gamma'} \Phi_{\gamma}^{\gamma'} \hu^{\alpha \beta} \del_{\alpha'} \big(\Phi_{\alpha}^{\alpha'} \Phi_{\beta}^{\beta'} \big)
- \hu_{\beta \gamma} \delu_{\alpha} \hu^{\alpha \beta}. 
\ee
Taking $\gamma = c =1,2,3$,
we analyze the left-hand side and observe that
$$
\aligned
\minu_{\beta c} \delu_{\alpha} \hu^{\alpha \beta}& =  \minu_{0c} \delu_0\hu^{00} + m_{bc} \delu_0\hu^{0b}
+\minu_{\beta c} \delu_a \hu^{a \beta}, 
\endaligned
$$
which leads us to
$
\minu_{0c} \delu_0\hu^{00} = \minu_{\beta c} \delu_{\alpha} \hu^{\alpha \beta} - \minu_{bc} \delu_0\hu^{0b}
 - \minu_{\beta c} \delu_a \hu^{a \beta}, 
$
so that 
$$
\aligned
\minu^{0c} \minu_{0c} \delu_0\hu^{00} & =  \minu^{0c} \minu_{\beta c} \delu_{\alpha} \hu^{\alpha \beta}
- \minu^{0c} \minu_{bc} \delu_0\hu^{0b} - \minu^{0c} \minu_{\beta c} \delu_a \hu^{a \beta'}.
\endaligned
$$
An explicit calculation shows that
$
\minu^{0c} \minu_{0c} = \frac{r^2}{t^2}, \quad \minu^{0c} \minu_{bc} = -(s/t)^2(x^b/t)
$
and thus
\bel{eq wave-condition 3}
\aligned
(r/t)^2\delu_0\hu^{00} & =  \minu^{0c} \minu_{\beta c} \delu_{\alpha} \hu^{\alpha \beta}
+(s/t)^2\sum_b(x^b/t) \delu_0\hu^{0b} - \minu^{0c} \minu_{\beta c} \delu_a \hu^{a \beta'}.
\endaligned
\ee
Combining \eqref{eq wave-condition 2} and \eqref{eq wave-condition 3}, we find 
\be
\aligned
&(r/t)^2\delu_0\hu^{00} = 
(s/t)^2\sum_b(x^b/t) \delu_0\hu^{0b} - \minu^{0c} \minu_{\beta c} \delu_a \hu^{a \beta'}
\\
& + \minu^{0c} \bigg(\frac{1}{2} \gu_{\alpha \beta} \delu_c\hu^{\alpha \beta} +  \frac{1}{2}g_{\alpha \beta} \hu^{\alpha'\beta'} \delu_c\big(\Phi_{\alpha'}^{\alpha} \Phi_{\beta'}^{\beta} \big)
- g_{\beta'\gamma'} \Phi_c^{\gamma'} \hu^{\alpha \beta} \del_{\alpha'} \big(\Phi_{\alpha}^{\alpha'} \Phi_{\beta}^{\beta'} \big)
- \hu_{\beta c} \delu_{\alpha} \hu^{\alpha \beta} \bigg),
\endaligned
\ee
which leads us to the terms in \eqref{eq wave-condition 0.5}.
\end{proof}


We now proceed by deriving some estimates based on the wave gauge condition.
For convenience, we introduce the notation
$$
\aligned
&\big|\hu\big| := \max_{\alpha, \beta} \big|\hu_{\alpha \beta} \big|, \qquad \big|\del\hu\big|:=\max_{\alpha, \beta, \gamma} \big|\del_{\gamma} \hu_{\alpha \beta} \big|,
\qquad
\big|\delu \hu\big| := \max_{c, \alpha, \beta} \big|\delu_c\hu_{\alpha \beta} \big|, \quad c= 1,2,3.
\endaligned
$$
Observe that $\big|\delu \hu\big|$ contains only the ``good'' derivatives of $\hu_{\alpha \beta}$. 
When $\big|\del \hu\big|$ and $\big| \hu\big|$ are supposed to be small enough, 
and, the rest of this section, 
we express the corresponding bound in the form $\vep_w\leq 1$, the algebraic relation  between $\hu^{\alpha \beta}$ and $\hu_{\alpha \beta}$ leads us to the following basic estimates:
\bel{eq wave-condition 4}
\max_{\alpha, \beta} \big|\hu^{\alpha \beta} \big| \leq C\big|\hu\big|,
\qquad
\max_{\alpha, \beta, \gamma} \big|\del_{\gamma} \hu^{\alpha \beta} \big| \leq C\big|\del \hu\big|,
\qquad
\max_{c, \alpha, \beta} \big|\delu_c\hu^{\alpha \beta} \big| \leq C\big|\delu\hu\big|. 
\ee
With the above preparation, the following estimate is immediate from Lemma \ref{lem wave-condition 1}.

\begin{lemma}[Zero-order wave coordinate estimate]\label{lem wave-condition 3}
Let $g^{\alpha \beta} = m^{\alpha \beta} + h^{\alpha \beta}$ be a metric satisfying the wave gauge condition \eqref{eq wave-condition 0}. We suppose furthermore that  $\big|\del \hu\big|$ and $\big| \hu\big|$ are small enough so \eqref{eq wave-condition 4} hold. Then the following estimate holds:
\bel{ineq wave-condition 1}
\big|\delu_t\hu^{00} \big| \leq C(s/t)^2\big|\del \hu\big| + C\big|\delu\hu\big| + Ct^{-1} \big|\hu\big| + C\big|\del\hu\big| \, \big|\hu\big|.
\ee
\end{lemma}

The interest of this estimate is as follows: the ``bad'' derivative of $\hu^{00}$ is bounded by the ``good'' derivatives arising in the right-hand side of \eqref{ineq wave-condition 1}. Of course, the ``bad'' term $\big|\del \hu\big|$ still arise, but it is multiplied by the factor $(s/t)^2$ which provides us with extra decay and turns this term into a ``good'' term.

\begin{lemma}[$k$-order wave coordinate estimates]\label{lem wave-condition 4}
Let $g^{\alpha \beta} = m^{\alpha \beta} + h^{\alpha \beta}$ be a smooth metric satisfying the wave gauge condition \eqref{eq wave-condition 0}. We suppose furthermore that for a product $\del^IL^J$ with $|I|+|J| \leq N$, $\big|\del \del^IL^J \hu\big|$ and $\big|\del^IL^J \hu\big|$ are small enough so that the following bounds hold:
$
\max_{\alpha, \beta} \big|\del^IL^J\hu^{\alpha \beta} \big| \leq C\big|\del^IL^J \hu\big|$, 
$ 
\max_{\alpha, \beta, \gamma} \big|\del_{\gamma} \del^IL^J\hu^{\alpha \beta} \big| \leq C\big|\del \del^IL^J \hu\big|$, 
and 
$
\max_{c, \alpha, \beta} \big|\delu_c\del^IL^J\hu^{\alpha \beta} \big| \leq C\big|\delu\del^IL^J \hu\big|$.
Then the following estimate holds:
\bel{ineq wave-condition 2a}
\aligned
\big|\del^IL^J\del_t\hu^{00} \big| + \big|\del_t \del^IL^J\hu^{00} \big| 
\leq 
& C\sum_{|I'|+|J'| \leq |I|+|J|\atop |J'| \leq |J|} \big((s/t)^2\big|\del \del^{I'}L^{J'} \hu\big| + \big|\del^{I'}L^{J'} \delu\hu\big| + t^{-1} \big|\del^{I'}L^{J'} \hu\big|\big)
\\
&  + C\sum_{|I_1|+|I_2| \leq|I|\atop|J_1|+|J_2| \leq |J|} \big|\del^{I_1}L^{J_1} \hu\big| \, \big|\del\del^{I_2}L^{J_2} \hu\big|.
\endaligned
\ee 
\end{lemma}

\begin{proof}
This result is also a direct consequence of Lemma \ref{lem wave-condition 1}. We derive the expression of $\del_t \hu^{00}$ which is a linear combination of the terms in \eqref{eq wave-condition 0.5} with smooth and homogeneous coefficients of degree $\leq 0$. So, $\del^IL^J\del_t \hu^{00}$ is again a linear combination of the following terms with smooth and homogeneous coefficients of degree $\leq |I|$ (since $\del^IL^J$ acts on a $0$-homogeneous function gives a $|I|$-homogeneous function):

$\del^{I'}L^{J'} \big((s/t)^2\del_{\alpha} \hu^{\beta \gamma} \big)$,
$\del^{I'}L^{J'} \big(\delu_a \hu^{\beta \gamma} \big)$,
$t^{-1} \del^{I'}L^{J'} \big(\hu^{\alpha \beta} \big)$,
$\del^{I'}L^{J'} \big(\hu^{\alpha \beta} \del_{\gamma} \hu^{\alpha'\beta'} \big)$, 
$t^{-1} \del^{I'}L^{J'} \big(h_{\alpha \beta} \hu^{\alpha'\beta'} \big)$

\vskip.15cm

\noindent with $|I'| \leq |I|$ and $|J'| \leq |J|$. We observe that 
$$
|\del^{I'}L^{J'} \big((s/t)^2\del_{\alpha} \hu^{\beta \gamma} \big)| \leq C(s/t)^2\sum_{|I''| \leq|I'|\atop |J''| \leq |J'|}|\del^{I''}L^{J''} \big(\del_{\alpha} \hu^{\beta \gamma} \big)|. 
$$
The second, fourth, and last terms are to be bounded by the commutator estimates in Lemma \ref{lem com 2}.
The estimate for $\del_t\del^IL^J\hu^{00}$ is deduced from \eqref{ineq wave-condition 2a} and the commutator estimates.
\end{proof}


\subsection{Revisiting the structure of the quasi-null terms}

In this section, we consider the estimates on quasi-null terms $P_{\alpha \beta}$ together with the wave gauge condition and we use wave coordinate estimates. We treat first the term $\gu^{\alpha \alpha'} \del_t\gu_{\beta \beta'}$ and formulate the wave gauge condition in the form:
\bel{eq wave-condition 5} 
g^{\alpha \beta} \del_{\alpha}h_{\beta \gamma} = \frac{1}{2}g^{\alpha \beta} \del_{\gamma}h_{\alpha \beta}.
\ee

\begin{lemma} \label{lem wave-condition 2}
There exists a positive constant $\vep_w\geq 0$ such that if
$
|h| + |\del h| \leq \vep_w,
$
and the wave gauge condition \eqref{eq wave-condition 5} holds, then the quasi-null term $\gu^{\alpha \alpha'} \gu^{\beta \beta'} \del_t\gu_{\alpha \alpha'} \del_t\gu_{\beta \beta'}$ is a linear combination of terms  
\bel{eq wave-condition lem 2}
{\sl GQS}_h(0,0), \quad  {\sl Com}(0,0), \quad {\sl Cub}(0,0), \quad \gu^{0a} \delu_0\gu_{0a} \gu^{0b} \delu_0\gu_{0b}
\ee
with smooth and homogeneous coefficients of degree $\leq 0$.
\end{lemma}

\begin{proof} The relation \eqref{eq wave-condition 5} can be written in the semi-hyperboloidal frame in the form:
\be
\gu^{\alpha \beta} \delu_{\alpha} \hu_{\beta \gamma} + \Phi_{\gamma}^{\gamma'}g^{\alpha \beta} \del_{\alpha} \left(\Psi_{\beta}^{\beta'} \Psi_{\gamma'}^{\gamma''} \right) \hu_{\beta'\gamma''}
= \frac{1}{2} \gu^{\alpha \beta} \delu_{\gamma} \hu_{\alpha \beta} + \frac{1}{2}g^{\alpha \beta} \delu_{\gamma} \left(\Psi_{\alpha}^{\alpha'} \Psi_{\beta}^{\beta'} \right) \hu_{\alpha'\beta'}.
\ee
We fix $\gamma = 0$ and see that
$$
\aligned
\gu^{\alpha \beta} \del_t\hu_{\alpha \beta} & =  2\gu^{\alpha \beta} \delu_{\alpha} \hu_{0\beta} + 2\Phi_0^{\gamma'}g^{\alpha \beta} \del_{\alpha} \left(\Psi_{\beta}^{\beta'} \Psi_{\gamma'}^{\gamma''} \right) \hu_{\beta'\gamma''} - g^{\alpha \beta} \del_t\left(\Psi_{\alpha}^{\alpha'} \Psi_{\beta}^{\beta'} \right) \hu_{\alpha'\beta'}.
\endaligned
$$
This identity can be written as
\be
\aligned
\gu^{\alpha \beta} \del_t\hu_{\alpha \beta} & =  2\minu^{\alpha \beta} \delu_{\alpha} \hu_{\beta 0}
 +  2\hu^{\alpha \beta} \delu_\alpha \hu_{\beta 0}
+2\Phi_0^{\gamma'}m^{\alpha \beta} \del_{\alpha} \left(\Psi_{\beta}^{\beta'} \Psi_{\gamma'}^{\gamma''} \right) \hu_{\beta'\gamma''}
- m^{\alpha \beta} \del_t\left(\Psi_{\alpha}^{\alpha'} \Psi_{\beta}^{\beta'} \right) \hu_{\alpha'\beta'}
\\
& \quad + 2\Phi_0^{\gamma'}h^{\alpha \beta} \del_{\alpha} \left(\Psi_{\beta}^{\beta'} \Psi_{\gamma'}^{\gamma''} \right) \hu_{\beta'\gamma''}
- h^{\alpha \beta} \del_t\left(\Psi_{\alpha}^{\alpha'} \Psi_{\beta}^{\beta'} \right) \hu_{\alpha'\beta'}. 
\endaligned
\ee
In the right-hand side, except for the first term, we have at least quadratic terms or terms containing an extra decay factor such as $\del_{\alpha} \left(\Psi_{\beta}^{\beta'} \Psi_{\gamma}^{\gamma''} \right)$. So, we see that in
$\gu^{\alpha \alpha'} \gu^{\beta \beta'} \del_t\gu_{\alpha \alpha'} \del_t\gu_{\beta \beta'}$
the only term to be concerned about is
$$
4\minu^{\alpha \alpha'} \minu^{\beta \beta'} \delu_\alpha \hu_{\alpha' 0} \delu_{\beta} \hu_{\beta'0}.
$$
The remaining terms are quadratic in $\hu^{\alpha \beta}$, $\hu_{\alpha \beta}$ or linear terms on $\hu_{\alpha \beta}$ with decreasing coefficients such as $\del_{\alpha} \left(\Psi_{\beta}^{\beta'} \Psi_{\gamma}^{\gamma''} \right)$. Then we also see that when $|\hu|$ sufficiently small,  $\hu^{\alpha \beta}$ can be expressed as a power series of $\hu_{\alpha \beta}$  (without zero order), which is itself a linear combination of $h_{\alpha \beta}$ with smooth and homogeneous coefficients of degree $\leq 0$. So, when $|h|$ sufficiently small, $\hu^{\alpha \beta}$ can be expressed as a power series of $h_{\alpha \beta}$ (without $0$ order) with smooth and homogeneous coefficients of degree $\leq 0$. We conclude that in the product $\gu^{\alpha \alpha'} \gu^{\beta \beta'} \del_t\gu_{\alpha \alpha'} \del_t\gu_{\beta \beta'}$, the remaining terms apart from $4\minu^{\alpha \alpha'} \minu^{\beta \beta'} \delu_\alpha \hu_{\alpha' 0} \delu_{\beta} \hu_{\beta'0}$ are contained in ${\sl Cub}(0,0)$ or ${\sl Com}(0,0)$.

We focus on the term $4\minu^{\alpha \alpha'} \minu^{\beta \beta'} \delu_\alpha \hu_{\alpha' 0} \delu_{\beta} \hu_{\beta'0}$. We have 
$$
\aligned
&4\big(\minu^{\alpha \alpha'} \delu_{\alpha} \hu_{\alpha'0} \big) \big(\minu^{\beta \beta'} \delu_{\beta} \hu_{\beta'0} \big)
\\
& =  4\big(\minu^{a \alpha'} \delu_a \hu_{\alpha'0} + \minu^{00} \delu_0\hu_{00} + \minu^{0a'} \delu_0\hu_{0a'} \big)
\times
\big(\minu^{b\beta'} \delu_b\hu_{\beta'0} + \minu^{00} \delu_0\hu_{00} + \minu^{0b} \delu_0\hu_{0b} \big)
\\
& =  4\big(\minu^{a \alpha'} \delu_a \hu_{\alpha'0} + \minu^{00} \delu_0\hu_{00} \big) \big(\minu^{b\beta'} \delu_b\hu_{\beta'0} + \minu^{00} \delu_0\hu_{00} + \minu^{0b} \delu_0\hu_{0b} \big)
\\
& \quad + 4\minu^{0a'} \delu_0\hu_{0a'} \big(\hu^{b\beta'} \delu_b\hu_{\beta'0} + \minu^{00} \delu_0\hu_{00} \big)
   + 4 \minu^{0a'} \delu_0\hu_{0a'} \minu^{0b} \delu_0\hu_{0b}.
\endaligned
$$
The last term is already presented in the \eqref{eq wave-condition lem 2}. The remaining terms are null quadratic terms (recall that $\minu^{00} = (s/t)^2$). 
\end{proof}

Now we combine Lemma \ref{lem P 1} with Lemmas \ref{lem wave-condition 1} and \ref{lem wave-condition 2}.

\begin{lemma} \label{lem P 2}
There exists a positive constant $\vep_w>0$ such that if
$
|h| + |\del h| \leq \vep_w,
$
then the quasi-null term $\Pu_{00}$ is a linear combination of the following terms with smooth and homogeneous coefficients of order $\leq 0$:
\bel{eq P 2}
{\sl GQS}_h(0,0), \quad {\sl Cub}(0,0), \quad {\sl Com}(0,0), \quad \del_t\hu_{a \alpha} \del_t\hu_{b\beta}.
\ee
The term $\Pu_{a \beta}$ is a linear combination of the following terms with smooth and homogeneous coefficients of order $\leq 0$:
\bel{eq P 2.5}
{\sl GQS}_h(0,0), \quad {\sl Cub}(0,0), \quad {\sl Com}(0,0).
\ee
\end{lemma}

\begin{proof} In view of Lemma \ref{lem P 1}, we need to focus on 
$
\gu^{\gamma \gamma'} \gu^{\delta \delta'} \del_t\hu_{\gamma \gamma'} \del_t\hu_{\delta \delta'}$
and 
$\minu^{\gamma \gamma'} \minu^{\delta \delta'} \del_t\hu_{\gamma \delta} \del_t\hu_{\gamma'\delta'}$. 
The first term is covered by Lemma \ref{lem wave-condition 2} and the second term is bounded as follows: we recall that
$$
\big|\del^IL^J \minu^{00} \big| = C(I,J)(s/t)^2, \quad \big|\minu^{\alpha \beta} \big| \leq C.
$$
Then, when $(\gamma, \gamma') =(0,0)$ or $(\delta, \delta') =(0,0)$, we have $\minu^{\gamma \gamma'} \minu^{\delta \delta'} \del_t\hu_{\gamma \delta} \del_t\hu_{\gamma'\delta'}$ becomes a null term. When $(\gamma, \gamma') \neq (0,0)$ and $(\delta, \delta') \neq (0,0)$, we denote by $(\gamma, \gamma') =(a, \alpha)$ and $(\delta, \delta') =(b, \beta)$, so we see that $\minu^{\gamma \gamma'} \minu^{\delta \delta'} \del_t\hu_{\gamma \delta} \del_t\hu_{\gamma'\delta'}$ is a linear combination of $\del_t\hu_{a \alpha} \del_t\hu_{b\beta}$ with homogeneous coefficients of degree zero.
\end{proof}

Finally, we emphasize that, in order to control the quasi-null terms, we must control the term $\del_t\hu_{a \alpha} \del_t\hu_{b\beta}$ which is {\sl not} a null term. This term will be bounded by refined decay estimates on $\del\hu_{a \alpha}$, and we refer to our forthcoming analysis in Section \ref{section-7}. 


\section{Initialization of the Bootstrap Argument}
\label{section-4}

\subsection{The bootstrap assumption and the basic estimates}

\subsubsection*{The bootstrap assumption}

 From now on, we assume that in a hyperbolic time interval $[2,s^*]$, the following energy bounds hold for $|I|+|J| \leq N$. Here $N\geq 14$, $(C_1, \vep)$ is a pair of positive constants and $1/50\leq \delta \leq 1/20$, say.
\begin{subequations} \label{ineq energy assumption 1}
\bel{ineq energy assumption 11}
E^*_M(s, \del^IL^Jh_{\alpha \beta})^{1/2} \leq C_1\vep s^{\delta},
\ee
\bel{ineq energy assumption 12}
E_{M, c^2}(s, \del^IL^J\phi)^{1/2} \leq C_1\vep s^{1/2 +\delta}. 
\ee
\end{subequations}
For $|I|+|J| \leq N-4$ we have (in which \eqref{ineq energy assumption 21} is repeated from \eqref{ineq energy assumption 11} for clarity in the presentation) 
\begin{subequations} \label{ineq energy assumption 2}
\bel{ineq energy assumption 21}
E^*_M(s, \del^IL^Jh_{\alpha \beta})^{1/2} \leq C_1\vep s^{\delta},
\ee
\bel{ineq energy assumption 22}
E_{M, c^2}(s, \del^IL^J\phi)^{1/2} \leq C_1\vep s^{\delta}.
\ee
\end{subequations}
In combination with Lemma \ref{lem energy 3}, we see that the total energy of $h_{\alpha \beta}$ on the hyperboloid $\Hcal_s$ is bounded by
\bel{ineq energy assumption 3}
E_M(s, \del^IL^Jh_{\alpha \beta}) \leq CC_1\vep s^{\delta} + C m_S\leq 2C_1\vep s^{\delta},
\ee 
where we take $m_S\leq \vep$. In the following discussion, except if specified otherwise, the letter $C$ always represents a constant depending only on $N$. This constant may change at each occurrence.

\subsubsection*{Basic $L^2$ estimates of the first generation}
These estimates come directly from the above energy bounds.

For $|I|+|J| \leq N$, we have 
\begin{subequations} \label{ineq basic-L2 1 generation 1}
\bel{ineq basic-L2 1 generation 1 a}
\|(s/t) \del_\gamma \del^IL^J h_{\alpha \beta} \|_{L_f^2(\Hcal_s)} + \|\delu_a \del^IL^J h_{\alpha \beta} \|_{L_f^2(\Hcal_s)} \leq CC_1\vep s^{\delta},
\ee
\bel{ineq basic-L2 1 generation 1 b}
\|(s/t) \del_{\alpha} \del^IL^J \phi\|_{L_f^2(\Hcal_s)} + \|\delu_a \del^IL^J \phi\|_{L_f^2(\Hcal_s)} \leq CC_1\vep s^{1/2 +\delta},
\ee
\bel{ineq basic-L2 1 generation 1 c}
\|\del^IL^J \phi\|_{L_f^2(\Hcal_s)} \leq CC_1\vep s^{1/2 +\delta}.
\ee
\end{subequations}
For $|I|+|J| \leq N-1$, we have 
\bel{ineq basic-L2 1 generation 2}
\|\del_{\alpha} \del^IL^J \phi\|_{L_f^2(\Hcal_s)} \leq CC_1\vep s^{1/2 +\delta}.
\ee
For $|I|+|J| \leq N-4$, we have 
\bel{ineq basic-L2 1 generation 3}
\|(s/t) \del_{\alpha} \del^IL^J \phi\|_{L_f^2(\Hcal_s)} + \|\delu_a \del^IL^J \phi\|_{L_f^2(\Hcal_s)} \leq CC_1\vep s^{\delta}
\ee
and, for $|I|+|J| \leq N-5$, 
\bel{ineq basic-L2 1 generation 4}
\|\del_{\alpha} \del^IL^J \phi\|_{L_f^2(\Hcal_s)} \leq CC_1\vep s^{\delta}.
\ee


\subsubsection*{Basic $L^2$ estimates of the second generation}

These estimates come from the above $L^2$ bounds of the first generation combined with the commutator estimates presented in Lemma \ref{lem com 2}. 
For $|I|+|J| \leq N$, we obtain 
\begin{subequations} \label{ineq basic-L2 2 generation 1}
\bel{ineq basic-L2 2 generation 1 a}
\|(s/t) \del^IL^J\del_\gamma h_{\alpha \beta} \|_{L_f^2(\Hcal_s)} + \|\del^IL^J \delu_ah_{\alpha \beta} \|_{L_f^2(\Hcal_s)} \leq CC_1\vep s^{\delta},
\ee
\bel{ineq basic-L2 2 generation 1 b}
\|(s/t) \del^IL^J\del_{\alpha} \phi\|_{L_f^2(\Hcal_s)} + \|\del^IL^J \delu_a \phi\|_{L_f^2(\Hcal_s)} \leq CC_1\vep s^{1/2 +\delta},
\ee
\end{subequations}
while for $|I|+|J| \leq N-1$ (the second term in the left-hand side being bounded by \eqref{ineq 1 homo}) 
\bel{ineq basic-L2 2 generation 2}
\|\del^IL^J\del_{\alpha} \phi\|_{L_f^2(\Hcal_s)} + \|t \del^IL^J\delu_a \phi\|_{L_f^2(\Hcal_s)} \leq CC_1\vep s^{1/2 +\delta}. 
\ee

For $|I|+|J| \leq N-4$, we have 
\bel{ineq basic-L2 2 generation 3}
\|(s/t) \del^IL^J \del_{\alpha} \phi\|_{L_f^2(\Hcal_s)} + \|\del^IL^J\delu_a \phi\|_{L_f^2(\Hcal_s)} \leq CC_1\vep s^{\delta},
\ee
while, for $|I|+|J| \leq N-5$, again from \eqref{ineq 1 homo})
\bel{ineq basic-L2 2 generation 4}
\|\del^IL^J\del_{\alpha} \phi\|_{L_f^2(\Hcal_s)} + \|t \del^IL^J\delu_a \phi\|_{L_f^2(\Hcal_s)} \leq CC_1\vep s^{\delta}.
\ee


\subsubsection*{Basic $L^\infty$ estimates of the first generation}

For $|I|+|J| \leq N-2$, we obtain 
\begin{subequations} \label{ineq basic-sup 1 generation 1}
\bel{ineq basic-sup 1 generation 1 a}
\sup_{\Hcal_s^*} \big( t^{3/2} (s/t) \del_\gamma \del^IL^J h_{\alpha \beta} \big)
+ \sup_{\Hcal_s^*} \big( t^{3/2}  \delu_a \del^IL^J h_{\alpha \beta} \big) \leq CC_1\vep s^{\delta},
\ee
\bel{ineq basic-sup 1 generation 1 b}
\sup_{\Hcal_s} \big(t^{3/2}  (s/t) \del_{\alpha} \del^IL^J \phi\big)
+ \sup_{\Hcal_s} \big(t^{3/2}  \delu_a \del^IL^J \phi\big) \leq CC_1\vep s^{1/2 +\delta},
\ee
\bel{ineq basic-sup 1 generation 1 c}
\sup_{\Hcal_s} \big(t^{3/2} \del^IL^J \phi\big) \leq CC_1\vep s^{1/2 +\delta}.
\ee
\end{subequations}
For $|I|+|J| \leq N-3$, we have 
\bel{ineq basic-sup 1 generation 2}
\sup_{\Hcal_s} \big( t^{3/2}  \del_{\alpha} \del^IL^J \phi\big)
+\sup_{\Hcal_s} \big( t^{5/2} \delu_a \del^IL^J \phi\big) \leq CC_1\vep s^{1/2 +\delta}.
\ee
Here, the second term in the left-hand side is bounded by applying \eqref{ineq 1 homo} once more.
For $|I|+|J| \leq N-6$, we have 
\bel{ineq basic-sup 1 generation 3}
\sup_{\Hcal_s} \big(t^{3/2}  (s/t) \del_{\alpha} \del^IL^J \phi\big)
+\sup_{\Hcal_s} \big( t^{3/2}  \delu_a \del^IL^J \phi\big) \leq CC_1\vep s^{\delta}, 
\ee
while, for $|I|+|J| \leq N-7$,
\bel{ineq basic-sup 1 generation 4}
\sup_{\Hcal_s} \big( t^{3/2}  \del_{\alpha} \del^IL^J \phi\big) + \sup_{\Hcal_s} \big(t^{5/2} \delu_a \del^IL^J \phi\big) \leq CC_1\vep s^{\delta}.
\ee


\subsubsection*{Basic $L^\infty$ estimates of the second generation}

For $|I|+|J| \leq N-2$, we obtain 
\begin{subequations} \label{ineq basic-sup 2 generation 1}
\bel{ineq basic-sup 2 generation 1 a}
\sup_{\Hcal_s^*} \big(t^{1/2}|\del^IL^J\del_\gamma h_{\alpha \beta}|\big) \leq CC_1\vep s^{-1 + \delta},
\qquad
\sup_{\Hcal_s^*} \big(t^{3/2}|\del^IL^J\delu_a h_{\alpha \beta}|\big) \leq CC_1\vep s^{\delta},
\ee
\bel{ineq basic-sup 2 generation 1 b}
\sup_{\Hcal_s} \big(t^{1/2}|\del^IL^J\del_{\alpha} \phi|\big) \leq CC_1\vep s^{-1/2 + \delta},
\qquad
\sup_{\Hcal_s} \big(t^{3/2}|\del^IL^J\delu_a \phi|\big) \leq CC_1\vep s^{1/2 +\delta},
\ee
\bel{ineq basic-sup 2 generation 1 c}
\sup_{\Hcal_s} \big(t^{3/2}|\del^IL^J \phi|\big) \leq CC_1\vep s^{1/2 +\delta}.
\ee
\end{subequations}
For $|I|+|J| \leq N-3$, we have 
\bel{ineq basic-sup 2 generation 2}
\sup_{\Hcal_s} \big(t^{3/2}|\del^IL^J\del_{\alpha} \phi|\big)+\sup_{\Hcal_s} \big(t^{5/2}|\del^IL^J\delu_a \phi|\big) \leq CC_1\vep s^{1/2 +\delta}, 
\ee
while, for $|I|+|J| \leq N-6$,
\begin{subequations} \label{ineq basic-sup 2 generation 3}
\bel{ineq basic-sup 2 generation 3 a}
\sup_{\Hcal_s} \big(t^{1/2}|\del^IL^J \del_{\alpha} \phi|\big) \leq CC_1\vep s^{-1+\delta},
 \qquad
\sup_{\Hcal_s} \big(t^{3/2}|\del^IL^J \delu_a \phi|\big) \leq CC_1\vep s^{\delta},
\ee
\bel{ineq basic-sup 2 generation 3 b}
\sup_{\Hcal_s} \big(t^{3/2}|\del^IL^J \phi|\big) \leq CC_1\vep s^{\delta}.
\ee
\end{subequations}
For $|I|+|J| \leq N-7$, we find 
\bel{ineq basic-sup 2 generation 4}
\sup_{\Hcal_s} \big(t^{3/2}|\del^IL^J \del_{\alpha} \phi|\big)+
\sup_{\Hcal_s} \big(t^{5/2}|\del^IL^J\delu_a \phi|\big) \leq CC_1\vep s^{\delta}.
\ee
By \eqref{ineq 2 homo} and \eqref{ineq 3 homo}, the following bounds hold:
\bel{ineq L-2 homo-ineq 1}
\|\del^IL^J\delu_a \del_{\beta'}h_{\alpha \beta} \|_{L^2(\Hcal_s^*)} + \|\del^IL^J\del_{\beta'} \delu_ah_{\alpha \beta} \|_{L^2(\Hcal_s^*)} \leq CC_1\vep s^{-1+\delta},
\ee
\bel{ineq sup homo-ineq 1}
\sup_{\Hcal_s^*} \left(t^{3/2} \big|\del^IL^J\delu_a \del_{\beta'}h_{\alpha \beta} \big|\right) + \sup_{\Hcal_s^*} \left(t^{3/2} \big|\del^IL^J\del_{\beta'} \delu_ah_{\alpha \beta} \big|\right) \leq CC_1\vep s^{-1+\delta}.
\ee


\subsection{Estimates based on integration along radial rays} \label{subsec ray-inte}
For $|I|+|J| \leq N-2$,
\bel{ineq basic-sup-h}
|\del^IL^J h_{\alpha \beta}(t, x)| \leq CC_1\vep (s/t)t^{-1/2}s^{\delta} + Cm_St^{-1} \leq CC_1\vep(s/t)t^{-1/2}s^{\delta}.
\ee
This estimate is based on the following observation: 
$$
\big|\del_r\del^IL^J h_{\alpha \beta}(t, x) \big| \leq C\big|\del_\gamma \del^IL^J h_{\alpha \beta}(t, x) \big| 
\leq CC_1\vep t^{-1/2}s^{-1+\delta} \simeq CC_1\vep t^{-1+\delta/2}(t-r)^{-1/2 +\delta/2}.
$$
Then we integrate $\del_r\del^IL^J h_{\alpha \beta}$ along the radial rays $\{(t, \lambda x)|1\leq \lambda \leq (t-1)/|x|\}$. We see when $\lambda = (t-1)/|x|$, $\del_r\del^IL^J h_{\alpha \beta}(t, \lambda x) \simeq Cm_St^{-1}$ since $h_{\alpha \beta}$ coincides with the Schwarzschild metric and, by integration, \eqref{ineq basic-sup-h} holds. 


\section{Direct Control of Nonlinearities in the Einstein Equations}
\label{section-5}

\subsection{$L^\infty$ estimates}

With the above estimates, we are in a position to control the good nonlinear terms:
${\sl GQQ}_{hh}$, ${\sl GQQ}_{h\phi}$, ${\sl GQS}_h$, ${\sl QS}_{\phi}$, ${\sl Com}$, and ${\sl Cub}$.  

\begin{lemma} \label{lem 0 biliner}
When the basic sup-norm estimates hold, the following sup-norm estimates are valid for $k\leq N-2$: 
\bel{ineq 1 bilinear-sup}
\aligned
|{\sl GQS}_h(N-2,k)| \leq C(C_1\vep)^2t^{-2}s^{-1+2\delta}, \quad
|{\sl GQQ}_{hh}(N-2,k)| \leq C(C_1\vep)^2t^{-3}s^{2\delta},
\endaligned
\ee
\bel{ineq 2 bilinear-sup}
|{\sl QS}_{\phi}(N-2,k)| \leq C(C_1\vep)^2t^{-2}s^{-1/2 +2\delta},
\ee
\bel{ineq 2.5 bilinear-sup}
|{\sl GQQ}_{h\phi}(N-2,k)| \leq C(C_1\vep)^2 t^{-3}s^{2\delta}, 
\ee
\bel{ineq 3 bilinear-sup}
|{\sl Com}(N-2,k)| \leq C(C_1\vep)^2t^{-5/2}s^{-1+2\delta},
\ee
\bel{ineq 4 bilinear-sup}
%
|{\sl Cub}(N-2,k)| \leq C(C_1\vep)^2 t^{-5/2}s^{3\delta}.
\ee
\end{lemma}

\begin{proof} We directly substitute the basic $L^\infty$ estimates, and we begin
$$
|{\sl GQS}_h(N-2,k) |
\leq  |(s/t)^2 \del_t h \del_t h |
+ \sum_{I_1+I_2 = I\atop J_1+J_2=J} 
|\del^{I_1}L^{J_1} \delu_ah_{\alpha \beta} \del^{I_2}L^{J_2} \delu_{\nu}h_{\alpha'\beta'}|. 
$$
By the basic decay estimate \eqref{ineq basic-sup 2 generation 1 a}, $|{\sl GQS}_h(N-2,k)|$ is bounded by $C(C_1\vep)^2t^{-2}s^{-1+2\delta}$.
The estimate for ${\sl GQQ}_{hh}$ is similar, where \eqref{ineq sup homo-ineq 1} is applied, and we omit the details. 
The estimate for ${\sl QS}_{\phi}$ is more delicate and we have 
$
\del^IL^J\left(\del_{\mu} \phi\del_{\nu} \phi\right)
= \sum_{I_1+I_2=I\atop J_1+J_2=J} \del^{I_1}L^{J_1} \del_{\mu} \phi\del^{I_2}L^{J_2} \del_{\nu} \phi$.
\begin{itemize}

\item $I_1=I, J_1=J$ then $|I_2|=|J_2|=0\leq N-7$. Then we apply \eqref{ineq basic-sup 2 generation 1 b} and
\eqref{ineq basic-sup 2 generation 4} we have
$$
\left|\del^{I_1}L^{J_1} \del_{\mu} \phi\del^{I_2}L^{J_2} \del_{\nu} \phi\right| \leq C(C_1\vep)^2t^{-2}s^{-1/2 +2\delta}.
$$

\item $N-3\geq |I_1|+|J_1|\geq N-5$ then $|I_2|+|J_2| \leq 3\leq N-6$, then we apply \eqref{ineq basic-sup 2 generation 2} and
\eqref{ineq basic-sup 2 generation 3 a}.

\item $|I_1|+|J_1| = N-6$, this leads us to $|I_2|+|J_2| \leq 4\leq N-3$, then we apply \eqref{ineq basic-sup 2 generation 3 a} and \eqref{ineq basic-sup 2 generation 2}.

\item $|I_1|+|J_1| \leq N-7$, this leads us to $|I_2|+|J_2| \leq N-2$, then we apply \eqref{ineq basic-sup 2 generation 4} and \eqref{ineq basic-sup 2 generation 1 b}.
\end{itemize}
The estimate of $\del^IL^J\left(\phi^2\right)$ is similar and we omit the details.

The estimate for ${\sl Com}$ is much simpler, due to the additional decay $t^{-1}$. We apply the above estimates to ${\sl QS}_\phi$ and the basic sup-norm estimate directly.
For the cubic term, we will not analyze each type but point out that the worst higher-order term is $h_{\alpha \beta}(\del\phi)^2$, since $\del^IL^J\del_{\alpha} \phi$ has a decay $\simeq t^{-3/2}s^{1/2 +\delta}$, but this term is found to be bounded by $t^{-5/2}(s/t)s^{3\delta}$.
\end{proof}


\subsection{$L^2$ estimates}

\begin{lemma} one has 
\bel{ineq 1 bilinear-L2}
\|{\sl GQQ}_{hh}(N,k) \|_{L^2(\Hcal_s^*)} \leq C(C_1\vep)^2s^{-3/2 +2\delta}, 
\ee
\bel{ineq 2 bilinear-L2}
\|{\sl GQS}_h(N,k) \|_{L^2(\Hcal_s^*)} \leq C(C_1\vep)^2s^{-3/2 +2\delta}, 
\ee
\bel{ineq 3 bilinear-L2}
\|{\sl QS}_\phi(N-4,k) \|_{L^2(\Hcal_s^*)} \leq C(C_1\vep)^2s^{-3/2 +2\delta}, 
\ee
\bel{ineq 4 bilinear-L2}
\|{\sl GQQ}_{h\phi}(N-4,k) \|_{L^2(\Hcal_s^*)} \leq C(C_1\vep)^2s^{-3/2 +2\delta},
\ee
\bel{ineq 5 bilinear-L2}
\|{\sl Cub} \|_{L_f^2(\Hcal_s)} \leq C(C_1\vep)^2s^{-3/2 +3\delta}.
\ee 
\end{lemma} 

\begin{proof}
For the term ${\sl GQQ}_{hh}$, we will only write the estimate of $\del^{I_1}L^{J_1}h_{\alpha'\beta'} \del^{I_2}L^{J_2} \delu_a \delu_{\nu}h_{\alpha \beta}$ in detail and, to this end, we distinguish between two main cases:

\vskip.15cm

\noindent {\sl Case 1.} $|I_1|\geq 1$.   {\sl Subcase 1.1} : When $|I_1|+|J_1| \leq N-2$, we obtain 
$$
\aligned
\big\|\del^{I_1}L^{J_1}h_{\alpha'\beta'} \del^{I_2}L^{J_2} \delu_a \delu_{\nu}h_{\alpha \beta} \big\|_{L^2(\Hcal_s^*)}
& \leq CC_1\vep \big\|t^{-1/2}s^{-1+\delta} \del^{I_2}L^{J_2} \delu_a \delu_{\nu}h_{\alpha \beta} \big\|_{L^2(\Hcal_s^*)}
\\
& \leq CC_1\vep s^{-3/2 +\delta}E_M^*(s, \del^{I_2}L^{J_2} \delu h)^{1/2}
\\
& \leq C(C_1\vep)^2s^{-3/2 +2\delta}.
\endaligned
$$

\noindent {\sl Subcase 1.2} : When $N\geq |I_1|+|J_2|\geq N-1$, we have $|I_2|+|J_2| \leq 1\leq N-3$,
then in view of \eqref{ineq L-2 homo-ineq 1}
$$
\aligned
\big\|\del^{I_1}L^{J_1}h_{\alpha'\beta'} \del^{I_2}L^{J_2} \delu_a \delu_{\nu}h_{\alpha \beta} \big\|_{L^2(\Hcal_s^*)}
& \leq  CC_1\vep \big\|t^{-3/2}s^{-1+\delta}(t/s) \big|(s/t) \del^{I_1}L^{J_1}h_{\alpha'\beta'} \big\|_{L^2(\Hcal_s^*)}
\\
& \leq  CC_1\vep s^{-3/2 +\delta} \big\|(s/t) \del^{I_1}L^{J_1}h_{\alpha'\beta'} \big\|_{L^2(\Hcal_s^*)}
  \leq  C(C_1\vep)^2 s^{-3/2 +2\delta}.
\endaligned
$$

\vskip.15cm

\noindent {\sl Case 2.} $|I_1|=0$. 
{\sl Subcase 2.1} : When $|J_1| \leq N-2$, then in view of \eqref{ineq L-2 homo-ineq 1} we obtain 
$$
\aligned
\big\|L^{J_1}h_{\alpha'\beta'} \del^IL^{J_2} \delu_a \delu_{\nu}h_{\alpha \beta} \big\|_{L^2(\Hcal_s^*)}
& \leq  CC_1\vep\big\|\big((s/t)t^{-1/2}s^{\delta} + t^{-1} \big) \del^IL^{J_2} \delu_a \delu_{\nu}h_{\alpha'\beta'} \|_{L^2(\Hcal_s^*)}
\\
& \leq  CC_1\vep\big\|\big((s/t)t^{-1/2}s^{\delta} + t^{-1} \big)s^{-1} \,  |s\del^IL^{J_2} \delu_a \delu_{\nu}h_{\alpha'\beta'}|\|_{L^2(\Hcal_s^*)}
\\
& \leq  CC_1\vep s^{-3/2 +\delta}E_M^*(s, \del^IL^{J_2} \delu h)^{1/2} \leq C(C_1\vep)^2s^{-3/2 +2\delta}.
\endaligned
$$

\noindent {\sl Subcase 2.2} : When $N\geq |J_1|\geq N-1\geq 1$, then we denote by $L^{J_1} = L_aL^{J_1'}$, we have $|I|+|J_2| \leq 1\leq N-3$. Then in view of \eqref{ineq sup homo-ineq 1}
$$
\aligned
\big\|L^{J_1}h_{\alpha'\beta'} \del^IL^{J_2} \delu_a \delu_{\nu}h_{\alpha \beta} \big\|_{L^2(\Hcal_s^*)}
& \leq CC_1\vep \big\|t^{-3/2}s^{-1+\delta}L_aL^{J_1'}h_{\alpha'\beta'} \big\|_{L^2(\Hcal_s^*)}
\\
& \leq  CC_1\vep \big\|t^{-1/2}s^{-1+\delta} \delu_aL^{J_1'}h_{\alpha'\beta'} \big\|_{L^2(\Hcal_s^*)}
\leq C(C_1\vep)^2 s^{-3/2 +2\delta}.
\endaligned
$$
The estimate on the term ${\sl GQS}_h$ is similar, and we omit the details.
For the estimate for ${\sl QS}_{\phi}(N-4,k)$, we will only writhe the proof on $\del^IL^J\left(\del_{\alpha} \phi\del_{\beta} \phi\right)$.  For $N\geq 9$, we have $\left[\frac{N-4}{2} \right]\leq N-7$. So, at least $|I_1|+|J_1| \leq N-7$ or $|I_2|+|J_2| \leq N-7$:
$$
\aligned
\big\|\del^{I_1}L^{J_1} \del_{\alpha} \phi \,   \del^{I_2}L^{J_2} \phi\big\|_{L^2(\Hcal_s^*)}
& \leq  CC_1\vep \big\|t^{-3/2}s^{\delta}(t/s) \,   (s/t) \del^{I_2}L^{J_2} \phi \big\|_{L^2(\Hcal_s^*)}
  \leq  C(C_1\vep)^2 s^{-3/2 +2\delta}.
\endaligned
$$
As far as ${\sl GQQ}_{h\phi}(N-4,k)$ is concerned, we only treat 
$\del^{I_1}L^{J_1}h_{\alpha'\beta'} \del^{I_2}L^{J_2} \delu_a \delu_{\mu} \phi.$
We observe that $|I_1|+|J_1| \leq N-4$ and by applying \eqref{ineq basic-sup-h}
$$
\aligned
\big\|\del^{I_1}L^{J_1}h_{\alpha'\beta'} \del^{I_2}L^{J_2} \delu_a \delu_{\mu} \phi\big\|_{L^2(\Hcal_s^*)}
& \leq  \Big\|\big((s/t)t^{-1/2}s^{\delta} + t^{-1} \big)s^{-1} \,  \big(s\del^{I_2}L^{J_2} \delu_a \delu_{\mu} \phi\big) \Big\|_{L^2(\Hcal_s^*)}
\\
& \leq  CC_1\vep s^{-3/2 +\delta} \big\|s\del^{I_2}L^{J_2} \delu_a \delu_{\mu} \phi\big\|_{L^2(\Hcal_s^*)}
\\
& \leq  CC_1\vep s^{-3/2 +\delta}E_{M,c^2} \big(s, \del^{I_2}L^{J_2} L_a \delu_{\mu} \phi\big)^{1/2} \leq C(C_1\vep)^2s^{-3/2 +2\delta}.
\endaligned
$$
The higher-order terms ${\sl Cub}$ are bounded as we did for the sup-norm: just observe that the worst term is again $h(\del\phi)^2$ and can be bounded as stated.
\end{proof}

\begin{lemma}
For $N\geq 7$, one has 
\bel{ineq 5 bilinear-L2-DEUX}
\|{\sl QS}_\phi(N,k) \|_{L^2(\Hcal_s^*)} \leq C(C_1\vep)^2s^{-1+2\delta}.
\ee
\end{lemma}

\begin{proof} We discuss the following cases:
\begin{itemize}

\item $|I_1|+|J_1| = N$, $N-7\geq 0$. So, in view of \eqref{ineq basic-L2 2 generation 1 b} and \eqref{ineq basic-sup 2 generation 4} :
$$
\aligned
\left\|\del^{I_1}L^{J_1} \del_{\gamma} \phi\del^{I_2}L^{J_2} \del_{\gamma'} \phi\right\|_{L^2(\Hcal_s^*)}
& \leq  CC_1\vep\left\|t^{-3/2}s^{\delta}(t/s) \,   (s/t) \del^{I_1}L^{J_1} \del_{\gamma} \phi\right\|_{L^2(\Hcal_s^*)}
\\
& \leq  CC_1\vep s^{-3/2 +\delta} \,   CC_1\vep s^{1/2 +\delta} \leq C(C_1\vep)^2 s^{-1+2\delta}.
\endaligned
$$

\item $|I_1|+|J_1|= N-1$, then $|I_2|+|J_2|= 1\leq N-6$. So, in view of \eqref{ineq basic-L2 2 generation 2} and \eqref{ineq basic-sup 2 generation 3 a}, we have 
$$
\aligned
\left\|\del^{I_1}L^{J_1} \del_{\gamma} \phi\del^{I_2}L^{J_2} \del_{\gamma'} \phi\right\|_{L^2(\Hcal_s^*)}
& \leq  CC_1\vep\left\|t^{-1/2}s^{-1+\delta} \,   \del^{I_1}L^{J_1} \del_{\gamma} \phi\right\|_{L^2(\Hcal_s^*)}
\\
& \leq  CC_1\vep s^{-3/2 +\delta} \,   CC_1\vep s^{1/2 +\delta} \leq C(C_1\vep)^2s^{-1+2\delta}.
\endaligned
$$

\item $|I_1|+|J_1| = N-2$, then $|I_2|+|J_2| =2\leq N-5$. So, in view of \eqref{ineq basic-sup 2 generation 1 a} and \eqref{ineq basic-L2 2 generation 4}, we have 
$$
\aligned
\left\|\del^{I_1}L^{J_1} \del_{\gamma} \phi\del^{I_2}L^{J_2} \del_{\gamma'} \phi\right\|_{L^2(\Hcal_s^*)}
& \leq  CC_1\vep\left\|t^{-1/2}s^{-1/2 +\delta} \,   \del^{I_2}L^{J_2} \del_{\gamma} \phi\right\|_{L^2(\Hcal_s^*)}
\\
& \leq  CC_1\vep s^{-1+\delta} \,   CC_1\vep s^{\delta} \leq C(C_1\vep)^2s^{-1+2\delta}.
\endaligned
$$

\item $|I_1|+|J_1| = N-3$, then $|I_2|+|J_2| =3\leq N-4$. So,  in view of \eqref{ineq basic-sup 2 generation 2} and \eqref{ineq basic-L2 2 generation 3}, we have 
$$
\aligned
\left\|\del^{I_1}L^{J_1} \del_{\gamma} \phi\del^{I_2}L^{J_2} \del_{\gamma'} \phi\right\|_{L^2(\Hcal_s^*)}
& \leq  CC_1\vep\left\|t^{-3/2}s^{1/2 +\delta}(t/s) \,   (s/t) \del^{I_2}L^{J_2} \del_{\gamma} \phi\right\|_{L^2(\Hcal_s^*)}
\\
& \leq  CC_1\vep s^{-1+\delta} \,   CC_1\vep s^{\delta} \leq C(C_1\vep)^2s^{-1+2\delta}.
\endaligned
$$

\item When $|I_1|+|J_1| \leq N-4\leq 3$, we exchange the role of $I_1,I_2$ and $J_1, J_2$, and apply the arguments above again.
\end{itemize} 
\end{proof}


\section{Direct Consequences of the Wave Gauge Condition}
\label{sect--7}

\subsection{$L^\infty$ estimates}
 
We now use the wave coordinate estimates \eqref{ineq wave-condition 1} and \eqref{ineq wave-condition 2a}. Combined with Proposition \ref{prop Hardy 2}, they provide us with rather precise $L^2$ estimates and $L^\infty$ estimate on the gradient of the metric coefficient $\hu^{00}$. In view of these estimates, we can say (as in \cite{PLF-YM-one}) that the quasi-linear terms
${\sl QQ}_{hh}$ and ${\sl QQ}_{h\phi}$ are essentially null terms.
In $\Kcal$, the gradient of a function $u$ can be written in the semi-hyperboloidal frame, that is
$
\del_\alpha u = \Psi_{\alpha}^{\alpha'} \delu_{\alpha'} u = \Psi_{\alpha}^{0} \del_t u + \Psi_{\alpha}^a \delu_a u.
$
The coefficients $\Psi_{\alpha}^{\beta}$ are smooth and homogeneous of degree $0$. And we observe that the derivatives $\delu_a$ are ``good'' derivatives. So our task is 
{\sl to get refined estimates on $\del_t u$,} which is the main purpose of the next subsections.
We begin with the $L^\infty$ estimates, whose derivation is simpler than the derivation of the $L^2$ estimates.

\begin{lemma} \label{lem h00-sup 1}
Assume that the bootstrap assumption \eqref{ineq energy assumption 1} holds with $C_1\vep$ sufficiently small so that Lemma \ref{lem wave-condition 4} holds, then the following estimates hold for $|I|+|J| \leq N-2$:
\bel{ineq h00-sup 1}
|\del^IL^J \del_\alpha \hu^{00}| + |\del_\alpha \del^IL^J \hu^{00}| \leq CC_1\vep t^{-3/2}s^{\delta},
\ee
\bel{ineq h00-sup 1.5}
|\del^IL^J \hu^{00}| \leq CC_1\vep t^{-1/2}(s/t)^2s^\delta + Cm_s t^{-1}.
\ee
\end{lemma}

\begin{proof} We derive
\eqref{lem h00-sup 1} by substituting the basic sup-norm estimates into \eqref{ineq wave-condition 2a}. Then we integrate \eqref{lem h00-sup 1} along radial rays, as we did in Section~\ref{subsec ray-inte} and we obtain \eqref{ineq h00-sup 1.5}.
\end{proof}

The following statements are direct consequences of the above sup-norm estimates and play an essential role in our analysis. Roughly speaking, these lemmas guarantee that the curved metric $g$ is sufficiently close to the Minkowski metric, so that the energy estimates in  Propositions \ref{prop 1 energy-W} and \ref{prop energy 2KG} hold, as well as a sup-norm estimate for the Klein-Gordon equation (discussed in Appendix~C).

\begin{lemma}[Equivalence between the curved energy and flat energy functionals]\label{lem h00-sup 2}
Under the bootstrap assumption with $C_1\vep$ sufficiently small so that Lemma \ref{lem wave-condition 3} holds, there exists a constant $\kappa>1$ such that
\be
\aligned
&\kappa^{-2}E^*_{M}(s, \del^IL^Jh_{\alpha \beta}) \leq E^*_{g}(s, \del^IL^Jh_{\alpha \beta}) \leq \kappa^{2}E^*_{M}(s, \del^IL^Jh_{\alpha \beta}),
\\
&\kappa^{-2} E_{M,c^2}(s, \del^IL^J\phi) \leq E_{g,c^2}(s, \del^IL^J\phi) \leq \kappa^2E_{M,c^2}(s, \del^IL^J\phi).
\endaligned
\ee
\end{lemma}

\begin{proof} We only show the first statement, since the proof of the second one is similar.  
From the identity 
$$
\aligned
E^*_{g}(s,u) - E^*_{M}(s,u)
& = \int_{\Hcal_s^*} \Big(-h^{00}|\del_t u|^2 + h^{ab} \del_au\del_bu + \sum_a \frac{2x^a}{t}h^{a \beta} \del_{\beta}u\del_t u \Big) \, dx
\\
& = \int_{\Hcal_s^*} \Big(h^{\alpha \beta} \del_{\alpha}u\del_{\beta}u + 2\sum_a \frac{x^a}{t}h^{a \beta} \del_t u\del_{\beta}u - 2h^{0\beta} \del_t u\del_{\beta}u\Big) \, dx
\\
& = \int_{\Hcal_s^*} \Big(\hu^{\alpha \beta} \delu_{\alpha}u\delu_{\beta}u + \sum_a \frac{2x^a}{t} \hu^{\alpha'\beta'} \Phi^a_{\alpha'} \Phi_{\beta'}^{\beta} \del_t u\del_{\beta}u -2\hu^{\alpha'\beta'} \Phi_{\alpha'}^0\Phi_{\beta'}^{\beta} \del_t u\del_{\beta}u\Big) \, dx
\endaligned
$$
and then 
$$
\aligned
& E^*_{g}(s,u) - E^*_{M}(s,u)
\\
& = \int_{\Hcal_s^*} \hu^{\alpha \beta} \delu_{\alpha}u\delu_{\beta}udx
+\int_{\Hcal_s^*} \Big(\frac{2x^a}{t} \hu^{a0}|\del_t u|^2 + \frac{2x^a}{t} \hu^{ab} \del_t u\delu_b u\Big) \, dx
\\
& \quad +  \int_{\Hcal_s^*} \Big(-2\hu^{00}|\del_t u|^2 - 2\hu^{0b} \del_t u\delu_b u - \frac{2 x^a}{t} \hu^{a0}|\del_t u|^2 - \frac{2x^a}{t} \hu^{ab} \del_t u\delu_b u\Big) \, dx
\\
& = \int_{\Hcal_s^*} \big(- \hu^{00}|\del_t u|^2 + \hu^{ab} \delu_au \delu_b u\big) \, dx
 = \int_{\Hcal_s^*} \big(-(t/s)^2\hu^{00}|(s/t) \del_t u|^2 + \hu^{ab} \delu_au \delu_b u\big) \, dx, 
\endaligned
$$
we obtain 
$$
|E^*_{g}(s,u) - E^*_{M}(s,u)| \leq C\Big(\|(t/s)^2\hu^{00} \|_{L^{\infty(\Hcal_s^*)}} + \sum_{a,b} \|\hu^{ab} \|_{L^\infty(\Hcal_s^*)} \Big) E^*_{M}(s,u). 
$$
Then, recall that in view of \eqref{ineq h00-sup 1.5}, $|h| \leq CC_1\vep (s/t)t^{-1/2}s^{\delta} + Cm_St^{-1}$. When $C_1\vep$ is sufficiently small, we have
\bel{ineq proof 1 lem h00-sup 2}
|\hu^{\alpha \beta}| \leq C\max_{\alpha, \beta}|h_{\alpha \beta}| \leq CC_1\vep (s/t)t^{-1/2}s^{\delta} + Cm_St^{-1}.
\ee
On the other hand, from \eqref{ineq h00-sup 1.5}, we obtain 
$
|\hu^{00}| \leq CC_1\vep (s/t)^2t^{-1/2}s^{\delta} + Cm_St^{-1},
$
which implies
\bel{ineq proof 2 lem h00-sup 2}
|(t/s)^2\hu^{00}| \leq CC_1\vep t^{-1/2}s^{\delta} + Cm_S.
\ee
Now, when $C_1\vep$ is sufficiently small, \eqref{ineq proof 1 lem h00-sup 2} and \eqref{ineq proof 2 lem h00-sup 2} imply that
$
|E^*_{g}(s,u) - E^*_{M}(s,u)| \leq (1/2) E^*_{M}(s,u), 
$
which leads us to the desired result.
\end{proof}

\begin{lemma}[Derivation of the uniform bound on $M_{\alpha \beta}$]\label{lem h00-sup 3}
Under the energy assumption \eqref{ineq energy assumption 2}, the following estimate holds:
\bel{ineq lem h00-sup 3}
M_{\alpha \beta}[\del^IL^J h]\leq C (C_1\vep)^2 s^{-3/2 +2\delta}, \quad |I|+|J| \leq N, 
\ee
and 
\begin{subequations} 
\label{ineq lem h00-sup 3.5}
\bel{ineq lem h00-sup 3.5 a}
M[\del^IL^J \phi]\leq C (C_1\vep)^2 s^{-3/2 +2\delta}, \quad |I|+|J| \leq N-4, 
\ee
\bel{ineq lem h00-sup 3.5 b}
M[\del^IL^J \phi]\leq C (C_1\vep)^2 s^{-1+2\delta}, \quad |I|+|J| \leq N. 
\ee
\end{subequations}
\end{lemma}

\begin{proof} We only provide the proof of the third inequality, since the other two are easier. 
Recall the definition of $M[\del^IL^J\phi]$
\be
\aligned
& \int_{\Hcal_s}(s/t) \big|\del_{\mu}g^{\mu\nu} \del_{\nu} \big(\del^IL^J\phi\big) \del_t\big(\del^IL^J\phi\big)
- \frac{1}{2} \del_tg^{\mu\nu} \del_{\mu} \big(\del^IL^J\phi\big) \del_{\nu} \big(\del^IL^J\phi\big) \big| \, dx
\\
& \leq M[\del^IL^J\phi](s)E_M(s, \del^IL^J\phi)^{1/2}.
\endaligned
\ee
We perform the following calculation:
\bel{eq 1 proof lem h00-sup 3}
\aligned
&(s/t) \del_{\mu}g^{\mu\nu} \del_{\nu} \left(\del^IL^J\phi\right) \del_t\left(\del^IL^J\phi\right)
 =  (s/t) \del_{\mu}h^{\mu\nu} \del_{\nu} \left(\del^IL^J\phi\right) \del_t\left(\del^IL^J\phi\right)
\\
& =  (s/t) \delu_{\mu} \hu^{\mu\nu} \delu_{\nu} \left(\del^IL^J\phi\right) \del_t\left(\del^IL^J\phi\right)
- (s/t) \del_{\mu'} \left(\Psi_{\mu}^{\mu'} \Psi_{\nu}^{\nu'} \right)h^{\mu\nu} \del_{\nu'} \left(\del^IL^J\phi\right) \del_t\left(\del^IL^J\phi\right)
\\
& =  (s/t) \del_t\hu^{00} \del_t\left(\del^IL^J\phi\right) \del_t\left(\del^IL^J\phi\right)
\\
& \quad +  (s/t) \del_t\hu^{0a} \delu_a \left(\del^IL^J\phi\right) \del_t\left(\del^IL^J\phi\right)
 + (s/t) \delu_b\hu^{b0} \del_t\left(\del^IL^J\phi\right) \del_t\left(\del^IL^J\phi\right)
\\
& \quad +  (s/t) \delu_a \hu^{ab} \delu_b\left(\del^IL^J\phi\right) \del_t\left(\del^IL^J\phi\right)
\\
& \quad -  (s/t) \del_{\mu'} \left(\Psi_{\mu}^{\mu'} \Psi_{\nu}^{\nu'} \right)h^{\mu\nu} \del_{\nu'} \left(\del^IL^J\phi\right) \del_t\left(\del^IL^J\phi\right)
\endaligned
\ee
and then observe that  
$$
\aligned
&\int_{\Hcal_s}(s/t) \left|\del_t\hu^{00} \del_t\left(\del^IL^J\phi\right) \del_t\left(\del^IL^J\phi\right) \right| \, dx
 = \int_{\Hcal_s}(t/s) \left|\del_t\hu^{00} \right| \, \left|(s/t) \del_t\left(\del^IL^J\phi\right) \right|^2dx
\\
& \leq CC_1\vep \int_{\Hcal_s}(t/s)t^{-3/2}s^{\delta} \, \left|(s/t) \del_t\left(\del^IL^J\phi\right) \right|^2dx
\\
& \leq CC_1\vep s^{-3/2 +\delta}E_M(s, \del^IL^J\phi)
\\
& \leq 
\left\{
\aligned
 &C(C_1\vep)^2s^{-3/2 +2\delta}E_M(s, \del^IL^J\phi)^{1/2}, \quad |I|+|J| \leq N-4,
 \\
 &C(C_1\vep)^2s^{-1+2\delta}E_M(s, \del^IL^J\phi)^{1/2}, N-3\leq |I|+|J| \leq N,
 \endaligned
 \right.
\endaligned
$$
where we have used \eqref{ineq h00-sup 1}, \eqref{ineq energy assumption 12} and \eqref{ineq energy assumption 22}.
The second, third, and fourth terms in the right-hand side of \eqref{eq 1 proof lem h00-sup 3} are null terms, we observe that the second term is bounded as follows:
$$
\aligned
&\int_{\Hcal_s} \left|(s/t) \del_t\hu^{0a} \delu_a \left(\del^IL^J\phi\right) \del_t\left(\del^IL^J\phi\right) \right| \, dx
 \leq \int_{\Hcal_s} \left|\del_t\hu^{0a} \right| \, \left|\delu_a \left(\del^IL^J\phi\right)(s/t) \del_t\left(\del^IL^J\phi\right) \right| \, dx
\\
& \leq  CC_1\vep s^{-3/2 +\delta}E_M(s, \del^IL^J\phi)
\\
& \leq 
\left\{
\aligned
 &C(C_1\vep)^2s^{-3/2 +2\delta}E_M(s, \del^IL^J\phi)^{1/2}, \quad |I|+|J| \leq N-4,
 \\
 &C(C_1\vep)^2s^{-1+2\delta}E_M(s, \del^IL^J\phi)^{1/2}, \, N-3\leq |I|+|J| \leq N.
 \endaligned
 \right.
\endaligned
$$
The third  and fourth terms are bounded similarly and we omit the details.

The last term is bounded by applying the additional decay provided by $\del_{\mu'} \left(\Psi_{\mu}^{\mu'} \Psi_{\nu}^{\nu'} \right)$. This term is bounded by $t^{-1}$. We have
$$
\aligned
&\int_{\Hcal_s} \left|(s/t) \del_{\mu'} \left(\Psi_{\mu}^{\mu'} \Psi_{\nu}^{\nu'} \right)h^{\mu\nu} \del_{\nu'} \left(\del^IL^J\phi\right) \del_t\left(\del^IL^J\phi\right) \right| \, dx
\\
& \leq CC_1\vep \int_{\Hcal_s}t^{-1}(t/s)|h^{\mu\nu}| \, \left|(s/t) \del_{\nu'} \left(\del^IL^J\phi\right)(s/t) \del_t\left(\del^IL^J\phi\right) \right| \, dx
\\
& \leq CC_1\vep \int_{\Hcal_s}s^{-1} \big(t^{-1} + t^{-1/2}(s/t)s^{\delta} \big)
\left|(s/t) \del_{\nu'} \left(\del^IL^J\phi\right)(s/t) \del_t\left(\del^IL^J\phi\right) \right| \, dx
\\
& \leq  CC_1\vep s^{-3/2 +\delta}E_M(s, \del^IL^J\phi)
\\
& \leq 
\left\{
\aligned
 &C(C_1\vep)^2s^{-3/2 +2\delta}E_M(s, \del^IL^J\phi)^{1/2}, \quad |I|+|J| \leq N-4,
 \\
 &C(C_1\vep)^2s^{-1+2\delta}E_M(s, \del^IL^J\phi)^{1/2}, N-3\leq |I|+|J| \leq N.
 \endaligned
 \right.
\endaligned
$$
We conclude that
$$
\aligned
& \int_{\Hcal_s} \left|(s/t) \del_{\mu}g^{\mu\nu} \del_{\nu} \left(\del^IL^J\phi\right) \del_t\left(\del^IL^J\phi\right) \right| \, dx
\\
& \leq
\left\{
\aligned
 &C(C_1\vep)^2s^{-3/2 +2\delta}E_M(s, \del^IL^J\phi)^{1/2}, \quad |I|+|J| \leq N-4,
 \\
 &C(C_1\vep)^2s^{-1+2\delta}E_M(s, \del^IL^J\phi)^{1/2}, N-3\leq |I|+|J| \leq N.
 \endaligned
 \right.
\endaligned
$$
The term $\del_tg^{\mu\nu} \del_{\mu} \big(\del^IL^J\phi\big) \del_{\nu} \big(\del^IL^J\phi\big)$ is bounded similarly and we omit the details.
\end{proof}

\begin{lemma} \label{lem h00-sup 4} 
Following the notation in Proposition \ref{Linfini KG}. When the bootstrap assumption \eqref{ineq energy assumption 1} holds, the following estimate holds:
\bel{ineq lem h00-sup 4 1}
|h_{t, x}'(\lambda)| \leq CC_1\vep (s/t)^{1/2} \lambda^{-3/2 +\delta} + CC_1\vep (s/t)^{-1} \lambda^{-2}.
\ee
\end{lemma}

\begin{proof}
Following the notation in Proposition \ref{Linfini KG}, we have 
$
h_{t, x}(\lambda) = \hb^{00} \bigg(\frac{\lambda t}{s}, \frac{\lambda x}{s} \bigg)
$
Recalling that
$
\hb^{00}=(t/s)^2\hu^{00}
$
we find
$
h_{t, x}(\lambda) = (t/s)^2\hu^{00} \bigg(\frac{\lambda t}{s}, \frac{\lambda x}{s} \bigg)
$
which leads us to
\bel{ineq 1 proof lem h00-sup 4}
h_{t, x}'(\lambda) = (t/s)^3\delu_{\perp} \hu^{00} \bigg(\frac{\lambda t}{s}, \frac{\lambda x}{s} \bigg).
\ee
Here we recall also that
$
\delu_{\perp} \hu^{00} = \frac{s^2}{t^2} \del_t\hu^{00} + \frac{x^a}{t} \delu_a \hu^{00} = \frac{s^2}{t^2} \del_t\hu^{00} + \frac{x}{t^2}L_a \hu^{00}.
$
We see that, in view of \eqref{ineq h00-sup 1},
$\big|(t/s) \del_t\hu^{00} \big| \leq CC_1\vep (s/t)^{1/2}s^{-3/2 +\delta}$
and, in view of \eqref{ineq h00-sup 1.5},
$$
\big|(t/s)^2s^{-1}L_a \hu^{00} \big| \leq CC_1\vep (s/t)^{1/2}s^{-3/2 +\delta} + Cm_Sts^{-3}.
$$ 
By combining this result with \eqref{ineq 1 proof lem h00-sup 4}, the desired conclusion is reached.
\end{proof}


\subsection{$L^2$ estimates}
\label{subsubsec L-2-h00}

We first establish an $L^2$ estimate on the gradient of $\del^IL^J \hu^{00}$.

\begin{lemma} \label{prop h00-sup 0}
Under the bootstrap assumptions \eqref{ineq energy assumption 1} and \eqref{ineq energy assumption 2}, the following estimate holds:
\bel{ineq h00-L-2 wave}
\left\|\del^IL^J\del_\alpha \hu^{00} \right\|_{L^2(\Hcal_s^*)} + \left\|\del_\alpha \del^IL^J\hu^{00} \right\|_{L^2(\Hcal_s^*)} \leq CC_1\vep s^{2\delta}.
\ee
\end{lemma}

\begin{proof} The estimate is immediate in view of \eqref{ineq wave-condition 2a}. Namely, thanks to the basic $L^2$ estimates, we have 
$$
\|(s/t)^2\del\del^{I'}L^{J'} \hu\|_{L^2(\Hcal_s^*)} + \|\delu\del^{I'}L^{J'}  \hu \|_{L^2(\Hcal_s^*)} \leq CC_1\vep s^{\delta}.
$$
By \eqref{ineq 0 Hardy}, we get 
\bel{eq pr1 pro h00-sup 0}
\|t^{-1} \del^IL^J  \hu^{00} \|_{L^2(\Hcal_s^*)} \leq C\sum_a \|\delu_a \del^IL^J \hu^{00} \|_{L^2(\Hcal_s^*)}
 + Cm_Ss^{-1} \leq CC_1\vep s^{\delta}.
\ee
Now, from \eqref{ineq wave-condition 2a}, we need to control the term
$
|\del^{I_1}L^{J_1} \hu\del\del^{I_2}L^{J_2} \hu|$.
When $|I_1|+|J_1| \leq N-2$, we apply \eqref{ineq basic-sup-h} and \eqref{ineq basic-L2 1 generation 1 a} :
$$
\|\del^{I_1}L^{J_1} \hu\del\del^{I_2}L^{J_2} \hu \|_{L^2(\Hcal_s^*)} 
\leq CC_1\vep s^{\delta} \|(s/t)t^{-1/2} \del^{I_2}L^{J_2} \hu\|_{L^2(\Hcal_s^*)} \leq CC_1\vep s^{\delta}.
$$
When $N-1\leq |I_1|+|J_1| \leq N$, we see that $|I_2|+|J_2| \leq 1$. We have 
$$
\aligned
\|\del^{I_1}L^{J_1} \hu\del\del^{I_2}L^{J_2} \hu \|_{L^2(\Hcal_s^*)} 
& \leq CC_1\vep s^{\delta} \|t^{-1/2}s^{-1} \del^{I_1}L^{J_1} \hu\|_{L^2(\Hcal_s^*)}
\\
& \leq CC_1\vep s^{\delta} \|t^{-1} \del^{I_1}L^{J_1} \hu\|_{L^2(\Hcal_s^*)}
\leq CC_1\vep s^{2\delta}, 
\endaligned
$$
where we have used \eqref{eq pr1 pro h00-sup 0}.  
\end{proof}

We are going now to derive the $L^2$ estimate on (the ``essential part'' of) $\del^IL^J \hu^{00}$. This is one of the most challenging terms and we first decompose $\hu^{00}$ as follows: 
$$
\hu^{00} := \chi(r/t) \hu^{00}_0 + \hu^{00}_1,
$$
where $\hu_0^{00} = \hu_S^{00}$ is the corresponding component of the Schwarzschild metric and the function $\chi$ is smooth
 with $\chi(\tau) = 0$ for $\tau \in [0, 1/3]$ while 
$\chi(\tau) = 1$ for $\tau \geq 2/3$.  
We introduce the notation
$
\hu^{00}_0:= \chi(r/t) \hu_S^{00}
$
and an explicit calculation shows that in $\Kcal_{[2,+\infty)}$
$$
\aligned
|\hu^{00}_0|  & \leq Cm_S t^{-1} \leq Cm_S(1+r)^{-1},
\qquad 
|\del_{\alpha} \hu_0^{00}| \leq Cm_S t^{-2} \leq Cm_S (1+r^2)^{-1}.
\endaligned
$$
This leads us to the estimate
\be
\|\del_a \hu^{00}_0\|_{L_f^2(\Hcal_s)} \leq Cm_S, \qquad \|\delu_a \hu^{00}_0\|_{L_f^2(\Hcal_s)} \leq Cm_S
\ee
and we are ready to establish the following result.

\begin{proposition} \label{prop h00-sup 1}
Assume that the bootstrap assumptions \eqref{ineq energy assumption 1} and \eqref{ineq energy assumption 2} hold  
with $C_1\vep$ sufficiently small (so that Lemma \ref{lem wave-condition 4} holds). Then, one has 
\be
|\del^IL^J \hu^{00}| \leq Cm_St^{-1} + |\del^IL^J \hu^{00}_1|
\ee
and 
\be 
\aligned
\|(s/t)^{-1+\delta}s^{-1} \del^IL^J \hu_1^{00} \|_{L^2(\Hcal_s^*)}
& \leq  CC_0 \, \vep
+ C\sum_{|I'| \leq |I|,|J'| \leq|J|\atop \alpha, \beta}E_M^*(s, \del^IL^J h_{\alpha \beta})^{1/2}
\\
& \quad +  C\sum_{|I'| \leq |I|,|J'| \leq|J|\atop \alpha, \beta} \int_2^s\tau^{-1}E_M^*(\tau, \del^{I'}L^{J'}h_{\alpha \beta})^{1/2}d\tau
  \leq  CC_1\vep s^{\delta}.
\endaligned
\ee
\end{proposition}

\begin{proof} In the decomposition of $\hu^{00}$, the term $\del_{\alpha} \del^IL^J\hu_1^{00}$ vanishes near the  
boundary of $\Kcal_{[2,s^*]}$, since in a neighborhood of this boundary, $\hu^{00} = \hu^{00}_S = \hu^{00}_0$. Furthermore, we have 
\bel{ineq 1 proof prop h00-sup 1}
\|(s/t)^{\delta} \del_\alpha \del^IL^J \hu_1^{00} \|_{L^2(\Hcal_s^*)} \leq \|(s/t)^{\delta} \del_{\alpha} \del^IL^J\hu^{00} \|_{L^2(\Hcal_s^*)} + \|(s/t)^{\delta} \del_{\alpha} \del^IL^J\hu_0^{00} \|_{L^2(\Hcal_s^*)}.
\ee

We recall that
$\del_a = - \frac{x^a}{t} \del_t + \delu_a$,
that is, $\del_{\alpha}$ is a linear combination of $\del_t$ and $\delu_a$ with homogeneous coefficients of degree $0$, so the following estimates are direct in view of \eqref{ineq wave-condition 2a} :
\be
\aligned
&\|(s/t)^{\delta} \del_{\alpha} \del^IL^J \hu^{00} \|_{L^2(\Hcal_s^*)}
\\
& \leq  C\sum_{|I'|+|J'| \leq |I|+|J|\atop |J'| \leq |J|} \Big(\|(s/t)^2\del\del^{I'}L^{J'} h\|_{L^2(\Hcal_s^*)} + \|\delu\del^{I'}L^{J'} h\|_{L^2(\Hcal_s^*)} +\|t^{-1} \del^{I'}L^{J'} h\|_{L^2(\Hcal_s^*)} \Big)
\\
& \quad + C\sum_{|I_1|+|I_2| \leq|I|\atop|J_1|+|J_2| \leq |J|} \big\|(s/t)^{\delta} \del^{I_1}L^{J_1} \hu\del\del^{I_2}L^{J_2} \hu\big\|_{L_f^2(\Hcal_s)}.
\endaligned
\ee
Here the first sum in the right-hand side is easily controlled by
$$
\sum_{|I'| \leq |I|,|J'| \leq|J|\atop \alpha, \beta}E_M^*(s, \del^{I'}L^{J'}h_{\alpha \beta})^{1/2} + C\|t^{-1} \del^{I'}L^{J'} h\|_{L^2(\Hcal_s^*)}.
$$
For the last term, we observe that when $N\geq3$, either $|I_1|+|J_1| \leq N-2$ or else $|I_2|+|J_2| \leq N-2$. When $|I_1|+|J_1| \leq N-2$, in view of \eqref{ineq basic-sup-h},
$$
\aligned
&\big\|(s/t)^{\delta} \del^{I_1}L^{J_1} \hu\del\del^{I_2}L^{J_2} \hu\big\|_{L_f^2(\Hcal_s)}
 \leq  CC_1\vep\big\|\big((s/t)t^{-1/2}s^{\delta} + t^{-1} \big) \del^{I_2}L^{J_2} \del\hu\big\|_{L_f^2(\Hcal_s)}
\\
& \leq  CC_1\vep\big\|(s/t) \del^{I_2}L^{J_2} \del\hu\big\|_{L^(\Hcal_s)}
 \leq CC_1\vep \sum_{|I'| \leq |I|,|J'| \leq|J|\atop \alpha, \beta}E_M^*(s, \del^{I'}L^{J'}h_{\alpha \beta})^{1/2}.
\endaligned
$$
When $|I_2|+|J_2| \leq N-2$, we see that $|I_1|+|J_1|\geq 1$. Then we need to distinguish between two different cases. If $|I_1|\geq 1$, then
$$
\aligned
&
\big\|(s/t)^{\delta} \del^{I_1}L^{J_1} \hu\del\del^{I_2}L^{J_2} \hu\big\|_{L_f^2(\Hcal_s)}
 \leq  CC_1\vep\big\|t^{-1/2}s^{-1+\delta}(s/t)^{\delta} \del^{I_1}L^{J_1} \hu\big\|_{L_f^2(\Hcal_s)}
\\
& \leq  CC_1\vep\|t^{1/2}s^{-2 +\delta}(s/t)^{\delta}(s/t) \del^{I_1}L^{J_1} \hu\|_{L^2(\Hcal_s^*)}
  \leq  CC_1\vep s^{-1} \sum_{|I'| \leq |I|,|J'| \leq|J|\atop \alpha, \beta}E_M^*(s, \del^{I'}L^{J'}h_{\alpha \beta})^{1/2}.
\endaligned
$$
When $|I_1|=0$, we see that $|J_1|\geq 1$. In this case we set $L^{J_1} = L_aL^{J_1'}$ with $|J_1'|\geq 1$. Then
$$
\aligned
&
\big\|(s/t)^{\delta} \del^{I_1}L^{J_1} \hu\del\del^{I_2}L^{J_2} \hu\big\|_{L_f^2(\Hcal_s)}
\\
& \leq  CC_1\vep\big\|(s/t)^{\delta}t^{-1/2}s^{-1+\delta}L_aL^{J_1'} \hu\big\|_{L_f^2(\Hcal_s)}
 =  CC_1\vep\big\|(s/t)^{\delta}t^{-1/2}s^{-1+\delta}t\delu_aL^{J_1'} \hu\big\|_{L_f^2(\Hcal_s)}
\\
& =  CC_1\vep\big\|t^{1/2- \delta}s^{-1+2\delta} \delu_aL^{J_1'} \hu\big\|_{L_f^2(\Hcal_s)}
 \leq CC_1\vep\sum_{|I'| \leq |I|,|J'| \leq|J|\atop \alpha, \beta}E_M^*(s, \del^{I'}L^{J'}h_{\alpha \beta})^{1/2}.
\endaligned
$$
Then the above discussion leads us to
\bel{eq pr1 prop h00-sup 1}
\|(s/t)^{\delta} \del_{\alpha} \del^IL^J \hu^{00} \|_{L^2(\Hcal_s^*)} \leq \sum_{|I'| \leq |I|,|J'| \leq|J|\atop \alpha, \beta}E_M^*(s, \del^{I'}L^{J'}h_{\alpha \beta})^{1/2} + C\|t^{-1} \del^{I'}L^{J'} h\|_{L^2(\Hcal_s^*)}
\ee

Now based on \eqref{eq pr1 prop h00-sup 1}, we continue our discussion. Recalling the adapted Hardy inequality \eqref{ineq 0 Hardy}, we obtain 
$$
\|t^{-1} \del^{I'}L^{J'} h\|_{L^2(\Hcal_s^*)} \leq \|r^{-1} \del^{I'}L^{J'} h\|_{L^2(\Hcal_s^*)}
\leq C\|\delu\del^IL^J h\|_{L^2(\Hcal_s^*)} + Cm_Ss^{-1}, 
$$
so that
$$
\|(s/t)^{\delta} \del_{\alpha} \del^IL^J \hu^{00} \|_{L^2(\Hcal_s^*)}
\leq  C\sum_{|I'| \leq |I|,|J'| \leq|J|\atop \alpha, \beta}E_M^*(s, \del^{I'}L^{J'}h_{\alpha \beta})^{1/2} + Cm_Ss^{-1}. 
$$
On the other hand,  by explicit calculation we have
$
\|\del_{\alpha} \del^IL^J\hu_0^{00} \|_{L^2(\Hcal_s^*)} \leq Cm_S s^{-1}.
$
So in view of \eqref{ineq 1 proof prop h00-sup 1}
$$
\|(s/t)^{\delta} \del_\alpha \del^IL^J \hu_1^{00} \|_{L^2(\Hcal_s^*)}
\leq C\sum_{|I'| \leq |I|,|J'| \leq|J|\atop \alpha, \beta}E_M^*(s, \del^{I'}L^{J'}h_{\alpha \beta})^{1/2} + Cm_Ss^{-1}.
$$
We also recall that by the basic $L^2$ estimate, $\|\delu_a \del^IL^J \hu_1^{00} \|_{L_f^2(\Hcal_s)} \leq CC_1\vep s^{\delta}$. By Proposition \ref{prop Hardy 2} with $\sigma  = 1- \delta$, the desired result is established.
\end{proof}


\subsection{Commutator estimates}

Next, we use the basic estimates and the estimate for $\hu^{00}$ in order to control the commutators
$
[\del^IL^J,h^{\mu\nu} \del_\mu\del_\nu]h_{\alpha \beta}.
$

\begin{lemma} \label{lem 1 commtator-sup I}
Assume that the bootstrap assumptions \eqref{ineq energy assumption 1} and \eqref{ineq energy assumption 2} holds, then for $|I|+|J| \leq N-2$, the following estimate holds in $\Kcal$:
\bel{ineq lem 1 commtator-sup I}
\aligned
&\left|[\del^IL^J,h^{\mu\nu} \del_\mu\del_\nu]h_{\alpha \beta} \right| 
\\
&\leq C(C_1\vep)^2t^{-2}s^{-1+2\delta}
+ CC_1\vep\left(t^{-1} + (s/t)^2t^{-1/2}s^{\delta} \right) \sum_{|J'|<|J|} \left|\del_t\del_t\del^IL^{J'} h_{\alpha \beta} \right|.
\endaligned
\ee
\end{lemma}

\begin{proof}
We recall Lemma \ref{lem 1 nonlinear}, to estimate $[\del^IL^J,h^{\mu\nu} \del_\mu\del_\nu]h_{\alpha \beta}$, we need to control the terms listed in \eqref{eq 1 lem 1 nolinear}. We see first that, in view of \eqref{ineq 1 bilinear-sup}, 
$\left|{\sl GQQ}_{hh}(p,k) \right| \leq C(C_1\vep)^2 t^{-3}s^{2\delta}$. 
For the term $t^{-1} \del^{I_3}L^{J_3}h_{\mu\nu} \del^{I_4}L^{J_4} \del_{\gamma}h_{\mu'\nu'}$, we observe that $|I_3|+|I_4| \leq N-2$ and $|I_4|+|J_4| \leq N-2$, so 
$$
\left|t^{-1} \del^{I_3}L^{J_3}h_{\mu\nu} \del^{I_4}L^{J_4} \del_{\gamma}h_{\mu'\nu'} \right| \leq C(C_1\vep)^2\left(t^{-1}+(s/t)t^{-1/2}s^\delta \right)t^{-1/2}s^{-1+\delta} \leq C(C_1\vep)^2t^{-3}s^{2\delta}.
$$
For the term $\del^{I_1}L^{J_1} \hu^{00} \del^{I_2}L^{J_2} \del_t\del_t h_{\alpha \beta}$, we see that $|I_1|+|J_1| \leq N-2$ and $|I_1|\geq 1$, $|I_2|+|J_2| \leq N-3$, so in view of \eqref{ineq h00-sup 1}
\bel{ineq pr1 lem 1 commtator-sup I}
\aligned
\left|\del^{I_1}L^{J_1} \hu^{00} \del^{I_2}L^{J_2} \del_t\del_t h_{\alpha \beta} \right|
& \leq  CC_1\vep s^{\delta}t^{-3/2}|\del^{I_2}L^{J_2} \del_t\del_th_{\alpha \beta}|.
\endaligned
\ee
For terms $L^{J_1'} \hu^{00} \del^IL^{J_2'} \del_t\del_t h_{\alpha \beta}$ and $\hu^{00} \del_{\gamma} \del_{\gamma'} \del^IL^{J'}h_{\alpha \beta}$, we first observe that by the condition $|J_2'|< |J|$ and $|J'|<|J|$, $|I|+|J_2'| \leq N-3, |I|+|J'| \leq N-3$. Then they are bounded by applying \eqref{ineq h00-sup 1.5}. We only write in detail $L^{J_1'} \hu^{00} \del^IL^{J_2'} \del_t\del_t h_{\alpha \beta}$: 
\bel{ineq 1 pr lem 1 commtator-sup I}
\aligned
\left|L^{J_1'} \hu^{00} \del^IL^{J_2'} \del_t\del_t h_{\alpha \beta} \right| & \leq  CC_1\vep \left((s/t)^2t^{-1/2}s^{\delta} + t^{-1} \right) \,   \sum_{|J'|<|J|} \left|\del^IL^{J'} \del_t\del_t h_{\alpha \beta} \right|.
\endaligned
\ee
In view of the commutator estimate \eqref{ineq com 2.4}, we have 
$
\left|\del^IL^{J'} \del_t\del_t h_{\alpha \beta} \right|
\leq C\sum_{\gamma, \gamma'\atop |J''| \leq|J'|} \left|\del_{\gamma} \del_{\gamma'} \del^IL^{J''}h_{\alpha \beta} \right|.
$
We observe that (and this is an argument frequently applied in the following discussion, as it says that $\del_t\del_t$ is the only ``bad'' component of the Hessian):
\bel{eq second-order-frame-1}
\aligned
\del_t\del_a u & =  \del_a \del_t u = \delu_a \del_t u - \frac{x^a}{t} \del_t\del_t u,
\\
\del_a \del_b u & =  \delu_a \delu_b u - \frac{x^a}{t} \del_t\delu_b u- \frac{x^b}{t} \delu_a \del_t u +\frac{x^ax^b}{t^2} \del_t\del_t u
   -  \delu_a \left(x^b/t\right) \del_t u +\frac{x^a}{t} \del_t\left(x^b/t\right) \del_t u. 
\endaligned
\ee
Here we observe that the term $\del_{\gamma} \del_{\gamma'} \del^IL^{J''}h_{\alpha \beta}$ is bounded by $\del_t\del_t\del^IL^{J''} h_{\alpha \beta}$ plus other ``good'' terms. We see that, in $\Kcal$,
$\left|\del_t\left(x^b/t\right) \right| \leq Ct^{-1}, \quad \delu_a \left(x^b/t\right) \leq Ct^{-1},$
so that 
$$
\left|\delu_a \left(x^b/t\right) \del_t\del^IL^{J''}h_{\alpha \beta} \right|
+ \left|\frac{x^a}{t} \del_t\left(x^b/t\right) \del_t\del^IL^{J''}h_{\alpha \beta} \right| \leq CC_1\vep t^{-3/2}s^{-1+\delta}.
$$
The terms $\delu_a \del_t\del^IL^{J''} h_{\alpha \beta}$, $\del_t\delu_a \del^IL^{J''}h_{\alpha \beta}$ and $\delu_a \delu_b\del^IL^{J''} h_{\alpha \beta}$ are the second-order derivatives, where at least one derivative is ``good'' (i.e. $\delu_a$). We apply \eqref{ineq 2 homo}, \eqref{ineq 3 homo} and \eqref{ineq 4 homo} and basic sup-norm estimate, then we conclude that these terms are bounded by $CC_1\vep t^{-3/2}s^{-1+\delta}$. 
We conclude that
\bel{ineq 2 pr lem 1 commtator-sup I}
\left|\del_{\gamma} \del_{\gamma'} \del^IL^{J''}h_{\alpha \beta} \right| \leq CC_1\vep t^{-3/2}s^{-1+\delta} + \left|\del_t\del_t\del^IL^{J''}h_{\alpha \beta} \right|.
\ee

Now we substitute this into \eqref{ineq 1 pr lem 1 commtator-sup I} and obtain 
$$
\left|L^{J_1'} \hu^{00} \del^IL^{J_2'} \del_t\del_t h_{\alpha \beta} \right| \leq C(C_1\vep)^2 t^{-3}s^{2\delta} + CC_1\vep \left((s/t)^2t^{-1/2}s^{\delta} + t^{-1} \right) \sum_{|J'|<|J|} \left|\del_t\del_t\del^IL^{J'} \right|.
$$
By combining the estimates above, the desired result is proven.
\end{proof}

\begin{lemma} \label{lem 1 commtator-L2 I} For $|I|+|J| \leq N$, one has 
\bel{ineq 1 lem 1 commtator-L2 I}
\aligned
\left\|s[\del^IL^J,h^{\mu\nu} \del_\mu\del_\nu]h_{\alpha \beta} \right\|_{L^2(\Hcal_s^*)}
& \leq  C(C_1\vep)^2s^{2\delta}
\\
& \quad +  CC_1\vep s^{\delta} \sum_{|J'| \leq 1}
\left\|s^2(s/t)^{1- \delta} \del^IL^{J'} \del_t\del_th_{\alpha \beta} \right\|_{L^\infty(\Hcal_s^*)}
\\
& \quad +  CC_1\vep s^{1/2 +\delta} \sum_{|J'|<|J|} \left\|(s/t)^{5/2} \del_t\del_t\del^IL^{J'}h_{\alpha \beta} \right\|_{L^2(\Hcal_s^*)}. 
\endaligned
\ee
\end{lemma}

\begin{proof}
The proof relies on Lemma \ref{lem 1 nonlinear} and we need to estimate the terms listed in \eqref{eq 1 lem 1 nolinear}. The term ${\sl GQQ}_{hh}$ is already bounded in view of \eqref{ineq 1 bilinear-L2}. For the term $t^{-1} \del^{I_1}L^{J_1}h_{\mu\nu} \del^{I_2}L^{J_2} \del_{\gamma}h_{\mu'\nu'}$, we have the following estimates. When $|I_1|+|J_1| \leq N-2$, we see that
$$
\aligned
\big\|st^{-1} \del^{I_1}L^{J_1}h_{\mu\nu} \del^{I_2}L^{J_2} \del_{\gamma}h_{\mu'\nu'} \big\|_{L_f^2(\Hcal_s)}
& \leq \big\|\left(t^{-1}+t^{-1/2}(s/t)s^{\delta} \right) \,  (s/t) \del^{I_2}L^{J_2} \del_{\gamma}h_{\mu'\nu'} \big\|_{L_f^2(\Hcal_s)}
\\
& \leq  C(C_1\vep)^2s^{-1/2 +2\delta}.
\endaligned
$$
When $|I_1|+|J_1|\geq N-1\geq 1$, we have $|I_2|+|J_2| \leq 1\leq N-2$. We distinguish between two subcases: when $|I_1|\geq 1$, we obtain 
$$
\aligned
\big\|st^{-1} \del^{I_1}L^{J_1}h_{\mu\nu} \del^{I_2}L^{J_2} \del_{\gamma}h_{\mu'\nu'} \big\|_{L_f^2(\Hcal_s)}
& \leq  CC_1\vep\big\|st^{-1} \del^{I_1}L^{J_1}h_{\mu\nu}t^{-1/2}s^{-1+\delta} \big\|_{L_f^2(\Hcal_s)}
\\
& \leq C(C_1\vep)^2s^{-3/2 +2\delta}.
\endaligned
$$
When $|I_1|=0$, then $|J_1|\geq 1$. We denote by $L^{J_1}= L_aL^{J_1'}$ and 
$$
\aligned
&\big\|st^{-1} \del^{I_1}L^{J_1}h_{\mu\nu} \del^{I_2}L^{J_2} \del_{\gamma}h_{\mu'\nu'} \big\|_{L_f^2(\Hcal_s)}
 = \big\|s\delu_aL^{J_1'}h_{\mu\nu} \del^{I_2}L^{J_2} \del_{\gamma}h_{\mu'\nu'} \big\|_{L_f^2(\Hcal_s)}
\\
& \leq CC_1\vep\big\|s\delu_aL^{J_1'}h_{\mu\nu}t^{-1/2}s^{-1+\delta} \big\|_{L_f^2(\Hcal_s)} \leq C(C_1\vep)^2s^{-1/2 +2\delta}.
\endaligned
$$

For the term $\del^{I_1}L^{J_1} \hu^{00} \del^{I_2}L^{J_2} \del_t\del_th_{\alpha \beta}$ with $|I_1|\geq 1$, we observe that
\begin{itemize}

\item When $1\leq |I_1|+|J_1| \leq N-1$ we apply \eqref{ineq h00-sup 1} :
$$
\aligned
    \left\|s\del^{I_1}L^{J_1} \hu^{00} \del^{I_2}L^{J_2} \del_t\del_th_{\alpha \beta} \right\|_{L^2(\Hcal_s^*)}
& \leq  CC_1\vep \left\|st^{-3/2}s^{\delta}(t/s) \,   (s/t) \del^{I_2}L^{J_2} \del_t\del_th_{\alpha \beta} \right\|_{L^2(\Hcal_s^*)}
\\
& \leq  C(C_1\vep)^2s^{-1/2 +2\delta}.
\endaligned
$$

\item When $|I_1|+|J_1|= N$, then $|I_2|+|J_2|=0 \leq N-3$. So
$$
\aligned
& \left\|s\del^{I_1}L^{J_1} \hu^{00} \del^{I_2}L^{J_2} \del_t\del_th_{\alpha \beta} \right\|_{L^2(\Hcal_s^*)}
\leq  CC_1\vep\left\|st^{-1/2}s^{-1+\delta} \,   \del^{I_1}L^{J_1} \hu^{00} \right\|_{L^2(\Hcal_s^*)}
\\
& \leq  CC_1\vep s^{-1/2 +\delta} \left\| \del^{I_1}L^{J_1} \hu^{00} \right\|_{L^2(\Hcal_s^*)} 
    \leq  C(C_1\vep)^2s^{-1/2 +3\delta}, 
\endaligned
$$
where we have applied \eqref{ineq h00-L-2 wave}.
\end{itemize}

For the term $L^{J_1'} \hu^{00} \del^IL^{J_2'} \del_t\del_th_{\alpha \beta}$, we apply the energy estimate to $L^J\hu^{00}$  by Proposition \ref{prop h00-sup 1} and the sup-norm estimate provided by Lemma \ref{lem h00-sup 1}.
\begin{itemize}

\item When $|J_1'| \leq N-2$, we apply \eqref{ineq h00-sup 1.5}
$$
\aligned
&
\left\|sL^{J_1'} \hu^{00} \del^IL^{J_2'} \del_t\del_th_{\alpha \beta} \right\|_{L^2(\Hcal_s^*)}
 \leq  CC_1\vep \left\|s\left(t^{-1}+(s/t)^2t^{-1/2}s^{\delta} \right) \,   \del^IL^{J_2'} \del_t\del_th_{\alpha \beta} \right\|_{L^2(\Hcal_s^*)}
\\
& \leq  CC_1\vep\left\|(s/t) \del^IL^{J_2'} \del_t\del_th_{\alpha \beta} \right\|_{L^2(\Hcal_s^*)}
 +  CC_1\vep s^{1/2 +\delta} \left\|(s/t)^{5/2} \del^IL^{J_2'} \del_t\del_th_{\alpha \beta} \right\|_{L^2(\Hcal_s^*)}
\\
& \leq  C(C_1\vep)^2 s^{\delta}
+ CC_1\vep s^{1/2 +\delta} \sum_{|J'|<|J|} \left\|(s/t)^{5/2} \del^IL^{J'} \del_t\del_th_{\alpha \beta} \right\|_{L^2(\Hcal_s^*)}
\endaligned
$$

\item When $|J_1'|\geq N-1$, we apply Proposition \ref{prop h00-sup 1}
$$
\aligned
&\left\|sL^{J_1'} \hu^{00} \del^IL^{J_2'} \del_t\del_th_{\alpha \beta} \right\|_{L^2(\Hcal_s^*)}
\leq 
CC_1\vep \left\|st^{-1} \del^IL^{J_2'} \del_t\del_th_{\alpha \beta} \right\|_{L^2(\Hcal_s^*)}
 + \left\|sL^{J_1'} \hu_1^{00} \del^IL^{J_2'} \del_t\del_th_{\alpha \beta} \right\|_{L^2(\Hcal_s^*)}
 \\
& \leq  C(C_1\vep)^2s^{\delta} + \left\|sL^{J_1'} \hu_1^{00} \del^IL^{J_2'} \del_t\del_th_{\alpha \beta} \right\|_{L^2(\Hcal_s^*)}
\\
&\leq  C(C_1\vep)^2s^{\delta} + \left\|(s/t)^{-1+\delta}s^{-1}L^{J_1'} \hu^{00}_1\right\|_{L^2(\Hcal_s^*)}
 \,  \left\|s^2(s/t)^{1- \delta} \del^IL^{J_2'} \del_t\del_th_{\alpha \beta} \right\|_{L^\infty(\Hcal_s^*)}
\\
&\leq C(C_1\vep)^2s^{\delta} + CC_1\vep s^{\delta}
\sum_{|J'| \leq 1} \left\|s^2(s/t)^{1- \delta} \del^IL^{J'} \del_t\del_th_{\alpha \beta} \right\|_{L^\infty(\Hcal_s^*)}.
\endaligned
$$
\end{itemize}

For the term $\hu^{00} \del_{\gamma} \del_{\gamma'} \del^IL^{J'}h_{\alpha \beta}$, the estimate is similar. We apply \eqref{ineq h00-sup 1.5} and 
$$
\aligned
&\left\|s\hu^{00} \del_{\gamma} \del_{\gamma'} \del^IL^{J'}h_{\alpha \beta} \right\|_{L^2(\Hcal_s^*)}
\\
& \leq CC_1\vep\left\|(s/t) \del_{\gamma} \del_{\gamma'} \del^IL^{J'}h_{\alpha \beta} \right\|_{L^2(\Hcal_s^*)}
+ \left\|(s/t)^2t^{-1/2}s^{1+\delta} \del_{\gamma} \del_{\gamma'} \del^IL^{J'}h_{\alpha \beta} \right\|_{L^2(\Hcal_s^*)}
\\
&\leq  C(C_1\vep)^2s^{\delta}
+ CC_1\vep s^{1/2 +\delta} \sum_{|J'|<|J|} \left\|(s/t)^{5/2} \del_{\gamma} \del_{\gamma'} \del^IL^{J'}h_{\alpha \beta} \right\|_{L^2(\Hcal_s^*)}.
\endaligned
$$
Now we need to treat the last term and bound it by $\|(s/t)^{5/2} \del_t\del_t\del^IL^{J'}h_{\alpha \beta} \|_{L^2(\Hcal_s^*)}$. We rely on the discussion after \eqref{eq second-order-frame-1} and conclude that
$$
\aligned
&\left\|\hu^{00} \del_{\gamma} \del_{\gamma'} \del^IL^{J'}h_{\alpha \beta} \right\|_{L^2(\Hcal_s^*)}
\\
& \leq \sum_{a, \mu\atop |J''|<|J'|} \|\hu^{00} \delu_a \del_\mu\del^IL^{J''}h_{\alpha \beta} \|_{L^2(\Hcal_s^*)} + C\sum_{|J''|<|J'|} \|\hu^{00} \del_t\del_t\del^IL^{J''}h_{\alpha \beta} \|_{L^2(\Hcal_s^*)}
\\
&\leq C(C_1\vep)^2s^{-1+\delta}
 + CC_1\vep s^{-1/2 +\delta} \sum_{|J''|<|J'|} \left\|(s/t)^{5/2} \del_t\del_t\del^IL^{J''}h_{\alpha \beta} \right\|_{L^2(\Hcal_s^*)}. \qedhere
\endaligned
$$ 
\end{proof}


\section{Second-Order Derivatives of the Spacetime Metric} \label{section-6.5}
\subsection{Preliminary}

We now establish $L^2$ and $L^\infty$ bounds for the terms
$
\del_t\del_t\del^IL^J h_{\alpha \beta}$ and $\del^IL^J \del_t\del_t h_{\alpha \beta}$, 
which contain at least two partial derivatives $\del_t$ and which we refer informally to as ``second-order derivatives''. We can now apply the method in \cite[Chapter 8]{PLF-YM-book}. However, we are here in a simpler situation, since the system is diagonalized with respect to second-order derivative terms. We recall the decomposition of the flat wave operator in the semi-hyperboloidal frame:
\bel{eq 0 second-order}
- \Box u = (s/t)^2\del_t\del_t u + 2\sum_a (x^a/t) \delu_a \del_t u - \sum_a \delu_a \delu_a u + \frac{r^2}{t^3} \del_t u + \frac{3}{t} \del_t u.
\ee
We also have the decomposition 
$
h^{\mu\nu} \del_{\mu} \del_{\nu}h_{\alpha \beta} = \hu^{\mu\nu} \delu_\mu\delu_\nu h_{\alpha \beta} + h^{\mu\nu} \del_{\mu} \left(\Psi_\nu^{\nu'} \right) \delu_{\nu'}h_{\alpha \beta} 
$
of the curved part of the reduced wave operator. 
The main equation \eqref{eq main PDE a} leads us to
\bel{eq 1 second-order}
\aligned
\left((s/t)^2- \hu^{00} \right) \del_t\del_th_{\alpha \beta}
& =   - 2\sum_a (x^a/t) \delu_a \del_t h_{\alpha \beta} + \sum_a \delu_a \delu_a h_{\alpha \beta} - \frac{r^2}{t^3} \del_th_{\alpha \beta} - \frac{3}{t} \del_th_{\alpha \beta}
\\
& \quad +  \hu^{0a} \del_t\delu_ah_{\alpha \beta} + \hu^{a0} \delu_a \del_th_{\alpha \beta}
+ \hu^{ab} \delu_a \delu_b h_{\alpha \beta} + h^{\mu\nu} \del_\mu\left(\Psi_\nu^{\nu'} \right) \delu_{\nu'}h_{\alpha \beta}
\\
& \quad -  F_{\alpha \beta} + 16\pi \del_{\alpha} \phi\del_{\beta} \phi + 8\pi c^2\phi^2g_{\alpha \beta}.
\endaligned
\ee

Let us differentiate the equation \eqref{eq main PDE a} with respect to $\del^IL^J$, then by a similar procedure in the above discussion,
\bel{eq 2 second-order}
\aligned
&\left((s/t)^2- \hu^{00} \right) \del_t\del_t\del^IL^J h_{\alpha \beta}
\\
& =   - 2\sum_a (x^a/t) \delu_a \del_t \del^IL^Jh_{\alpha \beta} + \sum_a \delu_a \delu_a \del^IL^Jh_{\alpha \beta}- \frac{r^2}{t^3} \del_t\del^IL^Jh_{\alpha \beta} - \frac{3}{t} \del_t\del^IL^Jh_{\alpha \beta}
\\
& \quad +  \hu^{0a} \del_t\delu_a \del^IL^Jh_{\alpha \beta} + \hu^{a0} \delu_a \del_t\del^IL^Jh_{\alpha \beta}
+ \hu^{ab} \delu_a \delu_b \del^IL^Jh_{\alpha \beta} + h^{\mu\nu} \del_\mu\left(\Psi_\nu^{\nu'} \right) \delu_{\nu'} \del^IL^Jh_{\alpha \beta}
\\
& \quad -  \del^IL^JF_{\alpha \beta} + [\del^IL^J,h^{\mu\nu} \del_\mu\del_\nu]h_{\alpha \beta}
 + 16\pi \del^IL^J\left(\del_{\alpha} \phi\del_{\beta} \phi\right)
 + 8\pi c^2\del^IL^J\left(\phi^2g_{\alpha \beta} \right).
\endaligned
\ee
For convenience, we introduce the notation
$$
\aligned
{Sc}_1[\del^IL^J u] :& =  - 2\sum_a (x^a/t) \delu_a \del_t \del^IL^Ju + \sum_a \delu_a \delu_a  \del^IL^Ju- \frac{r^2}{t^3} \del_t \del^IL^Ju - \frac{3}{t} \del_t \del^IL^Ju,
\\
{Sc}_2[\del^IL^J u] :& =   \hu^{0a} \del_t\delu_a \del^IL^Ju + \hu^{a0} \delu_a \del_t\del^IL^Ju
+ \hu^{ab} \delu_a \delu_b \del^IL^Ju + h^{\mu\nu} \del_\mu\left(\Psi_\nu^{\nu'} \right) \delu_{\nu'} \del^IL^Ju 
\endaligned
$$
and \eqref{eq 1 second-order} becomes
\bel{eq 3 second-order}
\aligned
& \left((s/t)^2- \hu^{00} \right) \del_t\del_t\del^IL^J h_{\alpha \beta}  =  {Sc}_1[\del^IL^J h_{\alpha \beta}] + {Sc}_2[\del^IL^J h_{\alpha \beta}]
\\
& \quad -  \del^IL^JF_{\alpha \beta} + [\del^IL^J,h^{\mu\nu} \del_\mu\del_\nu]h_{\alpha \beta}
  +  16\pi \del^IL^J\left(\del_{\alpha} \phi\del_{\beta} \phi\right)
 + 8\pi c^2\del^IL^J\left(\phi^2g_{\alpha \beta} \right).
\endaligned
\ee

Now we apply the estimate \eqref{ineq h00-sup 1.5} to $\hu^{00}$ and see that when $t\geq 2$ (which is the case if we are in $\Kcal$) and $C_1\vep$ sufficiently small, then 
$$
\aligned
(s/t)^2- \hu^{00}
&  \geq (s/t)^2 - CC_1\vep\big((s/t)^2t^{-1/2}s^{\delta} + t^{-1} \big) 
\\
&= (s/t)^2\left(1-CC_1\vep t^{-1/2}s^\delta - CC_1\vep ts^{-2} \right) \geq \frac{1}{2}(s/t)^2.
\endaligned
$$
This leads us to the following estimate. Later, this equation will be used to control the $L^2$ and $L^\infty$ norms of $\del_t\del_t\del^IL^J h_{\alpha \beta}$.

\begin{lemma} \label{lem 1 second-roder}
When $C_1\vep$ is sufficiently small, the following estimate holds for all multi-indices $(I,J)$:
\bel{eq 1 lem 1 second-order}
\aligned
\left|(s/t)^2\del_t\del_t\del^IL^Jh_{\alpha \beta} \right|
& \leq C\left(\left|Sc_1[\del^IL^J h_{\alpha \beta}]\right| + \left|Sc_2[\del^IL^J h_{\alpha \beta}]\right|\right) + \left|\del^IL^JF_{\alpha \beta} \right| + \left|{\sl QS}_\phi(p,k) \right|
\\
& \quad +  \left|[\del^IL^J,h^{\mu\nu} \del_\mu\del_\nu]h_{\alpha \beta} \right| + |{\sl Cub}(p,k)|.
\endaligned
\ee
\end{lemma}


\subsection{$L^\infty$ estimates}

In this section, we apply \eqref{eq 3 second-order} and the estimates of nonlinear terms presented in Lemma \ref{lem 0 biliner}. First we need to establish the following pointwise estimates
\begin{lemma} \label{lem 2 second-order}
For any $(I,J)$, the following pointwise estimate holds in $\Kcal$:
\bel{ineq 1 lem 2 second-order}
\left|Sc_1[\del^IL^J u]\right| + \left|Sc_2[\del^IL^J u]\right|
\leq Ct^{-1} \sum_{|I'| \leq|I|, \alpha} \left|\del_{\alpha} \del^{I'}L^{J} u\right| + Ct^{-1} \sum_{a, \alpha} \left|\del_\alpha \del^IL_aL^J u\right|. 
\ee
\end{lemma}

\begin{proof}
The estimate on the term $Sc_1$ is immediate by applying \eqref{ineq 2 homo} and \eqref{ineq 3 homo}. The bound on $Sc_2$ is due to the fact that $\hu^{\alpha \beta}$ are linear combinations of $h_{\alpha \beta}$ with smooth and homogeneous functions of degree zero plus higher-order corrections, which are bounded in $\Kcal$. 
\end{proof}

\begin{lemma} 
When the bootstrap assumption \eqref{ineq energy assumption 1} and \eqref{ineq energy assumption 2} hold, the following estimate holds in $\Kcal_{[2,s^*]}$:
\bel{ineq 1 second-sup}
|\del_t\del_t \del^IL^Jh_{\alpha \beta}| \leq CC_1\vep t^{1/2}s^{-3+2\delta}, \quad \text{for} \quad |I|+|J| \leq N-4.
\ee
\end{lemma}

\begin{proof}
The proof is a direct application of \eqref{eq 1 lem 1 second-order}, where we neglect the higher-order term ${\sl Cub}$. We just need to estimate each term in the right-hand side.
We first observe that by the basic sup-norm estimate \eqref{ineq basic-sup 1 generation 1 a} combined with \eqref{ineq 1 lem 2 second-order}
$$
\left|Sc_1[\del^IL^J u]\right| + \left|Sc_2[\del^IL^J u]\right| \leq CC_1\vep t^{-3/2}s^{-1+\delta}.
$$

The estimate for $\del^IL^JF_{\alpha \beta}$ can be expressed as  
$
{\sl QS}_{h}(p,k), \quad {\sl Cub}(p,k), 
$
which is bounded by
$
|\del^IL^JF_{\alpha \beta}| \leq C(C_1\vep)^2t^{-1}s^{-2 +2\delta}.
$
The estimate on the commutator $[\del^IL^J,h^{\mu\nu} \del_\mu\del_\nu]h_{\alpha \beta}$ is obtained by applying \eqref{ineq lem 1 commtator-sup I} :
$$
|[\del^IL^J,h^{\mu\nu} \del_\mu\del_\nu]h_{\alpha \beta}|
\leq C(C_1\vep)^2t^{-2}s^{-1+2\delta}
+ CC_1\vep\left(t^{-1} + (s/t)^2t^{-1/2}s^{\delta} \right) \sum_{|J'|<|J|} \left|\del_t\del_t\del^IL^{J'} h_{\alpha \beta} \right|.
$$
The estimate for ${\sl QS}_{\phi}$ is derived as follows. We only estimate $\del^IL^J\left(\del_{\alpha} \phi\del_\beta \phi\right)$, since 
dealing with the term $\del^IL^J\left(\phi^2\right)$ is easier: 
$$
\left|\del^IL^J\left(\del_{\alpha} \phi\del_\beta \phi\right) \right|
\leq \sum_{|I_1|+|I_2|=I\atop |J_1|+|J_2|=J} \left|\del^{I_1}L^{J_1} \del_{\alpha} \phi \,   \del^{I_2}L^{J_2} \del_{\beta} \phi\right|. 
$$
Recalling that $|I|+|J| \leq N-4$, we obtain:   
\begin{itemize}

\item When $|I_1|+|J_1| \leq N-7$,
$$
\big|\del^{I_1}L^{J_1} \del_{\alpha} \phi \,   \del^{I_2}L^{J_2} \del_{\beta} \phi\big|
\leq CC_1\vep\big|t^{-3/2}s^{\delta} \big| \,   CC_1\vep \big|t^{-1/2}s^{-1/2 +\delta} \big| \leq C(C_1\vep)^2t^{-2}s^{-1/2 +2\delta}.
$$

\item When $N-6\leq |I_1|+|J_1| \leq N-4$, we see that $|I_2|+|J_2| \leq 2\leq N-7$ and 
$$
\big|\del^{I_1}L^{J_1} \del_{\alpha} \phi \,   \del^{I_2}L^{J_2} \del_{\beta} \phi\big|
\leq CC_1\vep \big|t^{-1/2}s^{-1/2 +\delta} \big| \,   CC_1\vep\big|t^{-3/2}s^{\delta} \big| \leq C(C_1\vep)^2t^{-2}s^{-1/2 +2\delta}.
$$
\end{itemize}
So, we conclude that
$
|{\sl QS}_{\phi}(N-4,k)| \leq C(C_1\vep)^2(s/t)^2s^{-5/2 +2\delta}.
$
We thus have 
\bel{ineq 1 pr ineq 1 second-sup}
\aligned
|(s/t)^2\del_t\del_t \del^IL^Jh_{\alpha \beta}| & \leq  CC_1\vep t^{-3/2}s^{-1+\delta} + C(C_1\vep)^2(s/t)^2s^{-5/2 +2\delta}
\\
& \quad + CC_1\vep\left(t^{-1}
+ (s/t)^2t^{-1/2}s^{\delta} \right) \sum_{|J'|<|J|} \big|\del_t\del_t\del^IL^{J'} h_{\alpha \beta} \big|.
\endaligned
\ee
Observe that when $|J|=0$, the last term in the above estimate disappears and we conclude with \eqref{ineq 1 second-sup}. We proceed by induction on $|J|$. Assume that \eqref{ineq 1 second-sup} holds for all $|J| \leq m-1<N-4$. We will prove that it still holds for $|J|=m\leq N-4$. We substitute \eqref{ineq 1 second-sup} (case $|J'|<|J|=m$) into the last term of \eqref{ineq 1 pr ineq 1 second-sup}.
\end{proof}


\subsection{$L^2$ estimates}

The following two estimates are direct in view of \eqref{ineq 2 homo} and \eqref{ineq 3 homo} combined with the expression of the energy $E_M^*$.

\begin{lemma}
For {\sl all} multi-indices $(I,J)$, one has 
\bel{ineq pr2 Sc-L2}
\aligned
& \left\|\delu_a \del_{\alpha} \del^IL^J h_{\alpha \beta} \right\|_{L^2(\Hcal_s^*)}
 + \left\|\del_{\alpha} \delu_a \del^IL^J h_{\alpha \beta} \right\|_{L^2(\Hcal_s^*)}
\\
& \leq  Cs^{-1}E_M^*(s, \del^IL_aL^Jh_{\alpha \beta})^{1/2}
 + Cs^{-1} \sum_{|I'| \leq|I|, \gamma}E_M^*(s, \del^{I'}L^Jh_{\alpha \beta})^{1/2}.
\endaligned
\ee
\end{lemma}

A direct consequence of these bounds is that, 
for {\sl any} $(I,J)$,
\bel{ineq 1 Sc-L2}
\aligned
\left\|Sc_1[\del^IL^Jh_{\alpha \beta}]\right\|_{L^2(\Hcal_s^*)} 
& \leq Cs^{-1} \sum_aE_M^*(s, \del^IL_aL^Jh_{\alpha \beta})^{1/2} + Cs^{-1} \sum_{|I'| \leq|I|}E_M^*(s, \del^{I'}L^Jh_{\alpha \beta})^{1/2}.
\endaligned
\ee 
This estimate will play an essential role in our forthcoming analysis. 
Our next task is the derivation of an $L^2$ estimate for $Sc_2$. The term $h^{\mu\nu} \del_{\mu} \Psi_{\nu}^{\nu'}  \delu_{\nu'}h_{\alpha \beta}$ is bounded by the additional decay of $\left|\del_{\mu} \Psi_{\nu}^{\nu'}  \right| \leq t^{-1}$, and we thus focus on the first three quadratic terms. We provide the derive for the first term (but omit the second and third terms): 
$$
\aligned
& \left\|(t/s)^{3/2}h^{0a} \del_t\delu_a \del^IL^Jh_{\alpha \beta} \right\|_{L^2(\Hcal_s^*)}
\\
&\leq CC_1\vep \left\|(t/s)^{3/2} \left(t^{-1} + (s/t)t^{-1/2}s^{\delta} \right) \del_t\delu_a \del^IL^Jh_{\alpha \beta} \right\|_{L^2(\Hcal_s^*)}
\\
&\leq CC_1\vep s^{-1/2} \left\|\del_t\delu_a \del^IL^Jh_{\alpha \beta} \right\|_{L^2(\Hcal_s^*)}
+ CC_1\vep \left\|s^{-1/2 +\delta} \del_t\delu_a \del^IL^Jh_{\alpha \beta} \right\|_{L^2(\Hcal_s^*)}
\\
&\leq CC_1\vep s^{-1/2 +\delta} \left\|\del_t\delu_a \del^IL^Jh_{\alpha \beta} \right\|_{L^2(\Hcal_s^*)}. 
\endaligned
$$
Then we apply \eqref{ineq pr2 Sc-L2} and obtain 
\be
\aligned
\left\|(t/s)^{3/2}h^{0a} \del_t\delu_a \del^IL^Jh_{\alpha \beta} \right\|_{L^2(\Hcal_s^*)}
& \leq   CC_1\vep s^{-3/2 +\delta} \sum_a E_M^*(s, \del^IL_aL^Jh_{\alpha \beta})^{1/2}
\\
& \quad +  CC_1\vep s^{-3/2 +\delta} \sum_{|I'| \leq|I|, \gamma}E_M^*(s, \del^{I'}L^Jh_{\alpha \beta})^{1/2}. 
\endaligned
\ee 
We conclude that
\bel{ineq 2 Sc-L2}
\aligned
\big\|(t/s)^{3/2}Sc_2[\del^IL^Jh_{\alpha \beta}]\big\|_{L^2(\Hcal_s^*)} & \leq   CC_1\vep s^{-3/2 +\delta} \sum_a E_M^*(s, \del^IL_aL^Jh_{\alpha \beta})^{1/2}
\\
& \quad +  CC_1\vep s^{-3/2 +\delta} \sum_{|I'| \leq|I|, \gamma}E_M^*(s, \del^{I'}L^Jh_{\alpha \beta})^{1/2}.
\endaligned
\ee

With the above preparation,  in the rest of this subsection we will prove the following. 

\begin{lemma}
Under the bootstrap assumption \eqref{ineq energy assumption 1} and \eqref{ineq energy assumption 2}
\bel{ineq 1 second-L2}
\|s^3t^{-2} \del_t\del_t\del^IL^J h_{\alpha \beta} \|_{L^2(\Hcal_s^*)} \leq CC_1\vep s^{2\delta}, \qquad |I|+|J| \leq N-1.
\ee
\end{lemma}

\begin{proof}
{\sl Step I. Estimates for the nonlinear terms.}
The estimate of \eqref{ineq 1 second-L2} is also based on Lemma \ref{lem 1 second-roder}. 

\noindent 1.
This is done by direct application of \eqref{ineq 1 Sc-L2} combined with the energy assumption:
$$
\aligned
\left\|Sc_1[\del^IL^Jh_{\alpha \beta}]\right\|_{L^2(\Hcal_s^*)} 
& \leq  CC_1\vep s^{-1+\delta}.
\endaligned
$$ 

\noindent 2. For the term $Sc_2$ is bounded in view of \eqref{ineq 2 Sc-L2} combined with the energy assumption:
$$
\left\|Sc_1[\del^IL^Jh_{\alpha \beta}]\right\|_{L^2(\Hcal_s^*)} \leq C(C_1\vep)^2s^{-3/2 +2\delta}.
$$

\noindent 3. Now we are about to estimate $\del^IL^JF_{\alpha \beta}$. We observe that this term is a linear combination of
$
{\sl QS}_h(p,k)$ and ${\sl Cub}(p,k).
$
We see  that the term ${\sl QS}_h(p,k)$ is bounded as follows:
$$
\aligned
\left\|{\sl QS}_h(p,k) \right\|_{L^2(\Hcal_s^*)}
& \leq \sum_{\alpha, \beta, \alpha'\beta'\atop \gamma, \gamma'} \sum_{|I_1|+|I_2| \leq |I|\atop |J_1|+|J_2| \leq |J|}
\left\|\del^{I_1}L^{J_1} \del_{\gamma}h_{\alpha \beta} \,   \del^{I_2}L^{L_2} \del_{\gamma'}h_{\alpha'\beta'} \right\|_{L^2(\Hcal_s^*)}
\endaligned
$$
When $N\geq 3$, we must have either $|I_1|+|J_1| \leq N-2$ or $|I_2|+|J_2| \leq N-2$. So
$$
\aligned
&a \left\|\del^{I_1}L^{J_1} \del_{\gamma}h_{\alpha \beta} \,   \del^{I_2}L^{L_2} \del_{\gamma'}h_{\alpha'\beta'} \right\|_{L^2(\Hcal_s^*)}
\leq  CC_1\vep\left\|t^{-1/2}s^{-1+\delta} \del^{I_2}L^{L_2} \del_{\gamma'}h_{\alpha'\beta'} \right\|_{L^2(\Hcal_s^*)}
\\
& \leq  CC_1\vep s^{\delta} \left\|(t/s)t^{-1/2}s^{-1+\delta} \,   (s/t) \del^{I_2}L^{L_2} \del_{\gamma'}h_{\alpha'\beta'} \right\|_{L^2(\Hcal_s^*)}
\\
& \leq  CC_1\vep s^{-1+\delta}E_M^*(s, \del^{I_2}L^{J_2} h_{\alpha'\beta'})^{1/2} \leq C(C_1\vep)^2s^{-1+2\delta}.
\endaligned
$$
We can conclude that
$
\left\|\del^IL^JF_{\alpha \beta} \right\|_{L^2(\Hcal_s^*)} \leq C(C_1\vep)^2s^{-1+2\delta}.
$

\noindent 4. ${\sl QS}_\phi$ is bounded directly in view of \eqref{ineq 5 bilinear-L2}.

\noindent 5. The estimate on the commutator is the most difficult. We combine the sup-norm estimate \eqref{ineq 1 second-sup} and the estimate \eqref{ineq 1 lem 1 commtator-L2 I} :
$$
\aligned
\left\|s[\del^IL^J,h^{\mu\nu} \del_\mu\del_\nu]h_{\alpha \beta} \right\|_{L^2(\Hcal_s^*)}
& \leq  C(C_1\vep)^2s^{2\delta}
 + CC_1\vep s^{\delta} \sum_{|J'| \leq 1}
\left\|s^2(s/t)^{1- \delta} \del^IL^{J_2'} \del_t\del_th_{\alpha \beta} \right\|_{L^\infty(\Hcal_s^*)}
\\
& \quad +  CC_1\vep s^{1/2 +\delta} \sum_{|J'|<|J|} \left\|(s/t)^{5/2} \del_t\del_t\del^IL^{J'}h_{\alpha \beta} \right\|_{L^2(\Hcal_s^*)}
\\
& \leq  C(C_1\vep)^2s^{2\delta} + C(C_1\vep)^2s^{\delta} \|s^2(s/t)^{1- \delta}t^{1/2}s^{-3+2\delta} \|_{L^\infty(\Hcal_s^*)}
\\
& \quad +  CC_1\vep s^{1/2 +\delta} \sum_{|J'|<|J|} \left\|(s/t)^{5/2} \del_t\del_t\del^IL^{J'}h_{\alpha \beta} \right\|_{L^2(\Hcal_s^*)}
\\
& \leq C(CC_1\vep)^2s^{2\delta}
+ CC_1\vep s^{1/2 +\delta} \sum_{|J'|<|J|} \left\|(s/t)^{5/2} \del_t\del_t\del^IL^{J'}h_{\alpha \beta} \right\|_{L^2(\Hcal_s^*)}. 
\endaligned
$$

We thus conclude Step 1 with the inequality 
\bel{ineq 1 pr ineq 1 second-L2}
\left\|s^3t^{-2} \del^IL^J\del_t\del_th_{\alpha \beta} \right\|_{L^2(\Hcal_s^*)}
\leq CC_1\vep s^{2\delta}
+ CC_1\vep s^{1/2 +\delta} \sum_{|J'|<|J|} \left\|(s/t)^{5/2} \del_t\del_t\del^IL^{J'}h_{\alpha \beta} \right\|_{L^2(\Hcal_s^*)}
\ee
and we remark that when $|J|=0$ the last sum is empty.

\vskip.15cm

\noindent
{\sl Step II. Induction argument}
For $|I|+|J| \leq N-1$, we proceed by induction on $|J|$. When $|J|=0$, the last term in \eqref{ineq 1 pr ineq 1 second-L2} does not exist. Then in view of \eqref{eq 1 lem 1 second-order}, we have 
$$
\left\|s^3t^{-2} \del_t\del_t\del^IL^Jh_{\alpha \beta} \right\|_{L^2(\Hcal_s^*)} \leq C(C_1\vep)s^{2\delta}.
$$
Then we assume that \eqref{ineq 1 second-L2} holds for $|J| \leq n<N-1$, we want to prove that it still holds for $|J|=n$.  In this case, by our induction assumption, we have 
$$
\aligned
\left\|s^3t^{-2} \del^IL^J\del_t\del_th_{\alpha \beta} \right\|_{L^2(\Hcal_s^*)}
& \leq  C(CC_1\vep)^2s^{2\delta}
+ CC_1\vep s^{1/2 +\delta} \sum_{|J'|<|J|} \left\|(s/t)^{5/2} \del_t\del_t\del^IL^{J'}h_{\alpha \beta} \right\|_{L^2(\Hcal_s^*)}
\\
& \leq C(C_1\vep)^2s^{2\delta}.
\endaligned
$$
Then in view of \eqref{eq 1 lem 1 second-order}, the desired result is established.
\end{proof}


\subsection{Conclusion for general second-order derivatives}

In the above subsection we have only estimate the terms of the form $\del_t\del_t\del^IL^Jh_{\alpha \beta}$, but we observe that by the identities \eqref{eq second-order-frame-1} (and a similar argument below it in the proof of \eqref{lem 1 commtator-L2 I}) and the commutator estimates \eqref{ineq com 2.4} 
\bel{ineq 1 second-order}
| \del_\alpha \del_\beta \del^IL^Jh_{\alpha \beta}| \leq CC_1\vep t^{1/2}s^{-3+2\delta}, \quad  \quad |I|+|J| \leq N-4,
\ee
\bel{ineq 2 second-order}
\|s^3t^{-2} \del_\alpha \del_\beta \del^IL^J h_{\alpha \beta} \|_{L^2(\Hcal_s^*)} \leq CC_1\vep s^{2\delta}, \quad    \quad |I|+|J| \leq N-1,
\ee
\bel{ineq 3 second-order}
| \del^IL^J\del_\alpha \del_\beta h_{\alpha \beta}| \leq CC_1\vep t^{1/2}s^{-3+2\delta}, \quad   \quad |I|+|J| \leq N-4,
\ee
\bel{ineq 4 second-order}
\|s^3t^{-2} \del^IL^J\del_\alpha \del_\beta h_{\alpha \beta} \|_{L^2(\Hcal_s^*)} \leq CC_1\vep s^{2\delta}, \quad   \quad |I|+|J| \leq N-1.
\ee


\subsection{Commutator estimates}  
 
In this section, we improve the sup-norm and $L^2$ estimates for the commutators: our strategy is to apply Lemma \ref{lem 1 nonlinear}.

\begin{lemma} \label{lem 1 commtator sup-L2 II}
Assume that the energy assumptions \eqref{ineq energy assumption 1} and \eqref{ineq energy assumption 2} hold, then 
for all $|I|+|J| \leq N-4$
\bel{ineq sup lem 1 commtator II}
\left|[\del^IL^J,h^{\mu\nu} \del_\mu\del_\nu]h_{\alpha \beta} \right| \leq C(C_1\vep)^2 t^{-2}s^{-1+3\delta}
+ C(C_1\vep)^2t^{-1/2}s^{-3+2\delta}, 
\ee
while for all $|I|+|J| \leq N$
\bel{ineq L-2 lem 1 commtator II}
\left\|s[\del^IL^J,h^{\mu\nu} \del_\mu\del_\nu]h_{\alpha \beta} \right\|_{L^2(\Hcal_s^*)} \leq C(C_1\vep)^2 s^{-1/2 + 3\delta} + CC_1\vep \sum_{|J'|<|J|} \left\|s^3t^{-2} \del_t\del_t\del^IL^{J'}h_{\alpha \beta} \right\|_{L^2(\Hcal_s^*)}. 
\ee
\end{lemma}

\begin{proof}
The proof of \eqref{ineq sup lem 1 commtator II} is immediate by combining \eqref{ineq 1 second-order} with \eqref{ineq lem 1 commtator-sup I}. 
The proof of \eqref{ineq L-2 lem 1 commtator II} relies on a refinement of the proof of \eqref{ineq 1 lem 1 commtator-L2 I}. We will improve upon our estimates on $L^{J_1'} \hu^{00} \del^IL^{J_2'} \del_t\del_th_{\alpha \beta}$ and $\hu^{00} \del^IL^{J'}h_{\alpha \beta}$. First we observe that for $L^{J_1'} \hu^{00} \del^IL^{J_2'} \del_t\del_th_{\alpha \beta}$
\begin{itemize}

\item When $1\leq |J_1'| \leq N-2$
$$
\aligned
& \left\|sL^{J_1'} \hu^{00} \del^IL^{J_2'} \del_t\del_th_{\alpha \beta} \right\|_{L^2(\Hcal_s^*)}
  \leq CC_1\vep \left\|s\left(t^{-1}+(s/t)^2t^{-1/2}s^{\delta} \right) \,   \del^IL^{J_2'} \del_t\del_th_{\alpha \beta} \right\|_{L^2(\Hcal_s^*)}
\\
&\leq CC_1\vep\left\|(s/t) \del^IL^{J_2'} \del_t\del_th_{\alpha \beta} \right\|_{L^2(\Hcal_s^*)}
   + CC_1\vep s^{1/2 +\delta} \left\|(s/t)^{5/2} \del^IL^{J_2'} \del_t\del_th_{\alpha \beta} \right\|_{L^2(\Hcal_s^*)}
\\
&\leq CC_1\vep \left\|(s/t) \del^IL^{J_2'} \del_t\del_th_{\alpha \beta} \right\|_{L^2(\Hcal_s^*)}
  + CC_1\vep s^{1/2 +\delta} \sum_{|J'|<|J|} \left\|(s/t)^{5/2} \del^IL^{J'} \del_t\del_th_{\alpha \beta} \right\|_{L^2(\Hcal_s^*)}
\\
&\leq C(C_1\vep)^2s^{-1/2 +3\delta}+ CC_1\vep \left\|(s/t) \del^IL^{J_2'} \del_t\del_th_{\alpha \beta} \right\|_{L^2(\Hcal_s^*)}.
\endaligned
$$

\item When $|J_1'|\geq  N-1$, then $|J_2'|+|I| \leq 1 \leq N-4$, we apply \eqref{prop h00-sup 1} to $\del^{J_1'} \hu^{00}$:
$$
\aligned
&
\left\|sL^{J_1'} \hu^{00} \del^IL^{J_2'} \del_t\del_th_{\alpha \beta} \right\|_{L^2(\Hcal_s^*)}
 \\
& \leq 
CC_1\vep \left\|st^{-1} \del^IL^{J_2'} \del_t\del_th_{\alpha \beta} \right\|_{L^2(\Hcal_s^*)}
 + \left\|sL^{J_1'} \hu_1^{00} \del^IL^{J_2'} \del_t\del_th_{\alpha \beta} \right\|_{L^2(\Hcal_s^*)}
 \\ 
& \leq  C(C_1\vep) \left\| (s/t) \del^IL^{J_2'} \del_t\del_th_{\alpha \beta} \right\|_{L^2(\Hcal_s^*)}
  +  \left\|s^{-1}(s/t)^{-1+\delta}L^{J_1'} \hu^{00}_1\right\|_{L^2(\Hcal_s^*)} \,  \left\|s^2(s/t)^{1- \delta} \del^IL^{J_2'} \del_t\del_th_{\alpha \beta} \right\|_{L^\infty(\Hcal_s^*)}
\\
& \leq  C(C_1\vep)^2s^{-1/2 +3\delta}
 +  CC_1\vep\left\|(s/t) \del^IL^{J_2'} \del_t\del_th_{\alpha \beta} \right\|_{L^2(\Hcal_s^*)}
\endaligned
$$
\end{itemize}
 
For the term $\hu^{00} \del_{\gamma} \del_{\gamma'} \del^IL^{J'}h_{\alpha \beta}$, the estimate is similar:
$$
\aligned
& \left\|s\hu^{00} \del_{\gamma} \del_{\gamma'} \del^IL^{J'}h_{\alpha \beta} \right\|_{L^2(\Hcal_s^*)}
\\
& \leq CC_1\vep\left\|(s/t) \del_{\gamma} \del_{\gamma'} \del^IL^{J'}h_{\alpha \beta} \right\|_{L^2(\Hcal_s^*)}
+ \left\|(s/t)^2t^{-1/2}s^{1+\delta} \del_{\gamma} \del_{\gamma'} \del^IL^{J'}h_{\alpha \beta} \right\|_{L^2(\Hcal_s^*)}
\\
&\leq CC_1\vep\left\|(s/t) \del_\gamma \del_{\gamma'} \del^IL^{J'}h_{\alpha \beta} \right\|_{L^2(\Hcal_s^*)}
 +  CC_1\vep s^{1/2 +\delta} \sum_{|J'|<|J|} \left\|(s/t)^{5/2} \del_{\gamma} \del_{\gamma'} \del^IL^{J'}h_{\alpha \beta} \right\|_{L^2(\Hcal_s^*)}
\\
& \leq C(C_1\vep)^2s^{-1/2 +3\delta} +  CC_1\vep\left\|(s/t) \del_\gamma \del_{\gamma'} \del^IL^{J'}h_{\alpha \beta} \right\|_{L^2(\Hcal_s^*)}.
\endaligned
$$

Now,  
$\big|\del^IL^J\del_t\del_t h_{\alpha \beta} \big|
\leq \sum_{|J'| \leq |J|\atop \gamma, \gamma'} \big|\del_{\gamma} \del_{\gamma'} \del^IL^{J'}h_{\alpha \beta} \big|$
in view of the commutator estimates \eqref{ineq com 2.4},  
and, by the same argument after \eqref{eq second-order-frame-1},
$$
\left\|(s/t) \del_\gamma \del_{\gamma'} \del^IL^J h_{\alpha \beta} \right\|_{L^2(\Hcal_s^*)}
\leq
\sum_{|J'| \leq |J|} \left\|(s/t) \del_t\del_t\del^IL^{J'}h_{\alpha \beta} \right\|_{L^2(\Hcal_s^*)} + CC_1\vep s^{-1+\delta}.
$$
So, we conclude that
$$
\aligned
& \left\|(s/t) \del^IL^J\del_t\del_t h_{\alpha \beta} \right\|_{L^2(\Hcal_s^*)}
 + \left\|(s/t) \del_{\gamma} \del_{\gamma'} \del^IL^J h_{\alpha \beta} \right\|_{L^2(\Hcal_s^*)} 
\\
& \leq C\sum_{|J'| \leq|J|} \left\|s^3t^{-2} \del_t\del_t\del^IL^{J'}h_{\alpha \beta} \right\|_{L^2(\Hcal_s^*)} + CC_1\vep s^{-1+\delta}.
  \qedhere \endaligned
$$ 
\end{proof}


\section{Sup-Norm Estimate Based on Characteristics}
\label{section-7}

\subsection{Main statement in this section}

Our goal in this section is to control null derivatives, as now stated. 

\begin{proposition} \label{prop Ch-sup 1}
Assume that \eqref{ineq energy assumption 1} and \eqref{ineq energy assumption 2} hold with $C_1\vep$ sufficiently small, then for $|I|+|J| \leq N-4$, 
\bel{ineq Ch-sup 1}
|(\del_t- \del_r) \del^IL^J\del_{\alpha} \hu_{a \beta}| \leq CC_1\vep t^{-1+ C\vep},
\ee
\bel{ineq Ch-sup 1.5}
|(\del_t- \del_r) \del^I\hu_{a \beta}| \leq CC_1\vep t^{-1}.
\ee
\end{proposition}

\begin{proof} The proof relies on our earlier estimate along characteristics. We first write the estimate on the components $\hu_{a0}$ in details, and then we sketch the proof on $\hu_{ab}$.

\noindent{\sl Step I. Estimates for the correction terms.} We observe that the equation satisfied by $\hu_{0a} $:
$$
\Boxt_g\hu_{0a} = \Phi_0^{\alpha'} \Phi_a^{\beta'}Q_{\alpha'\beta'} + \Pu_{0a} - 16\pi \delu_a \phi\del_t\phi - 8\pi \minu_{a0} \phi^2
+ \frac{2}{t} \delu_a h_{00} - \frac{2x^a}{t^3}h_{00} + {\sl Cub}(0,0).
$$
Differentiating this equation with respect to $\del^IL^J$, we have
\bel{eq Ch-sup 2}
\aligned
\Boxt_g\big(\del^IL^J\hu_{0a} \big) & =  \del^IL^J\big(\Phi_0^{\alpha'} \Phi_a^{\beta'}Q_{\alpha'\beta'} \big) + \del^IL^J\big(\Pu_{0a} \big)
- 16\pi\del^IL^J \big(\delu_a \phi\del_t\phi\big) - 8\pi \del^IL^J\big(\minu_{a0} \phi^2\big)
\\
& \quad -  [\del^IL^J,h^{\mu\nu} \del_{\mu} \del_{\nu}]\hu_{a0}
 + \del^IL^J\bigg(\frac{2}{t} \delu_a h_{00} - \frac{2x^a}{t^3}h_{00} \bigg) + \del^IL^J{\sl Cub}(0,0).
\endaligned
\ee
Then we apply Lemma \ref{prop LR-sup 1} to this equation. We need to estimate the $L^\infty$ norm of the terms in the right-hand side and  the corrective $M_s[\del^IL^J\hu_{a0},h]$.

First of all, in view of \eqref{ineq 1 bilinear-sup}, the null terms $\Phi_0^{\alpha'} \Phi_a^{\beta'}Q_{\alpha'\beta'}$ decay like 
 $C(C_1\vep)^2 t^{-2}s^{-1+2\delta}$ and in view of \eqref{ineq 2 bilinear-sup}, the quadratic terms ${\sl QS}_{\phi}$ is bounded by $C(C_1\vep)^2 t^{-2}s^{-1/2 +2\delta}$. 
We also observe that by the tensorial structure of the Einstein equation, the term $\del^IL^JP_{a \beta}$ is also a null term, so it is bounded by $C(C_1\vep)^2 t^{-2}s^{-1+2\delta} $. We also point out that the high-order terms $\del^IL^J{\sl Cub}(0,0)$ enjoy also the sufficient decay $C(C_1\vep)^2 t^{-2}s^{-1+2\delta}$.

We focus on the linear correction terms  $\del^IL^J\big(\frac{2}{t} \delu_a h_{00} - \frac{2x^a}{t^3}h_{00} \big)$. We observe that this term is a linear combination of  
$
t^{-1} \del^IL^J \delu_a h_{00}$ and $t^{-2} \del^IL^J h_{00}$ with $|I|+|J| \leq N-4$
with smooth and homogeneous coefficients of degree $\leq 0$. Then, these terms can be bounded by $CC_1\vep t^{-5/2}s^{\delta}$.

Then, we analyze the commutator term $[\del^IL^J, h^{\mu\nu} \del_{\mu} \del_{\nu}]\hu_{a0}$. We recall that $\hu_{a0}$ is a linear combination of $h_{\alpha \beta}$ with smooth and homogeneous coefficients of degree zero, then the estimate for this term relies on  Lemma \ref{lem 1 nonlinear}. In the list \eqref{eq 1 lem 1 nolinear}, we observe that we need only to estimate the terms
$\del^{I_1}L^{J_1} \hu^{00} \del^{I_2}L^{J_2} \del_t\del_th_{\alpha \beta}$, 
$L^{J_1'} \hu^{00} \del^IL^{J_2'} \del_t\del_th_{\alpha \beta}$,
$\hu^{00} \del_{\gamma} \del_{\gamma'} \del^IL^{J'}h_{\alpha \beta}$, 
since the remaining terms can be bounded by $C(C_1\vep)^2t^{-2}s^{-1+2\delta}$ (see the proof of Lemma \ref{lem 1 commtator-sup I}). 
For the above three terms, we apply \eqref{ineq 1 second-order}, \eqref{ineq 3 second-order} and \eqref{ineq h00-sup 1.5} :
$$
\aligned
\left|L^{J_1'} \hu^{00} \del^IL^{J_2'} \del_t\del_th_{\alpha \beta} \right|
& \leq  CC_1\vep\left|\left(t^{-1}+(s/t)^2t^{-1/2}s^{\delta} \right) \del^IL^{J_2'} \del_t\del_th_{\alpha \beta} \right|
\\
& \leq  CC_1\vep t^{-1} \left|\del^IL^{J_2'} \del_t\del_th_{\alpha \beta} \right| + C(C_1\vep)^2t^{-2}s^{-1+3\delta}
\\
& \leq CC_1\vep t^{-1} \sum_{|J_1'| \leq|J'|\atop \gamma, \gamma'} \left|\del_{\gamma} \del_{\gamma'} \del^IL^{J_2'}h_{\alpha \beta} \right| + C(C_1\vep)^2t^{-2}s^{-1+3\delta},
\endaligned
$$
and
$
\aligned
\left|\hu^{00} \del_{\gamma} \del_{\gamma'} \del^IL^{J'}h_{\alpha \beta} \right|
& \leq  CC_1\vep t^{-1} \left|\del_\gamma \del_{\gamma'} \del^IL^{J'}h_{\alpha \beta} \right| + C(C_1\vep)^2t^{-2}s^{-1+3\delta},
\endaligned
$
where in the last inequality we applied \eqref{ineq 1 second-order}. Then thanks to \eqref{eq second-order-frame-1}  and the discussion below these identities in the proof of Lemma \ref{lem 1 commtator-sup I},
$
\left|\del_{\gamma} \del_{\gamma'} \del^IL^{J'}h_{\alpha \beta} \right| \leq CC_1\vep t^{-3/2}s^{-1+\delta}
 + \left|\del_t\del_t \del^IL^{J'}h_{\alpha \beta} \right|,
$
so that 
$$
\left|\hu^{00} \del_{\gamma} \del_{\gamma'} \del^IL^{J'}h_{\alpha \beta} \right| \leq C(C_1\vep)^2t^{-2}s^{-1+3\delta}
+ CC_1\vep t^{-1} \left|\del_t\del_t \del^IL^{J'}h_{\alpha \beta} \right|.
$$
Then, by combining this with the commutator estimates, we obtain
\bel{ineq 1 proof prop Ch-sup 1}
\aligned
\left|[\del^IL^J,h^{\mu\nu} \del_\mu\del_\nu]\hu_{a0} \right|
& \leq  Cm_S t^{-1} \sum_{|J'|<|J|}|\del^IL^{J'} \del_{\alpha} \del_{\beta} \hu_{a0}| + C(C_1\vep)^2t^{-2}s^{-1+3\delta}.
\endaligned
\ee 

Finally we analyze the correction term $M_s[\del^IL^J\hu_{a0},h]$. We recall that
$$
M_s[\del^IL^J \hu_{a0},h] = r\sum_{a<b} \left(r\Omega_{ab} \right)^2u + \hu^{00}W_1[\del^IL^J \hu_{a0}] 
+ rR[\del^IL^J \hu_{a0},h]. 
$$
We see that
$
r^{-1} \Omega_{ab} =\frac{x^a}{r} \delu_b - \frac{x^b}{r} \delu_a
$
is a linear combination of the ``good'' terms. So by a similar argument to \eqref{ineq 4 homo}, we have 
$
\big|\left(r^{-1} \Omega_{ab} \right)^2\del^IL^J\hu_{a0} \big| \leq CC_1\vep t^{-5/2}s^{\delta}.
$
The term $W_1$ is a linear combination of first- and second-order derivatives with  coefficients bounded in $\Kcal\backslash \Kical$. We apply \eqref{ineq h00-sup 1.5} to $\hu^{00}$, and we get 
$
\big|\hu^{00}W_1[\del^IL^J\hu_{a0}]\big| \leq C(C_1\vep)^2t^{-2}s^{2\delta}.
$
The term $R[\del^IL^J\hu_{a0},h]$ is bounded similarly, and is a linear combination of the quadratic terms of the following form with homogeneous coefficients:
$\hu^{\alpha \beta} \delu_a \delu_\beta \del^IL^J\hu_{a0}$,
$t^{-1} \hu^{\alpha \beta} \delu_{\beta} \del^IL^J\hu_{a0}$. 
For the first term, we apply \eqref{ineq 4 homo} and \eqref{ineq basic-sup-h} : the linear part of $\hu^{\alpha \beta}$ is a linear combination of $h_{\alpha \beta}$ with smooth and homogeneous coefficients of degree zero. The second term is bounded by the additional decreasing factor $t^{-1}$ and therefore 
$
\big|R[\del^IL^J\hu_{a0},h]\big| \leq C(C_1\vep)^2 t^{-3}s^{2\delta}.
$
Then we conclude that
$$
|M_s[\del^IL^J \hu_{a0},h](t, x)| \leq CC_1\vep t^{-3/2}s^{2\delta}, 
\qquad 3/5\leq r/t\leq 1, \quad |I|+|J| \leq N-4.
$$

\vskip.15cm

\noindent {\sl Step II. Case of $|J|=0$.} Now we substitute the above estimate into the inequality \eqref{ineq LR-sup 1} and observe that when $|J| = 0$, the first term in the right-hand side of \eqref{ineq 1 proof prop Ch-sup 1} disappears. Then, we have 
$$
\aligned
|(\del_t- \del_r) \del^I\hu_{a0}| & \leq  Ct^{-1} \sup_{\del_B\Kical_{[2,s^*]} \cup \del\Kcal} \{|(\del_t- \del_r)(r\del^I\hu_{a0})|\}
+ Ct^{-1}|\del^I\hu_{a0}(t, x)|
\\
& \quad +  C(C_1\vep)^2 t^{-1} \int_{a_0}^t\tau^{-5/4+3\delta}d\tau + CC_1\vep t^{-1} \int_{a_0}^t\tau^{-3/2 +3\delta}d\tau
\\
& \leq  CC_1\vep t^{-1} +  Ct^{-1} \sup_{\del_B\Kical_{[2,s_0]} \cup \del\Kcal} \{|(\del_t- \del_r)(r\del^I\hu_{a0})|\}.
\endaligned
$$
Observe that on the boundary $\del_B\Kical_{[2,s_0]}$, $r=3t/5$. We have 
$$
\aligned
|(\del_t- \del_r)( r\del^I\hu_{a0})| \leq 
& r|(\del_r- \del_t) \del^I\hu_{a0}| + |\del^I\hu_{a0}| 
\\
& \leq CC_1\vep rt^{-1/2}s^{-1+\delta} + Cm_S\vep t^{-1} + CC_1\vep (s/t)t^{-1/2}s^\delta
\\
&\leq CC_1\vep rt^{-3/2 +\delta/2} + CC_1\vep t^{-1} + CC_1\vep (s/t)t^{-1/2}s^\delta \leq CC_1\vep.
\endaligned
$$
We also observe that on $\del\Kcal$, $\hu_{a0} = {\hu_s}_{a0}$,
$$
\aligned
|(\del_t- \del_r)(r\del^I\hu_{a0})| & \leq  r|(\del_r- \del_t) \hu_{a0}| + |\hu_{a0}| \leq Cm_S\vep rt^{-1} + Cm_S\vep t^{-1} \leq CC_1\vep.
\endaligned
$$
This leads us to \eqref{ineq Ch-sup 1.5} for $\hu_{0a}$.

\vskip.3cm

\noindent{\sl Step III. Induction on $|J|$.}
The proof of \eqref{ineq Ch-sup 1} is done by induction on $|J|$. The initial case $|J| = 0$ is already guaranteed in view of \eqref{ineq Ch-sup 1.5}. We assume that \eqref{ineq Ch-sup 1} holds for all $0\leq |J'| \leq n< N-4$ and we will prove it with $|J|=n$. First, based on \eqref{prop Ch-sup 1}, the following result is immediate:
\bel{ineq Ch-sup 2}
|\del_{\alpha} \del^IL^J \hu_{a0}| + |\del^IL^I\del_{\alpha} \hu_{0a}| \leq CC_1\vep t^{-1+ C\vep}, \quad |I|+|J| \leq N-4,
\ee
\bel{ineq Ch-sup 2.5}
|\del_{\alpha} \del^I\hu_{a0}| \leq CC_1\vep t^{-1}, \quad |I| \leq N-4.
\ee
These are based on the identity 
$
\del_t = \frac{t-r}{t} \del_t + \frac{x^a}{t+r} \delu_a + \frac{r}{t+r}(\del_t- \del_r), 
$
where $\del_t$ can be expressed by the ``good'' derivatives and $\del_t - \del_r$. Furthermore, we have
$
\del_a = \delu_a - \frac{x^a}{t} \del_t
$
and, then, based on the basic $L^\infty$ estimate of the ``good'' derivatives and \eqref{ineq Ch-sup 1} and \eqref{ineq Ch-sup 1.5}, the derivation of  \eqref{ineq Ch-sup 2} and \eqref{ineq Ch-sup 2.5} is immediate.

Then we substitute the above estimates on the source terms and corrective term into \eqref{ineq LR-sup 1}. Observe that by the inductive assumption, \eqref{ineq 1 proof prop Ch-sup 1} becomes
$$
|[\del^IL^J, \hu^{00} \del_t\del_t]\hu_{a0}| \leq C(C_1\vep)^2t^{-2}s^{-1+3\delta} 
  + C(C_1\vep)^2t^{-2 + C\vep}, 
$$
where we have noticed that $\sum_{|J'|<|J|} \big|\del^IL^J\del_{\alpha} \del_{\beta} \hu_{a0} \big| \leq CC_1\vep s^{-1+ C\vep}$ (by the commutator estimates  and \eqref{ineq Ch-sup 2}). 
This leads us to (in view of \eqref{ineq LR-sup 1})
$$
\aligned
|(\del_t- \del_r) \del^IL^J\hu_{a0}| & \leq  Ct^{-1} \sup_{\del_B\Kical_{[2,s^*]} \cup \del\Kcal} \{|(\del_t- \del_r)(r\del^IL^J\hu_{a0})|\}
+ Ct^{-1}|\del^IL^J\hu_{a0}(t, x)|
\\
& \quad +  C(C_1\vep)^2 t^{-1} \int_{a_0}^t\tau^{-1+ C\vep}d\tau + CC_1\vep t^{-1} \int_{a_0}^t\tau^{-3/2 +2\delta}d\tau
\\
& \leq  CC_1\vep t^{-1+ C\vep} +  Ct^{-1} \sup_{\del_B\Kical_{[2,s_0]} \cup \del\Kcal} \{|(\del_t- \del_r)(r\hu_{a0})|\}.
\endaligned
$$
Then, similarly as in the argument above,  \eqref{ineq Ch-sup 1} is proved for $\hu_{0a}$.

The estimate for $\hu_{ab}$ is similar, where we also observe that the quasi-null terms $\Pu_{ab}$ are eventually null terms, and the correction terms behave the same decay as in the case of $\hu_{a0}$.
\end{proof}


\subsection{Application to quasi-null terms}

Our main application of the refined sup-norm estimate concerns the terms $P_{\alpha \beta}$.

\begin{lemma}
Let $(I,J)$ be a multi-index and $|I|+|J| \leq N$. Then, one has 
\bel{ineq Ch-sup 3}
\aligned
\left\|\del^IL^JP_{\alpha \beta} \right\|_{L^2(\Hcal_s^*)}
& \leq  CC_1\vep s^{-1} \sum_{\alpha', \beta'}E_M^*(s, \del^IL^Jh_{\alpha'\beta'})^{1/2}
 + CC_1\vep s^{-1} \sum_{|I'|<|I|\atop \alpha', \beta'}E_M^*(s, \del^{I'}L^Jh_{\alpha'\beta'})^{1/2}
\\
& \quad +  CC_1\vep s^{-1+ CC_1\vep} \sum_{|I'| \leq|I|,|J'|<|J|\atop \alpha', \beta'}E_M^*(s, \del^{I'}L^{J'}h_{\alpha'\beta'})^{1/2}
\  + C(C_1\vep)^2s^{-3/2 +2\delta}.
\endaligned
\ee
\end{lemma}

\begin{proof} We apply Lemma \ref{lem P 2} combined with the estimates \eqref{ineq Ch-sup 2} and \eqref{ineq Ch-sup 2.5}. We first observe that due to its tensorial structure, the estimate for $P_{\alpha \beta}$ can be relined on the estimates on $\Pu_{\alpha \beta}$. Furthermore, the components $\Pu_{a \beta}$ or $\Pu_{\alpha b}$ are essentially  null terms (see  \eqref{eq P 2.5}), so that 
$
\left\|\del^IL^J\Pu_{a \beta} \right\|_{L^2(\Hcal_s^*)} \leq C(C_1\vep)^2s^{-3/2 +2\delta}.
$
We focus on $\Pu_{00}$. We see that in the list \eqref{eq P 2}, the non-trivial term are linear combinations of $\del_t\hu_{a \alpha} \del_t\hu_{b\beta}$ with smooth and homogeneous coefficients of degree zero. Then we only need to estimate $\left\|\del^IL^J\left(\del_t\hu_{a \alpha} \del_t\hu_{b\beta} \right) \right\|_{L^2(\Hcal_s^*)}$ for $|I|+|J| \leq N$. We have
$$
\aligned
\left\|\del^IL^J\left(\del_t\hu_{a \alpha} \del_t\hu_{b\beta} \right) \right\|_{L^2(\Hcal_s^*)}
& \leq  \sum_{I_1+I_2=I\atop J_1+J_2=J}
\left\|\del^{I_1}L^{J_1} \del_t\hu_{a \alpha} \,   \del^{I_2}L^{J_2} \del_t\hu_{b\beta} \right\|_{L^2(\Hcal_s^*)}. 
\endaligned
$$
Recall that $N\geq 7$ then either $|I_1|+|J_1| \leq N-4$ or $|I_2|+|J_2| \leq N-4$. Without loss of generality, we suppose that $|I_1|+|J_1| \leq N-4$. Then
\begin{itemize}

\item When $J_1=0$, we apply \eqref{ineq Ch-sup 2.5}: 
$$
\aligned
& \left\|\del^{I_1} \del_t\hu_{a \alpha} \,   \del^{I_2}L^J\del_t\hu_{b\beta} \right\|_{L^2(\Hcal_s^*)}
\leq  CC_1\vep \left\|t^{-1} \,   \del^{I_2}L^J\del_t\hu_{b\beta} \right\|_{L^2(\Hcal_s^*)}
\\
&
 \leq  CC_1\vep s^{-1} \left\|(s/t) \del^{I_2}L^J\del_t\hu_{b\beta} \right\|_{L^2(\Hcal_s^*)}
 \leq  CC_1\vep s^{-1} \sum_{|I'| \leq|I|,|J'| \leq|J|\atop \gamma, \gamma' }E_M^*(s, \del^{I'}L^{J'}h_{\gamma \gamma'})^{1/2}.
\endaligned
$$

\item When $|J_1|\geq 1,1\leq |I_1|+|J_1| \leq N-4$, we apply \eqref{ineq Ch-sup 2}: 
$$
\aligned
\left\|\del^{I_1}L^{J_1} \del_t\hu_{a \alpha} \,   \del^{I_2}L^{J_2} \del_t\hu_{b\beta} \right\|_{L^2(\Hcal_s^*)}
& \leq  CC_1\vep s^{-1+ CC_1\vep} \left\|(s/t) \del^{I_2}L^{J_2} \del_t\hu_{b\beta} \right\|_{L^2(\Hcal_s^*)}
\\
& \leq  CC_1\vep s^{-1+ CC_1\vep} \sum_{|I'| \leq|I_2|,|J'| \leq|J_2|\atop \alpha, \beta }E_M^*(s, \del^{I'}L^{J'}h_{\gamma, \gamma'})^{1/2}
\\
& \leq  CC_1\vep s^{-1+ CC_1\vep} \sum_{|I'| \leq|I_2|,|J'|<|J|\atop \alpha, \beta }E_M^*(s, \del^{I'}L^{J'}h_{\gamma, \gamma'})^{1/2}.
  \qedhere \endaligned
$$
\end{itemize} 
\end{proof}


\section{Low-Order Refined Energy Estimate for the Spacetime Metric}
\label{section-8}

\subsection{Preliminary}

In this section, we improve the energy bounds on $E_M^*(s, \del^IL^Jh_{\alpha \beta})$ for $|I|+|J| \leq N-4$. We apply Proposition \ref{prop 1 energy-W}. In this case the $L^2$ norm of $\del^IL^J\left(\del_{\alpha} \phi\del_{\beta} \phi + \phi^2\right)$ is integrable with respect to $s$. We need to focus on the estimate of $F_{\alpha \beta}$ and the commutators $[\del^IL^J,h^{\mu\nu} \del_{\mu} \del_{\nu}]h_{\alpha \beta}$.

\begin{lemma} \label{lem 0 refined-energy-low-h}
Under the bootstrap assumption \eqref{ineq energy assumption 1} and \eqref{ineq energy assumption 2} with $C_1\vep$ sufficiently small, one has for $|I|+|J| \leq N$:
\bel{ineq 1 lem 0 refined-energy-low-h}
\aligned
\left\|\del^IL^J F_{\alpha \beta} \right\|_{L^2(\Hcal_s^*)} & \leq  C(C_1\vep)^2 s^{-3/2 +2\delta}
+ CC_1\vep s^{-1} \sum_{\alpha', \beta'}E_M^*(s, \del^IL^Jh_{\alpha'\beta'})^{1/2}
\\
& \quad
+ CC_1\vep s^{-1} \sum_{|I'|<|I|\atop \alpha', \beta'}E_M^*(s, \del^{I'}L^Jh_{\alpha'\beta'})^{1/2}
\\
& \quad +  CC_1\vep s^{-1+ CC_1\vep} \sum_{|I'| \leq|I|,|J'|<|J|\atop \alpha', \beta'}E_M^*(s, \del^{I'}L^{J'}h_{\alpha'\beta'})^{1/2}. 
\endaligned
\ee
\end{lemma}

\begin{proof} We use here \eqref{ineq Ch-sup 3}. We observe that $F_{\alpha \beta} = Q_{\alpha \beta} + P_{\alpha \beta}$, where $Q_{\alpha \beta}$ are null terms combined with higher-order (cubic) terms. Then trivial substitution of the basic $L^2$ and sup-norm estimates (see the proof of \eqref{ineq 2 bilinear-L2}) shows that
$
\left\|\del^IL^JQ_{\alpha \beta} \right\|_{L^2(\Hcal_s^*)} \leq C(C_1\vep)^2 s^{-3/2 +2\delta}.
$
The estimate for $P_{\alpha \beta}$ is provided by \eqref{ineq Ch-sup 3}.
\end{proof}

\begin{lemma} \label{lem 0.5 refined-energy-low-h}
Under the bootstrap assumption \eqref{ineq energy assumption 1} and \eqref{ineq energy assumption 2}, the following estimates hold for $|I|+|J| \leq N-4$:
\bel{ineq 1 lem 0.5 refined-energy-low-h}
\aligned
 \left\|[\del^IL^J,h^{\mu\nu} \del_\mu \del_\nu] h_{\alpha \beta} \right\|_{L^2(\Hcal_s^*)} 
& \leq  C(C_1\vep)^2s^{-3/2 +2\delta}
 +  CC_1\vep s^{-1} \sum_{a, |J'|<|J|}E_M^*(s, \del^IL_aL^{J'}h_{\alpha \beta})^{1/2}
\\
& \quad +  CC_1\vep s^{-1+ CC_1\vep} \sum_{|I'| \leq|I|\atop |J'|<|J|} \sum_{\alpha', \beta'}E_M^*(s, \del^{I'}L^{J'}h_{\alpha'\beta'})^{1/2}. 
\endaligned
\ee
\end{lemma}

\begin{proof}
This is based on \eqref{ineq L-2 lem 1 commtator II}. We need to estimate the term
$\left\|(s/t)^2\del_t\del_t\del^IL^{J'}h_{\alpha \beta} \right\|_{L^2(\Hcal_s^*)}$ with $|J'|<|J|$. We are going to use \eqref{eq 1 lem 1 second-order}. We see that in view of \eqref{ineq 1 Sc-L2} :
$$
\left\|Sc_1[\del^IL^{J'}h_{\alpha \beta}]\right\|_{L^2(\Hcal_s^*)} \leq Cs^{-1} \sum_aE_M^*(s, \del^IL_aL^{J'}h_{\alpha \beta})^{1/2} + Cs^{-1} \sum_{|I'| \leq|I|}E_M^*(s, \del^{I'}L^{J'}h_{\alpha \beta})^{1/2}.
$$
The term $Sc_2$ is bounded in view of \eqref{ineq 2 Sc-L2} :
$
\left\|Sc_2[\del^IL^{J'}h_{\alpha \beta}]\right\|_{L^2(\Hcal_s^*)} \leq C(C_1\vep)^2s^{-3/2 +2\delta}.
$
The term $F_{\alpha \beta}$ is bounded by Lemma \ref{lem 0 refined-energy-low-h}.

For the term ${\sl QS}_{\phi}$, we will only analyze in detail the term $\del_{\alpha} \phi\del_{\beta} \phi$ and omit the proof on $\phi^2$. We see first that
$
\del^IL^{J'} \left(\del_{\alpha} \phi\del_{\beta} \phi\right) = \sum_{I_1+I_2=I\atop J_1+J_2 = J'} \del^{I_1}L^{J_1} \del_{\alpha} \phi \,   \del^{I_2}L^{J_2} \del_{\beta} \phi.
$
We then observe that, for $N\geq 7$ and $|I|+|J'| \leq N-5$, 
either $|I_1|+|J_1| \leq N-6$ or $|I_2|+|J_2| \leq N-6$. Suppose without loss of generality that $|I_1|+|J_1| \leq N-6$. Then we have 
$
\left\|\del^IL^{J'} \left(\del_{\alpha} \phi\del_{\beta} \phi\right) \right\|_{L^2(\Hcal_s^*)} \leq \left\|\del^{I_1}L^{J_1} \del_{\alpha} \phi \,  \del^{I_2}L^{J_2} \del_{\beta} \phi\right\|_{L^2(\Hcal_s^*)}.
$
\begin{itemize}

\item when $I_1 = J_1=0$, we see that $0\leq N-7$, then we have 
$$
\aligned
\left\|\del^IL^{J'} \left(\del_{\alpha} \phi\del_{\beta} \phi\right) \right\|_{L^2(\Hcal_s^*)}
& \leq  \left\|(t/s) \del_{\alpha} \phi \,  (s/t) \del^{I_2}L^{J_2} \del_{\beta} \phi\right\|_{L^2(\Hcal_s^*)}
\\
& \leq  CC_1\vep \left\|(t/s)t^{-3/2}s^{\delta} \,  (s/t) \del^{I_2}L^{J_2} \del_{\beta} \phi\right\|_{L^2(\Hcal_s^*)}
\\
& \leq  CC_1\vep s^{-3/2 +\delta} \left\|(s/t) \del^{I_2}L^{J_2} \del_{\beta} \phi\right\|_{L^2(\Hcal_s^*)} \leq C(C_1\vep)^2s^{-3/2 +2\delta}.
\endaligned
$$

\item when $1\leq |I_1|+|I_2| \leq N-6$, we see that $|I_2|+|J_2| \leq N-5$. So we have 
$$
\aligned
\left\|\del^IL^{J'} \left(\del_{\alpha} \phi\del_{\beta} \phi\right) \right\|_{L^2(\Hcal_s^*)}
& \leq  \left\|\del^{I_1}L^{J_1} \del_{\alpha} \phi\right\|_{L^\infty(\Hcal_s^*)}
\left\|\del^{I_2}L^{J_2} \del_{\beta} \phi\right\|_{L^2(\Hcal_s^*)}
\\
& \leq  CC_1\vep s^{-3/2} \,   CC_1\vep s^{\delta} \leq C(C_1\vep)^2s^{-3/2 +2\delta}.
\endaligned
$$
\end{itemize}
We conclude that
\bel{ineq pr0 lem 0.5 refined-energy-low-h}
\left\|{\sl QS}_{\phi}(p,k) \right\|_{L^2(\Hcal_s^*)} \leq C(C_1\vep)^2s^{-3/2 +2\delta}, \qquad p\leq N-4.
\ee
The term $[\del^IL^{J'},h^{\mu\nu} \del_\mu\del_\nu]h_{\alpha \beta}$ is conserved. Then we see the following estimate are established:
\bel{ineq pr1 lem 0.5 refined-energy-low-h}
\aligned
& \left\|[\del^IL^J,h^{\mu\nu}]h_{\alpha \beta} \right\|_{L^2(\Hcal_s^*)}
\\
& \leq  CC_1\vep s^{-1} \sum_{\alpha', \beta',a \atop|J'|<|J|}E_M^*(s, \del^IL_aL^{J'}h_{\alpha'\beta'})^{1/2}
 +  CC_1\vep s^{-1+ CC_1\vep} \sum_{{\alpha', \beta'\atop|I'| \leq |I|} \atop|J'|<|J|}E_M^*(s, \del^{I'}L^{J'}h_{\alpha'\beta'})^{1/2} 
\\
& \quad + \sum_{\alpha', \beta'\atop |J'|<|J|} \big\|[\del^IL^{J'},h^{\mu\nu} \del_{\mu} \del_{\nu}]h_{\alpha'\beta'} \big\|_{L_f^2(\Hcal_s)}
+ C(C_1\vep)^2s^{-3/2 +2\delta}.
\endaligned
\ee
We proceed by induction on $|J|$. In \eqref{ineq pr1 lem 0.5 refined-energy-low-h}, if we take $|J|=0$, then only the last term in the right-hand side exists, this concludes \eqref{ineq 1 lem 0.5 refined-energy-low-h}. Assume that \eqref{ineq 1 lem 0.5 refined-energy-low-h} holds for $|J| \leq n-1\leq N-5$, we will prove that it still holds for $|J|=n\leq N-4$. We substitute \eqref{ineq 1 lem 0.5 refined-energy-low-h} into the last term in the right-hand side of \eqref{ineq pr1 lem 0.5 refined-energy-low-h}.
\end{proof}


\subsection{Main estimate established in this section}

\begin{proposition}[Lower order refined energy estimate for $h_{\alpha \beta}$]\label{prop 1 refined-energy-low-h}
There exists a constant $\vep_1>0$ determined by $C_1>2C_0$ such that assume that the bootstrap assumption \eqref{ineq energy assumption 1} holds with $(C_1, \vep)$, $0\leq \vep\leq \vep_1$, then the following refined estimate holds
\bel{ineq 1 prop 1 refined-energy-low-h}
E_M(s, \del^IL^J h_{\alpha \beta})^{1/2} \leq \frac{1}{2}C_1\vep s^{CC_1\vep}, \quad  \qquad  \alpha, \beta \leq 3, \quad |I|+|J| \leq N-4.
\ee
\end{proposition}

\begin{proof}
The proof relies on a direct application of Proposition \ref{prop 1 energy-W}. We need to bound the terms presented in the right-hand side of \eqref{eq energy 2}. The term $F_{\alpha \beta}$ is bounded by Lemma \ref{lem 0 refined-energy-low-h}, the term ${\sl QS}_\phi$ is bounded in view of \eqref{ineq pr0 lem 0.5 refined-energy-low-h}. The estimate for $[\del^IL^J,h^{\mu\nu} \del_\mu\del_\nu]h_{\alpha \beta}$ is obtained in view of \eqref{ineq 1 lem 0.5 refined-energy-low-h}.
By \eqref{ineq lem h00-sup 3}, the term $M_{\alpha \beta}[\del^IL^J h]$ is bounded by $C(C_1\vep)^2 s^{-3/2 +2\delta}$.
Then in view of \eqref{eq energy 2} :
\bel{ineq -1 proof prop1 refined-energy-low-h}
\aligned
\sum_{\alpha, \beta}E_M(s, \del^IL^Jh_{\alpha \beta})^{1/2} & \leq  CC_0 \, \vep + C(C_1\vep)^2
+ CC_1\vep\sum_{\alpha, \beta} \int_2^s\tau^{-1}E_M^*(\tau, \del^IL^Jh_{\alpha \beta})^{1/2}d\tau
\\
&
+ CC_1\vep\sum_{|I'|<|I|\atop \alpha, \beta} \int_2^s\tau^{-1}E_M^*(\tau, \del^{I'}L^Jh_{\alpha \beta})^{1/2}d\tau
\\
& \quad +  CC_1\vep\sum_{|I'| \leq|I|,|J'|<|J|\atop \alpha, \beta} \int_2^s\tau^{-1+ CC_1\vep}E_M^*(\tau, \del^{I'}L^{J'}h_{\alpha \beta})^{1/2}d\tau
\\
& \quad +  CC_1\vep\sum_{\alpha, \beta,a \atop |J'|<|J|} \int_2^s \tau^{-1}E_M^*(\tau, \del^IL_aL^{J'}h_{\alpha \beta})^{1/2}d\tau.
\endaligned
\ee

The rest of the proof is based on \eqref{ineq -1 proof prop1 refined-energy-low-h}. When $|J|=0$, the last two terms in the right-hand side of \eqref{ineq -1 proof prop1 refined-energy-low-h} disappears. Then, we have 
$$
\aligned
\sum_{\alpha, \beta \atop |I| \leq N-4}E_M(s, \del^Ih_{\alpha \beta})^{1/2} & \leq  C\left(C_0 \, \vep + (C_1\vep)^2\right)
+ CC_1\vep\sum_{\alpha, \beta \atop |I| \leq N-4} \int_2^s\tau^{-1} E_M(\tau, \del^Ih_{\alpha \beta})^{1/2}d\tau.
\endaligned
$$
Then by Gronwall's inequality, we have
\bel{ineq 0 proof prop1 refined-energy-low-h}
\sum_{\alpha, \beta \atop |I| \leq N-4}E_M(s, \del^Ih_{\alpha \beta})^{1/2} \leq C\big(C_0 \, \vep + (C_1\vep)^2\big)s^{CC_1\vep}.
\ee
Here we can already ensure that 
$\sum_{\alpha, \beta}E_M(s, \del^Ih_{\alpha \beta})^{1/2} \leq \frac{1}{2}C_1\vep s^{CC_1 \vep}$
by choosing ${\vep_1}_0 = \frac{C_1-2CC_0}{2C_1^2}$ with $C_1$ sufficiently large. 

We proceed by induction on $|J|$ and suppose that
\bel{ineq 1 proof prop1 refined-energy-low-h}
\sum_{\alpha, \beta \atop |I| \leq N-4}E_M(s, \del^Ih_{\alpha \beta})^{1/2} \leq C\big(C_0 \, \vep + (C_1\vep)^2\big)s^{CC_1\vep}
\ee
holds for $|J|< n\leq N-4$, we will prove that it still holds for $|J|=n$. Substitute \eqref{ineq 1 proof prop1 refined-energy-low-h} into the last two terms of the right-hand side of \eqref{ineq -1 proof prop1 refined-energy-low-h}, we see that
$$
\aligned
&\sum_{\alpha, \beta}E_M(s, \del^IL^Jh_{\alpha \beta})^{1/2}
\leq  CC_0 \, \vep + C(C_1\vep)^2
+ CC_1\vep\sum_{\alpha, \beta} \int_2^s\tau^{-1} E_M(\tau, \del^IL^Jh_{\alpha \beta})^{1/2}d\tau
\\
& \quad + CC_1\vep\sum_{\alpha, \beta \atop |I'|<|I|} \int_2^s\tau^{-1} E_M(\tau, \del^{I'}L^Jh_{\alpha \beta})^{1/2}d\tau
 + CC_1\vep\left(C_0 \, \vep + (C_1\vep)^2\right) \int_2^s\tau^{-1+ CC_1\vep}d\tau
\\
& \quad+ CC_1\vep \sum_{a, \alpha, \beta \atop |J'|=|J|-1} \int_2^s\tau^{-1}E_M^*(\tau, \del^IL_aL^{J'}h_{\alpha \beta})^{1/2}d\tau, 
\endaligned
$$ thus 
$$
\aligned
&\sum_{\alpha, \beta}E_M(s, \del^IL^Jh_{\alpha \beta})^{1/2}
\leq C\left(C_0+ (C_1\vep)^2\right)s^{CC_1\vep}
+ CC_1\vep\sum_{\alpha, \beta} \int_2^s\tau^{-1} E_M(\tau, \del^IL^Jh_{\alpha \beta})^{1/2}d\tau
\\
& \quad+ CC_1\vep\sum_{\alpha, \beta \atop |I'|<|I|} \int_2^s\tau^{-1} E_M(\tau, \del^{I'}L^Jh_{\alpha \beta})^{1/2}d\tau
+ CC_1\vep \sum_{\alpha, \beta \atop |J'|=|J|} \int_2^s\tau^{-1}E_M^*(\tau, \del^IL^{J'}h_{\alpha \beta})^{1/2}d\tau
\endaligned
$$
This leads us to
$$
\aligned
& \sum_{\alpha, \beta,|J|=n\atop |I| \leq N-4-n}E_M(s, \del^IL^Jh_{\alpha \beta})^{1/2}
\\
& \leq C\left(C_0 \, \vep + (C_1\vep)^2\right)s^{CC_1\vep}
+ CC_1\vep\sum_{\alpha, \beta,|J|=n\atop |I| \leq N-4-n} \int_2^s \tau^{-1} E_M(\tau, \del^IL^Jh_{\alpha \beta})^{1/2}d\tau.
\endaligned
$$
Then by Gronwall's inequality, we have (by taking some constant $C$ larger than the one provided the above estimate)
$$
\sum_{\alpha, \beta \atop |I| \leq N-4-|J|}E_M(s, \del^IL^Jh_{\alpha \beta})^{1/2} \leq C\left(C_0 \, \vep + (C_1\vep)^2\right)s^{CC_1\vep}.
$$
By choosing ${\vep_1}_n = \frac{C_1-2CC_0}{2C_1^2}$, we see that
$
\sum_{\alpha, \beta \atop |I| \leq N-4-|J|}E_M(s, \del^IL^Jh_{\alpha \beta})^{1/2} \leq \frac{1}{2}C_1\vep s^{CC_1\vep}.$
Then, we choose $\vep_1 = \min_{0\leq n\leq N-4} \{{\vep_1}_n\}$ and conclude that for $\vep\leq \vep_1$, \eqref{ineq 1 prop 1 refined-energy-low-h} is thus proven. 
\end{proof}


\subsection{Application of the refined energy estimate}
\label{subsec energy-low-app}

The improved low-order energy estimates on $h_{\alpha \beta}$ will lead us to a series of estimates. Based on \eqref{prop 1 refined-energy-low-h}, the  sup-norm estimates are direct by the global Sobolev inequality (for $|I|+|J| \leq N-6$): 
\bel{ineq 4 reined-sup low}
|\del^IL^J\del_{\gamma}h_{\alpha \beta}| + |\del_{\gamma} \del^IL^Jh_{\alpha \beta}| \leq CC_1\vep t^{-1/2}s^{-1 + CC_1\vep},
\ee
\bel{ineq 5 reined-sup low}
|\del^IL^J\delu_ah_{\alpha \beta}| + |\delu_a \del^IL^Jh_{\alpha \beta}| \leq CC_1\vep t^{-3/2}s^{CC_1\vep}.
\ee
Based on this improved sup-norm estimate, the following estimates are direct by integration along the radial rays $\{(t, \lambda x)|1\leq \lambda \leq t/|x|\}$:
\bel{ineq 6 reined-sup low}
|\del^IL^Jh_{\alpha \beta}| \leq CC_1\vep \left(t^{-1}+(s/t)t^{1/2}s^{CC_1\vep} \right).
\ee
We take the above bounds and substitute them into the proof of Lemma \ref{lem wave-condition 4}, following exactly the same procedure, we obtain for $|I|+|J| \leq N-6$:
\bel{ineq 7 reined-sup low}
\left|\del^IL^J\del_{\alpha} \hu^{00} \right| + \left|\del^IL^J\del_{\alpha} \hu^{00} \right| \leq CC_1\vep t^{-3/2}s^{CC_1\vep}
\ee
and also by integration along the rays $\{(t, \lambda x)|1\leq \lambda \leq t/|x|\}$ (and taking into account the exterior Schwarzschild metric):
\bel{ineq 8 reined-sup low}
\left|\del^IL^J\hu^{00} \right| \leq CC_1\vep\left(t^{-1} + (s/t)^2t^{1/2}s^{CC_1\vep} \right).
\ee

Two more delicate applications of this improved energy estimate for $h_{\alpha \beta}$ are now obtained. We begin with $F_{\alpha \beta}$, in view of  \eqref{ineq 4 reined-sup low}.

\begin{lemma} \label{lem 0 refined-sup}
For $|I|+|J| \leq N-6$, one has 
\bel{ineq 1 lem 0 refined-sup}
|\del^IL^JF_{\alpha \beta}| \leq C(C_1\vep)^2t^{-2 + CC_1\vep}(t-r)^{-1+ CC_1\vep}.
\ee
\end{lemma}

\begin{proof} Observe that $F_{\alpha \beta}$ is a linear combination of ${\sl GQS}_{h}$ and $P_{\alpha \beta}$ and in $P_{\alpha \beta}$ the only term to be concerned about (by Lemma \ref{lem P 2}) is $\minu^{0a} \minu^{0b} \del_t\hu_{0a} \del_t\hu_{0b}$, the remaining terms are ${\sl GQS}_{h}$, ${\sl Cub}$ or ${\sl Com}$ which have better decay. We observe that in view of \eqref{ineq 4 reined-sup low},
$$
\left|\del^IL^J\left(\del_t\hu_{a \alpha} \del_t\hu_{b\beta} \right) \right| \leq C(C_1\vep)^2t^{-1}s^{-2 + CC_1\vep}. \qedhere
$$ 
\end{proof}

Then, a second refined estimate can be established.  

\begin{lemma} \label{lem 1 refined-energy-low-h}
For $|I|+|J| \leq N-7$, one has 
\bel{ineq 1 lem 1 refined-energy-low-h}
\left|\del_t\del_t\del^IL^J h_{\alpha \beta} \right| \leq CC_1\vep t^{1/2}s^{-3+ CC_1\vep}.
\ee
\end{lemma}

\begin{proof} The proof is essentially a refinement of the proof of \eqref{ineq 1 second-sup}.
We see that when the energy is improved, in view of \eqref{ineq 4 reined-sup low}, $|Sc_1[\del^IL^Jh_{\alpha \beta}]|$ is bounded by $CC_1\vep t^{-3/2}s^{-1+ CC_1\vep}$
( in view of \eqref{ineq 1 lem 2 second-order}).
The term $F_{\alpha \beta}$ is bounded by the above estimate \eqref{ineq 1 lem 0 refined-sup}.
The terms $Sc_2$, ${\sl QS}_\phi$ and the commutator are bounded as in the proof of \eqref{ineq 1 second-sup}. Then we get the following estimate parallel to \eqref{ineq 1 pr ineq 1 second-sup} :
$$
\aligned
|(s/t)^2\del_t\del_t \del^IL^Jh_{\alpha \beta}| & \leq  CC_1\vep t^{-3/2}s^{-1+ CC_1\vep} + C(C_1\vep)^2 t^{-1}s^{-2 + CC_1\vep}
\\
& \quad + CC_1\vep\left(t^{-1}
+ (s/t)^2t^{-1/2}s^{\delta} \right) \sum_{|J'|<|J|} \left|\del_t\del_t\del^IL^{J'} h_{\alpha \beta} \right|.
\endaligned
$$
By induction, the desired result is thus established.
\end{proof}


\section{Low-Order Refined Sup-Norm Estimate for the Metric and Scalar Field}
\label{section-9}

\subsection{Main estimates established in this section}

Our  aim in this section is to establish the estimates:
$|I|+|J| \leq N-7$:
\bel{ineq 1 refined-sup low}
|L^Jh_{\alpha \beta}| \leq CC_1\vep t^{-1}s^{C(C_1\vep)^{1/2}},
\ee
\bel{ineq 2 refined-sup low}
(s/t)^{3\delta-2}|\del^IL^J\phi| + (s/t)^{3\delta-3}|\del^IL^J\delu_{\perp} \phi| \leq CC_1\vep s^{-3/2 + C(C_1\vep)^{1/2}},
\ee
\bel{ineq 3 refined-sup low}
(s/t)^{3\delta-2}|\del^I\phi| + (s/t)^{3\delta-3}|\delu_{\perp} \del^I\phi| \leq CC_1\vep s^{-3/2}.
\ee
Let us first point out some direct consequences of these three estimates, by noting the relations
$\del_t = (s/t)^{-2} \left(\delu_{\perp} - \frac{x^a}{t} \delu_a \right)$ and $\del_a = \delu_a - \frac{x^a}{t} \del_t$
and the sharp decay rate on $\delu_a$ (for $|I|+|J| \leq N-7$)
$$
|\delu_a \del^IL^J\phi(t, x)| \leq CC_1\vep t^{-5/2}s^{1/2 +\delta}.
$$ 
So, \eqref{ineq 1 refined-sup low}, \eqref{ineq 2 refined-sup low} and \eqref{ineq 3 refined-sup low} lead to
\bel{ineq 3.1 refined-sup low a}
\left|\del_\alpha \del^IL^J\phi(t, x) \right| \leq CC_1\vep (s/t)^{1-3\delta}s^{-3/2 + C(C_1\vep)^{1/2}}, \quad |I|+|J| \leq N-7, 
\ee 
\bel{ineq 3.1 refined-sup low b}
\left|\del_\alpha \del^IL^J\phi(t, x) \right| \leq CC_1\vep (s/t)^{2-3\delta}s^{-3/2 + C(C_1\vep)^{1/2}}, \quad |I|+|J| \leq N-8.
\ee 
We also have
\bel{ineq 3.2 refined-sup low a}
\left|\del_\alpha \del^I\phi(t, x) \right| \leq CC_1\vep (s/t)^{1-3\delta}s^{-3/2}, \quad |I| \leq N-7, 
\ee 
\bel{ineq 3.2 refined-sup low b}
\left|\del_\alpha \del^I\phi(t, x) \right| \leq CC_1\vep (s/t)^{2-3\delta}s^{-3/2}, \quad |I| \leq N-8.
\ee
In particular, we see that  
\bel{ineq 3.2 refined-sup low c}
\left|\del_\alpha \phi(t, x) \right| \leq CC_1\vep (s/t)^{2-3\delta}s^{-3/2}.
\ee
We observe that by the commutator estimates:
\bel{ineq 3.2 refined-sup low d}
\aligned
\big|\del^IL^J\del_{\alpha} \phi\big| & \leq  CC_1\vep (s/t)^{1-3\delta}s^{-3/2 + C(C_1\vep)^{1/2}}, \quad&&& |I|+|J| \leq N-7,
\\
\big|\del^IL^J\del_{\alpha} \phi\big| & \leq  CC_1\vep (s/t)^{2-3\delta}s^{-3/2}, \quad&&& |I|+|J| \leq N-8,
\\
\big|\del^IL^J\del_{\alpha} \del_{\beta} \phi\big| & \leq  CC_1\vep(s/t)^{1-3\delta}s^{-3/2 + C(C_1\vep)^{1/2}}
\quad &&& |I|+|J| \leq N-8.
\endaligned
\ee


\subsection{First refinement on the metric components}

We begin the proof of the refined sup-norm estimate by the following bound on $L^J\left(h^{\mu\nu} \del_{\mu} \del_{\nu}h_{\alpha \beta} \right)$.
\begin{lemma} \label{lem 0.5 refined-sup}
For all $|J| \leq N-7$, the following estimate holds:
\bel{ineq 1 lem 0.5 refined-sup}
\left|L^J\left(h^{\mu\nu} \del_{\mu} \del_{\nu}h_{\alpha \beta} \right) \right| \leq C(C_1\vep)^2t^{-2 + CC_1\vep}(t-r)^{-1+ CC_1\vep}. 
\ee
\end{lemma}

\begin{proof}
We have the following identity 
$$
\aligned
h^{\mu\nu} \del_\mu\del_\nu h_{\alpha \beta}
& =  \hu^{00} \del_t\del_t h_{\alpha \beta}
 + \hu^{a0} \delu_a \del_th_{\alpha \beta} + \hu^{0b} \del_t\delu_bh_{\alpha \beta} + \hu^{ab} \delu_a \delu_b h_{\alpha \beta}
   +  h^{\mu\nu} \del_\mu\left(\Psi^{\nu'}_{\nu} \right) \delu_{\nu'}h_{\alpha \beta}.
\endaligned
$$
We obtain
$$
\aligned
\left|L^J\left(h^{\mu\nu} \del_{\mu} \del_{\nu}h_{\alpha \beta} \right) \right|
& \leq  \left|L^J\left(\hu^{00} \del_t\del_th_{\alpha \beta} \right) \right| +  \left|L^J\left(\hu^{a0} \delu_a \del_th_{\alpha \beta} \right) \right| 
\\
& \quad + \left|L^J\left(\hu^{0b} \del_t\delu_bh_{\alpha \beta} \right) \right|
 + \left|L^J\left(\hu^{ab} \delu_a \delu_b h_{\alpha \beta} \right) \right|
 + \left|L^J\left(h^{\mu\nu} \del_\mu\left(\Psi^{\nu'}_{\nu} \right) \delu_{\nu'}h_{\alpha \beta} \right) \right|
\endaligned
$$
The second, third, and fourth terms are null terms, they contain at least one ``good'' derivative and can be bounded directly by applying the basic sup-norm estimates. We only treat $\hu^{a0} \delu_a \del_th_{\alpha \beta}$, since the third and fourth terms are bounded similarly: 
$$
\left|L^J\left(\hu^{a0} \delu_a \del_th_{\alpha \beta} \right) \right| \leq \sum_{J_1+J_2=J} \left|L^{J_1} \hu^{a0}L^{J_2} \delu_a \del_th_{\alpha \beta} \right|.
$$
We observe that
$$
\left|L^{J_2} \delu_a \del_th_{\alpha \beta} \right| = \left|L^{J_2} \left(t^{-1}L_a \del_th_{\alpha \beta} \right) \right|
\leq \sum_{J_3+J_4=J_2} \left|L^{J_3} \left(t^{-1} \right) L^{J_4}L_a \del_th_{\alpha \beta} \right|.
$$ 
Observe that $L^{J_3} \left(t^{-1} \right)$ is again smooth, homogenous of degree $-1$, which can be bounded by $Ct^{-1}$ in $\Kcal$. So the above sum is bounded by
$$
\sum_{|J'| \leq|J|+1}Ct^{-1} \left|L^{J'} \del_th_{\alpha \beta} \right| \leq CC_1\vep t^{-3/2}s^{-1+ CC_1\vep},
$$
where we have applied \eqref{ineq 4 reined-sup low}. On the other hand, in view of \eqref{ineq 6 reined-sup low}, we have 
$$
\left|L^{J_1} \hu^{a0} \right| \leq CC_1\vep\left(t^{-1} + (s/t)t^{-1/2}s^{CC_1\vep} \right),
$$
since $\hu^{a0}$ is a linear combination of $h_{\alpha \beta}$ with smooth and homogeneous coefficients of degree zero plus high order correction terms. We conclude that
$$
\left|L^J\left(\hu^{a0} \delu_a \del_th_{\alpha \beta} \right) \right| \leq C(C_1\vep)^2t^{-3}s^{CC_1\vep}.
$$
Furthermore, the term $\left|L^J\left(h^{\mu\nu} \del_\mu\left(\Psi^{\nu'}_{\nu} \right) \delu_{\nu'}h_{\alpha \beta} \right) \right|$ is bounded by making use of the additional decay provided by $\left|L^{J'} \del_\mu\left(\Psi^{\nu'}_{\nu} \right) \right| \leq C(J')t^{-1}$, and we omit the details and just state that 
$$
\left|L^J\left(h^{\mu\nu} \del_\mu\left(\Psi^{\nu'}_{\nu} \right) \delu_{\nu'}h_{\alpha \beta} \right) \right| \leq C(C_1\vep)^2 t^{-3}s^{CC_1\vep}.
$$

Now we focus on the most problematic term $L^J\left(\hu^{00} \del_t\del_th_{\alpha \beta} \right)$. We apply here the sharp decay of $\hu^{00}$ provided by \eqref{ineq 8 reined-sup low} and the refined second-order estimate \eqref{ineq 1 lem 1 refined-energy-low-h} :
$$
\aligned
&\left|L^J\left(\hu^{00} \del_t\del_th_{\alpha \beta} \right) \right|
\leq  \sum_{J_1+J_2 = J} \left|L^{J_1} \hu^{00}L^{J_2} \del_t\del_t h_{\alpha \beta} \right|
\\
& \leq  CC_1\vep \left(t^{-1} + (s/t)^2t^{-1/2}s^{CC_1\vep} \right) \,   CC_1\vep t^{1/2}s^{-3+ CC_1\vep}
\\
& \leq  C(C_1\vep)^2 t^{-1/2}s^{-3+ CC_1\vep} + C(C_1\vep)^2 t^{-2}s^{-1+ CC_1\vep}
\\
& \leq  C(C_1\vep)^2 t^{-2 + CC_1\vep}(t-r)^{-1+ CC_1\vep}.
\endaligned
$$ 
\end{proof}

\begin{lemma}[First refinement on $h_{\alpha \beta}$]\label{lem 1 refined-sup}
Assuming that the bootstrap assumption \eqref{ineq energy assumption 1} holds with $C_1\vep$ sufficiently small, one has
\bel{ineq 1 lem 1 refined-sup}
|h_{\alpha \beta}| \leq CC_1\vep  t^{-1}s^{2\delta}.
\ee
\end{lemma}

\begin{proof} We apply Proposition \ref{prop 1 sup-norm-W} and follow the notation therein. The wave equation satisfied by $h_{\alpha \beta}$
$$
\Boxt_g h_{\alpha \beta} = F_{\alpha \beta} - 16 \pi \phi\del_{\alpha} \phi\del_{\beta} \phi - 8\pi c^2\phi^2 _{\alpha \beta}
$$
leads us to
$$
\Box h_{\alpha \beta} = -h^{\mu\nu} \del_\mu\del_\nu h_{\alpha \beta} + F_{\alpha \beta} - 16 \pi \phi\del_{\alpha} \phi\del_{\beta} \phi - 8\pi c^2\phi^2 g_{\alpha \beta}.
$$
We can apply \eqref{ineq 1 lem 0.5 refined-sup} and \eqref{ineq 1 lem 0 refined-sup}, and we have
\bel{ineq pr2 lem 1 refined-sup}
|S_{I, \alpha \beta}^W| \leq C(C_1\vep)^2t^{-2 + CC_1\vep}(t-r)^{-1+ CC_1\vep}.
\ee
Second, by the basic sup-norm estimates, we have 
$$
|S_{\alpha \beta}^{KG,I,J}| \leq C(C_1\vep)^2t^{-2-1/2 +\delta}(t-r)^{-1/2 +\delta}, \quad |I|+|J| \leq N-6.
$$ 
We can choose $\vep_2>0$ sufficiently small so that $\vep\leq \vep_2$ and $CC_1\vep\leq \delta$, hence 
$$
|S_{I, \alpha \beta}^W[t, x, \del^IL^J]| \leq C(C_1\vep)^2 t^{-2 +\delta}(t-r)^{-1+\delta}
$$ 
and, by Proposition \ref{prop 1 sup-norm-W}, 
$$
\aligned
|h_{\alpha \beta}(t, x)| & \leq  C(C_1\vep)^2(t-r)^{2\delta}t^{-1} + CC_1\vep t^{-1} 
\leq  CC_1\vep (t-r)^{\delta}t^{-1+\delta}. \qedhere
\endaligned
$$
\end{proof}


\subsection{First refinement for the scalar field}

In this section, we apply Proposition \ref{Linfini KG} and consider first the correction terms.

\begin{lemma} \label{lem 2 refined-sup}
Assume the bootstrap assumption \eqref{ineq energy assumption 1}, \eqref{ineq energy assumption 2} and take the notation of Section \ref{subsec KG-sup} and Proposition \ref{Linfini KG}, then for $|I|+|J| \leq N-4$
\begin{subequations} \label{ineq 1 lem 2 refined-sup}
\be
|R_1[\del^IL^J\phi]| \leq CC_1\vep (s/t)^{3/2}s^{-3/2 +\delta},
\ee
\be
|R_2[\del^IL^J\phi]| \leq C(C_1\vep)^2(s/t)^{3/2}s^{-3/2 +3\delta},
\ee
\be
|R_3[\del^IL^J\phi]| \leq C(C_1\vep)^2(s/t)^{3/2}s^{-3/2 +3\delta}.
\ee
\end{subequations}
\end{lemma}

\begin{proof} We apply the basic sup-norm estimate to the corresponding expressions of $R_i$. For $R_1[\del^IL^J\phi]$, we apply \eqref{ineq 4 homo}. 
For the term $R_2[\del^IL^J\phi]$, we observe that
$\left|\hb^{00} \right| = \left|(t/s)^2\hu^{00} \right|$
and we recall that the linear part of $\hu^{00}$ is a linear combination of $h_{\alpha \beta}$ with smooth and homogeneous coefficients of degree zero. We see that, in view of \eqref{ineq 1 lem 1 refined-sup} (after neglecting the higher-order terms which vanish as $|h_{\alpha \beta}|^2$ at zero), 
$$
\big|\hb^{00} \big| \leq CC_1\vep (s/t)^{-1}s^{-1+2\delta}.
$$
Similarly, we have
$$
\big|\hb^{0b} \big| \leq \big|(t/s) \hu^{0b} \big|, 
$$
so that 
$$
\big|\hb^{0b} \big| \leq CC_1\vep s^{-1+2\delta}
$$
and, for $\hb^{ab} = \hu^{ab}$, we have 
$\big|\hb^{ab} \big| \leq CC_1\vep (s/t)^2s^{-1+2\delta}$. 
We also note that $\delb_0 \phi = (s/t) \del_t\phi$. Then, substituting the above bounds leads us to
$$
\left|R_2[\del^IL^J\phi]\right| \leq CC_1\vep (s/t)^{3/2}s^{-3/2 +3\delta}.
$$
A similar derivation allows us to control $\left|R_3[\del^IL^J\phi]\right| \leq CC_1\vep (s/t)^{3/2}s^{-3/2 +3\delta}$. 
\end{proof}

\begin{proposition}[Estimate on $\phi$ and $\del\phi$]\label{prop 1 refined-sup}
Assume the bootstrap assumption \eqref{ineq energy assumption 1} and \eqref{ineq energy assumption 2} hold with $C_1>C_0$ and $C_1\vep$ sufficiently small, then  
\bel{ineq 1 prop 1 refined-sup}
(s/t)^{3\delta - 2}|\phi(t, x)| + (s/t)^{3\delta-3}|\delu_{\perp} \phi(t, x)| \leq CC_1\vep s^{-3/2}.
\ee
\end{proposition}

\begin{proof}
We apply Proposition \ref{Linfty KG ineq} and follow the notation there. Recall that Lemma \ref{lem 2 refined-sup} and Lemma \ref{lem h00-sup 4}, we have
$$
|F(\tau)| \leq  \int_{s_0}^{\tau} \big|\sum_{i}R_i[\phi](\lambda t/s, \lambda x/s) \big|d\lambda
\leq  CC_1\vep (s/t)^{3/2} \int_{s_0}^\tau\lambda^{-3/2 +3\delta}d\lambda \leq CC_1\vep (s/t)^{3/2}s_0^{-1/2 +3\delta},
$$
$$
|h'_{t, x}(\lambda)| \leq CC_1\vep (s/t)^{1/2} \lambda^{-3/2 +\delta} + CC_1\vep (t/s) \lambda^{-2}.
$$
We observe that, in the inequality \eqref{Linfty KG ineq a} we need  
$$
\aligned
\int_{\tau}^s|h'_{t, x}(\lambda)d\lambda| & \leq  CC_1\vep(s/t)^{1/2} \int_{s_0}^s \lambda^{-3/2 +\delta} d\lambda
+ CC_1\vep(s/t)^{-1} \int_{s_0}^s\lambda^{-2}d\lambda
\\
& \leq  CC_1\vep (s/t)^{1/2}s_0^{-1/2 +\delta} + CC_1\vep (s/t)^{-1}s_0^{-1}.
\endaligned
$$ 
By \eqref{Linfty KG ineq a}, we have 
$$
|s^{3/2} \phi(t, x)| + \left|(s/t)^{-1}s^{3/2} \delu_{\perp} \phi(t, x) \right| \leq V(t, x)
$$
with
$$
V(t, x) \leq
\left\{
\aligned
&\left(\|v_0\|_{L^\infty} + \|v_1\|_{L^\infty} \right)
\left(1+\int_2^s|h'_{t, x}(\sbar)|e^{C\int_{\sbar}^s|h'_{t, x}(\lambda)|d\lambda} \right)
\\
&\quad\quad\quad+ F(s) + \int_2^sF(\sbar)|h'_{t, x}(\sbar)|e^{C\int_{\sbar}^s|h'_{t, x}(\lambda)|d\lambda}d\sbar, &&0\leq r/t\leq 3/5,
\\
&F(s) + \int_{s_0}^sF(\sbar)|h'_{t, x}(\sbar)|e^{C\int_{\sbar}^s|h'_{t, x}(\lambda)d\lambda|}d\sbar, &&3/5<r/t<1.
\endaligned
\right.
$$

When $0\leq r/t\leq 3/5$, we get $4/5\leq s/t\leq 1$ and $s_0=2$. This leads us to 
$$
V(t, x) \leq CC_1\vep + CC_1\vep\leq CC_1\vep,
$$
where we recall that $C_0\leq C_1$.
When $3/5\leq r/t< 1$, the estimate is more delicate. In this case, we have $s_0 = \sqrt{\frac{t+r}{t-r}} \simeq (s/t)^{-1}$. This leads us to the following bounds:
$$
|F(\tau)| \leq CC_1\vep (s/t)^{2-3\delta}, \quad \int_{\tau}^s|h'_{t, x}(\lambda)d\lambda| \leq CC_1\vep.
$$
Substituting these bounds into \eqref{Linfty KG ineq a}, we obtain
$$
|s^{3/2} \phi(t, x)|+|(s/t)^{-1}s^{3/2} \delu_{\perp} \phi(t, x)| \leq CC_1\vep (s/t)^{2-3\delta}.
$$
\end{proof}


\subsection{Second refinement for the scalar field and the metric}

In this section, we establish the following result.

\begin{lemma}[Second sup-norm refinement]\label{lem 3 refined-sup}
Assume that the bootstrap assumption \eqref{ineq energy assumption 1} and \eqref{ineq energy assumption 2} hold with $C_1>C_0$ and $C_1\vep$ sufficiently small, then for all $0\leq |I| \leq N-7$,
\bel{ineq 1 lem 3 refined-sup}
(s/t)^{3\delta-2}|\del^I\phi| + (s/t)^{3\delta-3}|\delu_{\perp} \del^I\phi| \leq CC_1\vep s^{-3/2},
\ee
\bel{ineq 2 lem 3 refined-sup}
\left|h_{\alpha \beta} \right| \leq CC_1\vep t^{-1}s^{C(C_1\vep)^{1/2}}.
\ee
\end{lemma}

We need to control the commutators first. 

\begin{lemma} \label{lem 3.25 refined-sup}
For $|I|+|J| \leq N-7$,
\bel{ineq 1 lem 3.25 refined-sup}
\aligned
\left|[\del^IL^J,h^{\mu\nu} \del_{\mu} \del_{\nu}]\phi\right| & \leq  C(C_1\vep)^2 (s/t)^2s^{-3+3\delta}
 \\
 & \quad +  \sum_{|J_1'|+|J_2'| \leq J\atop |J_2'|<|J|} \left|L^{J_1'} \hu^{00} \del_t\del_t\del^IL^{J_2'} \phi\right|
 + \sum_{ |J'|<|J|} \left|\hu^{00} \del_t\del_t\del^I L^{J'} \phi\right|. 
\endaligned
\ee
\end{lemma}

\begin{proof} We need to estimate all the terms listed in \eqref{eq 2 lem 1 nolinear}.
As far as the terms ${\sl GQQ}_{h\phi}$ are concerned, we will only treat in detail the term
$
\del^{I_1}L^{J_1}h_{\alpha'\beta'} \del^{I_2}L^{J_2} \delu_a \del_{\mu} \phi.
$
For $|I|+|J| \leq N-7$, we have 
$$
\aligned
& \left|\del^{I_1}L^{J_1}h_{\alpha'\beta'} \del^{I_2}L^{J_2} \delu_a \del_{\mu} \phi\right|
\leq  \left|\del^{I_1}L^{J_1}h_{\alpha'\beta'} \right| \,  \left|\del^{I_2}L^{J_2} \delu_a \del_{\mu} \phi\right|
\\
& \leq CC_1\vep \left((s/t)t^{-1/2}s^{\delta}+t^{-1} \right) \left|\del^{I_2}L^{J_2} \left(t^{-1}L_a \del_{\mu} \phi\right) \right|
\\
& \leq  CC_1\vep t^{-1} \left((s/t)t^{-1/2}s^{\delta}+t^{-1} \right)
\sum_{|I_2'| \leq |I_2|\atop |J_2'| \leq|J_2|} \left|\del^{I_2'}L^{J_2'}L_a \del_{\mu} \phi\right|
\\
& \leq  C(C_1\vep)^2t^{-3}s^{2\delta} = C(C_1\vep)^2(s/t)^3s^{-3+2\delta}.
\endaligned
$$
Other terms of ${\sl GQQ}_{h\phi}$ are bounded similarly, and we omit the details.

For the term $t^{-1} \del^{I_3}L^{J_3}h_{\alpha'\beta'} \del^{I_4}L^{J_4} \del_{\gamma} \phi$, due to its additional $t^{-1}$ decay, the basic sup-norm estimates are sufficient to get the following bound:
$$
\left|t^{-1} \del^{I_3}L^{J_3}h_{\alpha'\beta'} \del^{I_4}L^{J_4} \del_{\gamma} \phi\right| \leq C(C_1\vep)^2t^{-2}s^{-2 +\delta}= C(C_1\vep)^2(s/t)^2s^{-4+2\delta} \leq C(C_1\vep)^2 (s/t)^3s^{-3+2\delta}.
$$
For the term $\del^{I_1}L^{J_1} \hu^{00} \del^{I_2}L^{J_2} \del_t\del_t\phi$, we observe that $|I_1|\geq 1$, so it can be bounded in view of \eqref{ineq h00-sup 1} :
$$
\left|\del^{I_1}L^{J_1} \hu^{00} \del^{I_2}L^{J_2} \del_t\del_t\phi\right| \leq C(C_1\vep)^2 t^{-3/2}s^{\delta} \,   t^{-1/2}s^{-1+\delta} \leq C(C_1\vep)^2 (s/t)^2s^{-3+2\delta}.
$$

For the remaining terms in \eqref{eq 2 lem 1 nolinear} we observe that the term $\del^{I}L^{J_2'} \del_t\del_t\phi$ and $\del_\gamma \del_{\gamma'} \del^IL^{J'} \phi$ are bounded by $\del_t\del_t\del^{I'}L^{J'} \phi$ plus some corrections:
$
\left|\del^{I}L^{J_2'} \del_t\del_t\phi\right|
\leq C\sum_{\gamma, \gamma'\atop |J_2''| \leq|J_2'|} \left|\del_\gamma \del_{\gamma'} \del^IL^{J_2''} \phi\right|. 
$
Then in view of \eqref{eq second-order-frame-1} and the argument presented below it (but now $\phi$ plays the role of $h_{\alpha \beta}$ in \eqref{eq second-order-frame-1}), we have
$$
\left|\del^{I}L^{J_2'} \del_t\del_t\phi\right| \leq CC_1\vep t^{-5/2}s^{\delta}
+  C\sum_{|J_2''| \leq|J_2'|} \left|\del_t\del_t\del^IL^{J_2''} \phi\right|. 
$$
So the last two terms in \eqref{eq 2 lem 1 nolinear} is bounded by
$$
C(C_1\vep)^2t^{-3}s^{2\delta} + C\sum_{|J_1'|+|J_2'| \leq |I|\atop|J'|<|J| }|L^{J_1'} \hu^{00} \del_t\del_t\del^IL^{J_2'} \phi|
+ C\sum_{|J'|<|J| }|\hu^{00} \del_t\del_t\del^IL^{J'} \phi|.
$$
This yields us the conclusion.  
On the other hand, when $|J|=0$, the last two terms do not exist.
\end{proof}

\begin{proof}[Proof of Lemma \ref{lem 3 refined-sup}]
The proof of \eqref{ineq 1 lem 3 refined-sup} is similar to that of Proposition \ref{prop 1 refined-sup}. The only difference is that we need to bound the commutator $[\del^I,h^{\mu\nu} \del_\mu\del_\nu]\phi$ (which, with the notation in Proposition \ref{Linfini KG}, plays the role of $f$ in the definition of $F$). We apply \eqref{ineq 1 lem 3.25 refined-sup} with $|J|=0$ and, in this case, 
$
\left|[\del^I,h^{\mu\nu} \del_{\mu} \del_{\nu}]\phi\right| \leq C(C_1\vep)^2 (s/t)^2s^{-3+3\delta}.
$

Then (following the notation in Proposition \ref{Linfini KG}) in view of \eqref{lem 1 refined-sup} and by an argument similar to the one in the proof of Proposition \ref{prop 1 refined-sup}, we have 
$$
\aligned
|F(\tau)| &\leq CC_1\vep (s/t)^{3/2} s_0^{-1/2 +3\delta} + C(C_1\vep)^2(s/t)^2s_0^{-1/2 +3\delta}, 
\\
|h'_{t, x}(\lambda)| &\leq CC_1\vep (s/t)^{1/2} \lambda^{-3/2 +\delta} + CC_1\vep (t/s) \lambda^{-2}, 
\\
\int_{\tau}^s|h'_{t, x}(\lambda)d\lambda| 
& \leq CC_1\vep (s/t)^{1/2}s_0^{-1/2 +\delta} + CC_1\vep (s/t)^{-1}s_0^{-1}.
\endaligned
$$
In view of \eqref{Linfty KG ineq a}, the desired results are thus proven.

The proof of \eqref{ineq 2 lem 3 refined-sup} is an application of \eqref{ineq 1 lem 3 refined-sup}. We rely on the proof of Lemma \ref{lem 1 refined-sup} and we have that \eqref{ineq pr2 lem 1 refined-sup} still holds. We furthermore observe that in view of \eqref{ineq 1 lem 3 refined-sup},
$$
|S_{\alpha \beta}^{KG,I,J}| \leq C(C_1\vep)^2t^{-3}, \quad |I|+|J| \leq N-7.
$$
Furthermore, since $C_1\vep\leq 1$, we take, in view of \eqref{ineq pr2 lem 1 refined-sup}
$$
|S_{I, \alpha \beta}^W| \leq C(C_1\vep)^2t^{-2 + CC_1\vep}(t-r)^{-1+ CC_1\vep} \leq C(C_1\vep)^2t^{-2 + C(C_1\vep)^{1/2}}(t-r)^{-1+ C(C_1\vep)^{1/2}}. 
$$
In view of Proposition \ref{prop 1 sup-norm-W}, we arrive at 
$$
|h_{\alpha \beta}| \leq CC_1\vep t^{-1} + \frac{C(C_1\vep)^2}{CC_1\vep}t^{-1+ C(C_1\vep)^{1/2}}{(t-r)^{C(C_1\vep)^{1/2}}}
\leq C(C_1\vep)t^{-1}s^{C(C_1\vep)^{1/2}}. \qedhere
$$
\end{proof}


\subsection{A secondary bootstrap argument}

In this section, we improve the $L^\infty$ bounds of $\del^IL^J \phi$ and $\delu_{\perp} \del^IL^J\phi$ for $|I|+|J| \leq N-7$. 

\begin{proposition} \label{prop 2 refined-sup}
There exists a pair of positive constants $(C_1, \vep_2)$ with $C_1>C_0$ such that if \eqref{ineq energy assumption 1} and \eqref{ineq energy assumption 2} hold with $C_1$ and $0\leq \vep\leq \vep_2$, then for all $|I|+|J| \leq N-7$,  
\bel{ineq 1 prop 2 refined-sup}
(s/t)^{3\delta-2}|\del^IL^J \phi| + (s/t)^{3\delta - 3}|\delu_{\perp} \del^IL^J \phi| \leq CC_1\vep s^{-3/2 + C(C_1\vep)^{1/2}}, 
\ee
\bel{ineq 2 prop 2 refined-sup}
|L^Jh_{\alpha \beta}| \leq CC_1\vep t^{-1}s^{C(C_1\vep)^{1/2}}.
\ee
\end{proposition}

\begin{proof} We proceed by induction, by relying on a secondary bootstrap argument. Recall that the bootstrap assumptions \eqref{ineq energy assumption 1} and \eqref{ineq energy assumption 2} hold on $[2,s^*]$, and we suppose that there exist constants $K_{m-1}, C_{m-1}>0$ and $\vep'_{m-1}>0$ depending only on the structure of the main system such that
\bel{ineq pr0 prop 2 refined-sup}
(s/t)^{3\delta-2}|\del^IL^J \phi| + (s/t)^{3\delta - 3}|\delu_{\perp} \del^IL^J \phi| \leq K_{m-1}C_1\vep s^{-3/2 + C_{m-1}(C_1\vep)^{1/2}}, 
\ee
\bel{ineq pr0.5 prop 2 refined-sup}
|L^Jh_{\alpha \beta}(t, x)| \leq K_{m-1}C_1\vep t^{-1}s^{C_{m-1}(C_1\vep)^{1/2}}
\ee
holds on $[2,s^*]$ for all $0\leq \vep \leq \vep'_{m-1}$ and $|J| \leq m-1\leq N-7$ and $|I|+|J| \leq N-7$. 
This is true when $|J|=0$, guaranteed in view of \eqref{ineq 1 lem 3 refined-sup} and \eqref{ineq 2 lem 3 refined-sup} (since there the constant $C$ depends only on $N$ and the structure of the main system). We want prove that there exist constants $K_m,C_m, \vep'_m$ depending only on the structure of the main system such that
\bel{ineq pr1 prop 2 refined-sup}
(s/t)^{3\delta-2}|\del^IL^J \phi| + (s/t)^{3\delta - 3}|\delu_{\perp} \del^IL^J \phi| \leq K_mC_1\vep s^{-3/2 + C_m(C_1\vep)^{1/2}}, 
\ee
\bel{ineq pr2 prop 2 refined-sup}
|L^Jh_{\alpha \beta}(t, x)| \leq K_mC_1\vep t^{-1}s^{C_m(C_1\vep)^{1/2}}
\ee
hold for $0\leq\vep\leq \vep'_m$ and all $|J| \leq N-7$.

We observe that on the initial slice $\Hcal_2\cap \Kcal$, there exits a positive constant $K_{0,m}$ such that
$$
(s/t)^{3\delta-2}|\del^IL^J \phi| + (s/t)^{3\delta - 3}|\delu_{\perp} \del^IL^J \phi| \leq K_{0,m}C_0 \, \vep\leq K_{0,m}C_1\vep,
$$
We also denote by $K_{0,m}$ a positive constant such that
$
\sup_{t=2,|x| \leq 1} \{ts^{-C_m(C_1\vep)^{1/2}}|L^Jh_{\alpha \beta}(t, x)|\} \leq K_{0,m}C_0 \, \vep\leq K_{0,m}C_1\vep,
$
since we have chosen $C_1\geq C_0$. Here we observe that on $\{t=2\} \cap \Kcal$, $\sqrt{3} \leq s\leq 2$, so when $C_{m}>0$, the constant $K_{0,m}$ can be chosen independently of $C_m$.

So, first, we choose $K_m>K_{0,m}$ and set 
$s^{**} : = \sup_{s\in[2,s^*]} \big\{\eqref{ineq pr1 prop 2 refined-sup} \text{ and } \eqref{ineq pr2 prop 2 refined-sup} \text{ holds in } \Kcal_{[2,s^{**}]} \big\}.$
By continuity ($K_m>K_{0,m}$) we obtain $s^{**}>2$. We prove that if we choose $\vep'_m$ sufficiently small, then for all $\vep\leq \vep_m'$, $s^{**}=s^*$. This is done as follows. 

We take $K_m\geq K_{m-1}$, $C_m= 2C_{m-1}$ and see first that under the induction assumptions \eqref{ineq pr0 prop 2 refined-sup}, \eqref{ineq pr0.5 prop 2 refined-sup} and the bootstrap assumptions \eqref{ineq pr1 prop 2 refined-sup} and \eqref{ineq pr2 prop 2 refined-sup}, \eqref{ineq 1 lem 3.25 refined-sup} becomes (in $\Kcal_{[2,s^{**}]}$)
$$
\left|[\del^IL^J,h^{\mu\nu} \del_{\mu} \del_{\nu}]\phi\right| \leq C(C_1\vep)^2(s/t)^2 s^{-3+3\delta}
 + CK_m^2(C_1\vep)^2(s/t)^{2-3\delta}s^{-5/2 + C_m(C_1\vep)^{1/2}}.
$$
We observe that, in the right-hand side of \eqref{ineq 1 lem 3.25 refined-sup}, the last term is bounded directly by applying \eqref{ineq 2 lem 3 refined-sup} and \eqref{ineq pr2 prop 2 refined-sup}. The second term is more delicate. We distinguish between two different cases. When $|J_2'|=0$, we apply the bootstrap assumptions \eqref{ineq pr2 prop 2 refined-sup} and \eqref{ineq 1 lem 3 refined-sup}. When $0<|J_2'|<|J|$, we have $|J_1'| \leq m-1$, so we apply  \eqref{ineq pr0 prop 2 refined-sup} and \eqref{ineq pr0.5 prop 2 refined-sup} and observe that we have chosen $C_m= 2C_{m-1}$.

We then recall Lemma \ref{lem 2 refined-sup} and, by Proposition \ref{Linfini KG} (following the notation therein), we have in both cases $0\leq r/t\leq 3/5$ and $3/5<r/t<1$,
$$
\aligned
|F(s)| & \leq  CC_1\vep(s/t)^{3/2} \int_{s_0}^s \tau^{-3/2 +3\delta}d\tau + CK_m^2(C_1\vep)^2\int_{s_0}^s\tau^{-1+ C_m(C_1\vep)^{1/2}}d\tau
\\
& \leq  CC_1\vep (s/t)^{3/2}s_0^{-1/2 + 3\delta} + CC_m^{-1}K_m^2(C_1\vep)^{3/2}(s/t)^{2-3\delta} s^{C_m(C_1\vep)^{1/2}}
\\
& \leq  CC_1\vep (s/t)^{2-3\delta} + CC_m^{-1}K_m^2(C_1\vep)^{3/2}(s/t)^{2-3\delta} s^{C_m(C_1\vep)^{1/2}}.
\endaligned
$$
We also have, in view of \eqref{ineq lem h00-sup 4 1}, $\left|h_{t, x}(\lambda) \right| 
\leq CC_1\vep(s/t)^{1/2} \lambda^{-3/2 +\delta} + CC_1\vep(s/t)^{-1} \lambda^{-2}$
and then, in both cases $0\leq r/t\leq 3/5$ and $3/5<r/t<1$,
$$
\aligned
\int_{s_0}^s\left|h_{t, x}(\lambda) \right|
& \leq  CC_1\vep(s/t)^{1/2} \int_{s_0}^s\lambda^{-3/2 +\delta} d\lambda + CC_1\vep(s/t)^{-1} \int_{s_0}^s\lambda^{-2}d\lambda
\\
& \leq  CC_1\vep \left((s/t)^{1/2}s_0^{-1+\delta} + (s/t)^{-1}s_0^{-1} \right) \leq CC_1\vep.
\endaligned
$$
By Proposition \ref{Linfini KG}, we have 
$$
\aligned
& (s/t)^{3\delta - 2}s^{-3/2} \left|\del^IL^J\phi\right| + (s/t)^{3\delta-3}s^{-3/2} \left|\delu_{\perp} \del^IL^J\phi\right|
\\
& \leq CK_{0,m}C_1\vep + CC_1\vep + CC_m^{-1}K_m^2(C_1\vep)^{3/2}s^{C_m(C_1\vep)^{1/2}}. 
\endaligned
$$
We can choose $K_m$ sufficiently large and fix
$
\vep'_m = \frac{C_m^2}{C_1} \left(\frac{K_m - 2CK_{0,m}-2C}{2CK_m^2} \right)^2>0,
$
and 
then we see that on $[2,s^{**}]$:
\bel{ineq pr3 prop 2 refined-sup}
(s/t)^{3\delta - 2}s^{-3/2} \left|\del^IL^J\phi\right| + (s/t)^{3\delta-3}s^{-3/2} \left|\delu_{\perp} \del^IL^J\phi\right| \leq \frac{1}{2}K_mC_1\vep s^{C_m(C_1\vep)^{1/2}}.
\ee
Here we need to emphaze that $C_m$ is determined only by $N$ and the structure of the system: we have $C_0$, determined in view of \eqref{ineq 2 lem 3 refined-sup} where the constant $C$ is determined by $N$ and the main system. Then,  $C_m = 2C_{m-1}$ thus $C_m$ are determined only by $N$ and the structure of the system.

In the same way, we follow the notation in Proposition \ref{prop 1 sup-norm-W} combined with following estimates deduced from \eqref{ineq pr1 prop 2 refined-sup} : as $|I|+|J| \leq N-7$
$$
\aligned
|S_{\alpha \beta}^{KG,I,J}| & \leq  C_m(C_1\vep)^2(s/t)^{4-6\delta}s^{-3+ C_m(C_1\vep)^{1/2}}
\\
& \leq  C(K_mC_1\vep)^2t^{-3+3\delta+\frac{1}{2}C_m(C_1\vep)^{1/2}}(t-r)^{-3\delta+\frac{1}{2}C_m(C_1\vep)^{1/2}},
\endaligned
$$
where we rely on a similar argument for the estimate of $\left|[\del^IL^J,h^{\mu\nu} \del_{\mu} \del_{\nu}]\phi\right|$.

We also recall \eqref{ineq pr2 lem 1 refined-sup} for $|I|+|J| \leq N-7$
$$
|S_{I, \alpha \beta}^W| \leq C(C_1\vep)^2t^{-2 + CC_1\vep}(t-r)^{-1+ CC_1\vep} \leq C(C_1\vep)^2t^{-2 + C(C_1\vep)^{1/2}}(t-r)^{-1+ C(C_1\vep)^{1/2}}. 
$$
This leads us to (by  Proposition \ref{prop 1 sup-norm-W})
$$
\aligned
&\left|\del^IL^Jh_{\alpha \beta} \right|
\\
& \leq  Cm_S\vep t^{-1} + \frac{C(C_1\vep)^2}{CC_1\vep} t^{-1+ C(C_1\vep)^{1/2}}(t-r)^{C(C_1\vep)^{1/2}}
 + C(K_mC_1\vep)^2t^{-1}s^{C_m(C_1\vep)^{1/2}}
\\
& \leq  CC_1K_{0,m} \vep t^{-1} + CC_1\vep t^{-1+ C(C_1\vep)^{1/2}}(t-r)^{C(C_1\vep)^{1/2}} + C(K_mC_1\vep)^2t^{-1}(t-r)^{C_m(C_1\vep)^{1/2}}
\\
& \leq  CC_1\vep \left(K_{0,m} + 1 + K_m^2C_1\vep\right)t{-1+ C_m(C_1\vep)^{1/2}}(t-r)^{C_m(C_1\vep)^{1/2}}.
\endaligned
$$
We check that when $\vep\leq \vep_m'$, on $[2,s^{**}]$:
\bel{ineq pr4 prop 2 refined-sup}
\left|\del^IL^Jh_{\alpha \beta} \right| \leq \frac{1}{2}K_mC_1\vep.
\ee

Now, in view of \eqref{ineq pr3 prop 2 refined-sup} and \eqref{ineq pr4 prop 2 refined-sup}, we make the following observation: when $s^{**}<s^{*}$, by continuity we must have
\be
(s/t)^{3\delta-2}|\del^IL^J \phi| + (s/t)^{3\delta - 3}|\delu_{\perp} \del^IL^J \phi|= K_mC_1\vep s^{-3/2 + C(C_1\vep)^{1/2}}
\ee
or
\be
|L^Jh_{\alpha \beta}(t, x)|=K_mC_1\vep t^{-1}s^{C(C_1\vep)^{1/2}}. 
\ee
This is a contradiction with \eqref{ineq pr3 prop 2 refined-sup} together with \eqref{ineq pr4 prop 2 refined-sup}. We conclude that $s^{**}=s^{*}$. That is, \eqref{ineq 1 prop 2 refined-sup} and \eqref{ineq 2 prop 2 refined-sup} are proved for $|J| = m$. By induction, \eqref{ineq 1 prop 2 refined-sup} and \eqref{ineq 2 prop 2 refined-sup} are proved for $|J| \leq N-7$. This concludes the argument, by taking $\vep_2 = \vep_{N-7}'$.
\end{proof}
 

\section{High-Order Refined $L^2$ Estimates}
\label{section-10}

\subsection{Objective of this section and preliminary}

In this section we improve the energy bounds of both $h_{\alpha \beta}$ and $\phi$ for $N-4\leq |I|+|J| \leq N$. We rely on the energy estimates Proposition \ref{prop 1 energy-W} and Proposition \ref{prop energy 2KG}. In order to apply these two propositions, we need a control of the source terms:
\begin{itemize}

\item For $\del^IL^Jh_{\alpha \beta}$, we have the terms
$
\del^IL^JF_{\alpha \beta}, \quad {\sl QS}_{\phi}, \quad [\del^IL^J,h^{\mu\nu} \del_\mu\del_\nu]h_{\alpha \beta}.
$

\item For $\del^IL^J\phi$, we have the terms 
$
[\del^IL^J,h^{\mu\nu} \del_\mu\del_\nu]\phi.
$
\end{itemize}
In this section, we derive the $L^2$ bounds and apply them (in the next subsection) in the proof of the main estimate.
Note that the estimate for $F_{\alpha \beta}$ is already covered by Lemma \ref{lem 0 refined-energy-low-h}.  We begin with ${\sl QS}_{\phi}$.

\begin{lemma} \label{lem 1 pre refined-energy-higher}
Assume the bootstrap assumptions \eqref{ineq energy assumption 1} and \eqref{ineq energy assumption 2} hold. Then the following estimates hold for $|I|+|J| \leq N$:
\bel{ineq 1 lem 1 pre refined-energy-higher}
\aligned
& \left\|\del^IL^J\left(\del_{\alpha} \phi\del_{\beta} \phi\right) \right\|_{L^2(\Hcal_s^*)} + \left\|\del^IL^J\left(\phi^2\right) \right\|_{L^2(\Hcal_s^*)}
\\
 & \leq  CC_1\vep s^{-3/2} \sum_{|I'| \leq |I|}E_{M,c^2}(s, \del^{I'}L^J\phi)^{1/2}
 +  CC_1\vep s^{-3/2 + C(C_1\vep)^{1/2}} \sum_{|I'| \leq |I|\atop |J'|<|J|}E_{M,c^2}(s, \del^{I'}L^{J'} \phi)^{1/2}.
\endaligned
\ee
\end{lemma}

\begin{proof}
We only treat $\del^IL^J\left(\del_{\alpha} \phi\del_{\beta} \phi\right)$ and omit the argument for $\del^IL^J\left(\phi^2\right)$ which is simpler. We have
$
\del^IL^J\left(\del_{\alpha} \phi\del_{\beta} \phi\right)
= \sum_{I_1+I_2=I\atop J_1+J_2=J} \del^{I_1}L^{J_1} \del_{\alpha} \phi \,   \del^{I_2}L^{J_2} \del_{\beta} \phi.
$
Assuming that $N\geq 13$, we have either $|I_1|+|J_1| \leq N-7$ or $|I_2|+|J_2| \leq N-7$. Without loss of generality, we suppose that $|I_1|+|J_1| \leq N-7$:
\begin{itemize}

\item When $|I_1| = |J_1|=0$. We apply \eqref{ineq 3.2 refined-sup low c} :
$$
\aligned
& \left\|\del^{I_1}L^{J_1} \del_{\alpha} \phi \,   \del^{I_2}L^{J_2} \del_{\beta} \phi\right\|_{L^2(\Hcal_s^*)}
 =  \left\|\del_{\alpha} \phi \,   \del^IL^J\del_{\beta} \phi\right\|_{L^2(\Hcal_s^*)}
\\
& \leq  CC_1\vep \left\|(s/t)^{2-3\delta}s^{-3/2}(t/s) \,   (s/t) \del^IL^J\del_{\beta} \phi\right\|_{L^2(\Hcal_s^*)}
 \leq  CC_1\vep s^{-3/2}E_{M,c^2}(s, \del^IL^J\phi)^{1/2}.
\endaligned
$$

\item When $|J_1|=0, 1\leq |I_1| \leq N-7$, then $|I_2|+|J_2| \leq N-1$. We apply \eqref{ineq 3.2 refined-sup low a} : 
$$
\aligned
& \left\|\del^{I_1}L^{J_1} \del_{\alpha} \phi \,   \del^{I_2}L^{J_2} \del_{\beta} \phi\right\|_{L^2(\Hcal_s^*)}
= \left\|\del^{I_1} \del_{\alpha} \phi \,   \del^{I_2}L^J\del_{\beta} \phi\right\|_{L^2(\Hcal_s^*)}
\\
& \leq CC_1\vep \left\|(s/t)^{1-3\delta}s^{-3/2} \,   \del^{I_2}L^J\del_{\beta} \phi\right\|_{L^2(\Hcal_s^*)}
\leq  CC_1\vep s^{-3/2} \sum_{|I'| \leq |I|} E_{M,c^2}(\del^{I'}L^J\phi)^{1/2}.
\endaligned
$$

\item When $1\leq|J_1|$ and $|I_1|+|J_1| \leq N-7$, then $|I_2|+|J_2| \leq N-1$ and $|J_2|<|J|$. We apply \eqref{ineq 3.1 refined-sup low a}
$$
\aligned
\left\|\del^{I_1}L^{J_1} \del_{\alpha} \phi \,   \del^{I_2}L^{J_2} \del_{\beta} \phi\right\|_{L^2(\Hcal_s^*)}
& \leq CC_1\vep \left\|(s/t)^{1-3\delta}s^{-3/2 + C(C_1\vep)^{1/2}} \,   \del^{I_2}L^{J_2} \del_{\beta} \phi\right\|_{L^2(\Hcal_s^*)}
\\
& \leq CC_1\vep s^{-1+ C(C_1\vep)^{1/2}} \,   s^{-1/2} \sum_{I'\leq |I|\atop |J'|<|J|}E_{M,c^2}(s, \del^{I'}L^{J'} \phi)^{1/2}. \qedhere
\endaligned 
$$
\end{itemize} 
\end{proof}

\begin{lemma} \label{lem 2 pre refined-energy-higher}
Under the bootstrap assumption, for $|I|+|J| \leq N$ one has 
\bel{ineq 1 lem 2 pre refined-energy-higher}
\aligned
&\left\|[\del^IL^J,h^{\mu\nu} \del_\mu\del_\nu]h_{\alpha \beta} \right\|_{L^2(\Hcal_s^*)}
\\
& \leq   CC_1\vep s^{-1} \hskip-.25cm  \sum_{\alpha', \beta',a,|I'| \leq|I|\atop|J'|<|J|}  \hskip-.25cm  E_M^*(s, \del^{I'}L_aL^{J'}h_{\alpha'\beta'})^{1/2}
+ CC_1\vep s^{-1+ C(C_1\vep)}  \hskip-.25cm \sum_{\alpha'\beta',|I'| \leq|I|\atop |J'|<|J|}  \hskip-.25cm  E_M^*(s, \del^{I'}L^{J'}h_{\alpha'\beta'})^{1/2}
\\
& \quad + CC_1\vep s^{-3/2} \sum_{|I'| \leq|I|}E_{M,c^2}^*(s, \del^{I'}L^{J} \phi)^{1/2}
+ CC_1\vep s^{-3/2 + C(C_1\vep)^{1/2}} \sum_{|I'| \leq|I|\atop |J'|<|J|}E_{M,c^2}^*(s, \del^{I'}L^{J'} \phi)^{1/2}
\\
& \quad + C(C_1\vep)^2s^{-3/2 +3\delta}
\endaligned
\ee
and, in particular, for $|J|=0$, 
$$
\left\|[\del^I,h^{\mu\nu} \del_\mu\del_\nu]h_{\alpha \beta} \right\|_{L^2(\Hcal_s^*)}
\leq CC_1\vep s^{-3/2} \sum_{|I'| \leq|I|}E_{M,c^2}^*(s, \del^{I'} \phi)^{1/2} + C(C_1\vep)^2s^{-3/2 +3\delta}.
$$
\end{lemma}

\begin{proof} We rely on the estimate \eqref{ineq L-2 lem 1 commtator II} and \eqref{eq 1 lem 1 second-order} combined with \eqref{ineq 1 lem 1 pre refined-energy-higher}. In view of \eqref{ineq L-2 lem 1 commtator II}, we need to estimate
$
\left\|(s/t)^2\del_t\del_t\del^IL^{J'}h_{\alpha \beta} \right\|_{L^2(\Hcal_s^*)}
$
for $|J'|<|J|$. Then, in view of in view of \eqref{eq 1 lem 1 second-order}, the above quantity is to be bounded by the $L^2$ norm of
$
Sc_1[\del^IL^{J'}h_{\alpha \beta}]$,
$Sc_2[\del^IL^{J'}h_{\alpha \beta}]$,
$\del^IL^{J'}F_{\alpha \beta}$,
and $\del^IL^{J'}{\sl QS}_{\phi}$. 
These terms are bounded respectively in view of \eqref{ineq 1 Sc-L2}, \eqref{ineq 2 Sc-L2}, Lemma \ref{lem 0 refined-energy-low-h} and \eqref{ineq 1 lem 1 pre refined-energy-higher}. With all these estimate substitute into \eqref{eq 1 lem 1 second-order}, we have for $|J'|<|J|$,
\be
\aligned
&\left\|(s/t)^2\del_t\del_t\del^IL^{J'}h_{\alpha \beta} \right\|_{L^2(\Hcal_s^*)}
\\
& \leq  C s^{-1}   \hskip-.25cm  \sum_{\alpha', \beta',a,|I'| \leq|I|\atop|J'|<|J|}  \hskip-.25cm  E_M^*(s, \del^{I'}L_aL^{J'}h_{\alpha'\beta'})^{1/2}
+ CC_1\vep s^{-1+ C(C_1\vep)}   \hskip-.25cm  \sum_{\alpha'\beta',|I'| \leq|I|\atop |J'|<|J|}  \hskip-.25cm  E_M^*(s, \del^{I'}L^{J'}h_{\alpha'\beta'})^{1/2}
\\
& \quad + CC_1\vep s^{-3/2} \sum_{|I'| \leq|I|}E_{M,c^2}^*(s, \del^{I'}L^{J} \phi)^{1/2}
+ CC_1\vep s^{-3/2 + C(C_1\vep)^{1/2}} \sum_{|I'| \leq|I|\atop |J'|<|J|}E_{M,c^2}^*(s, \del^{I'}L^{J'} \phi)^{1/2}
\\
& \quad + \sum_{|J'|<|J|} \|[\del^IL^{J'},h^{\mu\nu} \del_{\mu} \del_{\nu}]h_{\alpha \beta} \|_{L_f^2(\Hcal_s)}
+ C(C_1\vep)^2s^{-3/2 +2\delta}.
\endaligned
\ee
That is, we have 
$$
\aligned
&\left\|[\del^IL^J,h^{\mu\nu} \del_\mu\del_\nu]h_{\alpha \beta} \right\|_{L^2(\Hcal_s^*)}
\\
& \leq  C C_1\vep s^{-1}   \hskip-.25cm  \sum_{\alpha', \beta',a,|I'| \leq|I|\atop|J'|<|J|}  \hskip-.25cm  E_M^*(s, \del^{I'}L_aL^{J'}h_{\alpha'\beta'})^{1/2}
+ CC_1\vep s^{-1+ C(C_1\vep)}   \hskip-.25cm  \sum_{\alpha'\beta',|I'| \leq|I|\atop |J'|<|J|}  \hskip-.25cm  E_M^*(s, \del^{I'}L^{J'}h_{\alpha'\beta'})^{1/2}
\\
& \quad + CC_1\vep s^{-3/2} \sum_{|I'| \leq|I|}E_{M,c^2}^*(s, \del^{I'}L^{J} \phi)^{1/2}
+ CC_1\vep s^{-3/2 + C(C_1\vep)^{1/2}} \sum_{|I'| \leq|I|\atop |J'|<|J|}E_{M,c^2}^*(s, \del^{I'}L^{J'} \phi)^{1/2}
\\
& \quad + \sum_{|J'|<|J|} \|[\del^IL^{J'},h^{\mu\nu} \del_{\mu} \del_{\nu}]h_{\alpha \beta} \|_{L_f^2(\Hcal_s)}
+ C(C_1\vep)^2s^{-3/2 +2\delta}.
\endaligned
$$
Then, we proceed by induction on $J$ and the desired result is reached.
When $|J|=0$, in the right-hand side of the above estimate there exist only the third and the last term, this proves the desired result in this case. Then, by induction on $|J|$, the desired result is established for $|I|+|J| \leq N$.
\end{proof}

\begin{lemma} \label{lem 3 pre refined-energy-higher}
Under the bootstrap assumption, for all $|I|+|J| \leq N$ one has 
\bel{ineq 1 lem 3 pre refined-energy-higher}
\aligned
& \left\|[\del^IL^J,h^{\mu\nu} \del_\mu\del_\nu]\phi\right\|_{L_f^2(\Hcal_s)}
\\
& \leq CC_1\vep s^{-1/2} \sum_{|J'|=|J|\atop \alpha, \beta}E_M^*(s,L^{J'}h_{\alpha \beta})^{1/2}
+ CC_1\vep s^{-1/2} \sum_{|J'|= |J|\atop \alpha \beta} \int_2^s\tau^{-1}E_M^*(\tau,L^{J'}h_{\alpha \beta})^{1/2}d\tau
\\
& \quad + CC_1\vep s^{-1+ C(C_1\vep)^{1/2}}   \hskip-.25cm  \sum_{|I'| \leq|I|+1\atop |J'|<|J|}  \hskip-.25cm  E_M^*(s, \del^{I'}L^{J'} \phi)^{1/2}
 + CC_1\vep s^{-1/2 + C(C_1\vep)^{1/2}}   \hskip-.25cm  \sum_{|J_1'|<|J|\atop \alpha', \beta'}  \hskip-.25cm  E_M^*(s,L^{J_1'}h_{\alpha'\beta'})^{1/2}
\\
& \quad + CC_1\vep s^{-1/2 + C(C_1\vep)^{1/2}} \sum_{|J_1'|<|J|\atop \alpha', \beta'} \int_2^s\tau^{-1}E_M^*(\tau,L^{J_1'}h_{\alpha'\beta'})^{1/2}d\tau
 + C(C_1\vep)^2s^{-1/2 + C(C_1\vep)^{1/2}}.
\endaligned
\ee
When $|J|=0$, one has 
\bel{ineq 2 lem 3 pre refined-energy-higher}
\left\|[\del^I,h^{\mu\nu} \del_\mu\del_\nu]\phi\right\|_{L_f^2(\Hcal_s)} \leq C(C_1\vep)^2s^{-1+3\delta}.
\ee
\end{lemma}

\begin{proof}
We need to estimate the terms listed in \eqref{eq 2 lem 1 nolinear}. The estimates on first two terms are trivial: one is a null term and the other has a additional decay $t^{-1}$. We just point out that for the first term we need to apply \eqref{ineq 2 homo}, \eqref{ineq 3 homo} combined with \eqref{ineq basic-sup-h} or \eqref{ineq 0 Hardy} and write down their $L^2$ bounds 
\be
\|\del^IL^J{\sl GQQ}_{h\phi} \|_{L^2(\Hcal_s^*)} + \|t^{-1} \del^{I_1}L^{J_1}h_{\mu\nu} \del^{I_2}L^{J_2} \del_\gamma \phi\|_{L^2(\Hcal_s^*)} \leq C(C_1\vep)^2 s^{-1+2\delta}.
\ee
We focus on the last three terms.

\vskip.15cm

\noindent {\sl Term 1.} $\del^{I_1}L^{J_1} \hu^{00} \del^{I_2}L^{J_2} \del_t\del_t\phi$. Recall that $|I_1|\geq 1$. The $L^2$ norm of this term is bounded by a discussion on the following cases:

$\bullet$ Case $1\leq |I_1|+|J_1| \leq N-2$. We apply \eqref{ineq h00-sup 1} combined with the basic energy estimate: 
$$
\left\|\del^{I_1}L^{J_1} \hu^{00} \del^{I_2}L^{J_2} \del_t\del_t\phi\right\|_{L^2(\Hcal_s^*)}
\leq CC_1\vep \left\|t^{-3/2}s^{\delta}(t/s) \,   (s/t) \del^{I_2}L^{J_2} \del_t\del_t\phi\right\|_{L^2(\Hcal_s^*)}
\leq C(C_1\vep)^2 s^{-1+3\delta}.
$$

$\bullet$ Case $N-1 \leq |I_1|+|J_1| \leq N$, then $|I_2|+|J_2| \leq 1\leq N-8$. Then we apply \eqref{ineq h00-L-2 wave} combined with the basic sup-norm estimate for $\del^{I_2}L^{J_2} \del_t\del_t\phi$:
$$
\aligned
\left\|\del^{I_1}L^{J_1} \hu^{00} \del^{I_2}L^{J_2} \del_t\del_t\phi\right\|_{L^2(\Hcal_s^*)}
& \leq  CC_1\vep \left\|(s/t) \del^{I_1}L^{J_1} \hu^{00} \,   (t/s)t^{-3/2}s^{\delta} \right\|_{L^2(\Hcal_s^*)}
\\
& \leq  CC_1\vep s^{-3/2 +\delta} \left\|(s/t) \del^{I_1}L^{J_1} \hu^{00} \right\|_{L^2(\Hcal_s^*)} \leq C(C_1\vep)^2s^{-3/2 +3\delta}.
\endaligned
$$

\vskip.15cm

\noindent {\sl Term 2.} $L^{J_1} \hu^{00} \del^IL^{J_2} \del_t\del_t\phi$. Recall that $|J_1|\geq 1$ so that $|J_2| \leq |J|-1\leq N-1$.

$\bullet$ Case $1\leq |J_1| \leq N-7$. In this case, we apply \eqref{ineq 2 prop 2 refined-sup} to $L^{J_1} \hu^{00}$ (seen as a linear combination of $L^{J_1'}h_{\alpha \beta}$ with $|J_1'|$ plus higher-order corrections): 
$$
\aligned
\left\|L^{J_1} \hu^{00} \del^IL^{J_2} \del_t\del_t\phi\right\|_{L^2(\Hcal_s^*)}
& \leq  CC_1\vep\left\|t^{-1}s^{C(C_1\vep)^{1/2}} \del^IL^{J_2} \del_t\del_t\right\|_{L^2(\Hcal_s^*)}
\\
& \leq  CC_1\vep s^{-1+ C(C_1\vep)^{1/2}} \left\|(s/t) \del^IL^{J_2} \del_t\del_t\right\|_{L^2(\Hcal_s^*)}
\\
& \leq  CC_1\vep s^{-1+ C(C_1\vep)^{1/2}} \sum_{|J'|<|J|}E_{M,c^2}(s, \del^IL^{J'} \phi)^{1/2}.
\endaligned
$$

$\bullet$ Case $N-6\leq |J_1| \leq |J|-1\leq N-1$ then $|I|+|J_2| \leq 6\leq N-8$. 
In this case we apply Proposition
\ref{prop h00-sup 1} to $L^{J_1} \hu^{00}$ and \eqref{ineq 3.1 refined-sup low a}. First of all, by the estimates \eqref{ineq com 2.4} of commutators and  \eqref{ineq 3.1 refined-sup low a}, we deduce that 
$
\left|\del^IL^{J_2} \del_t\del_t\phi\right| \leq CC_1\vep(s/t)^{1-3\delta}s^{-3/2 + C(C_1\vep)^{1/2}}.
$
Then, we have
$$
\aligned
&\left\|L^{J_1} \hu^{00} \del^IL^{J_2} \del_t\del_t\phi\right\|_{L^2(\Hcal_s^*)}
\\
& \leq  \left\|L^{J_1} \hu_0^{00} \del^IL^{J_2} \del_t\del_t\phi\right\|_{L^2(\Hcal_s^*)} +  \left\|L^{J_1} \hu_1^{00} \del^IL^{J_2} \del_t\del_t\phi\right\|_{L^2(\Hcal_s^*)}
\\
& \leq  CC_1\vep\left\|t^{-1} \del^IL^{J_2} \del_t\del_t\phi\right\|_{L^2(\Hcal_s^*)}
 + CC_1\vep \left\|L^{J_1} \hu_1^{00} \,   (s/t)^{1-3\delta}s^{-3/2 + C(C_1\vep)^{1/2}} \right\|_{L^2(\Hcal_s^*)}
\\
& \leq  CC_1s^{-1} \sum_{|I'| \leq |I|+1\atop |J'|<|J|}E_{M,c^2}(s, \del^{I'}L^{J'} \phi)^{1/2}
+ CC_1\vep s^{-1/2 + C(C_1\vep)^{1/2}} \left\|s^{-1}(s/t)^{-1+\delta}L^{J_1} \hu^{00}_1 \right\|_{L^2(\Hcal_s^*)}
\\
& \leq CC_1s^{-1} \sum_{|I'| \leq |I|+1\atop |J'|<|J|}E_{M,c^2}(s, \del^{I'}L^{J'} \phi)^{1/2}
+  CC_1\vep s^{-1/2 + C(C_1\vep)^{1/2}} \left\|s^{-1}(s/t)^{-1+\delta}L^{J_1} \hu^{00}_1 \right\|_{L^2(\Hcal_s^*)}
\\ 
& \quad + CC_1\vep s^{-1/2 + C(C_1\vep)^{1/2}} \sum_{|J'| \leq|J|\atop \alpha, \beta} \int_2^s\tau^{-1}E_M^*(\tau,L^{J'}h_{\alpha \beta})^{1/2}d\tau
+ C(C_1\vep)^2s^{-1/2 + C(C_1\vep)^{1/2}},
\endaligned
$$
where in the last inequality we applied Proposition \ref{prop h00-sup 1}. 

$\bullet$ Case $1\leq J_1=J$ then $|J_2|=0$.

When $|J|\geq N-6$, we see that $|I| \leq 6\leq N-7$ provided by $N\geq 13$. In this case we apply \eqref{ineq 3.2 refined-sup low a} to $\del^IL^{J_2} \del_t\del_t\phi$ and Proposition
\ref{prop h00-sup 1} on $L^{J_1} \hu^{00}$:
$$
\aligned
\left\|L^{J_1} \hu^{00} \del^IL^{J_2} \del_t\del_t\phi\right\|_{L^2(\Hcal_s^*)}
& =  \left\|L^J\hu^{00} \del^I\del_t\del_t\phi\right\|_{L^2(\Hcal_s^*)}
\\
& \leq  CC_1\vep \left\|t^{-1} \del^I\del_t\del_t\phi\right\|_{L^2(\Hcal_s^*)}
 + CC_1\vep\left\|(s/t)^{1-3\delta}s^{-3/2}L^J\hu^{00}_1\right\|_{L^2(\Hcal_s^*)}.
\endaligned
$$
The first term is bounded by $CC_1\vep s^{-1} \sum_{|I'| \leq|I|+1}E_{M,c^2}(\del^{I'} \phi)^{1/2}$. For the second term, by applying Proposition \ref{prop h00-sup 1}, we have 
$$
\aligned
&\left\|(s/t)^{1-3\delta}s^{-3/2}L^J\hu^{00}_1\right\|_{L^2(\Hcal_s^*)}
\\
& \leq  \left\|(s/t)^{1-3\delta}s^{-3/2}s(s/t)^{1- \delta} \,   s^{-1}(s/t)^{-1+\delta}L^J\hu^{00}_1\right\|_{L^2(\Hcal_s^*)}
\\
& \leq CC_1\vep s^{-1/2} \sum_{|J_1'| \leq|J|\atop \alpha, \beta}E_M^*(s,L^{J_1'}h_{\alpha \beta})^{1/2}
\\
& \quad + CC_1\vep s^{-1/2} \sum_{|J_1'| \leq|J|\atop \alpha, \beta} \int_2^s\tau^{-1}E_M^*(\tau,L^{J_1'}h_{\alpha \beta})^{1/2}d\tau
  +  C(C_1\vep)^2s^{-1/2}.
\endaligned
$$

When $|J| \leq N-7$, we apply \eqref{ineq 2 prop 2 refined-sup} to $L^J\hu^{00}$:
$$
\aligned
\left\|L^{J_1} \hu^{00} \del^IL^{J_2} \del_t\del_t\phi\right\|_{L^2(\Hcal_s^*)}
& \leq  CC_1\vep s^{-1+ C(C_1\vep)^{1/2}} \|(s/t) \del^I\del_t\del_t\phi\|_{L^2(\Hcal_s^*)}
\\
& \leq  CC_1\vep s^{-1+ C(C_1\vep)^{1/2}}E_{M,c^2}(\del^I\del_t\phi)^{1/2}.
\endaligned
$$ 
We emphasize that such a term does not exist  when $|J|=0$ since the condition $1\leq |J_1| \leq |J|$ is then never satisfied.
 
\vskip.15cm

\noindent {\sl Term 3.} $\hu^{00} \del_{\gamma} \del_{\gamma'} \del^IL^{J'}$ with $|J'|<|J|$. This term is easier. We apply \eqref{ineq 2 lem 3 refined-sup} to $\hu^{00}$: 
$$
\aligned
\left\|\hu^{00} \del_{\gamma} \del_{\gamma'} \del^IL^{J'} \phi\right\|_{L^2(\Hcal_s^*)}
& \leq  CC_1\vep s^{-1+ C(C_1\vep)^{1/2}} \left\|(s/t) \del_{\gamma} \del_{\gamma'} \del^IL^{J'} \phi\right\|_{L^2(\Hcal_s^*)}
\\
& \leq  CC_1\vep s^{-1+ C(C_1\vep)^{1/2}} \sum_{|I'| \leq |I|+1\atop |J'|<|J|}E_M^*(s, \del^{I'}L^{J'} \phi)^{1/2}.
\endaligned
$$

We now collect all the above estimates together and the desired result \eqref{ineq 1 lem 3 pre refined-energy-higher} is proved. Furthermore, when $|J|=0$, the condition $|J'|<|J|$ in the sum of the third, the fourth and fifth term in the right-hand side of \eqref{ineq 1 lem 3 pre refined-energy-higher} indicate that these three terms disappear. Furthermore, when $|J|=0$, the term $L^{J_1} \hu^{00} \del^IL^{J_2} \del_t\del_t\phi$ and $\hu^{00} \del_{\gamma} \del_{\gamma'} \del^IL^{J'}$ do not exist (since they demand $|J_1|\geq 1$ and $|J'|<|J|$). So,  the only existent terms are $\del^{I_1} \hu^{00} \del^{I_2} \del_t\del_t\phi$, the null terms and the commutative terms with additional $t^{-1}$ decay.
They can be bounded by $C(C_1\vep)^2s^{-1+2\delta}$ and this concludes the derivation of \eqref{ineq 2 lem 3 pre refined-energy-higher}.
\end{proof}


\subsection{Main estimates in this section}

\begin{proposition} \label{prop 1 refined-energy higher}
Let the bootstrap assumptions \eqref{ineq energy assumption 1} and \eqref{ineq energy assumption 2} hold with $C_1/C_0$ sufficiently large, then there exists a positive constant $\vep_3$ sufficiently small so that for all $\vep\leq \vep_3$ and
for $N-3\leq |I|+|J| \leq N$
\bel{ineq 1 prop 1 refined-energy higher}
E^*_{M}(s, \del^IL^J h_{\alpha \beta})^{1/2} \leq \frac{1}{2}C_1\vep s^{C(C_1\vep)^{1/2}},
\ee
\bel{ineq 2 prop 1 refined-energy higher}
E_{M, c^2}(s, \del^IL^J\phi)^{1/2} \leq \frac{1}{2}C_1\vep s^{1/2 + C(C_1\vep)^{1/2}}.
\ee
\end{proposition}

The proof will be split into two parts. First, we will derive \eqref{ineq 1 prop 1 refined-energy higher} and \eqref{ineq 2 prop 1 refined-energy higher} in the case $|J|=0$. In a second part, we will propose an induction argument for the case $|J| \neq 0$.

\begin{proof}[Proof of Proposition \ref{prop 1 refined-energy higher} in the case $|J|=0$]
In this case, the following estimates are direct by Lemma \ref{lem 0 refined-energy-low-h}, \eqref{ineq 1 lem 1 pre refined-energy-higher}, \eqref{ineq 1 lem 2 pre refined-energy-higher} and \eqref{ineq 1 lem 3 pre refined-energy-higher} : 
$$
\|\del^IF_{\alpha \beta} \|_{L^2(\Hcal_s^*)}
\leq CC_1\vep s^{-1} \sum_{|I'| \leq |I|\atop \alpha', \beta'}E_M^*\big(s, \del^{I'}h_{\alpha'\beta'} \big)^{1/2} + C(C_1\vep)^2s^{-3/2 +2\delta}
$$ 
and 
$$
\aligned
&\left\|\del^I\big(\del_{\alpha} \phi\del_{\beta} \phi\big) \right\|_{L^2(\Hcal_s^*)}
+ \left\|\del^I\left(\phi^2\right) \right\|_{L^2(\Hcal_s^*)}
 \leq  C(C_1\vep)s^{-3/2} \sum_{|I'| \leq|I|}E_{M,c^2}(s, \del^{I'} \phi)^{1/2}
\\
& \leq C(C_1\vep)^2s^{-3/2 +\delta} + C(C_1\vep)s^{-3/2} \sum_{N-3\leq|I'| \leq|I|}E_{M,c^2}(s, \del^{I'} \phi)^{1/2}, 
\endaligned
$$
while 
$$
\|[\del^I,h^{\mu\nu} \del_\mu\del_\nu]h_{\alpha \beta} \|_{L^2(\Hcal_s^*)} \leq
C(C_1\vep)^2 s^{-3/2 +3\delta}+ CC_1\vep s^{-3/2} \sum_{N-3\leq |I'| \leq |I|}E_{M,c^2}(s, \del^{I'}L^{J} \phi)^{1/2}, 
$$ 
$$
\aligned
&\left\|[\del^I,h^{\mu\nu} \del_\mu\del_\nu]\phi\right\|_{L_f^2(\Hcal_s)} \leq
C(C_1\vep)^2 s^{-1+3\delta}.
\endaligned
$$
And by Lemma \ref{lem h00-sup 3}, we obtain 
$M_{\alpha \beta}[\del^IL^J h](s) \leq C(C_1\vep)^2s^{-3/2 +2\delta}$
and 
$$
M[\del^IL^J\phi](s) \leq C(C_1\vep)^2s^{-1+2\delta}.
$$ 
We conclude that in view of \eqref{eq energy 10} and \eqref{eq energy 2}
(by observe that \eqref{eq energy 1} is guaranteed by Lemma \ref{lem h00-sup 2}): 
\bel{ineq 2 proof prop 1 refined-energy higher}
\aligned
E_{M,c^2}(s, \del^I\phi)^{1/2}
\leq CC_0 \, \vep  + C(C_1\vep)^2s^{2\delta}.
\endaligned
\ee
\bel{ineq 1' proof prop 1 refined-energy higher}
\aligned
E_M^*(s, \del^Ih_{\alpha \beta})^{1/2}
& \leq  CC_0 \, \vep + C(C_1\vep)^2
 + CC_1\vep\sum_{|I'| \leq |I|\atop \alpha', \beta'} \int_2^s\tau^{-1}E_M^*\big(\tau, \del^{I'}h_{\alpha'\beta'} \big)^{1/2} d\tau
\\
& \quad +  CC_1\vep\sum_{ N-3\leq |I'| \leq |I|} \int_2^s\tau^{-3/2} E_{M,c^2}(\tau, \del^{I'} \phi)^{1/2}d\tau
\endaligned
\ee
Substituting \eqref{ineq 2 proof prop 1 refined-energy higher} into \eqref{ineq 1' proof prop 1 refined-energy higher}, we obtain
\bel{ineq 1 proof prop 1 refined-energy higher}
\aligned
E_M^*(s, \del^Ih_{\alpha \beta})^{1/2}
& \leq  CC_0 \, \vep + C(C_1\vep)^2
+ CC_1\vep\sum_{|I'| \leq |I|\atop \alpha', \beta'} \int_2^s\tau^{-1}E_M^*\big(\tau, \del^{I'}h_{\alpha'\beta'} \big)^{1/2} d\tau.
\endaligned
\ee

Now, in view of \eqref{ineq 1 proof prop 1 refined-energy higher}, we introduce the notation
$
Y(s) : = \sum_{|I| \leq N\atop \alpha, \beta}E_M^*(s, \del^I h_{\alpha \beta})^{1/2}. 
$
With this notation, the estimate \eqref{ineq 1 proof prop 1 refined-energy higher} transforms into  
\bel{ineq 3 proof prop 1 refined-energy higher a}
\aligned
Y(s)& \leq  CC_0 \, \vep + C(C_1\vep)^2
+ CC_1\vep \int_2^s \tau^{-1} Y(\tau) d\tau. 
\endaligned
\ee 
Then Gronwall's inequality leads us to
\bel{ineq 5 proof prop 1 refined-energy higher}
\sum_{|I| \leq N\atop \alpha, \beta}E_M(s, \del^{I'}h_{\alpha \beta})^{1/2} = Y(s) \leq C(C_0 \, \vep + (C_1\vep)^2)s^{CC_1\vep}. 
\ee 
In \eqref{ineq 2 proof prop 1 refined-energy higher} and \eqref{ineq 5 proof prop 1 refined-energy higher}, we take  ${\vep_2}_0  = \frac{C_1-2CC_0}{2C_1^2}$ and for all $0\leq \vep\leq {\vep_2}_0 $, we obtain
$$
E_M(s, \del^Ih_{\alpha \beta})^{1/2} \leq \frac{1}{2}C_1\vep s^{CC_1\vep}
$$
and 
$$
E_{M,c^2}(s, \del^Ih_{\alpha \beta})^{1/2} \leq \frac{1}{2}C_1\vep s^{CC_1\vep}.
$$ 
This yields the desired result for $|J|=0$.
\end{proof}

\begin{proof}[Proof of Proposition \ref{prop 1 refined-energy higher}, Case $1\leq |J| \leq N$]. We proceed by induction on $|J|$ and assume that for $|I|+|J'| \leq N-1$ and $|J'| \leq m-1<N$
\be
\label{eq induction refined-energy higher}
\aligned
E_M(s, \del^IL^{J'}h_{\alpha \beta})^{1/2}
& \leq C(C_0 \, \vep + (C_1\vep)^2)s^{C(C_1\vep)^{1/2}}, \quad
\\
E_{M,c^2}(s, \del^IL^{J'} \phi)^{1/2}
& \leq C(C_0 \, \vep + (C_1\vep))^2s^{1/2 + C(C_1\vep)^{1/2}}.
\endaligned
\ee
We will prove that it is again valid for $|J| =m \leq N$ by using Propositions \ref{prop 1 energy-W} and \ref{prop energy 2KG}. From the induction assumption,
$$
\aligned
\|\del^IL^JF_{\alpha \beta} \|_{L^2(\Hcal_s^*)}
& \leq  CC_1\vep s^{-1} \sum_{|I'| \leq|I|\atop \alpha, \beta} E_M^*(s, \del^{I'}L^Jh_{\alpha \beta})^{1/2}
  +  CC_1\vep \left(C_0 \, \vep + (C_1\vep)^2\right)s^{-1+ C(C_1\vep)^{1/2}}
\endaligned
$$
thanks to \eqref{ineq 1 lem 0 refined-energy-low-h}, 
$$
\aligned
& \left\|\del^IL^J\left(\del_{\alpha} \phi\del_\beta \phi\right) \right\|_{L^2(\Hcal_s^*)} + \left\|\del^IL^J\left(\phi^2\right) \right\|_{L^2(\Hcal_s^*)}
\\
& \leq  CC_1\vep s^{-3/2} \sum_{|I'| \leq|I|}E_{M,c^2}(s, \del^{I'}L^J\phi)^{1/2}
   +  CC_1\vep \left(C_0 \, \vep+ (C_1\vep)^2\right)s^{-1+ C(C_1\vep)^{1/2}} 
\endaligned
$$
thanks to \eqref{ineq 1 lem 1 pre refined-energy-higher}, and finally in view of \eqref{ineq 1 lem 2 pre refined-energy-higher}. 
$$
\aligned
& \left\|[\del^IL^J,h^{\mu\nu} \del_\mu\del_\nu]h_{\alpha \beta} \right\|_{L^2(\Hcal_s^*)}
\\
&
 \leq CC_1\vep s^{-1} \sum_{|J'|=|J|\atop |I'| \leq|I|}E_M^*(s, \del^IL^{J'}h_{\alpha \beta})^{1/2}
+ CC_1\vep\left(C_0 \, \vep + (C_1\vep)^2\right)s^{-1+ C(C_1\vep)^{1/2}}.
\endaligned
$$
On the other hand, in view of \eqref{ineq 1 lem 3 pre refined-energy-higher}, we have 
$$
\aligned
& \left\|[\del^IL^J,h^{\mu\nu} \del_\mu\del_\nu]\phi\right\|_{L^2(\Hcal_s^*)}
\\
& \leq CC_1\vep s^{-1/2} \sum_{|J'|=|J|\atop \alpha, \beta}E_M^*(s,L^{J'}h_{\alpha \beta})^{1/2}
+ CC_1\vep s^{-1/2} \sum_{|J'|=|J|\atop \alpha, \beta} \int_2^s\tau^{-1}E_M^*(\tau,J^{J'}h_{\alpha \beta})^{1/2}
\\
&\quad + CC_1\vep\left(C_0+(C_1\vep)^2\right)s^{-1/2 + C(C_1\vep)^{1/2}}
\\& \quad 
 + CC_1\vep\left(C_0 \, \vep + (C_1\vep)^2\right)s^{-1/2 + C(C_1\vep)^{1/2}}
\int_2^s\tau^{-1+ C(C_1\vep)^{1/2}}d\tau
 + C(C_1\vep)^2s^{-1/2 + C(C_1\vep)^{1/2}}
\\
& \leq CC_1\vep s^{-1/2} \sum_{|J'|=|J|\atop \alpha, \beta}E_M^*(s,L^{J'}h_{\alpha \beta})^{1/2}
+ CC_1\vep s^{-1/2} \sum_{|J'|=|J|\atop \alpha, \beta} \int_2^s\tau^{-1}E_M^*(\tau,J^{J'}h_{\alpha \beta})^{1/2}
\\
& \quad
+ C(C_1\vep)^2 s^{-1/2 + C(C_1\vep)^{1/2}}.
\endaligned
$$
Also, in view of \eqref{ineq lem h00-sup 3} we have
$M_{\alpha \beta}[\del^IL^J h]\leq C (C_1\vep)^2 s^{-3/2 +2\delta}$ for $|I|+|J| \leq N$.  

With 
$$
W_m(s) := \sum_{|J|=m, \alpha, \beta \atop |I|+|J| \leq N}E_M(s, \del^IL^Jh_{\alpha \beta})^{1/2}
$$
and 
$$
K_m(s) := s^{-1/2} \sum_{|J|=m\atop |I|+|J| \leq N}E_{M,c^2}(s, \del^IL^J\phi)^{1/2},
$$
the energy estimates \eqref{eq energy 2} and \eqref{eq energy 10} lead us to a system of integral inequalities:
\bel{ineq 8 proof prop 1 refined-energy higher}
\aligned
W_m(s)& \leq  C\left(C_0 \, \vep + (C_1\vep)^2\right)s^{C(C_1\vep)^{1/2}} + CC_1\vep\int_2^s\tau^{-1} \left(W_m(\tau) + K_m(\tau) \right) \, d\tau 
\\
K_m(s)& \leq C\left(C_0 \, \vep + (C_1\vep)^2\right)s^{ C(C_1\vep)^{1/2}}
+ CC_1\vep s^{-1/2} \int_2^s\tau^{-1/2}W_m(\tau) \, d\tau
\\
& \quad +  CC_1\vep s^{-1/2} \int_2^s \tau^{-1/2} \int_2^\tau \eta^{-1}W_m(\eta)d\eta d\tau. 
\endaligned
\ee
Lemma \ref{lem f refined-energy higher} stated and proven below will guarantee that \eqref{ineq 8 proof prop 1 refined-energy higher} leads us 
$$
W_m(s) + K_m(s) \leq C\left(C_0 \, \vep + (C_1\vep)^2\right) s^{C(C_1\vep)^{1/2}}.
$$
This leads us to the desired $|J|=m$ case. Then, by induction, \eqref{ineq 1 prop 1 refined-energy higher} is valid for all $|J|=m\leq N$. We see that we can choose $\vep_3 := \frac{C_1-2CC_0}{2CC_1^2}$ with $C_1> 2CC_0$, then   
$$
W_m(s) + K_m(s) \leq \frac{1}{2}C_1\vep s^{C(C_1\vep)^{1/2}}
$$
for $0\leq \vep \leq \vep_3$.
This concludes the proof of Proposition \ref{prop 1 refined-energy higher}.
\end{proof}

\begin{lemma} \label{lem f refined-energy higher}
Let $W$ and $K$ be positive, locally integrable functions defined in $[0, T]$, and suppose that
\bel{ineq 1 lem f refined-energy higher}
\aligned
W(s)& \leq  C\left(C_0 \, \vep + (C_1\vep)^2\right)s^{C(C_1\vep)^{1/2}} + CC_1\vep\int_2^s\tau^{-1} \left(W(\tau) + K(\tau) \right) \, d\tau, 
\\
K(s)& \leq C\left(C_0 \, \vep + (C_1\vep)^2\right)s^{C(C_1\vep)^{1/2}}
+ CC_1\vep s^{-1/2} \int_2^s\tau^{-1/2}W(\tau) \, d\tau
\\
&
\quad 
 +  CC_1\vep s^{-1/2} \int_2^s \tau^{-1/2} \int_2^\tau \eta^{-1}W(\eta)d\eta d\tau
\endaligned
\ee
hold for some constant $C>0$ and a sufficiently small constant $C_1\vep>0$.
Then, one has 
$$
W(s) + K(s) \leq C\left(C_1\vep + (C_1\vep)^2\right)s^{C(C_1\vep)^{1/2}}, \quad s\in[0,T].
$$
\end{lemma}

\begin{proof}
We define
$$
W^*(s) := \sup_{\tau\in[0,s]} \left\{\tau^{-C(C_1\vep)^{1/2}}W(\tau) \right\}
$$
 as well as 
$$
K^*(s) := \sup_{s\in[0,s]} \left\{\tau^{-C(C_1\vep)^{1/2}}K(\tau) \right\}.
$$
With this notation, \eqref{ineq 1 lem f refined-energy higher} yields us to (after taking the supremum over $s$)
$$
\aligned
W^*(s)& \leq  C\left(C_0 \, \vep + (C_1\vep)^2\right)
+ CC_1\vep s^{-C(C_1\vep)^{1/2}} \big( W^*(s) + K^*(s) \big) \int_2^s \tau^{-1+ C(C_1\vep)^{1/2}}d\tau, 
\endaligned
$$
which leads us to
$$
\aligned
W^*(s) & \leq  C\left(C_0 \, \vep + (C_1\vep)^2\right) + C(C_1\vep)^{1/2} \left(W^*(s) + K^*(s) \right).
\endaligned
$$
A similar argument can be applied to estimate $K$ and we also find 
\be
K^*(s) \leq C\left(C_0 \, \vep + (C_1\vep)^2\right) + CC_1\vep W^*(s) + C(C_1\vep)^{1/2}W^*(s).
\ee
By taking the sum of the above two estimates and when $(C_1\vep)$ is sufficiently small, there exists a constant $\vep _4>0$, such that if $\vep\leq C_1^{-1} \vep_4$,
\be
W^*(s) + K^*(s) \leq  C\left(C_0 \, \vep + (C_1\vep)^2\right) + C(C_1\vep)^{1/2} \left(W^*(s) + K^*(s) \right).
\ee
Since $C(C_1\vep)^{1/2} \leq 1/2$ (for $C_1\vep$ sufficiently small) we have
$$
W^*(s) + K^*(s) \leq  C\left(C_0 \, \vep + (C_1\vep)^2\right), 
$$
which leads us to the desired result.
\end{proof}


\subsection{Applications to the derivation of refined decay estimates}
\label{subsec energy-high-app}

With the refined energy at higher-order, we can establish some additional refined decay estimates. This subsection is totally parallel to Section \ref{subsec energy-low-app}. First, by the global Sobolev inequality, for $|I|+|J| \leq N-2$: 
\bel{ineq 1 app reined-energy high}
|\del^IL^J\del_{\gamma}h_{\alpha \beta}| + |\del_{\gamma} \del^IL^Jh_{\alpha \beta}| \leq CC_1\vep t^{-1/2}s^{-1 + C(C_1\vep)^{1/2}},
\ee
\bel{ineq 2 app reined-energy high}
|\del^IL^J\delu_ah_{\alpha \beta}| + |\delu_a \del^IL^Jh_{\alpha \beta}| \leq CC_1\vep t^{-3/2}s^{C(C_1\vep)^{1/2}}.
\ee
Based on this improved sup-norm estimate, the following estimates are direct by integration along the rays $\{(t, \lambda x)|1\leq \lambda \leq t/|x|\}$:
\bel{ineq 3 app reined-energy high}
|\del^IL^Jh_{\alpha \beta}| \leq CC_1\vep \left(t^{-1}+(s/t)t^{-1/2}s^{C(C_1\vep)^{1/2}} \right).
\ee
From the above estimates and Lemma \ref{lem wave-condition 4}, we have 
\bel{ineq 4 reined-energy high}
\left|\del^IL^J\del_{\alpha} \hu^{00} \right| + \left|\del^IL^J\del_{\alpha} \hu^{00} \right| \leq CC_1\vep t^{-3/2}s^{C(C_1\vep)^{1/2}}
\ee
and also by integration along the rays $\{(t, \lambda x)|1\leq \lambda \leq t/|x|\}$:
\bel{ineq 5 app reined-energy high}
\left|\del^IL^J\hu^{00} \right| \leq CC_1\vep\left(t^{-1} + (s/t)^2t^{-1/2}s^{C(C_1\vep)^{1/2}} \right).
\ee

Two more delicate applications of this higher-order, improved energy estimate are discussed in the following. They are also parallel to Lemmas \ref{lem 0 refined-sup} and \ref{lem 1 refined-energy-low-h}.

\begin{lemma} \label{lem 1 app refined-energy higher}
For $|I|+|J| \leq N-2$, one has 
\bel{ineq 3 lem 1 app refined-energy higher}
\left|\del^IL^J F_{\alpha \beta} \right| \leq C(C_1\vep)^2 t^{-1}s^{-2 + C(C_1\vep)^{1/2}}.
\ee
\end{lemma}

\begin{proof}
We focus on $F_{\alpha \beta}$. Recall that $F_{\alpha \beta} = Q_{\alpha \beta} + P_{\alpha \beta}$. We see that (omit cubic and higher-order terms, which have good decay), the quadratic part of $F_{\alpha \beta}$ are linear combinations of $\del_{\gamma}h_{\alpha \beta} \del_{\gamma'}h_{\alpha'\beta'}$. Then, we  apply \eqref{ineq 1 app reined-energy high} and see that, for $|I|+|J| \leq N-2$, we find 
$
\del^IL^J\left(\del_{\gamma}h_{\alpha \beta} \del_{\gamma'}h_{\alpha'\beta'} \right) \leq C(C_1\vep)^2t^{-1}s^{-2 + C(C_1\vep)^{1/2}}.
$
\end{proof}

A second refined estimate parallel to Lemma \ref{lem 1 refined-energy-low-h} can now be derived. 
The proof is essentially the same as that of Lemma \ref{lem 1 refined-energy-low-h}. The only difference is that we apply the sup-norm estimates presented in Lemma \ref{lem 1 app refined-energy higher} for $|I|+|J| \leq N-2$.

\begin{lemma} \label{lem 2 app refined-energy higher}
For $|I|+|J| \leq N-3$, one has 
\bel{ineq 1 lem 2 app refined-energy higher}
\left|\del_t\del_t\del^IL^J h_{\alpha \beta} \right| \leq CC_1\vep t^{1/2}s^{-3+(CC_1\vep)^{1/2}}.
\ee
\end{lemma}

By a similar argument as done below \eqref{eq second-order-frame-1}, we have
\bel{ineq 1' lem 2 app refined-energy higher}
|\del_{\alpha} \del_{\beta} \del^IL^J h_{\alpha \beta}| + |\del^IL^J\del_{\alpha} \del_{\beta} h_{\alpha \beta}|
 \leq CC_1\vep t^{1/2}s^{-3+(CC_1\vep)^{1/2}}.
\ee
Apart from the above refined decay on $h_{\alpha \beta}$, we also have the following refined decay for $\phi$, deduced from \eqref{ineq 2 prop 1 refined-energy higher}. For $|I|+|J| \leq N-2$, we have
\bel{ineq 6 app reined-energy high}
\aligned
\left|\del^IL^J\del_{\alpha} \phi\right| + \left|\del_{\alpha} \del^IL^J\phi\right| & \leq  CC_1\vep t^{-1/2}s^{-1/2 + C(C_1\vep)^{1/2}},
\\
\left|\del^IL^J\delu_a \phi\right| + \left|\delu_a \del^IL^J\phi\right| + \left|\del^IL^J\phi\right| & \leq  CC_1\vep t^{-3/2}s^{1/2 + C(C_1\vep)^{1/2}}, 
\endaligned
\ee
while, for $|I|+|J| \leq N-3$, we apply \eqref{ineq 1 homo}  and get 
\bel{ineq 7 app reined-energy high}
\left|\del^IL^J\delu_a \phi\right| + \left|\delu_a \del^IL^J\phi\right| \leq CC_1\vep t^{-5/2}s^{1/2 + C(C_1\vep)^{1/2}}. 
\ee
Finally, for $|I|+|J| \leq N-4$, we have 
\bel{ineq 8 app reined-energy high}
\left|\del^IL^J\del_{\beta} \delu_a \phi\right| + \left|\delu_a \del_{\beta} \del^IL^J\phi\right| \leq CC_1\vep t^{-5/2}s^{1/2 + C(C_1\vep)^{1/2}}, 
\ee
\bel{ineq 9 app reined-energy high}
\left|\del_\alpha \del_\beta \del^IL^J\phi\right| + \left|\del^IL^J\del_\alpha \del_\beta \phi\right| \leq CC_1\vep t^{-3/2}s^{1/2 + C(C_1\vep)^{1/2}}.
\ee


\section{High-Order Refined Sup-Norm Estimates}
\label{section-12}

\subsection{Preliminary}

We begin with our refined estimates for $\del^IL^J\left(h^{\mu\nu} \del_\mu\del_\nu h_{\alpha \beta} \right)$, ${\sl QS}_{\phi}$ and $[\del^IL^J,h^{\mu\nu} \del_\mu\del_\nu]\phi$ for $|I|+|J| \leq N-4$.

\begin{lemma} \label{lem 1 sup-norm-higher}
For all $|I|+|J| \leq N-4$, the following estimate holds:
\bel{ineq 1 sup-norm-higher}
\left|L^J\left(h^{\mu\nu} \del_\mu\del_\nu h_{\alpha \beta} \right) \right| \leq C(C_1\vep)^2t^{-2 + C(C_1\vep)^{1/2}}(t-r)^{-1+ C(C_1\vep)^{1/2}}.
\ee
\end{lemma}

\begin{proof}
The proof is is parallel to that of Lemma \ref{lem 0.5 refined-sup}. The only difference is that there we only have refined decay estimates on $\del^IL^J\del_t\del_th_{\alpha \beta}$ and $L^J\hu^{00}$ for $|I|+|J| \leq 7$ but here we have, in view of \eqref{ineq 1 lem 2 app refined-energy higher} and \eqref{ineq 1' lem 2 app refined-energy higher}, the parallel estimate for $|I|+|J| \leq N-3$.
\end{proof}

\begin{lemma} \label{lem 2 sup-norm-higher}
For $|I|+|J| \leq N-4$, the following estimate holds:
\bel{ineq 1 lem 2 sup-norm-higher}
\aligned
&\left|[\del^IL^J,h^{\mu\nu} \del_\nu\del_\nu]\phi\right|
 \leq  C(C_1\vep)^2(s/t)^{3}s^{-3+2\delta}
 + CC_1\vep(s/t)^{3/2}s^{-3/2 +\delta} \sum_{|I_2| \leq |I|-1\atop |J_2| \leq|J|} \left|\del_t\del_t\del^{I_2}L^{J_2} \phi\right|
\\
& \quad +  CC_1\vep t^{-1}s^{C(C_1\vep)^{1/2}} \sum_{|J'|<|J|} \left|\del_t\del_t\del^IL^{J'} \phi\right|
\\
& \quad +  CC_1\vep (s/t)^{1-3\delta}s^{-3/2 + C(C_1\vep)^{1/2}} \sum_{|J'|<|J|, \atop \alpha, \beta} \left|L^{J'}h_{\alpha \beta} \right|
  +  CC_1\vep (s/t)^{1-3\delta}s^{-3/2} \sum_{\alpha, \beta} \left|L^Jh_{\alpha \beta} \right|
\endaligned
\ee
and, when $|J|=0$,
\bel{ineq 2 lem 2 sup-norm-higher}
\left|[\del^I,h^{\mu\nu} \del_\nu\del_\nu]\phi\right| \leq C(C_1\vep)^2(s/t)^3s^{-3+2\delta}
 + CC_1\vep(s/t)^{3/2}s^{-3/2 +\delta} \sum_{|I_2| \leq |I|-1} \left|\del_t\del_t\del^{I_2} \phi\right|.
\ee
\end{lemma}

\begin{proof}
The proof relies on the decomposition presented in \eqref{eq 2 lem 1 nolinear} combined with the refined decay estimates on $\del h$, $\phi$ and $\del\phi$ presented in Section \ref{subsec energy-high-app}. We see that the null terms and the terms of commutators listed in $\eqref{eq 2 lem 1 nolinear}$ are bounded by trivial application of the refined decay estimates presented in Section \ref{subsec energy-high-app}. We only write the estimate on the null term $\del^{I_1}L^{J_1} \hu^{a0} \del^{I_2}L^{J_2} \delu_a \del_t\phi$ (and omit the treatement of the other terms). We see that $\hu^{a0}$ is a linear combination of $h_{\alpha \beta}$ with smooth and homogeneous coefficients plus higher-order correction terms:

\noindent {\sl Case 1.} When $|I_1|\geq 1$, we apply the basic sup-norm estimates \eqref{ineq basic-sup 1 generation 1 a} and \eqref{ineq 2 homo} : 
$$
\aligned
\left|\del^{I_1}L^{J_1} \hu^{a0} \del^{I_2}L^{J_2} \delu_a \del_t\phi\right|
& \leq  CC_1\vep t^{-1/2}s^{-1+\delta}  \,   CC_1\vep t^{-3/2}s^{1/2 + \delta}
  \leq  C(C_1\vep)^2(s/t)^2s^{-5/2 +2\delta}.
\endaligned
$$

\noindent{\sl Case 2.} When $|I_1|=0$, we apply \eqref{ineq basic-sup-h} and \eqref{ineq 2 homo} : 
$$
\aligned
&\left|\del^{I_1}L^{J_1} \hu^{a0} \del^{I_2}L^{J_2} \delu_a \del_t\phi\right|
 =  \left|L^{J_1} \hu^{a0} \del^{I}L^{J_2} \delu_a \del_t\phi\right|
\\
& \leq  CC_1\vep \left((s/t)t^{-1/2}s^{\delta} + t^{-1} \right) CC_1\vep t^{-5/2}s^{1/2 +\delta}
  \leq C(C_1\vep)^2(s/t)^4s^{-5/2 +2\delta}.
\endaligned
$$ 
We then focus on the estimates of the last three terms.

$\bullet$ We treat first the term $\del^{I_1}L^{J_1} \hu^{00} \del^{I_2}L^{J_2} \del_t\del_t\phi$ with $|I_1|\geq 1$. We apply the sharp estimate to $\del^{I_1}L^{J_1} \hu^{00}$ provided by \eqref{ineq h00-sup 1} : 
$$
\left|\del^{I_1}L^{J_1} \hu^{00} \del^{I_2}L^{J_2} \del_t\del_t\phi\right|
 \leq CC_1\vep (s/t)^{3/2}s^{-3/2 +\delta} \sum_{|I_2| \leq |I|\atop |J_2| \leq|J|} \left|\del^{I_2}L^{J_2} \del_t\del_t\phi\right|. 
$$
By the commutator estimate \eqref{ineq com 2.4}, we have
$$
\left|\del^{I_2}L^{J_2} \del_t\del_t\phi\right| \leq C\sum_{|J_2'| \leq |J_2|} \left|\del_{\gamma} \del_{\gamma'} \del^IL^{J_2'} \phi\right|.
$$
Then we rely on the decomposition \eqref{eq second-order-frame-1} and a similar argument and obtain
$$
\left|\del_{\gamma} \del_{\gamma'} \del^IL^{J_2'} \phi\right|
\leq \left|\del_t\del_t\del^IL^{J_2'} \phi\right| + CC_1\vep t^{-5/2}s^{1/2 +\delta}, 
$$
so that 
$$
\aligned
& \left|\del^{I_1}L^{J_1} \hu^{00} \del^{I_2}L^{J_2} \del_t\del_t\phi\right|
\\
& \leq CC_1\vep(s/t)^{3/2}s^{-3/2 +\delta} \sum_{|I_2| \leq |I|-1\atop |J_2| \leq|J|} \left|\del_t\del_t\del^{I_2}L^{J_2} \phi\right|
+ C(C_1\vep)^2(s/t)^4s^{-7/2 +2\delta}.
\endaligned
$$ 

$\bullet$ The term $L^{J_1'} \hu^{00} \del^IL^{J_2'} \phi$ is bounded as follows. We see that $|J_2'|<|J|$ and we will discuss the following cases:

\noindent {\sl Case 1.} When $1\leq |J_1'| \leq N-7$, we apply \eqref{ineq 2 prop 2 refined-sup} :
$$
\left|L^{J_1'} \hu^{00} \del^IL^{J_2'} \del_t\del_t\phi\right|
\leq CC_1\vep t^{-1}s^{C(C_1\vep)^{1/2}} \,   CC_1\vep \sum_{|J'|<|J|} \left|\del^IL^{J'} \del_t\del_t\phi\right|.
$$
Apply the same estimate for $\left|\del^IL^{J'} \del_t\del_t\phi\right|$ as above, we conclude that
$$
\left|L^{J_1'} \hu^{00} \del^IL^{J_2'} \del_t\del_t\phi\right|
\leq  CC_1\vep t^{-1}s^{C(C_1\vep)^{1/2}} \sum_{|J'|<|J|} \left|\del_t\del_t\del^IL^{J'} \phi\right|
 + C(C_1\vep)^2(s/t)^{7/2}s^{-3+ C(C_1\vep)^{1/2}}.
$$

\noindent{\sl Case 2.} When $N-6\leq |J_1'| \leq |J|-1$, we have $|I|+|J_2'| \leq 2\leq N-8$, then we apply the last inequality of \eqref{ineq 3.2 refined-sup low d} to $\del^IL^{J_2'} \del_t\del_t\phi$: 
$$
\left|L^{J_1'} \hu^{00} \del^IL^{J_2'} \del_t\del_t\phi\right|
\leq CC_1\vep (s/t)^{1-3\delta}s^{-3/2 + C(C_1\vep)^{1/2}} \sum_{|J'|<|J|, \atop \alpha, \beta} \left|L^{J'}h_{\alpha \beta} \right|.
$$
 
\noindent{\sl Case 3.} When $N-6\leq |J_1'|$ and $J_1'=J$, we have $|I| \leq 2\leq N-8$ and $|J_2'|=0$. We apply \eqref{ineq 3.2 refined-sup low a} :  
$$
\aligned
\left|L^{J_1'} \hu^{00} \del^IL^{J_2'} \del_t\del_t\phi\right| & =  \left|L^J\hu^{00} \del^I\del_t\del_t\phi\right|
\leq CC_1\vep (s/t)^{1-3\delta}s^{-3/2} \sum_{\alpha, \beta} \left|L^Jh_{\alpha \beta} \right|.
\endaligned
$$
The term $\hu^{00} \del_{\gamma} \del_{\gamma'} \del^IL^{J'} \phi$ is bounded by
$$
CC_1\vep t^{-1}s^{C(C_1\vep)^{1/2}} \sum_{|J'|<|J|} \left|\del_t\del_t\del^IL^{J'} \phi\right| + C(C_1\vep)^2(s/t)^{7/2}s^{-3+ C(C_1\vep)^{1/2}}.
$$
We omit the details of the proof which are essentially the same as in {\sl Case 1} for $\del^{I_1}L^{J_1} \hu^{00} \del^{I_2}L^{J_2} \phi$. Therefore, we have established \eqref{ineq 1 lem 2 sup-norm-higher}.

For \eqref{ineq 2 lem 2 sup-norm-higher}, when $|J|=0$, the third and fourth terms in the right-hand side of \eqref{ineq 1 lem 2 sup-norm-higher} disappear. The last term also disappear since, if we follow the proof of \eqref{ineq 1 lem 2 sup-norm-higher}, we see that when $|J|=0$, and the {\sl case 3} of $L^{J_1'} \hu^{00} \del^IL^{J_2'} \phi$ does not exist ($N-6\leq J_1'$ and $J_1=J$ is contradictory). This is the only place that the last term in the right-hand side of \eqref{ineq 1 lem 2 sup-norm-higher} appears. Therefore, we have established \eqref{ineq 2 lem 2 sup-norm-higher}.
\end{proof} 


\subsection{Main estimate in this section}  

\begin{proposition} \label{prop 1 refined-sup-higher}
There exist constants $C_1, \vep_4 >0$ such that if the bootstrap condition \eqref{ineq energy assumption 1}-\eqref{ineq energy assumption 2} holds with $C_1>C_0$ sufficiently large, then there exists a constant $\vep_4>0$ such that for any $\vep\in (0, \vep_4)$ and $N-6\leq |I|+|J| \leq N-4$:
\bel{ineq 1 prop 1 refined-sup-higher}
\left|L^Jh_{\alpha \beta} \right| \leq CC_1\vep t^{-1}s^{C(C_1\vep)^{1/2}},
\ee
\bel{ineq 2 prop 1 refined-sup-higher}
(s/t)^{3\delta-2}|\del^IL^J\phi| + (s/t)^{3\delta - 3}|\del^IL^J\delu_{\perp} \phi| \leq CC_1\vep t^{-3/2}s^{C(C_1\vep)^{1/2}}.
\ee
\end{proposition}

The proof is divided into two parts and we analyze first the case $|J|=0$.

\begin{proof}[Proof of Proposition \ref{prop 1 refined-sup-higher} in the case $|J|=0$]
We see that \eqref{ineq 1 prop 1 refined-sup-higher} is already guaranteed by \eqref{ineq 2 lem 3 refined-sup}. To establish \eqref{ineq 2 prop 1 refined-sup-higher}, we rely on Proposition \ref{Linfini KG} and follow the notation therein. The terms $R_i$ are already bounded by Lemma \ref{lem 2 refined-sup}, while the commutator term $[\del^I,h^{\mu\nu} \del_\mu\del_\nu]\phi$ is bounded in view of \eqref{ineq 2 lem 2 sup-norm-higher}. Hence, we have (always with $s=\sqrt{t^2 - r^2}$) 
$$
\aligned
F(t, x)& \leq  CC_1\vep (s/t)^{3/2} \int_{s_0}^s \tau^{-3/2 +3\delta}d\tau
+ C(C_1\vep)^2(s/t)^3\int_{s_0}^s\tau^{-3+2\delta} \,  \tau^{3/2}d\tau
\\
& \quad + CC_1\vep(s/t)^{3/2} \sum_{|I'| \leq|I|-1} \int_{s_0}^s\lambda^{\delta} \left|\del^{I'} \del_t\del_t \phi\right|(\lambda t/s, \lambda x/s)d\lambda
\endaligned
$$
so
$$
\aligned
F(t, x)
& \leq  CC_1\vep(s/t)^{3/2}s_0^{-1/2 +3\delta} + C(C_1\vep)^{2}(s/t)^3
\\
& \quad 
+ CC_1\vep(s/t)^{3/2} \sum_{|I'| \leq|I|-1}
\int_{s_0}^s\lambda^{\delta} \left|\del^{I'} \del_t\del_t \phi\right|(\lambda t/s, \lambda x/s)d\lambda
\\
& \leq  CC_1\vep(s/t)^{2-3\delta}
+ CC_1\vep(s/t)^{3/2} \sum_{|I'| \leq|I|-1}
\int_{s_0}^s\lambda^{\delta} \left|\del^{I'} \del_t\del_t \phi\right|(\lambda t/s, \lambda x/s)d\lambda,
\endaligned
$$
where we recall that $s_0 \simeq \frac{t}{s}$.

Setting 
$$
X_n(\tau) := \sum_{|I| \leq n} \sup_{\Kcal_{[2,\tau]}} \Big(
(s/t)^{3\delta-2}s^{3/2} \left|\del^I\phi\right| + (s/t)^{3\delta-3}s^{3/2} \left|\delu_{\perp} \del^I\phi\right|
\Big)(t,x),
$$ 
we claim that
\bel{eq:explain} 
\left|(s/t)^{3\delta-1} \del^{I'} \del_t\del_t \phi\right| (t, x) 
\leq Cs^{-3/2}X_n(s)
+ C t^{-1} \eps (s/t)^{3\delta - 1/2} s^{-1/2+\delta}, 
\ee
which will be explained at the end of this proof. Replacing $t$ by $\lambda t/s$ and integrating in $\lambda$, 
we then obtain 
\be
\aligned
& F(t, x)
\\
& \leq  C(C_1\vep)(s/t)^{2-3\delta} + CC_1\vep (s/t)^{5/2-3\delta} \int_{s_0}^s 
\Big( \lambda^{-3/2 +\delta} X_n(\lambda) 
+
  \eps (s/t)^{3\delta + 1/2} \lambda^{-3/2+2\delta}  
 \Big) \, d\lambda
\\
& \leq C(C_1\vep)(s/t)^{2-3\delta} + CC_1\vep (s/t)^{5/2-3\delta} 
\Big( 
X_n(s) \int_{s_0}^s \lambda^{-3/2 +\delta}d\lambda
+ \eps (s/t)^{3\delta + 1/2}
 \int_{s_0}^s \lambda^{-3/2 + 2\delta}   \, d\lambda
\Big) 
\\
& \leq C(C_1\vep)(s/t)^{2-3\delta} + CC_1\vep(s/t)^{3-4\delta}X_n(s) 
+ CC_1\vep^2 (s/t)^{7/2-2\delta}, 
\endaligned
\ee 
where we used that $X_n(\cdot)$ is non-decreasing and $s_0 \simeq \frac{t}{s}$.
Also, recall that \eqref{ineq lem h00-sup 4 1} gives the desired bound for $h'_{t, x}$ and, therefore, by Proposition \ref{Linfini KG} we deduce that 
$$
(s/t)^{3\delta-2}s^{3/2} \left|\del^I\phi\right| + (s/t)^{3-3\delta}s^{3/2} \left|\delu_{\perp} \del^I\phi\right| \leq CC_0 \, \vep + CC_1\vep + CC_1\vep X_n(s).
$$
Taking the sup-norm of the above inequality in $\Kcal_{[2,s]}$, we obtain
$
X_n(s) \leq CC_0 \, \vep  + CC_1\vep + CC_1\vep X_n(s).
$
Then, if we take in the bootstrap assumption that $\vep_0'$ sufficiently small so that  $CC_1\vep\leq 1/2$ for $0\leq \vep\leq \vep_0'$, we have
$
X_n(s) \leq CC_0 \, \vep  + CC_1\vep\leq CC_1\vep,
$
which is the desired result (since $C_1\geq C_0$).

It remains to derive \eqref{eq:explain} and, with the notation above, we write at any $(t,x)$ 
$$
\aligned
\big| \del^{I'} \del_t\del_t \phi \big| 
= \Big| (t/s)^2 \big( \delu_{\perp}  - (x^a/t) \delu_a \big) \del^{I'} \del_t \phi \big|
& \leq 
(t/s)^2 \big| \delu_{\perp} \del^{I'} \del_t \phi \big|
+(t/s)^2 \big| (x^a/t) \delu_a \del^{I'} \del_t \phi \big|
\\
&
\leq 
(s/t)^{1-3\delta} s^{-3/2} X_n(s) + (t/s)^2 t^{-1} \sum_a \big| L_a \del^{I'} \del_t \phi \big|, 
\endaligned
$$
in which we used the definition of $X_n$ and, on the other hand, the fact that $\del^{I'}$ is of order $|I|-1$ at most. 
Recalling  \eqref{ineq basic-sup 2 generation 1 b} (together with the commutator estimates), we obtain 
$$
\sum_a \big| L_a \del^{I'} \del_t \phi \big| \leq C C_1 \eps t^{-5/2} s^{1/2 + \delta}
=
 C C_1 \eps (s/t)^{5/2} s^{-2+\delta}, 
$$
which leads us to $\aligned
\big| \del^{I'} \del_t\del_t \phi \big|  
\leq 
(s/t)^{1-3\delta} s^{-3/2} X_n(s) + t^{-1} C C_1 \eps (s/t)^{1/2} s^{-2+\delta}.
\endaligned
$
\end{proof}

Before we can proceed with the proof of Proposition \ref{prop 1 refined-sup-higher} in the case $|J|\geq 1$, we need to establish the following result.

\begin{lemma} \label{lem 3 sup-norm-higher} 
For $|I|+|J| \leq N-4$, one has 
\be
\aligned
& \left|\del^IL^J\left(\del_{\alpha} \phi\del_{\beta} \phi\right) \right| + \left|\del^IL^J\left(\phi^2\right) \right|
  \leq  CC_1\vep (s/t)^{2-3\delta}s^{-3/2}  \sum_{|I'| \leq |I|\atop \gamma} \big|\del^{I'}L^J\del_\gamma \phi\big| + |\del^{I'}L^J \phi| 
\\
&\quad + CC_1\vep (s/t)s^{2-3\delta}s^{-3/2 + C(C_1\vep)^{1/2}}
 \sum_{|I'| \leq|I|,|J'|<|J|\atop \gamma} \left|\del^{I'}L^{J'} \del_\gamma \phi\right| + |\del^{I'}L^{J'} \phi|. 
\endaligned
\ee
\end{lemma}

\begin{proof} We only consider $\del_{\alpha} \phi\del_\beta \phi$, by relying on \eqref{ineq 2 prop 1 refined-sup-higher} in the case $|J|=0$. Observe that
$$
\left|\del^IL^J\left(\del_{\alpha} \phi\del_\beta \phi\right) \right|
\leq \sum_{I_1+I_2=I\atop J_1+J_2=J} \left|\del^{I_1}L^{J_1} \del_{\alpha} \phi\right| \,  \left|\del^{I_2}L^{J_2} \del_{\beta} \phi\right|.
$$
When $J_1 = 0$ or $ J_2 = 0$, thanks to \eqref{ineq 1 lem 3 refined-sup},
$$
\aligned
\left|\del^{I_1}L^{J_1} \del_{\alpha} \phi\right| \,  \left|\del^{I_2}L^{J_2} \del_{\beta} \phi\right|
& \leq  CC_1\vep (s/t)^{2-3\delta}s^{-3/2} \sum_{\gamma} \left|\del^IL^J\del_{\gamma} \phi\right|. 
\endaligned
$$
When $1\leq |J_1|$ or $1\leq |J_2|$ we see that $|J_2|<|J|$ and $|J_1|<|J|$ and it remains to apply \eqref{ineq 1 prop 2 refined-sup}.
\end{proof}

\begin{proof}[Proof of Proposition \ref{prop 1 refined-sup-higher} in the case $|J|\geq 1$]
We proceed by induction and with the help of a secondary  bootstrap argument (as in the proof of Proposition \ref{prop 2 refined-sup}). We will not rewrite the argument in full details, but only provide the key steps. Suppose that on the interval $[2,s^*]$ there exist positive constants $K_{m-1},C_{m-1}, \vep'_{m-1}$ (depending only on the structure of the main system and $N$) such that 
\bel{ineq pr1 prop 1 refined-sup-higher}
(s/t)^{3\delta-2}s^{3/2} \left|\del^IL^J\phi\right| + (s/t)^{3\delta-3}s^{3/2} \left|\delu_{\perp} \del^IL^J\phi\right|
\leq K_{m-1}C_1\vep s^{C_{m-1}(C_1\vep)^{1/2}}, 
\ee 
\bel{ineq pr2 prop 1 refined-sup-higher}
t\left|L^Jh_{\alpha \beta} \right| \leq K_{m-1}C_1\vep s^{C_{m-1}(C_1\vep)^{1/2}}
\ee
for $0\leq \vep \leq \vep'_{m-1}$ and $|I|+|J| \leq N-4$ and $|J| \leq m-1<N-4$. We will prove that there exist positive constants $K_m,C_m, \vep'_m$ (determined by the structure of the main system and the integer $N$) such that the following inequaities 
hold for $0\leq \vep\leq \vep_m'$: 
\bel{ineq pr3 prop 1 refined-sup-higher}
(s/t)^{3\delta-2}s^{3/2} \left|\del^IL^J\phi\right| + (s/t)^{3\delta-3}s^{3/2} \left|\delu_{\perp} \del^IL^J\phi\right| \leq K_mC_1\vep s^{C_m(C_1\vep)^{1/2}}, 
\ee 
\bel{ineq pr4 prop 1 refined-sup-higher}
t\left|L^Jh_{\alpha \beta} \right| \leq K_mC_1\vep s^{C_m(C_1\vep)^{1/2}}.
\ee

We begin the formulation of the secondary bootstrap argument and  set 
$$
s^{**} := \sup_{s\in[2,s^*]} \{s|\eqref{ineq pr3 prop 1 refined-sup-higher} \text{ and } \eqref{ineq pr4 prop 1 refined-sup-higher} \text{ hold in } \Kcal_{[2,s^*]} \}.
$$
Suppose the $K_m$ that we have taken is sufficiently large such that $s^{**}>2$ and $C_m=2C_{m-1}$ (see the argument in the proof of Proposition \ref{prop 2 refined-sup}.)

We substitute the assumptions \eqref{ineq pr1 prop 1 refined-sup-higher}, \eqref{ineq pr2 prop 1 refined-sup-higher}, \eqref{ineq pr3 prop 1 refined-sup-higher} and \eqref{ineq pr4 prop 1 refined-sup-higher} into \eqref{lem 2 sup-norm-higher}. This gives
\bel{ineq pr5 prop 1 refined-sup-higher}
\left|[\del^IL^J,h^{\mu\nu} \del_{\mu} \del_{\nu}]\phi\right| \leq C(C_1\vep)^2(s/t)^3 s^{-3+3\delta}
 + CK_m^2(C_1\vep)^2(s/t)^{2-3\delta}s^{-5/2 + C_m(C_1\vep)^{1/2}}.
\ee

With the notation in Proposition \ref{Linfini KG} (recalling that $h'_{t, x}$ is bounded in view of \eqref{ineq lem h00-sup 4 1} and $R_i$ are bounded by Lemma \ref{lem 2 refined-sup}), we obtain
$$
\left|F(s) \right| \leq CC_1\vep (s/t)^{3/2}s_0^{-1/2 +3\delta} + CC_m^{-1}K_m^2(C_1\vep)^{3/2}(s/t)^{2-3\delta}s^{C_m(C_1\vep)^{1/2}}.
$$
Then in view of \eqref{Linfty KG ineq}, we have
$$
\aligned
&
(s/t)^{3\delta-2}s^{3/2} \left|\del^IL^J\phi\right| + (s/t)^{3\delta-3}s^{3/2} \left|\delu_{\perp} \del^IL^J\phi\right|
\\
&
\leq CK_{0,m}C_1\vep  + CC_1\vep +  CC_m^{-1}K_m^2(C_1\vep)^{3/2}s^{C_m(C_1\vep)^{1/2}}.
\endaligned
$$
Then, as in the proof of Proposition \ref{prop 2 refined-sup}, we choose
$
\vep_m' = \frac{C_m^2}{C_1} \left(\frac{K_m-2CK_{0,m}-2C}{2CK_m^2} \right)^2. 
$
Then, for $0\leq \vep\leq\vep_m'$, we have 
$$
(s/t)^{3\delta-2}s^{3/2} \left|\del^IL^J\phi\right| + (s/t)^{3\delta-3}s^{3/2} \left|\delu_{\perp} \del^IL^J\phi\right|
\leq \frac{1}{2}K_mC_1\vep s^{C(C_1\vep)^{1/2}}.
$$

The estimate for $L^Jh_{\alpha \beta}$ is checked
as the argument in the proof of Proposition \ref{prop 2 refined-sup}. We omit the details and point out the estimates on ${\sl QS}_{\phi}$ is covered by Lemma \ref{lem 3 sup-norm-higher} and the induction-bootstrap assumption \eqref{ineq pr1 prop 1 refined-sup-higher}, \eqref{ineq pr2 prop 1 refined-sup-higher}, \eqref{ineq pr3 prop 1 refined-sup-higher} and \eqref{ineq pr4 prop 1 refined-sup-higher}.
Other nonlinear terms such as $F_{\alpha \beta}$ and $h^{\mu\nu} \del_\mu\del_\nu h_{\alpha\beta}$ are bounded in view of \eqref{ineq 3 app reined-energy high} and \eqref{ineq 1 sup-norm-higher}. The same argument as in the proof of Proposition \ref{prop 2 refined-sup} leads us to the desired result with $\vep_4 = \min(\vep'_m, \vep_0')$, where $\vep_0'$ was determined at the end of the proof for $|J|=0$.
\end{proof}


\section{Low-Order Refined Energy Estimate for the Scalar Field} \label{section-13}

It remains to establish the refined energy estimate in order to complete the proof of our main result.

\begin{proposition} \label{prop refined-energy-low}
Let $|I|+|J| \leq N-4$ and suppose that the bootstrap assumptions \eqref{ineq energy assumption 1} \eqref{ineq energy assumption 2} hold for  $C_1$ sufficiently large, then there exists some $\vep_5>0$ such that for all $0\leq \vep\leq \vep_5$;
\bel{ineq 1 prop refined-energy-low}
E_{M,c^2}(s, \del^IL^J\phi)^{1/2} \leq \frac{1}{2}C_1\vep s^{C(C_1\vep)^{1/2}}.
\ee
\end{proposition}

\begin{proof} Our argument now relies on the energy estimate in Proposition \ref{prop energy 2KG}, in which the coercivity condition \eqref{eq energy 1} is guaranteed by Lemma \ref{lem h00-sup 2}. The estimate for $M[\del^IL^J\phi]$ is provided by  \eqref{ineq lem h00-sup 3.5 b}. So the only issue still to be discussed is the estimate of the commutator $\left\|[\del^IL^J,h^{\mu\nu} \del_\mu\del_\nu]\phi\right\|_{L^2(\Hcal_s^*)}$.
Here, we use \eqref{eq 2 lem 1 nolinear} and, in view of \eqref{ineq 3 bilinear-L2}, obtain 
$$
\|{\sl GQQ}_{h\phi}(N-4,k) \|_{L_f^2(\Hcal_s)} \leq C(C_1\vep)^2s^{-3/2 +2\delta}.
$$
For $t^{-1} \del^{I_3}L^{J_3}h_{\alpha'\beta'} \del^{I_4}L^{J_4} \del_{\gamma} \phi$, we have 
$$
\aligned
\|t^{-1} \del^{I_3}L^{J_3}h_{\alpha'\beta'} \del^{I_4}L^{J_4} \del_{\gamma} \phi\|_{L_f^2(\Hcal_s)} 
& \leq \left\|t^{-1}(t^{-1}+(s/t)t^{-1/2}s^{\delta}) \del^{I_4}L^{J_4} \del_{\gamma} \phi\right\|_{L_f^2(\Hcal_s)} 
\\
& \leq C(C_1\vep)^2s^{-3/2 +2\delta},
\endaligned
$$
while the term $\del^{I_1}L^{J_1} \hu^{00} \del^{I_2}L^{J_2} \del_t\del_t\phi$ is bounded by applying \eqref{ineq h00-sup 1} :
$$
\aligned
\|\del^{I_1}L^{J_1} \hu^{00} \del^{I_2}L^{J_2} \del_t\del_t\phi\|_{L^2(\Hcal_s^*)}
& \leq  CC_1\vep s^{-3/2 +\delta} \|(s/t)^{3/2} \del^{I_2}L^{J_2} \del_t\del_t\phi\|_{L^2(\Hcal_s^*)}
  \leq  CC_1\vep s^{-3/2 +2\delta}.
\endaligned
$$
The term $L^{J'_1} \hu^{00} \del^{I}L^{J'_2} \del_t\del_t\phi$ is bounded by applying \eqref{ineq 1 prop 1 refined-sup-higher} and observing that $|J'_1|>0$:
$$
\aligned
\big\|L^{J'_1} \hu^{00} \del^IL^{J'_2} \del_t\del_t\phi\big\|_{L_f^2(\Hcal_s)}
& \leq  CC_1\vep\big\|t^{-1}s^{C(C_1\vep)^{1/2}} \del^IL^{J'_2} \del_t\del_t\phi\big\|_{L_f^2(\Hcal_s)}
\\
& \leq  CC_1\vep s^{-1+ C(C_1\vep)^{1/2}} \big\|(s/t) \del^IL^{J'_2} \del_t\del_t\phi \big\|_{L_f^2(\Hcal_s)}
\\
& \leq  CC_1\vep s^{-1+ C(C_1\vep)^{1/2}} \sum_{|J'|<|J|}E_{M,c^2}(s, \del^IL^{J'} \phi)^{1/2}.
\endaligned
$$
And for the term $\hu^{00} \del_{\alpha} \del_{\beta}$, we apply \eqref{ineq 2 lem 3 refined-sup} :
$$
\big\|\hu^{00} \del_{\alpha} \del_{\beta} \del^IL^{J'} \big\|_{L_f^2(\Hcal_s)}
\leq CC_1\vep s^{-1} \sum_{|J'|<|J|}E_{M,c^2}(\del^IL^{J'} \phi)^{1/2}, 
$$
so that 
$
\left\|[\del^IL^J,h^{\mu\nu} \del_\mu\del_\nu]\phi\right\|_{L^2(\Hcal_s^*)}
\leq CC_1\vep s^{-1+ C(C_1\vep)^{1/2}} \sum_{|J'|<|J|}E_{M,c^2}(s, \del^IL^{J'} \phi)^{1/2}.
$
So by Proposition \ref{prop energy 2KG}, we have 
\bel{ineq pr0 prop refined-energy-low}
\aligned
E_{M,c^2}(s, \del^IL^J\phi)^{1/2} & \leq  C_0 \, \vep + C(C_1\vep)^2\int_{2}^s \tau^{-3/2 +2\delta}d\tau
\\
& \quad + CC_1\vep\sum_{|J'|<|J|} \int_2^s \tau^{-1+ C(C_1\vep)^{1/2}}E_{M,c^2}(\tau, \del^IL^{J'} \phi)^{1/2} d\tau. 
\endaligned
\ee
When $|J|=0$, the last term disappears. We have
\bel{ineq pr1 prop refined-energy-low}
E_{M,c^2}(s, \del^I\phi)^{1/2} \leq CC_0 \, \vep  + C(C_1\vep)^2.
\ee

We are going to prove that for all $|I|+|J| \leq N-4$,
\bel{ineq pr2 prop refined-energy-low}
E_{M,c^2}(s, \del^IL^J\phi)^{1/2} \leq CC_0 \, \vep + C(C_1\vep)^{3/2}s^{C(C_1\vep)^{1/2}}.
\ee
When $|J|\geq 1$, we proceed by induction on $|J|$ and see that \eqref{ineq pr2 prop refined-energy-low} is guaranteed by \eqref{ineq pr1 prop refined-energy-low} ($C_1\vep$ smaller that $1$). Assume that \eqref{ineq pr2 prop refined-energy-low} holds for $|J| \leq m-1< n-4$, we will prove it for $|J|=m\leq N-4$. We directly apply the induction assumption in \eqref{ineq pr0 prop refined-energy-low} and conclude that
$
E_{M,c^2}(s, \del^IL^J\phi)^{1/2} \leq CC_0 \, \vep + C(C_1\vep)^{3/2}s^{C(C_1\vep)^{1/2}}
$
for $|I|+|J| \leq N-4$ and, by taking
$
\vep_5 = \left(\frac{C_1-2CC_0}{2CC_1^{3/2}} \right)^2,
$
the desired result is proven.
\end{proof}

In conclusion, in view of \eqref{ineq 1 prop 1 refined-energy-low-h}, \eqref{ineq 1 prop 1 refined-energy higher}, \eqref{ineq 2 prop 1 refined-energy higher} and \eqref{ineq 1 prop refined-energy-low}, if the bootstrap assumption holds for $C_1>C_0$ sufficiently large, then there exists some $\vep_0:=\min\{\vep_1\, \vep_2, \vep_3, \vep_4, \vep_5\}$ such that
$$
\aligned
E_M(s, \del^IL^J h_{\alpha \beta})^{1/2} & \leq  \frac{1}{2}C_1\vep s^{C(C_1\vep)^{1/2}}, \quad &&|I|+|J| \leq N,
\\
E_M(s, \del^IL^J\phi)^{1/2} & \leq  \frac{1}{2}C_1\vep s^{1/2 + C(C_1\vep)^{1/2}}, \quad &&N-3\leq |I|+|J| \leq N,
\\
E_M(s, \del^IL^J\phi)^{1/2} & \leq  \frac{1}{2}C_1\vep s^{C(C_1\vep)^{1/2}}, \quad&& |I|+|J| \leq N-4.
\endaligned
$$
This improves the bootstrap assumption \eqref{ineq energy assumption 1}--\eqref{ineq energy assumption 2}. We see that \eqref{ineq energy assumption 1}--\eqref{ineq energy assumption 2} hold on the time interval where the solution exists. In view of the local existence theory for the hyperboloidal foliation (see the last chapter of \cite{PLF-YM-book}) the global existence result is thus established.


\appendix 

\section{Revisiting the wave-Klein-Gordon model}

A {\sl wave-Klein-Gordon model} was ``extracted'' from the Einstein equations by the authors in \cite{PLF-YM-CRAS,PLF-YM-one} when they were beginning to analyze the Einstein equations via the Hyperboloidal Foliation Method introduced in \cite{PLF-YM-book}. This model\footnote{A.D. Ionescu and B. Pausader recently further investigated our model via Fourier techniques; see ArXiv:1703.02846.}
 provided to the authors a simple, yet highly not trivial, example of coupling between a wave equation and a Klein-Gordon equation, before developing the method for the full Einstein system, as we do in the present monograph.
We revisit here the proof of existence in \cite{PLF-YM-one} since our presentation missed one bootstrap condition in the list (5.1) which however turns out to be necessary for dealing with the (comparatively easier) wave component when $k=0$ in (5.1).  

When $k=0$, the first bound in (5.1) in \cite{PLF-YM-one} should be weakened to
\bel{ineq energy assumption 0}
E_m(s,\del^I u)^{1/2}\leq C_1\vep s^{\delta},
\qquad |I|\leq N,
\ee
while a similar remark applies to (5.2). Doing so has no effect on the derivation of the sup-norm bounds (in Section 6.2, on which Section 7 is based), since in the application of the Klainerman-Sobolev inequality one uses one boost at least, and the additional growth allowed by \eqref{ineq energy assumption 0} is negligible. Note in passing also that, in Section 6.5 of \cite{PLF-YM-one}, the Hardy-based estimate  (6.20a) is valid for $k= |J| \geq 1$ only, while we already pointed out in \cite{PLF-YM-one}  the next inequality (6.20b) is never used.

In Lemma 8.1, the estimate (8.4) can be improved to 
\bel{eq 1 21-06-2017}
M(s)\lesssim C_1\vep s^{-3/2+k\delta}.
\ee
which is checked for $|I|+|J|\leq N-1$ by writing 
$$
\aligned
\int_{\Hcal_s}\big|\del_{\gamma}h^{\alpha\beta}\del_\alpha\del^IL^Jv\big|^2dx
& 
\lesssim C_1\vep \int_{\Hcal_s}t^{-1}s^{-2} \,  \big|\del_\alpha\del^IL^J v \big|^2dx
\lesssim C_1\vep s^{-3}\int_{\Hcal_s}\big|\del_\alpha\del^IL^J v \big|^2dx
\\
& \lesssim C_1\vep s^{-3}\sum_{\alpha}E_{g,c}(s,\del_{\alpha}\del^IL^Jv). 
\endaligned
$$
In Lemma 8.2, when $k=0$ (8.6) can be improved to
\bel{ineq lem energy lower 2'}
\|[H^{\alpha\beta}u\del_\alpha\del_\beta,\del^I ]v\|_{L_f^2(\Hcal_s)}\lesssim (C_1\vep)^2 s^{-3/2+2\delta},
\qquad |I|\leq N.
\ee
Namely, only the term $\del^{I_1}L^{J_1}u\del^{I_2}L^{J_2}\del_\alpha\del_\beta v$ 
with $|I_1|=1$ and $J_1=0$ need to be considered:
\be
\aligned
\|\del_{\gamma}u\del^{I_2}L^J\del_\alpha\del_\beta v\|_{L_f^2(\Hcal_s)}
\leq& \|(s/t) \del_\gamma u\|_{L^2(\Hcal_s)}\, \|(t/s)\del^{I_2}L^J\del_\alpha\del_\beta v\|_{L^\infty(\Hcal_s)}
\\
\lesssim &(C_1\vep)^2 \|(s/t)^{-\delta}s^{-3/2}\|_{L^{\infty}(\Hcal_s)}\lesssim (C_1\vep)^2s^{-3/2+\delta}.
\endaligned
\ee
In Lemma 8.3, when $k=0$ (8.7) can be improved to
\bel{ineq lem energy higher 0}
\left\|\del^I\left(P^{\alpha\beta}\del_{\alpha}v\del_{\beta}v + Rv^2\right)\right\|\lesssim (C_1\vep)^2 s^{-1+\delta},
\qquad |I|\leq N.
\ee
Namely,  in $
\del^I\left(\del_{\alpha}v\del_{\beta}v\right) = \sum_{I_1+I_2=I}\del^{I_1}\del_{\alpha}v \,  \del^{I_2}\del_{\beta}v$ we can assume that $|I_1|\leq|I_2|$, 
hence $|I_1|\leq \left[|I| / 2\right]\leq N-5$, and then by (7.23b)  (with $\del^I\del_{\alpha}$ of order $\leq N-4$) and (6.5) (third and last inequalities):
\be
\left\|\del^I\left(\del_{\alpha}v\del_{\beta}v\right)\right\|_{L_f^2(\Hcal_s)}\lesssim (C_1\vep)^2\|(s/t)^{1/2-4\delta}t^{-3/2}(t/s) \  (s/t)\del^{I_2}\del_{\beta}v\|_{L_f^2(\Hcal_s)}\lesssim (C_1\vep)^2s^{-1+\delta}.
\ee
\newcommand \Cbar {\overline C}
In the proof of Proposition 5.1, when $|J|=0$ thanks to \eqref{ineq lem energy higher 0}
\bel{pr2 prop bootstrap 0}
\aligned
E_m(s,\del^Iu)^{1/2}\leq \,& \Cbar C_0\vep +  \Cbar (C_1\vep)^2\int_2^s \sbar^{-1+\delta} \, d\sbar
\leq \Cbar C_0\vep + \Cbar (C_1\vep)^2s^{\delta}, 
\endaligned
\ee
and for (8.14) with $k=0$, one has $|I|\leq N-4$ and we can apply \eqref{eq 1 21-06-2017}-\eqref{ineq lem energy lower 2'}:
\bel{pr3 prop bootstrap'}
\aligned
E_{m,c}(s,\del^I v)^{1/2} \leq \, &  \Cbar C_0\vep + \Cbar (C_1\vep)^2\int_2^s \sbar^{-3/2+k\delta} \, d\sbar
\leq \,  \Cbar C_0\vep + \Cbar (C_1\vep)^2.
\endaligned
\ee


\section{Sup-norm estimate for the wave equations}
\label{sec supnorm}

\begin{proposition} 
\label{Linfini wave}
Let $u$ be a spatially compactly supported solution to the wave equation
\be
\aligned
&-\Box u = f,\\
& u|_{t=2} = 0,\qquad \del_t u|_{t=2} = 0,
\endaligned
\ee
in which $f$ is spatially compactly supported in $\Kcal$  
and satisfies  
\be
|f|\leq C_f t^{-2-\nu}(t-r)^{-1+\mu}
\ee
for some $C_f>0$, $0<\mu\leq 1/2$, and $0< |\nu|\leq 1/2$. 
Then, one has 
\be
\label{Linfini wave ineq}
|u(t,x)|
\lesssim
\begin{cases}
\frac{C_f}{\nu\mu}(t-r)^{\mu-\nu} t^{-1}, \qquad & 0< \nu\leq 1/2,
\\
\frac{C_f}{|\nu|\mu}(t-r)^{\mu} t^{-1 -\nu}, &-1/2\leq \nu < 0.
\end{cases}
\ee
\end{proposition}

We denote by $d\sigma$ the Lebesgue measure on the sphere $\{|y| = 1-\lambda \}$ and $x\in \mathbb{R}^3$ with $r = |x|$, and consider the integral term
$$
I(\lambda)=I(\lambda, t, x/t):=\int_{|y| = 1-\lambda,|\frac{x}{t}-y|\leq \lambda-t^{-1}}
 \frac{d\sigma(y)}{\big(\lambda-\big|\frac{x}{t} - y\big|\big)^{1-\mu}}.
$$
Clearly, when $0< \lambda\leq \frac{t-r+1}{2t}$, we have $I(\lambda) = 0$.

\begin{lemma}\label{wave lemma1}
When $\frac{t-r+1}{2t}\leq \lambda\leq 1$, we obtain 
$$
I(\lambda) \lesssim
\left\{
\aligned
&\frac{\lambda t(1-\lambda)}{\mu r}\left(\frac{t-r}{t}\right)^{\mu},
\qquad
&&\frac{t-r+1}{2t}\leq \lambda\leq \frac{t+r+1}{2t},
\\
&(1-\lambda)\left(\frac{t+r}{t}-\lambda\right)\left(2\lambda - \frac{t+r}{t}\right)^{-1+\mu}, \qquad
&&\frac{t+r+1}{2t}\leq \lambda\leq \frac{t-r}{t},
\\
& &&  \hskip.cm  \text{ provided } \frac{t+r+1}{2t}\leq  \frac{t-r}{t},
\\
&\frac{(1-\lambda)t}{\mu r}\left(\frac{t-r}{t}\right)^\mu,\qquad
&&\max\left(\frac{t-r}{t},\frac{t+r+1}{2t}\right) \leq \lambda\leq 1.
\endaligned
\right.
$$
\end{lemma}

\begin{proof}[Proof of Proposition \ref{Linfini wave}]
From the expression
\bel{eq:form-onde}
u(t,x) = \frac{1}{4\pi}\int_2^t\frac{1}{t-\sbar}\int_{|y| = t-\sbar} f(\sbar,x-y) \, d\sigma d\sbar,
\ee
in which the integration is on the intersection of the cone $\big\{(\sbar,y) \, / \, |y-x| = t-\sbar, 2\leq \sbar\leq t\big\}$ and $\big\{(t,x) \, / \, r<t-1, t^2-r^2\leq s^2, t\geq 2 \big\}$, we obtain
$$
\aligned
|u(t,x)|
&\leq \frac{C_f}{4\pi}\int_2^t\int_{|y| = t-\sbar,|x-y|\leq \sbar-1}\frac{\sbar^{-2-\nu}(\sbar-|x-y|)^{-1+\mu}}{t-\sbar}d\sigma d\sbar
\\
& = \frac{C_f}{4\pi t^{1+\nu-\mu}}\int_{\frac{2}{t}}^1\int_{|y'|= 1-\lambda,|\frac{x}{t}-y'|\leq \lambda-t^{-1}}
\frac{(1-\lambda)^{-1}\lambda^{-2-\nu}d\sigma d\lambda}{\big(\lambda-\big|\frac{x}{t} - y'\big|\big)^{1-\mu}}
\quad(\lambda:= \sbar/t,\quad y' := y/t)
\\
& = \frac{C_f}{4\pi t^{1+\nu-\mu}}\int_{\frac{2}{t}}^1(1-\lambda)^{-1}\lambda^{-2-\nu} 
\int_{|y'| = 1-\lambda,|\frac{x}{t}-y'|\leq \lambda-t^{-1}} \frac{d\sigma}{\big(\lambda-\big|\frac{x}{t} - y'\big|\big)^{1-\mu}} \, d\lambda. 
\endaligned
$$
When $|\frac{x}{t} - y'|\leq \lambda - t^{-1}$, we get $\frac{t-r+1}{2t}\leq\lambda\leq 1$. In the following, we replace $y'$ by $y$.  
We distinguish between two cases:

\

\noindent{\bf Case 1:} $\frac{t-r}{t}> \frac{t+r+1}{2t}\Leftrightarrow r\leq \frac{t-1}{3}$. 
We write 
$$ 
\aligned
|u(t,x)|
\leq
& \frac{C_f}{4\pi t^{1+\nu-\mu}}\int_{\frac{t-r+1}{2t}}^1(1-\lambda)^{-1}\lambda^{-2-\nu} 
\int_{|y| = 1-\lambda,|\frac{x}{t}-y|\leq \lambda-t^{-1}} \frac{d\sigma}{\big(\lambda-\big|\frac{x}{t} - y\big|\big)^{1-\mu}} \, d\lambda
\\
\lesssim &\frac{C_f}{\mu t^{1+\nu-\mu}} \int_{\frac{t-r+1}{2t}}^{\frac{t+r+1}{2t}}(1-\lambda)^{-1}\lambda^{-2-\nu} \frac{\lambda t(1-\lambda)}{r}\left(\frac{t-r}{t}\right)^{\mu} d\lambda
\\
&+ \frac{C_f}{ t^{1+\nu-\mu}}\int_{\frac{t+r+1}{2t}}^{\frac{t-r}{t}}
(1-\lambda)^{-1}\lambda^{-2-\nu}(1-\lambda)\left(\frac{t+r}{t}-\lambda\right)\left(2\lambda - \frac{t+r}{t}\right)^{-1+\mu} \, d\lambda
\\
&+ \frac{C_f}{\mu t^{1+\nu-\mu}}\int_{\frac{t-r}{t}}^{1}
(1-\lambda)^{-1}\lambda^{-2-\nu}\frac{(1-\lambda)t}{ r}\left(\frac{t-r}{t}\right)^\mu d\lambda, 
\endaligned
$$
and therefore 
$$ 
\aligned
|u(t,x)|
\lesssim
&\frac{C_f}{\mu t^{1+\nu-\mu}}\frac{t}{r}\bigg(\frac{t-r}{t}\bigg)^{\mu}\int_{\frac{t-r+1}{2t}}^{\frac{t+r+1}{2t}}\lambda^{-1-\nu} \, d\lambda
\\
&+\frac{C_f}{t^{1+\nu-\mu}}\int_{\frac{t+r+1}{2t}}^{\frac{t-r}{t}}\lambda^{-2-\nu}\left(\frac{t+r}{t}-\lambda\right)\left(2\lambda - \frac{t+r}{t}\right)^{-1+\mu} \, d\lambda
\\
&+\frac{C_f}{\mu t^{1+\nu-\mu}}\frac{t}{r}\left(\frac{t-r}{t}\right)^{\mu}\int_{\frac{t-r}{t}}^1\lambda^{-2-\nu} \, d\lambda.
\endaligned
$$
Recall that $r\leq \frac{t-1}{3}$ and  that $0<|\nu|\leq1/2$, we have  
$$
\frac{t}{r}\int_{\frac{t-r+1}{2t}}^{\frac{t+r+1}{2t}}\lambda^{-1-\nu} \, d\lambda
\lesssim \bigg(\frac{t}{t-r}\bigg)^{1+\nu} \lesssim 1, 
$$
and 
$$
\left|\frac{C_f}{\mu t^{1+\nu-\mu}}\frac{t}{r}\bigg(\frac{t-r}{t}\bigg)^{\mu}\int_{\frac{t-r+1}{2t}}^{\frac{t+r+1}{2t}}\lambda^{-1-\nu} \, d\lambda\right | \lesssim C_f\mu^{-1}(t-r)^{\mu}t^{-1-\nu}.
$$
For the second integral term, we note that 
$$
\aligned
&\int_{\frac{t+r+1}{2t}}^{\frac{t-r}{t}}\lambda^{-2-\nu}\left(\frac{t+r}{t}-\lambda\right)\left(2\lambda - \frac{t+r}{t}\right)^{-1+\mu} \, d\lambda
\\
&\lesssim \int_{\frac{t+r+1}{2t}}^{\frac{t-r}{t}}\left(2\lambda - \frac{t+r}{t}\right)^{-1+\mu} \, d\lambda
 = \frac{1}{\mu}\left(2\lambda - \frac{t+r}{t}\right)^{\mu}\bigg|_{\frac{t+r+1}{2t}}^{\frac{t-r}{t}} 
\lesssim \frac{1}{\mu}, 
\endaligned
$$
thus
$$
\frac{C_f}{t^{1+\nu-\mu}}\int_{\frac{t+r+1}{2t}}^{\frac{t-r}{t}}\lambda^{-2-\nu}\left(\frac{t+r}{t}-\lambda\right)\left(2\lambda - \frac{t+r}{t}\right)^{-1+\mu} \, d\lambda
\lesssim
 \frac{C_f}{\mu t^{1+\nu-\mu}}.
$$
For the third term, from $\frac{t-r}{t}\geq \frac{t+r+t1}{2t}\geq \frac{1}{2}$ we get
$$
\aligned
\frac{C_f}{\mu t^{1+\nu-\mu}}\frac{t}{r}\left(\frac{t-r}{t}\right)^{\mu}\int_{\frac{t-r}{t}}^1\lambda^{-2-\nu} \, d\lambda
\lesssim &\frac{C_f}{\mu t^{1+\nu-\mu}}\frac{t}{r}\left(\frac{t-r}{t}\right)^{\mu}\int_{\frac{t-r}{t}}^1 2^{2+\mu} \, d\lambda
\\
\lesssim & C_f\mu^{-1} (t-r)^{\mu}t^{-1-\nu}.
\endaligned
$$
Hence, in the case $0<r\leq \frac{t-1}{3}$, $|u(t,x)|\lesssim C_f\mu^{-1}(t-r)^{\mu}t^{-1-\nu}$.

\

\noindent {\bf Case 2:} $\frac{t+r+1}{2t}\geq \frac{t-r}{t}\Leftrightarrow r\geq \frac{t-1}{3}$.
The second case in Lemma \ref{wave lemma1} can not occur. We have 
$$
|u(t,x)|
\lesssim\frac{C_f}{\mu t^{1+\nu-\mu}}\left(\frac{t-r}{t}\right)^{\mu}
\left(\int_{\frac{t-r+1}{2t}}^{\frac{t+r+1}{2t}}\lambda^{-1-\nu} \, d\lambda + \int_{\frac{t+r+1}{2t}}^1\lambda^{-2-\nu} \, d\lambda\right).
$$
Since $\frac{t+r+1}{2t}\geq 1/2$, the second integral term is bounded by some constant $C$. For the first integral, when $\nu> 0$,
$$
\int_{\frac{t-r+1}{2t}}^{\frac{t+r+1}{2t}}\lambda^{-1-\nu} \, d\lambda
\lesssim
 \frac{1}{\nu}\left(\frac{t-r+1}{t}\right)^{-\nu}, 
$$
thus
$|u(t,x)|\lesssim C_f(\mu\nu)^{-1}(t-r)^{\mu-\nu}t^{-1}$.

When $\nu<0$, we write 
$$
\int_{\frac{t-r+1}{2t}}^{\frac{t+r+1}{2t}}\lambda^{-1-\nu} \, d\lambda
\lesssim
\frac{1}{|\nu|}\left(\frac{t+r+1}{t}\right)^{-\nu}
\lesssim \frac{1}{|\nu|}
$$
and  obtain
$|u(t,x)|\lesssim C_f(\mu|\nu|)^{-1}(t-r)^{\mu}t^{-1-\nu}$. 
\end{proof}

\begin{proof}[Proof of Lemma \ref{wave lemma1}] When $r=0$, the estimate is trivial. When $r>0$,
we can set $x = (r,0,0)$. The surface $S_{\lambda} := \{|y| = 1-\lambda\}\cap \{\left|\frac{x}{t} - y\right|\leq \lambda - t^{-1}\}$ is parameterized by:
\begin{itemize}

\item $\theta$: angle from $(1,0,0)$ to $y$ with $0\leq \theta\leq \pi$,

\item $\phi$: angle from the plane determined by $(1,0,0)$ and $(0,1,0)$ and the plane determined by $y$ and $(1,0,0)$ with $0\leq \phi\leq 2\pi$. 

\end{itemize} 
We have
$y = (1-\lambda)\big(\cos\theta,\sin\theta  \cos\phi,\sin\theta \sin\phi\big)$ and distinguish between two cases: 

\

\noindent {\bf Case 1.}
When $\frac{t-r+1}{2t}\leq\lambda\leq \frac{t+r+1}{2t}$, 
we only have a part of the sphere $\{|y|=1-\lambda\}$ contained in the ball $\{\left|\frac{x}{t}-y\right|\leq \lambda - t^{-1}\}$ where $\cos(\theta) \geq \frac{(r/t)^2+(1-\lambda)^2 - \left(\lambda-t^{-1}\right)^2}{(2r/t)(1-\lambda)}$. So we set 
$\theta_0 := \arccos\left(\frac{(r/t)^2+(1-\lambda)^2 - \left(\lambda-t^{-1}\right)^2}{(2r/t)(1-\lambda)}\right)$
and see that 
$$
\lambda - \big|\frac{x}{t} - y\big|  = \lambda -\sqrt{\frac{r^2}{t^2} + (1-\lambda)^2 - 2\frac{r}{t}(1-\lambda)\cos\theta}
$$
and
$d\sigma = (1-\lambda)^2\sin(\theta)d\theta d\phi$. 
The integral is estimated as follows:
$$
\aligned
& \hskip-.3cm \int_{|y| = 1-\lambda,|\frac{x}{t}-y|\leq \lambda-t^{-1}} \frac{d\sigma}{\big(\lambda-\big|\frac{x}{t} - y\big|\big)^{1-\mu}}
\\
=&\int_0^{2\pi}d\phi\int_0^{\theta_0}(1-\lambda)^2\sin\theta
\bigg(\lambda -\sqrt{\frac{r^2}{t^2} + (1-\lambda)^2 - 2\frac{r}{t}(1-\lambda)\cos\theta}\bigg)^{-1+\mu}d\theta
\\
=&2\pi \int_0^{\theta_0}(1-\lambda)^2\sin\theta
\bigg(\lambda -\sqrt{\frac{r^2}{t^2} + (1-\lambda)^2 - 2\frac{r}{t}(1-\lambda)\cos\theta}\bigg)^{-1+\mu}d\theta
\\
=&-2\pi(1-\lambda)^2\int_0^{\theta_0}
\bigg(\lambda -\sqrt{\frac{r^2}{t^2} + (1-\lambda)^2 - 2\frac{r}{t}(1-\lambda)\cos\theta}\bigg)^{-1+\mu}d\cos\theta
\endaligned
$$
thus, with $\omega = \cos \theta$, 
$$
\aligned
& \hskip-.3cm \int_{|y| = 1-\lambda,|\frac{x}{t}-y|\leq \lambda-t^{-1}} \frac{d\sigma}{\big(\lambda-\big|\frac{x}{t} - y\big|\big)^{1-\mu}}
\\
=&2\pi(1-\lambda)^2\int_{\cos\theta_0}^1
\bigg(\lambda -\sqrt{\frac{r^2}{t^2} + (1-\lambda)^2 - 2\frac{r}{t}(1-\lambda)\omega}\bigg)^{-1+\mu}d\omega
\\
=&\frac{\pi t(1-\lambda)}{r}\int_{|\frac{r}{t} - (1-\lambda)|^2}^{(\lambda-t^{-1})^2}\big(\lambda - \sqrt{\gamma}\big)^{-1+\mu} \, d\gamma
=2 \frac{\pi t(1-\lambda)}{r}\int_{t^{-1}}^{\lambda-|\frac{r}{t}-(1-\lambda)|}\zeta^{-1+\mu}(\lambda-\zeta) \, d\zeta, 
\endaligned
$$
where $\gamma = \frac{r^2}{t^2} + (1-\lambda)^2 - 2\frac{r}{t}(1-\lambda)\omega$ and
$\zeta := \lambda - \sqrt{\gamma}$. We distinguish between two sub-cases.

\

\noindent {\bf Case 1.1:} $\frac{r}{t}\leq 1-\lambda$ or, equivalently, $\lambda \leq \frac{t-r}{t}$. We have 
$$
\aligned
 &2 \frac{\pi t(1-\lambda)}{r}\int_{t^{-1}}^{\lambda-|\frac{r}{t}-(1-\lambda)|}\zeta^{-1+\mu}(\lambda-\zeta) \, d\zeta
\\
&=2 \frac{\pi t(1-\lambda)}{r}\int_{t^{-1}}^{2(\lambda-\frac{t-r}{2t})}\zeta^{-1+\mu}(\lambda-\zeta) \, d\zeta
\lesssim \frac{\lambda t(1-\lambda)}{\mu r}\frac{(t-r)^{\mu}}{t^{\mu}}.
\endaligned
$$

\

\noindent {\bf Case 1.2:}  $1-\lambda<\frac{r}{t}$ or, equivalently, $\lambda>\frac{t-r}{t}$.  We have 
$$
\aligned
&2 \frac{\pi t(1-\lambda)}{r}\int_{t^{-1}}^{\lambda-|\frac{r}{t}-(1-\lambda)|}\zeta^{-1+\mu}(\lambda-\zeta) \, d\zeta
\\
&=2 \frac{\pi t(1-\lambda)}{r}\int_{t^{-1}}^{\frac{t-r}{t}}\zeta^{-1+\mu}(\lambda-\zeta) \, d\zeta
\lesssim
 \frac{\lambda t(1-\lambda)}{\mu r}\frac{(t-r)^{\mu}}{t^{\mu}}.
\endaligned
$$

\

\noindent {\bf Case 2.} When $\frac{t+r+1}{2t}\leq \lambda\leq 1$, the sphere $\{|y|=1-\lambda\}$ is contained in $\{\left|(x/t)-y\right|\leq \lambda-t^{-1} \}$ and 
$$
\aligned
&\int_{|y| = 1-\lambda,|\frac{x}{t}-y|\leq \lambda-t^{-1}} \frac{d\sigma}{\big(\lambda-\big|\frac{x}{t} - y\big|\big)^{1-\mu}}
=\int_{|y| = 1-\lambda}\frac{d\sigma}{\big(\lambda-\big|\frac{x}{t} - y\big|\big)^{1-\mu}}
\\
&=2\pi \int_0^{\pi}(1-\lambda)^2\sin\theta
\bigg(\lambda -\sqrt{\frac{r^2}{t^2} + (1-\lambda)^2 - 2\frac{r}{t}(1-\lambda)\cos\theta}\bigg)^{-1+\mu}d\theta
\\
&=2\pi(1-\lambda)^2\int_{-1}^1
\bigg(\lambda -\sqrt{\frac{r^2}{t^2} + (1-\lambda)^2 - 2\frac{r}{t}(1-\lambda)\omega}\bigg)^{-1+\mu}d\omega. 
\endaligned
$$
Therefore, we have 
$$
\aligned
\int_{|y| = 1-\lambda,|\frac{x}{t}-y|\leq \lambda-t^{-1}} \frac{d\sigma}{\big(\lambda-\big|\frac{x}{t} - y\big|\big)^{1-\mu}}
&=2 \frac{\pi t(1-\lambda)}{r}\int_{\lambda - (\frac{r}{t} + (1-\lambda))}^{\lambda-|\frac{r}{t}-(1-\lambda)|}\zeta^{-1+\mu}(\lambda-\zeta) \, d\zeta
\\
&=2 \frac{\pi t(1-\lambda)}{r}\int_{2\lambda - \frac{t+r}{t}}^{\lambda-|\frac{r}{t}-(1-\lambda)|}\zeta^{-1+\mu}(\lambda-\zeta) \, d\zeta.
\endaligned
$$
We distinguish between two sub-cases.

\

\noindent {\bf Case 2.1: } When $\frac{r}{t}\leq 1-\lambda$ or, equivalently, $\lambda \leq \frac{t-r}{t}$, we have
$$
\aligned
&2 \frac{\pi t(1-\lambda)}{r}\int_{2\lambda - \frac{t+r}{t}}^{\lambda-|\frac{r}{t}-(1-\lambda)|}\zeta^{-1+\mu}(\lambda-\zeta) \, d\zeta
\\
&=2 \frac{\pi t(1-\lambda)}{r}\int_{2\lambda - \frac{t+r}{t}}^{2\lambda-\frac{t-r}{t}}\zeta^{-1+\mu}(\lambda-\zeta) \, d\zeta
\leq C(1-\lambda)\left(\frac{t+r}{t}-\lambda\right)\left(2\lambda - \frac{t+r}{t}\right)^{-1+\mu},
\endaligned
$$
where the function $\zeta^{-1+\mu}(\lambda-\zeta)$ is decreasing and we can bound this integral by the value at the inferior boundary (which is $2\lambda - \frac{t+r}{t}$) times the length of the interval $2r/t$.

\

\noindent {\bf Case 2.2:} When $1-\lambda<\frac{r}{t}$ or, equivalently, $\lambda>\frac{t-r}{t}$, we have
$$
\aligned
& 2 \frac{\pi t(1-\lambda)}{r}\int_{2\lambda - \frac{t+r}{t}}^{\lambda-|\frac{r}{t}-(1-\lambda)|}\zeta^{-1+\mu}(\lambda-\zeta) \, d\zeta
\\
&=2 \frac{\pi t(1-\lambda)}{r}\int_{2\lambda - \frac{t+r}{t}}^{\frac{t-r}{t}}\zeta^{-1+\mu}(\lambda-\zeta) \, d\zeta
\leq C(1-\lambda)\frac{t}{r}\int_{2\lambda - \frac{t+r}{t}}^{\frac{t-r}{t}}\zeta^{-1+\mu}d\zeta
\\
&\leq\frac{C(1-\lambda)t}{\mu r}\zeta^{\mu}\bigg|_0^{\frac{t-r}{r}} =\frac{ C(1-\lambda)t}{\mu r}\left(\frac{t-r}{t}\right)^\mu.
\endaligned
$$
When $\frac{t+r+1}{2t}\leq \frac{t-r}{t}$, both case above may occur, while only Case 2.2 is possible if the opposite inequality holds true. 
\end{proof}


\section{Sup-norm estimate for the Klein-Gordon equation}
\label{sec subsec KG-supnorm}

We provide here a proof of Proposition~\ref{Linfini KG}. 

\begin{lemma}[A decomposition of the Klein-Gordon operator]
\label{lem 0 K-G}
For sufficiently smooth solutions $v$ to \eqref{Linfini KG eq}, the function
$
w_{t,x}(\lambda) := \lambda^{3/2}v(\lambda t/s, \lambda x/s) 
$ 
is a solution to the second-order ODE in $\lambda$
$$
\aligned
&
\frac{d^2}{d\lambda^2}w_{t,x}(\lambda) + \frac{c^2}{1+\hb^{00}(\lambda t/s,\lambda x/s)} w_{t,x}(\lambda)
\\
&= \big(1+\hb^{00}(\lambda t/s,\lambda x/s)\big)^{-1}\big(R_1[v] + R_2[v] + R_3[v] +s^{3/2}f\big)(\lambda t/s,\lambda x/s).
\endaligned
$$
\end{lemma}

\begin{lemma}[Technical ODE estimate]
\label{lem 1 K-G}
Let the function $G$ be defined on some interval $[s_0,s_1]$ and satisfying $\sup | G| \leq 1/3$ and let $k$ be some integrable function defined on $[s_0,s_1]$. The solution $z$ to 
\be
\aligned
& z''(\lambda) + \frac{c^2}{1+ G(\lambda)} z(\lambda) = k(\lambda),
\\
& z(s_0) = z_0, \qquad z'(s_0) = z_1,
\endaligned
\ee
(for some initial data $z_0, z_1$) 
satisfies the uniform estimate for $s \in [s_0, s_1]$
\bel{Linfini ineq ODE}
\aligned
&
|z(s)| + |z'(s)|
\lesssim  \big(|z_0| + |z_1| + K(s)\big) + \int_{s_0}^s\Big(|z_0|+|z_1| + K(\sbar)\Big) \, | G'(\sbar)|e^{C\int_\sbar^s| G'(\lambda)|d\lambda} \, d\sbar
\endaligned
\ee 
with $K(s) := \int_{s_0}^s |k(\sbar)| \, d\sbar$ and a constant $C>0$.
\end{lemma}

\begin{proof}[Proof of Lemma~\ref{lem 0 K-G}]
{\bf 1. Flat wave operator.} Recall $s=\sqrt{t^2-r^2}$ and $r=|x|$.  an elementary The flat wave operator $\Box$ in the hyperboloidal frame reads
\bel{Hyper box1}
-\Box = \delb_0 \delb_0 - \sum_a\delb_a\delb_a + 2 \sum_a \frac{x^a}{s}\delb_0\delb_a + \frac{3}{s}\delb_0.
\ee
Given any function $v$, we write 
$$
w(t,x) = s^{3/2}v(t,x) = (t^2-|x|^2)^{3/4}v(t,x),
$$
and 
\bel{Hyper box2}
-s^{3/2}\Box v = \delb_0\delb_0 w - \sum_a\delb_a\delb_a w + 2\sum_a \frac{x^a}{s}\delb_0\delb_a w - \frac{3w}{4s^2} - \sum_a\frac{3x^a\delb_a w}{s^2}.
\ee
Consider the function of a single variable 
$$
w_{t,x}(\lambda) := w(\lambda t/s,\lambda x/s)
= \lambda^{3/2}v(\lambda t/s,\lambda x/s),
$$
so that
$$
\frac{d}{d\lambda}w_{t,x}(\lambda) = \big(\delb_0 + s^{-1}x^a\delb_a\big)w(\lambda t/s,\lambda x/s)
= \frac{t}{s}\newperp  w\left(\lambda t/s,\lambda x/s\right)
$$
and
\bel{Hyper 2order}
\frac{d^2}{d\lambda^2}w_{t,x}(\lambda) = \bigg(\delb_0\delb_0 + 2 \frac{x^a}{s}\delb_0\delb_a +\frac{x^ax^b}{s^2}\delb_a\delb_b\bigg)w(\lambda t/s,\lambda x/s).
\ee
Combining with \eqref{Hyper box2} and recalling $w(t,x) = s^{3/2}v(t,x)$, we obtain 
\bel{Hyper 2order-box1}
\aligned
&\bigg(\delb_0\delb_0 + 2 \frac{x^a}{s}\delb_0\delb_a +\frac{x^ax^b}{s^2}\delb_a\delb_b\bigg)w
\\
& = -s^{3/2}\Box v + \sum_a\delb_a\delb_a w + \frac{x^ax^b}{s^2}\delb_a\delb_b w + \frac{3}{4s^2}w + \sum_a\frac{3x^a}{s^2}\delb_a w
  = - s^{3/2}\Box v + R_1[v].
\endaligned
\ee


\noindent {\bf 2. Curved wave operator.} We write
$$
-\Box v = h^{\alpha\beta}\del_\alpha\del_\beta v  - c^2v + f
$$
and perform a change of frame: 
$$
\aligned
h^{\alpha\beta}\del_\alpha\del_\beta v =& \hb^{\alpha\beta}\delb_\alpha\delb_\beta v + h^{\alpha\beta}\del_\alpha\Psib^{\beta'}_\beta\,\delb_{\beta'}v
\\
=& \hb^{00}\delb_0\delb_0 v
+ 2\hb^{0b}\delb_0\delb_bv + \hb^{ab}\delb_a\delb_bv + h^{\alpha\beta}\del_\alpha\Psib^{\beta'}_\beta\,\delb_{\beta'}v.
\endaligned
$$
We get
$$
\aligned
-s^{3/2}\Box v
=& - s^{3/2}\hb^{00}\delb_0\delb_0 v
      - s^{3/2}\big(2\hb^{0b}\delb_0\delb_bv + \hb^{ab}\delb_a\delb_bv + h^{\alpha\beta}\del_\alpha\Psib^{\beta'}_\beta\,\delb_{\beta'}v\big)
 - c^2s^{3/2}v + s^{3/2}f
\\
=& - \hb^{00}\delb_0\delb_0 \big(s^{3/2}v\big)- c^2s^{3/2}v
\\
& + \hb^{00}\bigg(\frac{3v}{4s^{1/2}} + 3s^{1/2}\delb_0 v\bigg)
 -  s^{3/2}\big(2\hb^{0b}\delb_0\delb_bv + \hb^{ab}\delb_a\delb_bv + h^{\alpha\beta}\del_\alpha\Psib^{\beta'}_\beta\,\delb_{\beta'}v\big) + s^{3/2}f,
\endaligned
$$
and conclude that 
\bel{Hyper KG}
\aligned
-s^{3/2}\Box v
& =  - \hb^{00}\delb_0\delb_0 w - c^2w + \hb^{00}\bigg(\frac{3v}{4s^{1/2}} + 3s^{1/2}\delb_0 v\bigg)
\\
& \quad -  s^{3/2}\big(2\hb^{0b}\delb_0\delb_bv + \hb^{ab}\delb_a\delb_bv + h^{\alpha\beta}\del_\alpha\Psib^{\beta'}_\beta\,\delb_{\beta'}v\big) + s^{3/2}f
\\
&= - \hb^{00}\delb_0\delb_0 w - c^2w + R_2[v] + s^{3/2}f.
\endaligned
\ee
Combining \eqref{Hyper 2order-box1} and \eqref{Hyper KG}, we get 
\bel{Hyper 2order-box2}
\delb_0\delb_0w + 2 \frac{x^a}{s}\delb_0\delb_a w+\frac{x^ax^b}{s^2}\delb_a\delb_bw
- \hb^{00}\delb_0\delb_0 w + c^2w
= R_1[v] + R_2[v] + s^{3/2}f.
\ee

\

\noindent {\bf 3. Conclusion.} We now write
$$
\aligned
& \big(1 + \hb^{00}\big)\bigg(\delb_0\delb_0 + 2 \frac{x^a}{s}\delb_0\delb_a +\frac{x^ax^b}{s^2}\delb_a\delb_b\bigg)w + c^2w
\\
& = \hb^{00}\bigg( 2 \frac{x^a}{s}\delb_0\delb_a +\frac{x^ax^b}{s^2}\delb_a\delb_b\bigg)w + R_1[v] + R_2[v] + s^{3/2}f
\endaligned
$$
and so
\bel{Linfini eq1}
\aligned
& \bigg(\delb_0\delb_0 + 2 \frac{x^a}{s}\delb_0\delb_a +\frac{x^ax^b}{s^2}\delb_a\delb_b\bigg)w
+ \frac{c^2w}{1 + \hb^{00}} 
\\
& = \big(1 + \hb^{00}\big)^{-1}\big(R_1[v] + R_2[v] + R_3[v] +s^{3/2}f\big).
\endaligned
\ee
This implies that
\be
\label{Linfini ineq ODE0}
\aligned
& \frac{d^2}{d\lambda^2}w_{t,x}(\lambda) + \frac{c^2w_{t,x}(\lambda)}{1 + \hb^{00}(\lambda t/s,\lambda x/s)}
\\
&= \big(1 + \hb^{00}(\lambda t/s,\lambda x/s)\big)^{-1}\big(R_1[v] + R_2[v] + R_3[v] +s^{3/2}f\big)(\lambda t/s,\lambda x/s).
\endaligned
\ee
\end{proof}


\begin{proof}[Proof of Lemma~\ref{lem 1 K-G}] We consider the vector field $b(\lambda) = \big(z(\lambda),z'(\lambda)\big)^T$ and the matrix $A(\lambda) := \left(
\begin{array}{cc}
0 &1
\\
-c^2(1+G)^{-1} &0
\end{array}
\right)
$
and write
$
b' = Ab +
\left(
\begin{array}{c}
0\\
k
\end{array}
\right). 
$
Consider the diagonalization $A = PQP^{-1}$ with
$$
Q =
\left(
\begin{array}{cc}
ic\big(1+G\big)^{-1/2} &0
\\
0 &-ic\big(1+G \big)^{-1/2}
\end{array}
\right)
$$
and
$$
P =
\left(
\begin{array}{cc}
1 & 1
\\
\frac{ic}{(1+G)^{1/2}} &-\frac{ic}{(1+G)^{1/2}}
\end{array}
\right),
\qquad
\quad
P^{-1} =
\left(
\begin{array}{cc}
1/2 & \frac{(1+G)^{1/2}}{2ic}
\\
1/2 & -\frac{(1+G)^{1/2}}{2ic}
\end{array}
\right).
$$
We thus have
$b' = PQP^{-1}b +
\left(
\begin{array}{c}
0
\\
k
\end{array}
\right)$,
leading us to
$$
\big(P^{-1}b\big)' = Q\big(P^{-1}b\big) + \big(P^{-1}\big)'b + P^{-1}
\left(
\begin{array}{c}
0
\\
k
\end{array}
\right).
$$
We regard $\big(P^{-1}\big)'b$ as a source and write 
$$
\aligned
P^{-1}b(s) = e^{\int_{s_0}^s Q(\sbar)d\sbar} P^{-1}b(s_0)
& + \int_{s_0}^s e^{\int_{\lambda}^s Q(\sbar)d\sbar}P^{-1}
\left(
\begin{array}{c}
0
\\
k
\end{array}
\right)
d\lambda
\\
& + \int_{s_0}^s e^{\int_{\lambda}^s Q(\sbar)d\sbar} \big(P^{-1}\big)'(\lambda)\,b(\lambda) \, d\lambda.
\endaligned
$$
When $\sup_{\lambda\in[1,s]} |G(\lambda)|\leq 1/3$, the norm of $P(\lambda)$ and $P^{-1}(\lambda)$ are bounded for $\lambda\in[s_0,s]$. 
The norm of $\big(P^{-1}\big)'(\lambda)$ is bounded by $C|G'(\lambda)|$ for a constant $C$ depending only on $c$. The norm of $Q$ is bounded by a constant $C>0$. Observe also that
$$
\int_{\lambda}^sQ(\sbar)d\sbar =
\left(
\begin{array}{cc}
ic\int_\lambda^s (1+G)^{-1/2}(\sbar)d\sbar & 0
\\
0 & -ic\int_\lambda^s (1+G)^{-1/2}(\sbar)d\sbar
\end{array}
\right), 
$$
therefore
$$
e^{\int_{\lambda}^sQ(\sbar)d\sbar} =
\left(
\begin{array}{cc}
e^{ic\int_\lambda^s (1+G)^{-1/2}(\sbar)d\sbar} & 0
\\
0 & e^{-ic\int_\lambda^s (1+G)^{-1/2}(\sbar)d\sbar}
\end{array}
\right).
$$
The norm of $e^{\int_{\lambda}^sQ(\sbar)d\sbar}$ is uniformly bounded and we have proven:
$$
|z(s)| + |z'(s)|\leq C (|z(s_0)|+|z'(s_0)|) + C \, K(s) + C\int_{s_0}^s| G'(\lambda)|\big(|z(\lambda)| + |z'(\lambda)|\big) \, d\lambda, 
$$
and it remains to apply Gronwall's lemma.
\end{proof}

\begin{proof}[Proof of Proposition \ref{Linfini KG}] 
We have 
$$
\aligned
&w_{t,x}(\lambda) = \lambda^{3/2}v(\lambda t/s, \lambda x/s),
\\
&w'_{t,x}(\lambda) = \frac{3}{2}\lambda^{1/2}v(\lambda t/s, \lambda x/s) + \frac{t}{s}\lambda^{3/2}\newperp v(\lambda t/s, \lambda x/s). 
\endaligned
$$
The function $w_{t,x}$ is the restriction of $w(t,x) = s^{3/2}v(t,x)$ to the segment $\big\{(\lambda t/s, \lambda x/s),\lambda\in[s_0,s] \big\}$.
Apply \eqref{Linfini ineq ODE} and \eqref{Linfini ineq ODE0} to this segment, with
$$
s_0 =\left
\{
\aligned
& 2, \quad 0\leq r/t \leq 3/5,
\\
& \sqrt{\frac{t+r}{t-r}},\quad 3/5\leq r/t\leq 1.
\endaligned
\right.
$$
This is the line $\{(\lambda t/s, \lambda x/s)\}$ between $(t,x)$ and the boundary of $\Kcal_{[s_0,+\infty)}$.

The function $v$ is supported in $\Kcal$ and the restriction of $v$ to the hyperboloid $\Hcal_2$ is supported in $\Hcal_2 \cap \Kcal$. We recall that when $3/5 \leq r/t\leq 1$, $w_{t,x}(s_0) = 0$ and when $0\leq r/t\leq 3/5$, $w_{t,x}(s_0)$ is determined by $v_0$.

When $0\leq r/t\leq 3/5$, we apply \eqref{Linfini ineq ODE} with $s_0 = 2$. When $\lambda = 2$, we write 
$
w_{t,x}(2) = w(2t/s,2x/s) = 2^{3/2}v(2t/s,2x/s) = 2^{3/2}v_0(2x/s),
$
and
$$
\aligned
w'_{s,x}(2) =& \frac{d}{d\lambda}\big(\lambda^{3/2}v(\lambda t/s,\lambda x/s)\big)\big|_{\lambda = 2}
\\
=& \frac{3\sqrt{2}}{2}v(2t/s,2x/s) + 2^{3/2}(s/t)^{-1}\newperp v(2t/s,2x/s)
\\
=& \frac{3\sqrt{2}}{2}v(2t/s,2x/s) + 2^{3/2}(s/t)^{-1}\del_tv(2t/s,2x/s) + 2^{3/2}(x^a/s)\del_av(2t/s,2x/s)
\\
=& \frac{3\sqrt{2}}{2}v_0(2x/s) + 2^{3/2}(x^a/s)\del_av_0(2x/s) + 2^{3/2}(s/t)^{-1}v_1(2t/s,2x/s).
\endaligned
$$
Recall that when $0\leq r/t\leq 3/5$, we have $4/5\leq s/t\leq 1$. So we see that $|w_{t,x}(s_0)| + |w'_{t,x}(s_0)|\leq C(\|v_0\|_{L^\infty(\Hcal_2)} + \|v_1\|_{L^\infty(\Hcal_2)})$. Then by \eqref{Linfini ineq ODE} and \eqref{Linfini ineq ODE0} we have
$$
\aligned
|w_{t,x}(s)| + |w'_{t,x}(s)| \leq
\, & C(\|v_0\|_{L^\infty(\Hcal_2)} + \|v_1\|_{L^\infty(\Hcal_2)}) + C F(s)
\\
&+ C(\|v_0\|_{L^\infty(\Hcal_2)} + \|v_1\|_{L^\infty(\Hcal_2)})\int_2^s|h_{t,x}'(\sbar)|e^{C\int_{\sbar}^s |h_{t,x}'(\lambda)| d\lambda} \, d\sbar
\\
&+ C\int_2^s F(\sbar)|h_{t,x}'(\sbar)|e^{C\int_\sbar^s |h_{t,x}'(\lambda)|d\lambda} \, d\sbar.
\endaligned
$$

When $3/5\leq r/t \leq 1$, $w_{t,x}(s_0) = w_{t,x}'(s_0) = 0$ and so  
$$
\aligned
|w_{t,x}(s)| + |w'_{t,x}(s)| \leq \, & C F(s) + C\int_{s_0}^s F(\sbar)|h_{t,x}'(\sbar)|e^{C\int_\sbar^s |h_{t,x}'(\lambda)|d\lambda} \, d\sbar, 
\endaligned
$$
which leads to 
$
|w_{t,x}(s)| + |w'_{t,x}(s)|\lesssim V(t,x).
$
Recall finally 
$v(t,x) = s^{3/2}w_{t,x}(s)$
and
$$
(s/t)^{-1}s^{3/2}\newperp  v(t,x) = w_{t,x}'(s) - \frac{3}{2}s^{1/2}v(t,x) = w_{t,x}'(s) - \frac{3}{2}s^{-1}w_{t,x}(s). 
$$ 
\end{proof}


\section{Commutator estimates for the hyperboloidal frame}
\label{appendix-COM} 

In this appendix, we provide some further details on some important properties shared by the commutators arising in our problem. The vector fields $\del_\alpha, $ and $L_a$ are Killing for the wave operator $\Box$, so that
\be
[\del_\alpha, \,\Box]=0, \qquad [L_a,\, \Box] =0.
\ee
By introducing 
\be
\label{pre commutator base L-P}
\aligned
\, [L_a,\del_\beta]
& =: \Theta_{a\beta}^{\gamma}\del_{\gamma},
\qquad
[\del_\alpha,\delu_\beta]
 =: t^{-1}\Gammau_{\alpha\beta}^{\gamma}\del_{\gamma},
\qquad
[L_a,\delu_\beta]
 =: \Thetau_{a\beta}^{\gamma}\delu_{\gamma},
\endaligned
\ee
we find 
\bel{pre commutator base'}
\aligned
&\Theta_{a0}^{\gamma} = -\delta_a^{\gamma},
\qquad
&&\Theta_{ab}^{\gamma} = -\delta_{ab}\delta_0^{\gamma},
\\
&\Gammau_{0b}^{\gamma} = -\frac{x^b}{t}\delta_0^{\gamma} = \Psi^0_b\delta_0^{\gamma},
\quad&&
 \Gammau_{\alpha0}^{\gamma}  = 0,
\qquad
&&&\Gammau_{ab}^{\gamma}= \delta_{ab}\delta_0^{\gamma},\quad
\\
&\Thetau_{a0}^{\gamma} = -\delta^{\gamma}_a + \frac{x^a}{t}\delta^{\gamma}_0 =  -\delta^{\gamma}_a + \Phi_0^a\delta^{\gamma}_0,
\qquad
&&\Thetau_{ab}^{\gamma} = -\frac{x^b}{t}\delta^{\gamma}_a = \Psi^0_b\delta^{\gamma}_a.
\endaligned
\ee 
All of these coefficients are smooth in the (open) cone $\Kcal$ and homogeneous of degree $0$.
Furthermore, we also get 
\bel{pre commutator base''}
\Thetau_{ab}^0 = 0, \quad \text{ so that } \quad
[L_a,\delu_b] =
 \Thetau_{ab}^c\delu_c, 
\ee
which means that the commutator of a ``good'' derivative $\delu_b$ with $L_a$ is again a ``good'' derivative. 

\begin{lemma}
[Algebraic decomposition of commutators. I]
There exist constants $\lambda_{aJ}^I$ such that
\bel{pre lem commutator pr1}
[\del^I, L_a] = \sum_{|J|\leq|I|}\lambda^I_{aJ}\del^J.
\ee
\end{lemma}

\begin{proof} We proceed by induction and, for $|I| = 1$, this is \eqref{pre commutator base L-P}. Assuming that \eqref{pre lem commutator pr1} holds for all $|I_1|\leq m$, we are going to prove that it is still valid for $|I|\leq m+1$. Let $I = (\alpha,\alpha_m,\alpha_{m-1},\dots,\alpha_1)$ and  $I_1=(\alpha_m,\alpha_{m-1}, \ldots, \alpha_1)$, so that $\del^I = \del_\alpha\del^{I_1}$. We find 
$$
\aligned
\,[\del^I,L_a]
= \, & [\del_\alpha\del^{I_1},L_a] = \del_\alpha\big([\del^{I_1},L_a]\big)
 + [\del_\alpha,L_a]\del^{I_1}
= \del_\alpha\bigg(\sum_{|J|\leq |I_1|}\lambda_{aJ}^{I_1}\del^J \bigg) - \Theta_{a\alpha}^\gamma\del_{\gamma}\del^{I_1}
\\
= \, &\sum_{|J|\leq |I_1|}\lambda_{aJ}^{I_1}\del_\alpha\del^J - \Theta_{a\alpha}^\gamma\del_{\gamma}\del^{I_1},
\endaligned
$$
which yields the statement for $|I| = m+1$.
\end{proof}

\begin{lemma}
[Algebraic decomposition of commutators. II]
There exist constants $\theta_{\alpha J}^{I\gamma}$ so that
\be
\label{pre lem commutator pr2}
[L^I, \del_\alpha] = \sum_{|J|\leq|I|-1,\gamma}\theta_{\alpha J}^{I\gamma}\del_{\gamma}L^J.
\ee
\end{lemma}

\begin{proof} The case $|I|=1$ is already covered by \eqref{pre commutator base L-P}. Assuming that \eqref{pre lem commutator pr2} is valid for $|I|\leq m$, we are going to prove that it is still valid when $|I|=m+1$. We write 
$
L^I = L_aL^{I_1}$ with $|I_1| = m$, and find 
$$
\aligned
\,[L^I,\del_\alpha]
 = &[L_aL^{I_1},\del_\alpha]
 = L_a\big([L^{I_1},\del_\alpha]\big) + [L_a,\del_\alpha]L^{I_1}
\\
=& L_a\bigg(\sum_{|J|\leq |I_1|-1,\gamma}\theta_{\alpha J}^{I_1\gamma}\del_{\gamma}L^J \bigg)
+ \sum_{\gamma}\Theta_{a\alpha}^{\gamma}\del_{\gamma}L^{I_1}
\\
=&\sum_{|J|\leq |I_1|-1,\gamma}\theta_{\alpha J}^{I_1\gamma}L_a\del_{\gamma}L^J
 + \sum_{\gamma}\Theta_{a\alpha}^{\gamma}\del_{\gamma}L^{I_1}
\endaligned
$$
so 
$$
\aligned
\,[L^I,\del_\alpha]
=&\sum_{|J|\leq |I_1|-1,\gamma}\theta_{\alpha J}^{I_1\gamma}\del_{\gamma}L_aJ^J
+\sum_{|J|\leq |I_1|-1,\gamma}\theta_{\alpha J}^{I_1\gamma}[L_a,\del_{\gamma}]J^J
+ \sum_{\gamma}\Theta_{a\alpha}^{\gamma}\del_{\gamma}L^{I_1}
\\
=&\sum_{|J|\leq |I_1|-1,\gamma}\theta_{\alpha J}^{I_1\gamma}\del_{\gamma}L_aJ^J
+\sum_{|J|\leq |I_1|-1,\gamma}\theta_{\alpha J}^{I_1\gamma}\Theta_{a\gamma}^{\gamma'}\del_{\gamma'}L^J
+ \sum_{\gamma}\Theta_{a\alpha}^{\gamma}\del_{\gamma}L^{I_1}.
\endaligned
$$
\end{proof}

As a consequence of \eqref{pre lem commutator pr2}, we have 
\be
[\del^IL^J,\del_\alpha]u = \sum_{|J'|<|J|,\gamma}\theta_{\alpha J'}^{J\gamma}\del_{\gamma}\del^{I}L^{J'} u.
\ee

\begin{lemma}
[Algebraic decomposition of commutators. III]
\label{lem-com2}
One has 
\bel{pre lem commutator pr2 NEW}
[\del^IL^J,\delu_\beta]
 = \sum_{|J'| \leq |J|,|I'|\leq|I|\atop |I'|+|J'|<|I|+|J|}\thetau_{\beta I'J'}^{IJ \gamma}\del_{\gamma}\del^{I'}L^{J'},
\ee
where $\thetau_{\beta I'J'}^{IJ\gamma}$ are smooth functions satisfying
\bel{pre lem commutator pr4a}
 \big|\del^{I_1}L^{J_1}\thetau_{\beta I'J'}^{IJ\gamma}\big| \leq
\begin{cases}
C\big(|I|,|J|,|I_1|,|J_1|\big) \, t^{-|I_1|}
  &\text{ when } |J'| < |J|,
\\
C\big(|I|,|J|,|I_1|,|J_1|\big) \, t^{-|I_1|-1}     &\text{ when } |I'| < |I|.
\end{cases}
\ee
\end{lemma}

\begin{proof} Consider the  identity
$$
\aligned
\,[\del^IL^J,\delu_\beta]
= [\del^IL^J, \Phi_\beta^{\gamma}\del_{\gamma}]
=&
\Phi_\beta^{\gamma}[\del^IL^J,\del_{\gamma}]
+ \sum_{I_1+I_2=I, J_1+J_2=J \atop |I_1|+|J_1|<|I|+|J|}
\del^{I_1}L^{J_1}\Phi_\beta^{\gamma}\del^{I_2}L^{J_2}\del_{\gamma}.
\endaligned
$$
Commuting $\del^{I_2}L^{J_2}$ and $\del_{\gamma}$, we obtain  
$$
\aligned
\,[\del^IL^J,\delu_\beta]
=& \Phi_\beta^{\gamma}[\del^IL^J,\del_{\gamma}]
\\
& + \sum_{I_1+I_2=I,J_1+J_2 = J \atop |I_1|+|J_1|<|I|+|J|}
\del^{I_1}L^{J_1}\Phi_\beta^{\gamma}\del_{\gamma}\del^{I_2}L^{J_2}
 + \sum_{I_1+I_2=I,J_1+J_2= J\atop |I_1|+|J_1|<|I|+|J|}
\del^{I_1}L^{J_1}\Phi_\beta^{\gamma}[\del^{I_2}L^{J_2},\del_{\gamma}]
\\
=&  \sum_{I_1+I_2=I,J_1+J_2= J\atop |I_1|+|J_1|<|I|+|J|}
\del^{I_1}L^{J_1}\Phi_\beta^{\gamma}\del_{\gamma}\del^{I_2}L^{J_2}
 + \sum_{I_1+I_2=I\atop J_1+J_2=J}
\del^{I_1}L^{J_1}\Phi_\beta^{\gamma}[\del^{I_2}L^{J_2},\del_{\gamma}]
\\
=& \sum_{I_1+I_2=I,J_1+J_2= J\atop |I_1|+|J_1|<|I|+|J|}
\del^{I_1}L^{J_1}\Phi_\beta^{\gamma}\del_{\gamma}\del^{I_2}L^{J_2}
+ \sum_{I_1+I_2=I\atop J_1+J_2=J}\sum_{|J_2'|<|J_2|}
\big(\del^{I_1} L^{J_1}\Phi_\beta^{\gamma} \big) \, \theta_{\gamma J_2'}^{J_2\delta}\del_{\delta} \del^{I_2}L^{J_2'}.
\endaligned
$$
Hence, $\thetau_{\gamma I'J'}^{IJ\alpha}$ are linear combinations of $\del^{I_1}L^{J_1} \Phi_\beta^{\gamma}$ and
$\big( \del^{I_1}L^{J_1}\Phi_\beta^{\gamma} \big) \theta_{\gamma J_2'}^{J_2\delta}$ and $J_1+J_2=J$,
which yields \eqref{pre lem commutator pr2 NEW}.
Note that $\theta_{\gamma J_2'}^{J_2\delta}$ are constants, so that
$$
\del^{I_3}L^{J_3}\big(\del^{I_1}L^{J_1}\Phi_{\beta}^{\gamma}\theta_{\gamma J_2'}^{J_2\delta}\big)
= \theta_{\gamma J_2'}^{J_2\delta}\del^{I_3}L^{J_3}\del^{I_1}L^{J_1}\Phi_{\beta}^{\gamma}.
$$ 
By definition, $\Phi_{\beta}^{\gamma}$ is a homogeneous function of degree zero, so that $\del^{I_1}L^{J_1}\Phi_{\beta}^{\gamma}$ is again homogeneous but with degree $\leq 0$. We thus arrive at  \eqref{pre lem commutator pr4a}.
\end{proof}

\begin{lemma}[Algebraic decomposition of commutators. IV] 
One has 
\bel{pre lem commutator pr3}
[L^I,\delu_c]
 = \sum_{|J|<|I|}\sigma^{Ia}_{cJ}\delu_aL^J,
\ee
where the coefficients $\sigma_{c J}^{Ia}$ are
smooth functions and
satisfy (in $\Kcal$)
\bel{pre lem commutator pr3b}
\big|\del^{I_1}L^{J_1}\sigma_{c J}^{Ia}\big| \leq C(|I|,|J|,|I_1|,|J_1|)t^{-|I_1|}.
\ee
\end{lemma}

\begin{proof}
This is also by induction. Again, when $|I|=1$, \eqref{pre lem commutator pr3} together with \eqref{pre lem commutator pr3b} are guaranteed by \eqref{pre commutator base''}. Assume that \eqref{pre lem commutator pr3} and \eqref{pre lem commutator pr3b} hold for $|I|\leq m$, we will prove that they are valid for $|I| = m+1$. We take $L^I = L_aL^J$ with $|J| = m$, and obtain
$$
\aligned
\,[L^I,\delu_c]  =&[L_aL^J,\delu_c]  = L_a\big([L^J,\delu_c] \big) + [L_a,\delu_c]L^J
\\
=&L_a\bigg(\sum_{|J'|<|J|}\sigma^{Ja}_{cJ'}\delu_aL^{J'}  \bigg) + \Thetau_{ac}^b\delu_b L^J
\\
=&\sum_{|J'|<|J|}L_a\sigma^{Jb}_{cJ'}\delu_bL^{J'}
 + \sum_{|J'|<|J|}\sigma^{Jb}_{cJ'}L_a\delu_bL^{J'}  + \Thetau_{ac}^b\delu_b L^J, 
\endaligned
$$
so that 
$$
\aligned
\,[L^I,\delu_c]  
=&\sum_{|J'|<|J|}L_a\sigma^{Jb}_{cJ'}\delu_bL^{J'}  + \sum_{|J'|<|J|}\sigma^{Jb}_{cJ'}\delu_b L_a L^{J'}
+ \sum_{|J'|<|J|}\sigma^{Jb}_{cJ'}[L_a,\delu_b]L^{J'}  + \Thetau_{ac}^b\delu_b L^J
\\
=&\sum_{|J'|<|J|}L_a\sigma^{Jb}_{cJ'}\delu_bL^{J'}  + \sum_{|J'|<|J|}\sigma^{Jb}_{cJ'}\delu_b L_a L^{J'}
+ \sum_{|J'|<|J|}\sigma^{Jb}_{cJ'}\Thetau_{ab}^d\delu_dL^{J'}  + \Thetau_{ac}^b\delu_b L^J.
\endaligned
$$
In each term the coefficients are homogeneous of degree $0$ (by applying \eqref{pre lem commutator pr3b}),
and the desired result is proven.
\end{proof}

The following result is also checked by induction along the same lines as above, and so its proof is omitted. 

\begin{lemma}[Algebraic decomposition of commutators. V] 
One has 
\bel{pre lem commutator pr4}
[\del^I,\delu_c]
=  t^{-1}\!\!\!\!\sum_{|J|\leq|I|}\rho_{cJ}^{I}\del^{J},
\ee
where $\rho_{cJ}^{I}$ are smooth functions satisfying
\bel{pre lem commutator pr4b}
\big|\del^{I_1}L^{J_1}\rho_{cJ}^{I}\big| \leq C(|I|,|J|,|I_1|,|J_1|)t^{-|I_1|}.
\ee
\end{lemma}


The following statements are now immediate in view of \eqref{pre lem commutator pr1}, \eqref{pre lem commutator pr2}, and \eqref{pre lem commutator pr3}, and \eqref{pre lem commutator pr4}.

\begin{proposition}[Estimates on commutators. I]
\label{lem commutator esti I}
For all sufficiently regular functions $u$ defined in the future cone $\Kcal$, one has 
\bel{pre lem commutator pr5}
\big|[\del^IL^J,\del_\alpha]u\big|\leq C(|I|, |J|)\sum_{|J'|<|J|,\beta}|\del_\beta\del^IL^{J'}u|,
\ee
\bel{pre lem commutator pr5 NEW}
\big|[\del^IL^J,\delu_c]u\big|
\leq
C(|I|,|J|)
\Bigg(
\sum_{|J'|<|J|,a\atop |I'|\leq |I|} |\delu_a \del^{I'}L^{J'}u|
+ t^{-1}\!\!\!\!\sum_{|I|\leq|I'|\atop |J|\leq|J'|}|\del^{I'}L^{J'}u| \Bigg),
\ee
\bel{pre lem commutator trivial} 
\left|[\del^IL^J,\delu_{\alpha}u]\right|
\leq C(|I|,|J|)t^{-1}\sum_{\beta,|I'|<|I|\atop |J'|\leq|J|}\left|\del_{\beta}\del^{I'}L^{J'}u\right|
 +C(|I|,|J|)\sum_{\beta,|I'|\leq|I|\atop |J'|<|J|}\left|\del_{\beta}\del^{I'}L^{J'}u\right|, 
\ee
\bel{pre lem commutator second-order}
\big|[\del^IL^J,\del_\alpha \del_\beta] u \big|
\leq  C(|I|,|J|)\sum_{\gamma,\gamma'\atop |I|\leq|I'|,|J'|<|I|} \big|\del_{\gamma}\del_{\gamma'}\del^{I'}L^{J'} u\big|,
\ee
\bel{pre lem commutator second-order bar}
\aligned
&\big|[\del^IL^J, \delu_a\delu_\beta] u\big| + \big|[\del^IL^J, \delu_\alpha \delu_b] u\big|
\\
&\leq C(|I|,|J|) \Bigg(
\sum_{c,\gamma,|I'|\leq |I|\atop |J'| < |J|}\big|\delu_c \delu_{\gamma} \del^{I'}L^{J'}u\big|
+
 t^{-1} \sum_{c,\gamma,|I'| < |I|\atop |J'| \leq |J|}\big|\delu_c \delu_{\gamma} \del^{I'}L^{J'}u\big|
+ t^{-1}\sum_{\gamma,|I'|\leq|I|\atop |J'|\leq|J|}\big|\del_{\gamma}\del^{I'}L^{J'}u\big|
\Bigg).
\endaligned
\ee 
\end{proposition}

Further estimates will be also needed, as now stated.

\begin{proposition}[[Estimates on commutators. II]
\label{pre lem commutator s/t}
For all sufficiently regular functions $u$ defined in the future cone $\Kcal$, one has 
 (for all $I, J, \alpha$)
\bel{pre lem commutator T/t'}
\big|\del^IL^J\big((s/t)\del_\alpha u\big)\big| \leq \big|(s/t)\del_\alpha \del^IL^J u\big|
+ C(|I|,|J|)\sum_{\beta,|I'|\leq|I|\atop |J'|\leq |J|}\big|(s/t)\del_\beta\del^{I'}L^{J'}u\big|.
\ee
\end{proposition}

Finally, recall from \cite{PLF-YM-book}) the following technical observation concerning  products of first-order linear operators with homogeneous coefficients of order $0$ or $1$.

\begin{lemma}
\label{pre lem lem commutator s/t}
For all multi-indices $I$, the function
$$
\Xi^{I,J}  := (t/s) \del^I L^J (s/t),
$$
defined in the closed cone $\overline{\Kcal} = \{|x|\leq t-1\}$, is smooth and all of its derivatives (of any order)
are bounded in $\overline{\Kcal}$. Furthermore, it is homogeneous of degree $\eta$ with $\eta\leq 0$.
\end{lemma}


%
%
%
%
%
%
%
%
%
%


\begin{thebibliography}{99}

\bibitem{AnsorgMacedo} {\sc M. Ansorg and R.P. Macedo,}
Spectral decomposition of black-hole perturbations on hyperboloidal slices, 
Physical Rev. D 93 (2016), 124016. 

\bibitem{Asanov} {\sc R.A. Asanov,}
The Schwarzschild metric and de Donder condition,
Gen. Relat. Grav. 21 (1989), 149--154.

\bibitem{Bachelot88} {\sc A. Bachelot,}
Probl\`eme de Cauchy global pour des syst\`emes de Dirac-Klein-Gordon,
Ann. Inst. Henri Poincar\'e 48 (1988), 387--422.

\bibitem{Bachelot94} {\sc A. Bachelot,}
Asymptotic completeness for the Klein-Gordon equation on the Schwarzschild metric,
Ann. Inst. Henri Poincar\'e: Phys. Th\'eor. 61 (1994), 411--441.

\bibitem{BLF} {\sc F. Beyer and P.G. LeFloch,}
Dynamics of self--gravitating fluids in Gowdy-symmetric spacetimes near cosmological singularities, 
Comm. Part. Diff. Equa. (2017). Preprint ArXiv:1512.07187.

\bibitem{BieriZipser} {\sc L. Bieri and N. Zipser,}
{\sl Extensions of the stability theorem of the Minkowski space in General Relativity,}
AMS/IP Studies Adv. Math. 45. Amer. Math. Soc., International Press, Cambridge,  2009.

\bibitem{BD} {\sc A.Y. Burtscher and R. Donninger,}
Hyperboloidal Evolution and Global Dynamics for the Focusing Cubic Wave Equation,
Comm. Math. Phys. 353 (2017), 549--596. 

\bibitem{BuLF} {\sc A.Y. Burtscher and P.G. LeFloch,} 
The formation of trapped surfaces in spherically-symmetric Einstein-Euler spacetimes with bounded variation,  
J. Math. Pures Appl. 102 (2014), 1164--1217. 

\bibitem{CB}{\sc Y. Choquet-Bruhat},
{\sl General relativity and the Einstein equations}, Oxford Math. Monograph,
Oxford Univ. Press, 2009

\bibitem{Christodoulou86} {\sc D. Christodoulou,}
Global solutions of nonlinear hyperbolic equations for small initial data,
Comm. Pure Appl. Math. 39 (1986), 267--282.

\bibitem{Christodoulou} {\sc D. Christodoulou,}
{\sl The formation of black holes in general relativity,}
Eur. Math. Soc. (EMS) series, Z\"urich, 2008.

\bibitem{CK} {\sc D. Christodoulou and S. Klainerman,}
{\sl The global nonlinear stability of the Minkowski space,}
Princeton Math. Ser. 41, Princeton University, 1993.

\bibitem{Chrusciel} {\sc P.T. Chrusciel,} 
Anti-gravit\'e \`a la Carlotto-Schoen, 
Bourbaki seminar, Paris, November 2016, to appear. 

\bibitem{CD} {\sc P.T. Chrusciel and E. Delay,} 
Exotic hyperbolic gluings, 
Jour. Differ. Geom. (2017). Preprint ArXiv:1511.07858. 

\bibitem{CorvinoSchoen} {\sc J. Corvino and R. Schoen,}
On the asymptotics for the vacuum Einstein constraint equations, 
J. Diff. Geom. 73 (2006), 185--217.

\bibitem{Delort01}
{\sc J.-M. Delort,}
Existence globale et comportement asymptotique pour l'\'equation de Klein-Gordon quasi-lin\'eaire \`a donn\'ees petites en dimension $1$,
Ann. Sci. \'Ecole Norm. Sup. 34 (2001), 1--61.

\bibitem{Delort04}
{\sc J.-M. Delort, D. Fang, and R. Xue,}
Global existence of small solutions for quadratic quasilinear Klein-Gordon systems in two space dimensions,
J. Funct. Anal. 211 (2004), 288--323.

\bibitem{FJS} {\sc D. Fajman, J. Joudioux, and J. Smulevici,}
A vector field method for relativistic transport equations with applications, 
Analysis \& PDE (2017), to appear. 
See also ArXiv:1510.04939. 

\bibitem{FJS2} {\sc D. Fajman, J. Joudioux, and J. Smulevici,}
Sharp asymptotics for small data solutions of the Vlasov-Nordström system in three dimensions,
Preprint ArXiv:1704.05353. 

\bibitem{Frauendiener} {\sc J. Frauendiener,}
Numerical treatment of the hyperboloidal initial value problem for the vacuum Einstein equations. II. The evolution equations,
Phys. Rev. D 58 (1998), 064003.

\bibitem{FrauendienerH} {\sc J. Frauendiener and M. Hein,}
Numerical evolution of axisymmetric, isolated systems in general relativity,
Phys. Rev. D 66 (2002), 124004.

\bibitem{Friedrich81} {\sc H. Friedrich,}
On the regular and the asymptotic characteristic initial value problem for Einstein's vacuum field equations,
Proc. R. Soc. London Ser.
A 375 (1981), 169--184.

\bibitem{Friedrich83} {\sc H. Friedrich,}
Cauchy problems for the conformal vacuum field equations in General Relativity,
Commun. Math. Phys. 91 (1983), 445--472.

\bibitem{G} {\sc P. Germain,}
Global existence for coupled Klein--Gordon equations with different speeds,
Ann. Inst. Fourier 61  (2011), 2463--2506.

\bibitem{GLF} {\sc N. Grubic and P.G. LeFloch,}
On the area of the symmetry orbits in weakly regular Einstein-Euler spacetimes with Gowdy symmetry, 
SIAM J. Math. Anal. 47 (2015), 669--683.

\bibitem{Hilditch} {\sc D. Hilditch, E. Harms, M. Bugner, H. Rueter, and B. Bruegmann,}
The evolution of hyperboloidal data with the dual foliation formalism: Mathematical analysis and wave equation tests, 
Preprint ArXiv:1609.08949. 

\bibitem{Hormander}{\sc L. H\"ormander,}
{\sl Lectures on nonlinear hyperbolic differential equations,}
Springer Verlag, Berlin, 1997.

\bibitem{HoshigaKubo} {\sc A. Hoshiga and H. Kubo,}
Global small amplitude solutions of nonlinear hyperbolic systems with a critical exponent under the null condition,
SIAM J. Math. Anal. 31 (2000), 486--513.

\bibitem{IP} {\sc A.D. Ionescu and B. Pausader,} 
Global solutions of quasilinear systems of Klein-Gordon equations in 3D, 
J. Eur. Math. Soc. 16, 2355--2431.

\bibitem{Katayama12a} {\sc S. Katayama,}
Global existence for coupled systems of nonlinear wave and Klein-Gordon equations in three space dimensions,
Math. Z. 270 (2012), 487--513.

\bibitem{Katayama12b}{\sc S. Katayama,}
Asymptotic pointwise behavior for systems of semilinear wave equations in three space dimensions,
J. Hyperbolic Differ. Equ. 9 (2012), 263--323.

\bibitem{Klainerman80}{\sc S. Klainerman,}
Global existence for nonlinear wave equations,
 Comm. Pure Appl. Math. 33 (1980), 43--101.

\bibitem{Klainerman85}{\sc S. Klainerman,}
Global existence of small amplitude solutions to nonlinear Klein-Gordon equations in four spacetime dimensions,
Comm. Pure Appl. Math. 38 (1985), 631--641.

\bibitem{Klainerman86}{\sc S. Klainerman,}
The null condition and global existence to nonlinear wave equations,
in ``Nonlinear systems of partial differential equations in applied mathematics'', Part 1, Santa Fe, N.M., 1984,
Lectures in Appl. Math., Vol.~23, Amer. Math. Soc., Providence, RI, 1986, pp.~293--326.

\bibitem{Klainerman87}{\sc S. Klainerman,}
Remarks on the global Sobolev inequalities in the Minkowski space $\RR^{n+1}$,
Comm. Pure Appl. Math. 40 (1987), 111--117.

\bibitem{KR} {\sc S. Klainerman and I. Rodnianski,}
On the formation of trapped surfaces,
Acta Math. 208 (2012), 211--333. 

\bibitem{KRS} {\sc S. Klainerman, I. Rodnianski, and J. Szeftel,}
The bounded L2 curvature conjecture,
Invent. Math. 202 (2015), 91--216.

\bibitem{LeFloch-lectures}  {\sc P.G. LeFloch,} 
An introduction to self-gravitating matter, Graduate course given at
Institute Henri Poincar\'e, Paris, Fall 2015,  available at {\sl http://www.youtube.com/user/PoincareInstitute.}

\bibitem{PLF-YM-book} {\sc P.G. LeFloch and Y. Ma,}
{\sl The hyperboloidal foliation method,}  World Scientific, 2014.

\bibitem{PLF-YM-CRAS}{\sc P.G. LeFloch and Y. Ma}, 
The global nonlinear stability of Minkowski spacetime for the Einstein equations in presence of massive fields, 
Note C.R. Acad. Sc. Paris 354 (2016), 948--953.

\bibitem{PLF-YM-one}{\sc P.G. LeFloch and Y. Ma}, 
The global nonlinear stability of Minkowski space for self-gravitating massive fields. The wave-Klein-Gordon model,  
Comm. Math. Phys. 346 (2016), 603--665.  

\bibitem{PLF-YM-three}{\sc P.G. LeFloch and Y. Ma},
The global nonlinear stability of Minkowski space for the $f(R)$-theory of modified gravity, 
in preparation.  

\bibitem{LFR} {\sc P.G. LeFloch and A.D. Rendall,}
A global foliation of Einstein-Euler spacetimes with Gowdy-symmetry on $T^3$, 
Arch. Rational Mech. Anal. 201 (2011), 841--870.  

\bibitem{LNS} {\sc H. Lindblad, M. Nakamura, and C.D. Sogge,}
Remarks on global solutions for nonlinear wave equations under the standard null conditions,
 J. Differential Equations 254 (2013), 1396--1436.

\bibitem{LR1} {\sc H. Lindblad and I. Rodnianski,}
Global existence for the Einstein vacuum equations in wave coordinates,
Comm. Math. Phys. 256 (2005), 43--110.

\bibitem{LR2} {\sc H. Lindblad and I. Rodnianski,}
The global stability of Minkowski spacetime in harmonic gauge,
Ann. of Math. 171 (2010), 1401--1477.

\bibitem{Ma}{\sc Y. Ma},
A quasi-linear wave-Klein-Gordon system in $2 + 1$ dimensions, in preparation.

\bibitem{MoncriefRinne} {\sc V. Moncrief and O. Rinne,}
Regularity of the Einstein equations at future null infinity,
Class. Quant. Grav. 26 (2009), 125010.

\bibitem{OCP} {\sc H. Okawa, V. Cardoso, and P. Pani,}
Collapse of self-interacting fields in asymptotically flat spacetimes: do self-interactions render Minkowski spacetime unstable?,
Phys. Rev. D 89 (2014), 041502.

\bibitem{Rinne} {\sc O. Rinne,}
An axisymmetric evolution code for the Einstein equations on hyperboloidal slices,
Class. Quantum Grav. 27 (2010), 035014. 

\bibitem{Shatah82} {\sc J. Shatah,}
Global existence of small solutions to nonlinear evolution equations,
J. Differential Equations 46 (1982), 409--425.

\bibitem{Shatah85} {\sc J. Shatah,}
Normal forms and quadratic nonlinear Klein--Gordon equations,
Comm. Pure Appl. Math. 38 (1985), 685--696.

\bibitem{Smulevici} {\sc J. Smulevici,} 
Small data solutions of the Vlasov-Poisson system and the vector field method, 
Preprint ArXiv:1504.02195. 

\bibitem{Speck} {\sc J. Speck},  
The global stability of the Minkowski spacetime solution to the Einstein-nonlinear system in wave coordinates, 
Analysis \& PDE 7 (2014), 771--901. 

\bibitem{Strauss} {\sc W.A. Strauss,}
{\sl Nonlinear wave equations,}
CBMS 73, Amer. Math. Soc., Providence, 1989.

\bibitem{VHH} {\sc A. Van\'o-Vinuales, S. Husa, and D. Hilditch,}
Spherical symmetry as a test case for unconstrained hyperboloidal evolution,
Class. Quantum Gravity 32 (2015), 175010. 

\bibitem{Zenginoglu} {\sc A. Zenginoglu,}
Hyperboloidal evolution with the Einstein equations,
Class. Quantum Grav. 25 (2008), 195025.

\bibitem{Zenginoglu2} {\sc A. Zenginoglu,}
Hyperboloidal layers for hyperbolic equations on unbounded domains,
J. Comput. Phys. 230 (2011), 2286--2302.

\end{thebibliography}
\end{document}